\documentclass[a4paper,notitlepage,fleqn,12pt]{article}

\usepackage[tbtags]{amsmath}	
\usepackage{amsfonts}
\usepackage{amsthm}
\usepackage{color}				
\usepackage{graphicx,amssymb}	
\usepackage{multirow}
\usepackage{rotating}
\usepackage{mathrsfs}
\usepackage{epstopdf}
\usepackage{enumerate}
\usepackage{enumitem}
\usepackage{pdflscape}
\usepackage{geometry}
\usepackage{natbib}
\usepackage{changepage}
\usepackage{dcolumn}
\usepackage{tabularx}
\usepackage{subfig}
\usepackage{afterpage}
\usepackage{xr}

\newcolumntype{d}[1]{D{.}{.}{#1}}
\newcolumntype{Y}{>{\raggedleft\arraybackslash}X}
\newcolumntype{Z}{>{\centering\arraybackslash}X}

\usepackage{subfig}
\newtheorem{lemma}{Lemma}[section]

\newtheorem{cor}{Corollary}[section]
\newtheorem{assum}{Assumption}
\newtheorem{thm}{Theorem}[section]
\newtheoremstyle{remark}
{2ex}
{2ex}
{}
{}
{\bfseries}
{.}
{.5em}
{}
\theoremstyle{remark}
\newtheorem{remark}{Remark}[section]

\DeclareMathOperator*{\argmin}{argmin}

\DeclareMathOperator*{\var}{var}
\newcommand{\bm}[1]{\mbox{\boldmath{$#1$}}}

\renewcommand{\arraystretch}{0.85}

\numberwithin{equation}{section}

\usepackage{tabularx}

\title{Quantile autoregressive conditional heteroscedasticity}
\author{Qianqian Zhu$^{a}$, Songhua Tan$^{a}$, Yao Zheng$^{b}$ and Guodong Li$^{c}$\\
	\textit{$^{a}$Shanghai University of Finance and Economics,}\\
	\textit{$^{b}$University of Connecticut and $^{c}$University of Hong Kong}}

\date{}

\begin{document}
	\newcolumntype{L}[1]{>{\raggedright\arraybackslash}p{#1}}
	\newcolumntype{C}[1]{>{\centering\arraybackslash}p{#1}}
	\newcolumntype{R}[1]{>{\raggedleft\arraybackslash}p{#1}}
	
	\maketitle
	\begin{abstract}
		This paper proposes a novel conditional heteroscedastic time series model by applying the framework of quantile regression processes to the ARCH$(\infty)$ form of the GARCH model. This model can provide varying structures for conditional quantiles of the time series across different quantile levels, while including the commonly used GARCH model as a special case. The strict stationarity of the model is discussed. For robustness against heavy-tailed distributions,  a self-weighted quantile regression (QR) estimator is proposed. 
		While  QR performs satisfactorily at intermediate quantile levels, its accuracy deteriorates at high quantile levels due to data scarcity. 
		As a remedy,  a self-weighted composite quantile regression (CQR) estimator is further introduced and,  based on an approximate GARCH model with a flexible Tukey-lambda distribution for the innovations, we can extrapolate the  high quantile levels by borrowing information from intermediate ones.
		Asymptotic properties for the proposed estimators are established.
		Simulation experiments are carried out to access the finite sample performance of the proposed methods, and an empirical example is presented to illustrate the usefulness of the new model.
	\end{abstract}
	
	\newpage
	\section{Introduction}
	
	Since the appearance of autoregressive conditional heteroscedastic (ARCH) \citep{Engle1982} and generalized  ARCH (GARCH) models \citep{Bollerslev1986}, GARCH-type models have become popular and powerful tools to capture the volatility of financial time series; see \cite{Francq_Zakoian2010} for an overview. 
	Volatility modeling plays an important role in financial risk management. In particular, it is a key ingredient for the calculation of quantile-based risk measures such as the value-at-risk (VaR) and expected shortfall. As estimating these measures is essentially a quantile estimation problem  \citep{Artzner_Delbaen_Eber_Heath1999,Wu_Xiao2002,Francq_Zakoian2015}, considerable research has been devoted to the development of quantile regression (QR) methods for GARCH-type models, such as \citeauthor{Taylor2008}'s \citeyearpar{Taylor2008} linear ARCH \citep{Koenker_Zhao1996} and linear GARCH models \citep{Xiao_Koenker2009}, \citeauthor{Bollerslev1986}'s \citeyearpar{Bollerslev1986} GARCH model \citep{Lee_Noh2013, Zheng_Zhu_Li_Xiao2018}, and asymmetric power GARCH model \citep{Wang_Zhu_Li_Li2022}.
	
	A common feature of the above research is that  the global structure of the volatility process is captured by a parametric GARCH-type model with distribution-free innovations. This implies that the conditional quantile process will be the product of the volatility process and the quantile of the innovation. 
	Consider the following linear GARCH($1,1$) model \citep{Taylor2008}:
	\begin{equation}\label{linearGARCH}
		y_{t}=\varepsilon_{t}h_{t},\hspace{5mm} h_{t}=a_0+a_1|y_{t-1}|+b_1h_{t-1},
	\end{equation}
	where $\{y_t\}$ is the observed series, and $\{\varepsilon_{t}\}$ are independent and identically distributed ($i.i.d.$) innovations with mean zero. 
	The $\tau$th conditional quantile function of $y_t$ is
	\begin{equation*}
		Q_{\tau}(y_t|y_{t-1}, y_{t-2}, \dots)=\left(a_0+a_1|y_{t-1}|+b_1h_{t-1}\right)Q_{\tau}(\varepsilon_{t})=\bm\theta_{\tau}^{\prime}\bm z_t,
	\end{equation*}
	where $Q_{\tau}(\varepsilon_{t})$ is the $\tau$th quantile of $\varepsilon_t$, $\bm\theta_{\tau}=(a_{0},a_{1},b_{1})^{\prime}Q_{\tau}(\varepsilon_{t})$, and $\bm z_t=(1,|y_{t-1}|,h_{t-1})^{\prime}$. Thus, $Q_{\tau}(y_t|y_{t-1}, y_{t-2}, \dots)$ can be estimated by replacing $\bm\theta_{\tau}$ and the volatility $h_{t}$ with their estimates; see \cite{Xiao_Koenker2009} and \cite{Zheng_Zhu_Li_Xiao2018}. 
	Note that $Q_{\tau}(y_t|y_{t-1}, y_{t-2}, \dots)$ is dependent on $\tau$ only through $Q_{\tau}(\varepsilon_{t})$, whereas the GARCH parameters remain invariant across different $\tau$.
	However,  in practice  the GARCH parameters may vary across quantile levels. The above framework would fail to capture this phenomenon, potentially resulting in poor forecast accuracy; see Section \ref{Sec-realdata} for empirical evidence. 
	To address this limitation, a natural idea is to allow the GARCH parameters to be $\tau$-dependent.
	
	Recently random-coefficient time series models  built upon  quantile regression have attracted growing attention. 
	By assuming that the AR coefficients are functions of a standard uniform random variable, the quantile AR model in \cite{Koenker_Xiao2006} allows for asymmetric dynamic structures across quantile levels; 
	see, e.g., \cite{Ferreira2011} and \cite{Baur_Dimpfl_Jung2012} for various empirical applications of this model. There have been many extensions of the quantile AR model, such as the quantile self-exciting threshold AR model \citep{Cai_Stander2008}, the threshold quantile AR model \citep{Galvao2011}, and the quantile double AR model  \citep{Zhu_Li2022}.  
	However, as far as we know, the approach of \cite{Koenker_Xiao2006} has not been explored for GARCH-type models.  
	To fill this gap, this paper proposes the quantile GARCH model, where the GARCH parameters are allowed to vary across  quantile levels.

	Our main contributions are threefold. 
	First, we develop a more flexible QR framework for conditional heteroscedastic time series, namely the quantile GARCH model, and establish a sufficient condition for its strict stationarity. As the volatility process of the GARCH model is latent and defined recursively, a direct extension of \cite{Koenker_Xiao2006} would be infeasible. Instead, by exploiting the ARCH($\infty$) form \citep{Zaffaroni_2004} of the GARCH model, we introduce a random-coefficient GARCH process, where the GARCH parameters are functions of a standard uniform random variable. It can be written as a weighted sum of past information across all lags, where the weights are exponentially decaying random-coefficient functions. 
	The proposed model can capture asymmetric dynamic structures and varying persistence across different quantile levels, while including the linear GARCH model as a special case.

	Secondly, for the proposed quantile GARCH model, we introduce the self-weighted QR estimator. The uniform convergence theory of the estimator, including uniform consistency and weak convergence, is established for the quantile process with respect to the quantile level $\tau$. Note that the weak convergence of the unweighted QR estimator would require $E(|y_t|^3)< \infty$. By contrast, the self-weighted estimator only requires  $E(|y_t|^s)< \infty$ for an arbitrarily small $s>0$ and thus is applicable to very heavy-tailed financial data. 
	The major theoretical difficulty comes from the non-convex and non-differentiable  objective function of self-weighted QR estimator. To overcome it, we adopt the bracketing method in \cite{Pollard1985} to derive the pointwise Bahadur representation of the self-weighted QR estimator for each fixed $\tau$, hence the pointwise $\sqrt{n}$-consistency and asymptotic normality. Then, we strengthen  the pointwise convergence to uniform convergence for all $\tau$, by deriving the  Bahadur representation uniformly in $\tau$ and proving the asymptotic tightness of its leading term. In addition,  to check whether the persistence coefficient is $\tau$-independent, we construct a Cram\'{e}r-von Misses (CvM) test. Based on the weak convergence result, we obtain the limiting null distribution of the CvM test statistic and propose a feasible subsampling method to calculate its critical values.
	
	Finally, to remedy the possible inefficiency of the QR at high quantile levels due to data scarcity, we further introduce the self-weighted composite quantile regression (CQR) estimator. High quantile levels are of great interest in financial risk management. 
	A common approach to extremal QR \citep{Chernozhukov2005} is to  estimate the   quantiles  at multiple intermediate levels and then  extrapolate those at high levels \citep{Wang_Li_He2012, Li_Wang2019}. 
	We adopt such an   approach for the quantile GARCH model. Since this model is similar to  \citeauthor{Taylor2008}'s \citeyearpar{Taylor2008}  GARCH model, we can conveniently make use of the latter for the extrapolation under a chosen innovation distribution such that an explicit quantile function is available. We choose the  Tukey-lambda distribution \citep{Joiner_Rosenblatt1971}, since it not only has an explicit quantile function, but is flexible in fitting heavy tails and approximating many common distributions such as the Gaussian distribution \citep{Gilchrist2000}. 
	For the proposed weighted CQR estimator, we  derive asymptotic properties  under possible model misspecification and provide practical suggestions for computational issues.
	In addition, our simulation studies and empirical analysis indicate that the CQR outperforms the QR at high quantile levels. 
	
	The rest of this paper is organized as follows. Section \ref{Sec-QGARCH11} introduces the quantile GARCH$($1,1$)$ model and studies its strict stationarity. 
	Section \ref{Sec-estimation} proposes the self-weighted QR estimator, together with the convergence theory for the corresponding quantile process  and a CvM test for checking the constancy of the persistence coefficient across all quantile levels. Section \ref{Sec-CQR} introduces the CQR estimator and derives its asymptotic properties.
	Simulation studies and an empirical example are provided in 
	Sections \ref{Sec-simulation} and \ref{Sec-realdata}, respectively. 
	Conclusion and discussion are given in Section \ref{Sec-conclusion}. 
	A section on the generalization to the quantile GARCH($p,q$) model, all technical proofs, and additional numerical results are given in  the Appendix. 
	Throughout the paper, $\rightarrow_d$ denotes the convergence in distribution, $\rightsquigarrow$ denotes weak convergence, and $o_p(1)$ denotes the convergence in probability. Moreover, $\|\cdot\|$ denotes the norm of a matrix or column vector, defined as $\|A\|=\sqrt{\text{tr}(AA^\prime)}=\sqrt{\sum_{i,j}a_{ij}^2}$. 
	In addition, $\ell^{\infty}(\mathcal{T})$ denotes the space of all uniformly bounded functions on $\mathcal{T}$.
	The dataset in Section \ref{Sec-realdata} and computer programs for the analysis are available at https://github.com/Tansonghua-sufe/QGARCH.
	
	\section{Proposed quantile GARCH$($1,1$)$ model}\label{Sec-QGARCH11}
	
	\subsection{Motivation}
	
	For succinctness, we restrict our attention to the quantile GARCH$($1,1$)$ model in the main paper, while the generalization to the quantile GARCH$(p,q)$ model is detailed in the Appendix.
	
	To motivate the proposed model, first consider a strictly stationary GARCH($1,1$) process in the form of
	\begin{equation}\label{garch11}
		x_t=\eta_t h_t^{1/2}, \hspace{5mm} h_t=a_0+a_1x_{t-1}^2+b_1 h_{t-1},
	\end{equation}
	where $a_0>0$, $a_1\geq 0$, $b_1\geq 0$, and the innovations $\{\eta_t\}$ are $i.i.d.$ random variables with mean zero and variance one. The ARCH($\infty$) representation \citep{Zaffaroni_2004} of model \eqref{garch11} can be written as 
	\begin{equation}\label{garch11b}
		x_t = \eta_t\left(\frac{a_0}{1-b_1}+a_1\sum_{j=1}^{\infty} b_1^{j-1}x_{t-j}^2\right)^{1/2}.
	\end{equation}
	Then, the $\tau$th conditional quantile function of $x_t$ in model \eqref{garch11b} is given by
	\begin{equation}\label{garch11q}
		Q_\tau(x_t|x_{t-1}, x_{t-2}, \dots)=Q_\tau(\eta_t)\left(\frac{a_0}{1-b_1}+a_1\sum_{j=1}^{\infty} b_1^{j-1}x_{t-j}^2\right)^{1/2}, \quad \tau\in(0, 1),
	\end{equation}
	where $Q_\tau(\eta_t)$ denotes the $\tau$th quantile of $\eta_t$. The parameters $a_0, a_1$ and $b_1$, which are independent of the specified quantile level $\tau$, control the scale of the conditional distribution of $x_t$, while the distribution of $\eta_t$ determines its shape.
	As a result, if the GARCH coefficients are allowed to vary with $\tau$ and thus capable of altering both the scale and shape of the conditional distribution, we will have a more flexible model that can accommodate asymmetric dynamic structures across different quantile levels.
	
	However, note that \eqref{garch11q} is nonlinear in the coefficients of the $x_{t-j}^2$'s.
	Consequently, a direct extension from \eqref{garch11} to a varying-coefficient model is undesirable, since it will result in a nonlinear conditional quantile function whose estimation  is computationally challenging. 
	Alternatively, we will consider the linear GARCH($1,1$) model in \eqref{linearGARCH}, in which case \eqref{garch11b} is revised to
	\begin{equation}\label{lgarch11}
		y_t = \varepsilon_t\left(\frac{a_0}{1-b_1}+a_1\sum_{j=1}^{\infty} b_1^{j-1}|y_{t-j}|\right).
	\end{equation}
	Then, its corresponding conditional quantile function has the following linear form:
	\begin{equation}\label{qlgarch11}
		Q_\tau(y_t|y_{t-1}, y_{t-2}, \dots)=Q_\tau(\varepsilon_t)\left(\frac{a_0}{1-b_1}+a_1\sum_{j=1}^{\infty} b_1^{j-1}|y_{t-j}|\right), \quad \tau\in(0, 1).
	\end{equation}
	We will adopt \eqref{qlgarch11} to formulate the proposed quantile GARCH model.
	\begin{remark}

		As shown in \cite{Zheng_Zhu_Li_Xiao2018}, the traditional GARCH($1,1$) model in \eqref{garch11} has an equivalent form of the linear GARCH($1,1$) model in \eqref{linearGARCH} up to a one-to-one transformation $T(\cdot)$. Specifically, for any $x_t$ following model \eqref{garch11b}, if we take the transformation $y_t=T(x_t)=x_t^2 \text{sgn}(x_t)$, then it can be shown that $y_t$ satisfies  \eqref{lgarch11} with $\varepsilon_t=T(\eta_t)=\eta_t^2 \text{sgn}(\eta_t)$.
		Note that $E(\varepsilon_t)$ may not be zero although $E(\eta_t)=0$, and this will not affect our derivation since the conditional quantile function at \eqref{qlgarch11} depends on $Q_\tau(\varepsilon_t)$ rather than $E(\varepsilon_t)$.
	\end{remark}
	
	\subsection{The proposed model}\label{sec:model}
	
	Let $\mathcal{F}_{t}$ be the $\sigma$-field generated by $\{y_{t}, y_{t-1}, \dots\}$. 
	To allow the GARCH parameters to vary with $\tau$, we extend model \eqref{qlgarch11} to the following conditional quantile model:
	\begin{equation}\label{QGARCH11-quantile}
		Q_\tau(y_t|\mathcal{F}_{t-1})=\omega(\tau)+\alpha_1(\tau)\sum_{j=1}^{\infty}[\beta_1(\tau)]^{j-1}|y_{t-j}|,  \quad \tau\in(0, 1),
	\end{equation}
	where $\omega: (0, 1)\rightarrow\mathbb{R}$ and $\alpha_1: (0, 1)\rightarrow\mathbb{R}$ are unknown monotonic increasing functions, and $\beta_1: (0, 1)\rightarrow [0, 1)$ is a non-negative real-valued function. 
	Note that both the scale and shape of the conditional distribution of $y_t$ can be altered by the past information $|y_{t-j}|$. 
	Assuming that the right hand side of \eqref{QGARCH11-quantile} is monotonic increasing in $\tau$, then \eqref{QGARCH11-quantile} is equivalent to the following random-coefficient process:
	\begin{equation}\label{qgarch11}
		y_t=\omega(U_t)+\alpha_1(U_t)\sum_{j=1}^{\infty}[\beta_1(U_t)]^{j-1}|y_{t-j}|,
	\end{equation}
	where $\{U_t\}$ is a sequence of $i.i.d.$ standard uniform random variables; see a discussion on the  monotonicity of $Q_\tau(y_t|\mathcal{F}_{t-1})$ with respect to $\tau$ in Remark \ref{remark-monotonicity}.
	We call model \eqref{QGARCH11-quantile} or \eqref{qgarch11} the quantile GARCH($1,1$) model.

	Similar to the GARCH model which requires the innovations to have mean zero, the quantile GARCH model also needs a location constraint. For the conditional quantile function \eqref{QGARCH11-quantile}, we may impose that
	\begin{equation}\label{location11}
		Q_{0.5}(y_t|\mathcal{F}_{t-1})=0.
	\end{equation}
	Since $\beta_1(\cdot)$ is non-negative, condition \eqref{location11} holds if and only if
	\begin{equation}\label{constraint11}
		\omega(0.5)=\alpha_1(0.5)=0.
	\end{equation}
	For the quantile GARCH($1,1$) model, we impose condition \eqref{constraint11} throughout this paper. 
	
	Recall that the functions $\omega(\cdot)$ and $\alpha_1(\cdot)$ are monotonic increasing and $\beta_1(\cdot)$ is non-negative. Under \eqref{constraint11} the quantile GARCH($1,1$) model \eqref{qgarch11} can be rewritten into
	\begin{align*}
		y_t &= \text{sgn}(U_t-0.5)|y_t|,\\
		|y_t|&= |\omega(U_t)|+\sum_{j=1}^{\infty}|\alpha_1(U_t)|[\beta_1(U_t)]^{j-1}|y_{t-j}|,
	\end{align*}
	where $y_t$, $U_t-0.5$, $\omega(U_t)$ and $\alpha_1(U_t)$ have the same sign at each time $t$. 
	For simplicity, denote $\phi_{0, t}=|\omega(U_t)|$ and $\phi_{j,t}= |\alpha_1(U_t)|[\beta_1(U_t)]^{j-1}$ for $j\geq1$. 
	Then the quantile GARCH($1,1$) model \eqref{qgarch11} is equivalent to 
	\begin{align}\label{equ-qgarch11}
		y_t = \text{sgn}(U_t-0.5)|y_t|, \quad |y_t|= \phi_{0, t} + \sum_{j=1}^{\infty}\phi_{j,t}|y_{t-j}|, \quad j\geq1.
	\end{align}
	This enables us to establish a sufficient condition for the existence of a strictly stationary solution of the quantile GARCH($1,1$) model in the following theorem. 
	\begin{thm}\label{thm-stationarity11}
		Suppose that condition \eqref{constraint11} holds. If there exists $s\in(0,1]$ such that 
		\begin{equation}\label{StationaryCondition11_fractional}
			E(\phi_{0, t}^s)<\infty \quad\text{and}\quad \sum_{j=1}^{\infty}E(\phi_{j,t}^s)<1,
		\end{equation}
		or $s>1$ such that 
		\begin{equation}\label{StationaryCondition11_integer}
			E(\phi_{0, t}^s)<\infty \quad\text{and}\quad \sum_{j=1}^{\infty}[E(\phi_{j,t}^s)]^{1/s}<1,
		\end{equation}
		then there exists a strictly stationary solution of the quantile GARCH($1,1$) equations in \eqref{equ-qgarch11}, and the process $\{y_t\}$ defined by
		\begin{equation}\label{solution}
			y_t=\textup{sgn}(U_t-0.5)\left(\phi_{0,t}+\sum_{\ell=1}^{\infty}\sum_{j_1, \dots, j_\ell=1}^{\infty}	\phi_{0,  t-j_1-\cdots-j_\ell}\phi_{j_1, t}\phi_{j_2, t-j_1}\cdots\phi_{j_\ell, t-j_1-\cdots-j_{\ell-1}}\right)
		\end{equation}
		is the unique strictly stationary and $\mathcal{F}_{t}^{U}$-measurable solution to \eqref{equ-qgarch11} such that $E|y_t|^s<\infty$, where $\mathcal{F}_{t}^{U}$ is the $\sigma$-field generated by $\{U_{t}, U_{t-1}, \dots\}$.
	\end{thm}
	
	Theorem \ref{thm-stationarity11} gives a sufficient condition for the existence of a unique strictly stationary solution satisfying $E|y_t|^s<\infty$. The proof relies on a method similar to that of Theorem 1 in \cite{Douc_Roueff_Soulier2008}; see also \cite{Giraitis_Kokoszka_Leipus_2000} and \cite{Royer2022}.

	\begin{remark}[Monotonicity conditions for quantile and coefficient functions]\label{remark-monotonicity}
		As discussed in \cite{Koenker_Xiao2006} and \cite{Phillips2015}, it is very difficult to derive a necessary and sufficient condition on random-coefficient functions to ensure the monotonicity of $Q_{\tau}(y_{t}|\mathcal{F}_{t-1})$ in $\tau$ for the quantile GARCH($1,1$) model in \eqref{QGARCH11-quantile}.
		Given that $\omega(\cdot)$ and $\alpha_1(\cdot)$ are monotonic increasing,  a sufficient condition for monotonicity of $Q_{\tau}(y_{t}|\mathcal{F}_{t-1})$  is that the non-negative function $\beta_1(\cdot)$ is monotonic decreasing on $(0, 0.5)$ and monotonic increasing on $(0.5, 1)$. However, since $Q_{\tau}(y_{t}|\mathcal{F}_{t-1})$ could be monotonic increasing even if $\beta_1(\cdot)$ does not satisfy the above constraint (e.g., if $\beta_1(\tau)$ is constant over $\tau$), we refrain from imposing any monotonicity constraint on $\beta_1(\cdot)$  in order to avoid overly restricting the function space. 
	\end{remark}

	\begin{remark}[Special cases of Theorem \ref{thm-stationarity11}] \label{2nd-orderGARCH11}
		When   $\omega(U_t)=a_0\varepsilon_{t}/(1-b_1), \alpha_1(U_t)=a_1\varepsilon_{t}$, and $\beta_1(U_t)=b_1$,  the quantile GARCH($1,1$) model in \eqref{qgarch11} reduces to  the linear GARCH($1,1$) model in \eqref{lgarch11}. Then, \eqref{StationaryCondition11_fractional} can be simply written as $a_1^sE|\varepsilon_t|^s+b_1^s<1$ for $s\in(0,1]$, while \eqref{StationaryCondition11_integer} reduces to $a_1(E|\varepsilon_t|^s)^{1/s}+b_1<1$ with $E|\varepsilon_t|^s<\infty$ for $s>1$.	In particular, when $s=1$, the stationarity condition becomes $a_1+b_1<1$, which is exactly the necessary and sufficient condition for the existence of a second-order stationary solution to the GARCH($1,1$) model in \eqref{garch11}. 	If $s=2$, then the condition becomes $a_1[E(\eta_t^4)]^{1/2}+b_1<1$ with $E(\eta_t^4)<\infty$, which is slightly stronger than the necessary and sufficient condition for the existence of a fourth-order stationary solution to the GARCH($1,1$) model in  \eqref{garch11}; see also \cite{Bollerslev1986} and \cite{Zaffaroni_2004}.
	\end{remark}

	\begin{remark}[Extension to asymmetric quantile GARCH models]\label{remark-asymmetric-extension}	
		There are numerous variants of the GARCH model, such as the exponential GARCH \citep{nelson1991conditional} and threshold GARCH \citep{zakoian1994threshold} models. The quantile GARCH model in this paper can be extended along the lines of these variants. For example, to capture leverage effects in quantile dynamics, as the quantile counterpart of the threshold GARCH model \citep{zakoian1994threshold}, the threshold quantile GARCH($1,1$)  model can be defined  as 
		\[Q_\tau(y_t|\mathcal{F}_{t-1})=\omega(\tau)+\alpha_1^{+}(\tau)\sum_{j=1}^{\infty}[\beta_1(\tau)]^{j-1}y_{t-j}^{+}-\alpha_1^{-}(\tau)\sum_{j=1}^{\infty}[\beta_1(\tau)]^{j-1}y_{t-j}^{-},\]
		where $\omega: (0,1)\rightarrow\mathbb{R}$ and $\alpha_1^{+}, \alpha_1^{-}: (0,1)\rightarrow\mathbb{R}$ are monotonic increasing, $\beta_1: (0,1)\rightarrow [0,1)$, $y_{t-j}^{-}=\min\{y_{t-j}, 0\}$, and $y_{t-j}^{+}=\max\{y_{t-j}, 0\}$. 
		We leave this interesting extension for future research.
	\end{remark}

	\section{Quantile regression}\label{Sec-estimation}
	
	\subsection{Self-weighted estimation}\label{subsec-WCQE}
	Let $\bm\theta=(\omega, \alpha_{1},\beta_{1})^{\prime}\in \Theta$ be the parameter vector of the quantile GARCH($1,1$) model, which belongs to the parameter space $\Theta \subset \mathbb{R}^{2}\times [0,1)$. From \eqref{QGARCH11-quantile}, we can define the conditional quantile function below,
	\[
	q_t(\bm\theta) =\omega + \alpha_{1}\sum_{j=1}^\infty \beta_{1}^{j-1}|y_{t-j}|.
	\]
	Since the function $q_t(\bm\theta)$ depends on  observations in the infinite past, initial values are required in practice. In this paper, we set $y_t=0$ for $t\leq 0$, and denote the resulting function by $\widetilde{q}_t(\bm\theta)$, that is, $\widetilde{q}_t(\bm\theta)=\omega + \alpha_{1}\sum_{j=1}^{t-1} \beta_{1}^{j-1}|y_{t-j}|$.  We will prove that the effect of the initial values on the estimation and inference is asymptotically negligible.

	Let  $\psi_{\tau}(x)=\tau-I(x<0)$, where the indicator function $I(\cdot)=1$ if the condition is true and 0 otherwise.
	For any $\tau\in \mathcal{T}\subset (0,1)$, 
	we propose the self-weighted quantile regression (QR) estimator as follows,
	\begin{align}\label{PairwiseWCQE}
		\widetilde{\bm\theta}_{wn}(\tau)&= (\widetilde{\omega}_{wn}(\tau), \widetilde{\alpha}_{1wn}(\tau), \widetilde{\beta}_{1wn}(\tau))^{\prime} =\argmin_{\bm\theta\in\Theta}\sum_{t=1}^n w_t\rho_{\tau} \left(y_t-\widetilde{q}_t(\bm\theta)\right),
	\end{align} 
	where $\{w_t\}$ are nonnegative random weights, and $\rho_{\tau}(x)=x\psi_{\tau}(x)=x[\tau-I(x<0)]$ is the check function; see also  \cite{Ling2005}, \cite{Zhu_Ling2011}, and \cite{Zhu_Zheng_Li2018}.

	When $w_t=1$ for all $t$, \eqref{PairwiseWCQE} reduces to the unweighted QR estimator.  In this case,  the consistency and asymptotic normality of the estimator would require $E|y_t|<\infty$ and $E|y_t|^3<\infty$, respectively. A sufficient condition for the existence of these moments is provided in Theorem \ref{thm-stationarity11}. 
	However, higher order moment conditions will make the stationarity region much narrower. 
	Moreover, financial time series are usually  heavy-tailed, so these moment conditions can be easily violated. By contrast, using the self-weighting approach \citep{Ling2005}, we  only need a finite fractional moment of $|y_t|$.

	Denote the true parameter vector by $\bm\theta(\tau)=(\omega(\tau), \alpha_1(\tau), \beta_1(\tau))^{\prime}$. 
	Let $F_{t-1}(\cdot)$ and $f_{t-1}(\cdot)$ be the distribution and density functions of $y_t$ conditional on $\mathcal{F}_{t-1}$, respectively. 
	To establish the asymptotic properties of $\widetilde{\bm\theta}_{wn}(\tau)$, we need the following assumptions. 
	
	\begin{assum}\label{assum-Process}
		$\{y_t\}$ is strictly stationary and ergodic. 
	\end{assum}
	
	\begin{assum}\label{assum-Space}
		(i) The parameter space $\Theta$ is compact; 
		(ii) $\bm\theta(\tau)$ is an interior point of $\Theta$.
	\end{assum}
	
	\begin{assum}\label{assum-ConditionalDensity}
		With probability one, $f_{t-1}(\cdot)$ and its derivative function $\dot{f}_{t-1}(\cdot)$ are uniformly bounded, and $f_{t-1}(\cdot)$ is positive on the support $\{x:0<F_{t-1}(x)<1\}$.	
	\end{assum}
	
	\begin{assum}\label{assum-RandomWeight}
		$\{w_{t}\}$ is strictly stationary and ergodic, and $w_{t}$ is nonnegative and measurable with respect to $\mathcal{F}_{t-1}$ such that $E(w_{t})<\infty$ and $E(w_{t}|y_{t-j}|^3)<\infty$ for $j\geq 1$.
	\end{assum}
	\begin{assum}\label{assum-Tightness}
		The functions $\omega(\cdot),\alpha_1(\cdot)$ and $\beta_1(\cdot)$ are Lipschitz continuous. 
	\end{assum}		
	
	Theorem \ref{thm-stationarity11} provides a sufficient condition for Assumption \ref{assum-Process}.  In Assumption \ref{assum-Space}, condition (i) is standard for the consistency of estimator, while condition (ii) is needed for the asymptotic normality; see also \cite{Francq_Zakoian2010} and \cite{Zheng_Zhu_Li_Xiao2018}.
	Assumption \ref{assum-ConditionalDensity} is commonly required for QR processes whose coefficients are functions of a uniform random variable; see Assumption A.3 in \cite{Koenker_Xiao2006} for quantile AR models and Assumption 4 in \cite{Zhu_Li2022} for quantile double AR models. 
	Specifically, the positiveness and continuity of $f_{t-1}(\cdot)$ are required to show the uniform consistency of $\widetilde{\bm\theta}_{wn}(\tau)$ in Theorem \ref{thm-WCQE-uniform-consistency}, while the boundedness of $f_{t-1}(\cdot)$ and $\dot{f}_{t-1}(\cdot)$ is needed for the weak convergence in Theorem \ref{thm-WCQE-weak-convergence}.
	In the special case where the quantile GARCH($1,1$) model in \eqref{qgarch11} reduces to model \eqref{lgarch11}, Assumption \ref{assum-ConditionalDensity} can be simplified to conditions similar to Assumption (A2) in \cite{Lee_Noh2013} and Assumption 4 in \cite{Zhu_Li_Xiao2021QRGARCHX}.  
	Assumption \ref{assum-RandomWeight} on the self-weights $\{w_t\}$ is used to reduce the moment requirement on $\{y_t\}$ in establishing asymptotic properties of $\widetilde{\bm\theta}_{wn}(\tau)$; see more discussions on  $\{w_t\}$ in Remark \ref{remark-self-weighting}.
	Assumption \ref{assum-Tightness} is required to establish the stochastic equicontinuity for weak convergence in Theorem \ref{thm-WCQE-weak-convergence}.

	Let $\bm T_n(\tau)=n^{-1/2}\sum_{t=1}^{n}w_t\dot{q}_t(\bm\theta(\tau))\psi_{\tau}(y_t-q_t(\bm\theta(\tau)))$ and $\Sigma_{w}(\tau_1, \tau_2)=(\min\{\tau_1, \tau_2\}-\tau_1\tau_2)\Omega_{1w}^{-1}(\tau_1)\Omega_{0w}(\tau_1, \tau_2)\Omega_{1w}^{-1}(\tau_2)$, where $\Omega_{0w}(\tau_1, \tau_2)=E\left[w_t^2\dot{q}_t(\bm\theta(\tau_1))\dot{q}_t^{\prime}(\bm\theta(\tau_2))\right]$ and $\Omega_{1w}(\tau)=E\left[f_{t-1}(F_{t-1}^{-1}(\tau))w_{t}\dot{q}_t(\bm\theta(\tau))\dot{q}_t^{\prime}(\bm\theta(\tau))\right]$. Theorems \ref{thm-WCQE-uniform-consistency} and \ref{thm-WCQE-weak-convergence} below establish the uniform consistency and weak convergence for the QR process $\widetilde{\bm\theta}_{wn}(\cdot)$, respectively. 
	
	\begin{thm}\label{thm-WCQE-uniform-consistency}
		For $\{y_t\}$ generated by model \eqref{qgarch11} with condition \eqref{constraint11}, suppose $E|y_t|^s<\infty$ for some $s\in (0,1)$. If Assumptions \ref{assum-Process}, \ref{assum-Space}(i), \ref{assum-ConditionalDensity} and \ref{assum-RandomWeight} hold, then $\sup_{\tau\in\mathcal{T}}\|\widetilde{\bm\theta}_{wn}(\tau)-\bm\theta(\tau)\| \rightarrow_p 0$ as $n\rightarrow\infty$.	 
	\end{thm}
	
	\begin{thm}\label{thm-WCQE-weak-convergence}
		For $\{y_t\}$ generated by model \eqref{qgarch11} with condition \eqref{constraint11}, suppose $E|y_t|^s<\infty$ for some $s\in (0,1)$ and the covariance kernel $\Sigma_w(\tau_1, \tau_2)$ is positive definite uniformly for $\tau_1=\tau_2=\tau\in\mathcal{T}$. 
		If Assumptions \ref{assum-Process}--\ref{assum-Tightness} hold, as $n\rightarrow\infty$, then we have 
		\begin{align}\label{Bahadur-representation}
			\sqrt{n}(\widetilde{\bm\theta}_{wn}(\cdot)-\bm\theta(\cdot)) = \Omega_{1w}^{-1}(\cdot)\bm T_n(\cdot) + o_p(1)
			\rightsquigarrow \mathbb{G}(\bm{\cdot}) \:\: \text{in} \:\: (\ell^{\infty}(\mathcal{T}))^3,
		\end{align}
		where the remainder term is uniform in $\tau\in\mathcal{T}$, and $\mathbb{G}(\bm{\cdot})$ is a zero mean Gaussian process with covariance kernel $\Sigma_w(\tau_1, \tau_2)$. 	
	\end{thm}	
	
	Owing to the self-weights, the above results hold for very heavy-tailed data with a finite fractional moment. 
	The proof of Theorem \ref{thm-WCQE-weak-convergence} is nontrivial. The first challenge comes from the non-convex and non-differentiable objective function of QR. Specifically, we need to prove the finite dimensional convergence of $\widetilde{\bm\theta}_{wn}(\tau)$, i.e., the $\sqrt{n}$-consistency of $\widetilde{\bm\theta}_{wn}(\tau)$ for each $\tau$ in the form of $\sqrt{n}(\widetilde{\bm\theta}_{wn}(\tau)-\bm\theta(\tau))=O_p(1)$. 
	We overcome this challenge by adopting the bracketing method in \cite{Pollard1985}. The second challenge is to obtain the Bahadur representation uniformly in $\tau\in\mathcal{T}$ and prove the asymptotic tightness of the leading term $\Omega_{1w}^{-1}(\cdot)\bm T_n(\cdot)$ in this representation. The key to accomplishing this is to verify the stochastic equicontinuity for all remainder terms and $\bm T_n(\cdot)$.

	In particular, when a fixed quantile level $\tau\in\mathcal{T}$ is considered, by the martingale central limit theorem (CLT), we can obtain the asymptotic normality of $\widetilde{\bm\theta}_{wn}(\tau)$ without the Lipschitz condition in Assumption \ref{assum-Tightness} as follows. 
	\begin{cor}\label{thm-WCQE}
		For $\{y_t\}$ generated by model \eqref{qgarch11} with condition \eqref{constraint11}, suppose $E|y_t|^s<\infty$ for some $s\in (0,1)$ and $\Sigma_{w}(\tau,\tau)$ is positive definite. 
		If Assumptions \ref{assum-Process}--\ref{assum-RandomWeight} hold, then $\sqrt{n}(\widetilde{\bm\theta}_{wn}(\tau)-\bm\theta(\tau))\rightarrow_d N\left(\bm 0,\Sigma_{w}(\tau,\tau)\right)$ as $n\rightarrow\infty$. 
	\end{cor}

	To estimate the asymptotic covariance $\Sigma_{w}(\tau,\tau)$ in Corollary \ref{thm-WCQE}, 
	we first estimate  $f_{t-1}(F_{t-1}^{-1}(\tau))$ in $\Omega_{1w}(\tau)$ using the difference quotient method \citep{Koenker2005}. Let $\widetilde{Q}_{\tau}(y_t|\mathcal{F}_{t-1})=\widetilde{q}_t(\widetilde{\bm\theta}_{wn}(\tau))$ be the fitted $\tau$th conditional quantile. We employ the estimator $\widetilde{f}_{t-1}(F_{t-1}^{-1}(\tau))=2\ell[\widetilde{Q}_{\tau+\ell}(y_t|\mathcal{F}_{t-1})-\widetilde{Q}_{\tau-\ell}(y_t|\mathcal{F}_{t-1})]^{-1}$, where $\ell$ is the bandwidth. As in \cite{Koenker_Xiao2006}, we consider two commonly used bandwidths for $\ell$ as follows: 
	\begin{equation}\label{bandwidths}
		\ell_{B}=n^{-1/5}\left\{\dfrac{4.5f_N^4(F_N^{-1}(\tau))}{[2F_N^{-2}(\tau)+1]^2}\right\}^{1/5}
		\quad\text{and}\quad
		\ell_{HS}=n^{-1/3}z_{\alpha}^{2/3}\left\{\dfrac{1.5f_N^2(F_N^{-1}(\tau))}{2F_N^{-2}(\tau)+1}\right\}^{1/3},
	\end{equation}
	where $f_N(\cdot)$ and $F_N(\cdot)$ are the standard normal density and distribution functions, respectively, and $z_{\alpha}=F_N^{-1}(1-\alpha/2)$ with $\alpha=0.05$. 
	Then the matrices $\Omega_{0w}(\tau,\tau)$ and $\Omega_{1w}(\tau)$ can be approximated by the sample averages:
	\begin{align*}
		\widetilde{\Omega}_{0w}(\tau,\tau)&=\dfrac{1}{n}\sum_{t=1}^{n}w_t^2\dot{\widetilde{q}}_t(\widetilde{\bm\theta}_{wn}(\tau))\dot{\widetilde{q}}_t^{\prime}(\widetilde{\bm\theta}_{wn}(\tau)) \quad\text{and} \\
		\widetilde{\Omega}_{1w}(\tau)&=\dfrac{1}{n}\sum_{t=1}^{n}\widetilde{f}_{t-1}(F_{t-1}^{-1}(\tau))w_{t}\dot{\widetilde{q}}_t(\widetilde{\bm\theta}_{wn}(\tau))\dot{\widetilde{q}}_t^{\prime}(\widetilde{\bm\theta}_{wn}(\tau)),
	\end{align*}
	where $\dot{\widetilde{q}}_t(\bm\theta)=(1,\sum_{j=1}^{t-1}\beta^{j-1}_{1}|y_{t-j}|,\alpha_{1}\sum_{j=2}^{t-1}(j-1)\beta^{j-2}_{1}|y_{t-j}|)^{\prime}$. Consequently, a consistent estimator of $\Sigma_{w}(\tau,\tau)$ can be constructed as $\widetilde{\Sigma}_{w}(\tau,\tau)=\tau(1-\tau)\widetilde{\Omega}_{1w}^{-1}(\tau)\widetilde{\Omega}_{0w}(\tau,\tau)\widetilde{\Omega}_{1w}^{-1}(\tau)$.
	
	\begin{remark}[Choices of self-weights]\label{remark-self-weighting}
		The goal of the self-weights $\{w_t\}$ is to relax the moment condition  from $E|y_t|^3<\infty$ to $E|y_t|^s<\infty$ for $s\in(0,1)$. If there is empirical evidence that $E|y_t|^3<\infty$ holds, then we can simply let $w_t=1$ for all $t$. Otherwise, the self-weights are needed. 
		There are many choices of random weights $\{w_t\}$ that satisfy Assumption \ref{assum-RandomWeight}. Note that the main role of $\{w_t\}$ in our technical proofs is to bound the term  $w_ty_{t-j}^{\delta}$ for $\delta\geq 1$ by $O(|y_{t-j}|^{s})$ for some $s\in (0,1)$. Following \cite{He_Hou_Peng_Shen2020}, we may consider 
		\begin{equation}\label{choice1_of_selfweights}
			w_t=\left(\sum_{i=0}^{\infty} e^{-\log^2(i+1)}\left\{I\left[|y_{t-i-1}|\leq c\right]+c^{-1}|y_{t-i-1}|I\left[|y_{t-i-1}|>c\right]\right\}\right)^{-3}
		\end{equation}
		for some given $c>0$, where $y_s$ is set to zero for $s\leq 0$. In our simulation and empirical studies, we take $c$ to be the 95\% sample quantile of $\{y_t\}_{t=1}^n$. 
	\end{remark}

	\begin{remark}[The quantile crossing problem]\label{remark-crossing}
		If we are only interested in estimating $Q_{\tau}(y_{t}|\mathcal{F}_{t-1})$ at a specific quantile level $\tau$, the L-BFGS-B algorithm \citep{zhu_Byrd_Lu_Nocedal1997} can be used to solve \eqref{PairwiseWCQE} with the constraint $\beta_1\in(0, 1)$. Then the estimate $\widetilde{Q}_{\tau}(y_{t}|\mathcal{F}_{t-1})=\widetilde{q}_t(\widetilde{\bm\theta}_{wn}(\tau))$ can be obtained for $Q_{\tau}(y_{t}|\mathcal{F}_{t-1})$. 
		As a more flexible approach, we may study multiple quantile levels simultaneously, say $\tau_1<\tau_2<\cdots<\tau_K$. However, the pointwise estimates $\{\widetilde{Q}_{\tau_k}(y_{t}|\mathcal{F}_{t-1})\}_{k=1}^K$ in practice may not be a monotonic increasing sequence even if $Q_{\tau}(y_{t}|\mathcal{F}_{t-1})$ is monotonic increasing in $\tau$. To overcome the  quantile crossing problem, we adopt the easy-to-implement rearrangement method \citep{Chernozhukov2010} to enforce the monotonicity of pointwise quantile estimates $\{\widetilde{Q}_{\tau_k}(y_{t}|\mathcal{F}_{t-1})\}_{k=1}^K$. 
		By Proposition 4 in \cite{Chernozhukov2010}, it can be shown that the rearranged quantile curve has smaller estimation error than the original one whenever the latter is not monotone; see also the simulation experiment in Section \ref{supp::Simulation rearrangement} of the Appendix.
	\end{remark}
	
	\begin{remark}[Rearranging coefficient functions]
		The proposed model in \eqref{QGARCH11-quantile} assumes that  $\omega(\cdot)$ and $\alpha_1(\cdot)$ are monotonic increasing. In practice, we can apply the method in  \cite{Chernozhukov2009} to  rearrange the estimates $\{\widetilde{\omega}_{wn}(\tau_k)\}_{k=1}^K$ and $\{\widetilde{\alpha}_{1wn}(\tau_k)\}_{k=1}^K$ to ensure the monotonicity of the curves across $\tau_k$'s. It is shown in \cite{Chernozhukov2009} that the rearranged confidence intervals are monotonic and narrower than the original ones.
	\end{remark}

	\subsection{Testing for constant persistence coefficient}\label{subsec-CvM}	   
	
	In this subsection, we present a test to determine if the persistence coefficient $\beta_1(\tau)$ is independent of the quantile level $\tau$ for $\tau\in \mathcal{T}\subset (0,1)$. This problem can be cast as a more general hypothesis testing problem as follows: 
	\begin{equation}\label{Hypotheses_constant}
		H_0:\ \forall\tau\in\mathcal{T}, \; R\bm\theta(\tau)=r  \quad \text{versus} \quad H_1:\ \exists\tau\in\mathcal{T}, \;R\bm\theta(\tau)\neq r,
	\end{equation}
	where $R$ is a predetermined row vector, and  $r \in \Gamma$ denotes a parameter whose specific value is unknown, but it is known to be independent of  $\tau$. Here the parameter space $\Gamma$ contains all values $R\bm\theta(\tau)$ can take under the proposed model.  
	Then, we can write the hypotheses for testing  the constancy of $\beta_1(\tau)$ in the form of \eqref{Hypotheses_constant} by setting $R=(0,0,1)$ and $r=\beta_1\in \Gamma = (0,1)$. In this case,  the null hypothesis in \eqref{Hypotheses_constant} means that $\beta_1(\tau)$ does not vary cross quantiles.

	For generality, we present the result for the general problem in  \eqref{Hypotheses_constant}. Under $H_0$, we can estimate the unknown $r$ using $\widetilde{r}=\int_{\mathcal{T}}R\widetilde{\bm\theta}_{wn}(\tau)d\tau$.  
	Define the inference process 
	$v_{n}(\tau)=R\widetilde{\bm\theta}_{wn}(\tau)-\widetilde{r}=R[\widetilde{\bm\theta}_{wn}(\tau)-\int_{\mathcal{T}}\widetilde{\bm\theta}_{wn}(\tau)d\tau]$.
	To test $H_0$, we construct the Cram\'{e}r-von Misses (CvM) test statistic as follows:
	\begin{align}\label{CvM_constant}
		S_{n}=n\int_{\mathcal{T}}v_{n}^2(\tau)d\tau.
	\end{align}
	
	Let $\sigma(\tau_1,\tau_2)=R[\Sigma_{w}(\tau_1,\tau_2)+\int_{\mathcal{T}}\int_{\mathcal{T}}\Sigma_{w}(\tau,\tau^{\prime})d\tau d\tau^{\prime}-\int_{\mathcal{T}}\Sigma_{w}(\tau_1,\tau)d\tau-\int_{\mathcal{T}}\Sigma_{w}(\tau,\tau_2)d\tau]R^{\prime}$. 
	Denote $v_{0}(\tau)=R[\mathbb{G}(\tau)-\int_{\mathcal{T}}\mathbb{G}(\tau)d\tau]$ with $\mathbb{G}(\tau)$ defined in Theorem \ref{thm-WCQE-weak-convergence}. 
	
	\begin{cor}\label{thm-test1}
		Suppose the conditions of Theorem \ref{thm-WCQE-weak-convergence} hold. 
		Under $H_0$, then we have $S_{n}\to_d S\equiv\int_{\mathcal{T}}v_{0}^2(\tau)d\tau$ as $n\to\infty$.
		If the covariance function of $v_{0}(\bm{\cdot})$ is nondegenerate, that is, $\sigma(\tau,\tau)>0$ uniformly in $\tau\in\mathcal{T}$, then $\text{Pr}(S_{n}>c_{\alpha})\to\text{Pr}(S>c_{\alpha})=\alpha$, where the critical value $c_{\alpha}$ is chosen such that $\text{Pr}(S>c_{\alpha})=\alpha$.
	\end{cor}
	Corollary \ref{thm-test1} indicates that we can reject $H_0$ if $S_{n}>c_{\alpha}$ at the significance level $\alpha$. In practice, we can use a grid of values $\mathcal{T}_n$ in place of $\mathcal{T}$. Similar to Corollary 3 in \cite{Chernozhukov2006}, we can verify that Corollary \ref{thm-test1} still holds for the discretization if the largest cell size of $\mathcal{T}_n$, denoted as $\delta_n$, satisfies $\delta_n\to 0$ as $n\to\infty$.

	Note that the CvM test in \eqref{CvM_constant} is not asymptotically distribution-free due to the estimation of $r$, which is commonly known as the Durbin problem \citep{Durbin1973}. This complicates the approximation of the limiting null distribution of $S_{n}$ and the resulting critical value $c_{\alpha}$. We suggest approximating the limiting null distribution by subsampling the linear approximation of the inference process $v_{n}(\tau)$; see also \cite{Chernozhukov2006}. This approach is computationally efficient as it avoids the repeated estimation  over the resampling steps for many values of $\tau$. 	
	Specifically, by Theorem \ref{thm-WCQE-weak-convergence}, under $H_0$ we have
	\begin{align}
		\sqrt{n}v_{n}(\tau)= \dfrac{1}{\sqrt{n}}\sum_{t=1}^n z_{t}(\tau) + o_p(1),
	\end{align}
	where $z_{t}(\tau)=R[m_t(\tau)-\int_{\mathcal{T}}m_t(\tau)d\tau]$, with $m_t(\tau)=w_t\Omega_{1w}^{-1}(\tau)\dot{q}_t(\bm\theta(\tau))\psi_{\tau}(y_t-q_t(\bm\theta(\tau)))$.	
	By the consistency of $\widetilde{\bm\theta}_{wn}(\tau)$ in  Theorem \ref{thm-WCQE-uniform-consistency}, we can estimate $z_{t}(\tau)$ using $\widetilde{z}_{t}(\tau)=R[\widetilde{m}_t(\tau)-\int_{\mathcal{T}}\widetilde{m}_t(\tau)d\tau]$, where $\widetilde{m}_t(\tau)=w_t\widetilde{\Omega}_{1w}^{-1}(\tau)\dot{\widetilde{q}}_t(\widetilde{\bm\theta}_{wn}(\tau))\psi_{\tau}(y_t-\widetilde{q}_t(\widetilde{\bm\theta}_{wn}(\tau)))$. Thus, a sample of estimated scores $\{\widetilde{z}_{t}(\tau), \tau\in\mathcal{T}, 1\leq t\leq n\}$ is obtained, where $n$ is the sample size. Then a subsampling procedure is conducted as follows. Given a block size $b_n$, we consider $L_n=n-b_n+1$ overlapping blocks  of the sample, indexed by $B_k=\{k, k+1,\ldots,k+b_n-1\}$ for $k=1,\ldots, L_n$. For each block $B_k$, we compute the inference process $v_{k,b_n}(\tau)=b_n^{-1}\sum_{t\in B_k}\widetilde{z}_{t}(\tau)$ and define $S_{k,b_n}=b_n\int_{\mathcal{T}}v_{k,b_n}^2(\tau)d\tau$. Then the critical value $c_{\alpha}$ can be calculated as the $(1-\alpha)$th empirical quantile of $\{S_{k,b_n}\}_{k=1}^{L_n}$.
	
	To establish the asymptotic validity of the subsampling procedure above, we can use a method similar to the proof of Theorem 5 in \cite{Chernozhukov2006}. This is possible under the conditions of Theorem \ref{thm-WCQE-weak-convergence} and an $\alpha$-mixing condition on ${y_t}$, provided that $L_n\to\infty$, $b_n\to\infty$, and $b_n/n\to 0$ as $n\to\infty$. However, we leave the rigorous proof for future research.
	Following \cite{Shao2011bootstrap}, we consider $b_n=\lfloor cn^{1/2}\rfloor$ with a positive constant $c$, where $\lfloor x \rfloor$ stands for the integer part of $x$. Our simulation study shows that the CvM test has reasonable size and power when  $c=0.5,1$ or 2.

	\section{Composite quantile regression}\label{Sec-CQR}
	
	\subsection{Self-weighted estimation}\label{subsec-CQR}
	
	It is well known that the QR can be unstable when $\tau$ is very close to zero or one due to data scarcity \citep{Li_Wang2019}.   
	However, estimating high conditional quantiles is of great interest in financial risk management. As a remedy, this section proposes the composite quantile regression (CQR).  To estimate the conditional quantile at a target level $\tau_0 \in \mathcal{T}\subset (0,0.01]\cup[0.99,1)$, the main idea is to conduct extrapolation based on estimation results of intermediate quantile levels at the one-sided neighbourhood of $\tau_0$.
	
	Suppose that $\{y_t\}$ follows the quantile GARCH($1,1$) model in \eqref{qgarch11}.
	Note that the conditional quantile function $Q_{\tau}(y_t|\mathcal{F}_{t-1})$  cannot be extrapolated directly due to the unknown nonparametric coefficient functions. To develop a feasible and easy-to-use extrapolation approach, we leverage the close connection between the linear GARCH($1,1$) process in \eqref{lgarch11} and quantile GARCH($1,1$) process in \eqref{qgarch11}. First, we approximate $y_t$ in \eqref{qgarch11} by the linear GARCH($1,1$) model in \eqref{lgarch11}. 
	Then, the $\tau$th conditional quantile of $y_t$ in \eqref{QGARCH11-quantile} can be approximated  by that of the linear GARCH($1,1$) model in \eqref{qlgarch11}:
	\begin{equation}\label{approxqlgarch11}
		Q_{\tau}(y_t|\mathcal{F}_{t-1})\approx Q_{\tau}(\varepsilon_t)\left(\frac{a_0}{1-b_1}+a_1\sum_{j=1}^{\infty} b_1^{j-1}|y_{t-j}|\right), 
	\end{equation}
	where $\varepsilon_t$'s are the $i.i.d.$ innovations of the linear GARCH($1,1$) model. 
	If the quantile function $Q_{\tau}(\varepsilon_t)$ has an explicit parametric form, then \eqref{approxqlgarch11} will be fully parametric and hence can be easily used for extrapolation of conditional quantiles of $y_t$ at high levels. 
	While this parametric approximation will induce a  bias, the gain is greater estimation efficiency at high quantile levels; see more discussions on the bias-variance trade-off in Section \ref{remark-bandwidth}. 
	
	Next we need  a suitable distribution of $\varepsilon_t$ such that the tail behavior can be flexibly captured. 
	There are many choices such that $Q_{\tau}(\varepsilon_t)$ has an explicit form,  including  distributions in lambda and Burr families \citep{Gilchrist2000}. We choose the Tukey-lambda distribution since it provides a wide range of shapes. It can not only approximate Gaussian and Logistic distributions but also fit heavy Pareto tails  well. 
	Given that $\varepsilon_t$ follows the Tukey-lambda distribution with shape parameter $\lambda \neq 0$ \citep{Joiner_Rosenblatt1971}, $Q_{\tau}(\varepsilon_t)$ has a simple explicit form given by
	\begin{equation}\label{TukeyLambda}
		Q_{\tau}(\lambda) := Q_{\tau}(\varepsilon_t;\lambda) = \dfrac{\tau^{\lambda}-(1-\tau)^{\lambda}}{\lambda}. 
	\end{equation}
	Combining \eqref{approxqlgarch11} and \eqref{TukeyLambda}, we can approximate the conditional quantile $Q_{\tau}(y_t|\mathcal{F}_{t-1})$ by 
	\[
	q_{t,\tau}(\bm\varphi) = Q_{\tau}(\lambda) \left(\frac{a_0}{1-b_1}+a_1\sum_{j=1}^{\infty} b_1^{j-1}|y_{t-j}|\right):= Q_{\tau}(\lambda)h_t(\bm\phi), 
	\]
	where $\bm\varphi=(\bm\phi^{\prime}, \lambda)^{\prime}=(a_0, a_1, b_1, \lambda)^{\prime}$ is the parameter vector of linear GARCH($1,1$) model with $\varepsilon_t$ following the Tukey-lambda distribution.  
	Note that $Q_{0.5}(\lambda)=0$ for any $\lambda$. Thus, $q_{t,0.5}(\bm\varphi)=0$ holds for any $\bm\varphi$, i.e.,  the location constraint on $Q_{\tau}(y_t|\mathcal{F}_{t-1})$ in \eqref{location11} is satisfied.

	Since $q_{t,\tau}(\bm\varphi)$ depends on unobservable values of $y_t$ in the infinite past, in practice we initialize $y_t=0$ for $t\leq 0$ and define its feasible counterpart as 
	\[\widetilde{q}_{t,\tau}(\bm\varphi) = Q_{\tau}(\lambda)\left(\frac{a_0}{1-b_1}+a_1\sum_{j=1}^{t-1} b_1^{j-1}|y_{t-j}|\right):=Q_{\tau}(\lambda)\widetilde{h}_t(\bm\phi).\]
	The initialization effect is asymptotically negligible, as we verify in our technical proofs.
	Note that $\widetilde{q}_{t,\tau}(\bm\varphi)$ is fully parametric. Since $\bm\varphi$ is independent of  $\tau$, we can approximate the nonparametric  function $Q_{\tau_0}(y_t|\mathcal{F}_{t-1})$ by the parametric function $\widetilde{q}_{t,\tau_0}(\bm\varphi)$, where we replace  $\bm\varphi$ with an estimator obtained by fitting the above Tukey-lambda linear GARCH($1,1$) model at lower quantile levels.

	Let $\Phi \subset (0,\infty)\times[0,\infty)\times [0,1)\times \Lambda$ be the parameter space of $\bm\varphi$, where $\Lambda=(-\infty,0)\cup (0,\infty)$ is the parameter space of $\lambda$. 
	To estimate $\bm\varphi$ locally for the target level $\tau_0$, we utilize the information at lower quantile levels  in the one-sided neighborhood of $\tau_0$, namely $\mathcal{T}_h=[\tau_0,\tau_0+h]\subset (0,0.5)$ if $\tau_0$ is close to zero and $\mathcal{T}_h=[\tau_0-h,\tau_0]\subset (0.5,1)$ if $\tau_0$ is close to one, where $h>0$ is a fixed bandwidth; see Section \ref{remark-bandwidth} for discussions on the selection of bandwidth $h$.
	If $Q_{\tau}(y_t|\mathcal{F}_{t-1})$ is well approximated by $q_{t,\tau}(\bm\varphi)$ for $\tau\in\mathcal{T}_h$, then we can estimate $\bm\varphi$ by the weighted CQR as follows:
	\begin{equation}\label{CQR} 
		\check{\bm\varphi}_{wn} = (\check{\bm\phi}_{wn}^{\prime}, \check{\lambda}_{wn})^{\prime} =\argmin_{\bm\varphi\in\Phi}\sum_{t=1}^n \sum_{k=1}^K w_t\rho_{\tau_k} \left(y_t-\widetilde{q}_{t,\tau_k}(\bm\varphi)\right),
	\end{equation} 
	where $\{w_t\}$ are the  self-weights defined as in \eqref{PairwiseWCQE}, and $\tau_1 < \cdots < \tau_K$ are fixed quantile levels with $\tau_k \in \mathcal{T}_h$ for all $1\leq k\leq K$; see also \cite{Zou_Yuan2008}. 
	In practice, equally spaced levels are typically used. That is, $\tau_k=\tau_0+h(k-1)/(K-1)$ if $\tau_0$ is close to zero, whereas $\tau_k=\tau_0-h(k-1)/(K-1)$ if $\tau_0$ is close to one. 
	As a result, the conditional quantile $Q_{\tau_0}(y_t|\mathcal{F}_{t-1})$ can be approximated by $\widetilde{q}_{t,\tau_0}(\check{\bm\varphi}_{wn})$.
	
	\subsection{Asymptotic properties}
	Note that the approximate conditional quantile function $q_{t,\tau}(\bm\varphi)$ can be rewritten using the true conditional quantile function $q_{t}(\cdot)$ as follows:
	\begin{equation}\label{eq:qt}
		q_{t,\tau}(\bm\varphi) = \dfrac{a_{0}Q_{\tau}(\lambda)}{1-b_{1}}+a_{1}Q_{\tau}(\lambda)\sum_{j=1}^{\infty} b_{1}^{j-1}|y_{t-j}| := q_{t}(\bm\theta_{\tau}^*), 
	\end{equation}
	where $\bm\theta_{\tau}^*=g_{\tau}(\bm\varphi)=(a_{0}Q_{\tau}(\lambda)/(1-b_{1}),a_{1}Q_{\tau}(\lambda),b_{1})^{\prime}$, and  $g_{\tau}: \mathbb{R}^4\to\mathbb{R}^3$ is a measurable function such that $q_{t,\tau}=q_{t} \circ g_{\tau}$. 
	Let $\check{\bm\theta}_{wn}^*(\tau):=g_{\tau}(\check{\bm\varphi}_{wn})$ be the transformed CQR estimator. In view of \eqref{eq:qt} and the fact that  $Q_{\tau}(y_t|\mathcal{F}_{t-1})=q_t(\bm\theta(\tau))$,  $\check{\bm\theta}_{wn}^*(\tau)$ can be used as an estimator of $\bm\theta(\tau)$; see \eqref{QGARCH11-quantile} and the definition of $q_t(\cdot)$ in Section \ref{subsec-WCQE}. 
	The pseudo-true parameter vector $\bm\varphi_0^*=(\bm\phi_0^{\prime}, \lambda_0)^{\prime}=(a_{00}, a_{10}, b_{10}, \lambda_0)^{\prime}$ is defined as
	\begin{equation}\label{pseudo-true}
		\bm\varphi_0^*=\argmin_{\bm\varphi\in\Phi}\sum_{k=1}^K E[w_t\rho_{\tau_k} \left(y_t-q_{t,\tau_k}(\bm\varphi)\right)], \quad \tau_k\in\mathcal{T}_h.
	\end{equation}
	In other words, for $\tau\in\mathcal{T}_h$, the best approximation of the nonparametric function $Q_{\tau}(y_t|\mathcal{F}_{t-1})=q_{t}(\bm\theta(\tau))$ via the fully parametric function $q_{t,\tau}(\cdot)$ is given by $q_{t,\tau}(\bm\varphi_0^*)=q_{t}(g_{\tau}(\bm\varphi_0^*))$.

	In general, $Q_{\tau}(y_t|\mathcal{F}_{t-1})$ may be misspecified by $q_{t,\tau}(\bm\varphi_0^*)$, and $\bm\theta(\tau)=g_{\tau}(\bm\varphi_0^*)$ may not hold for all $\tau$. Thus, asymptotic properties of the CQR estimator $\check{\bm\varphi}_{wn}$ and its transformation $\check{\bm\theta}_{wn}^*(\tau)=g_{\tau}(\check{\bm\varphi}_{wn})$ should be established under possible model misspecification. The following assumptions will be required.
	
	\begin{assum}\label{assum-Process-Mixing} 
		$\{y_t\}$ is a strictly stationary and $\alpha$-mixing time series with the mixing coefficient $\alpha(n)$ satisfying $\sum_{n\geq 1}[\alpha(n)]^{1-2/\delta}<\infty$ for some $\delta>2$.
	\end{assum}
	
	\begin{assum}\label{assum-SpaceTukey}
		(i) The parameter space $\Phi$ is compact and $\bm\varphi_0^*$ is unique; 
		(ii) $\bm\varphi_0^*$ is an interior point of $\Phi$.
	\end{assum}
	
	Note that Assumption \ref{assum-Process} is insufficient for the asymptotic normality of $\check{\bm\varphi}_{wn}$ under model misspecification, since $E[\psi_{\tau}(y_t-q_{t,\tau}(\bm\varphi_0^*))|\mathcal{F}_{t-1}]\neq 0$  in this case, which renders the martingale CLT no longer applicable. Instead, we rely on Assumption \ref{assum-Process-Mixing} to  ensure  the ergodicity of $\{y_t\}$ and enable the use of the CLT for $\alpha$-mixing sequences;  see \cite{Fan_Yao2003} and more discussions in Remark \ref{mixing}.
	Assumption \ref{assum-SpaceTukey} is analogous to Assumption \ref{assum-Space}, which is standard in the literature on  GARCH models \citep{Francq_Zakoian2010, Zheng_Zhu_Li_Xiao2018}.
	If there is no model misspecification, i.e. $Q_{\tau}(y_t|\mathcal{F}_{t-1})$ is correctly specified by $q_{t,\tau}(\bm\varphi_0^*)$ for all $\tau\in\mathcal{T}_h$, then the uniqueness of $\bm\varphi_0^*$ can be guaranteed for $K\geq 3$ and $\lambda<1$.

	Let $\dot{q}_{t,\tau}(\bm\varphi)$ and $\ddot{q}_{t,\tau}(\bm\varphi)$ be the first and second derivatives of $q_{t,\tau}(\bm\varphi)$ with respect to $\bm\varphi$, respectively, given by
	\[\dot{q}_{t,\tau}(\bm\varphi)=(Q_{\tau}(\lambda)\dot{h}_t^{\prime}(\bm\phi),\dot{Q}_{\tau}(\lambda)h_t(\bm\phi))^{\prime}
	\quad\text{and}\quad
	\ddot{q}_{t,\tau}(\bm\varphi)=\begin{pmatrix}
		Q_{\tau}(\lambda)\ddot{h}_t(\bm\phi) & \dot{Q}_{\tau}(\lambda)\dot{h}_t(\bm\phi) \\ 
		\dot{Q}_{\tau}(\lambda)\dot{h}_t^{\prime}(\bm\phi) & \ddot{Q}_{\tau}(\lambda)h_t(\bm\phi)
	\end{pmatrix},
	\] 
	where $\dot{Q}_{\tau}(\lambda)$ and $\dot{h}_t(\bm\phi)$ (or $\ddot{Q}_{\tau}(\lambda)$ and $\ddot{h}_t(\bm\phi)$) are the first (or second) derivatives of $Q_{\tau}(\lambda)$ and $h_t(\bm\phi)$, respectively.    
	Denote $\bm X_t=\sum_{k=1}^K w_t\dot{q}_{t,\tau_k}(\bm\varphi_0^*)\psi_{\tau_k}(y_t-q_{t,\tau_k}(\bm\varphi_0^*))$ and $\Omega_{0w}^*=E(\bm X_t\bm X_t^{\prime})+n^{-1}\sum_{t\neq s}^nE(\bm X_t\bm X_s^{\prime})$. Define the matrices
	\[
	\Omega_{11}^*=\sum_{k=1}^KE\left[w_t\ddot{q}_{t,\tau_k}(\bm\varphi_0^*)\psi_{\tau_k}(y_t-q_{t,\tau_k}(\bm\varphi_0^*))\right]
	\;\;\text{and}\;\;
	\Omega_{12}^*=\sum_{k=1}^KE\left[w_tf_{t-1}(q_{t,\tau_k}(\bm\varphi_0^*))\dot{q}_{t,\tau_k}(\bm\varphi_0^*)\dot{q}_{t,\tau_k}^{\prime}(\bm\varphi_0^*)\right].
	\]
	Let $\Omega_{1w}^*=\Omega_{12}^*-\Omega_{11}^*$ and $\Sigma_w^*=\Omega_{1w}^{*-1}\Omega_{0w}^*\Omega_{1w}^{*-1}$.  
	
	\begin{thm}\label{thm-WCQR}
		For $\{y_t\}$ generated by model \eqref{qgarch11} with condition \eqref{constraint11}, suppose $E|y_t|^s<\infty$ for some $s\in (0,1)$ and $\Sigma_w^*$ is positive definite. 
		If Assumptions \ref{assum-ConditionalDensity}, \ref{assum-RandomWeight}, \ref{assum-Process-Mixing}, \ref{assum-SpaceTukey}(i) hold, then as $n\rightarrow\infty$, we have (i) $\check{\bm\varphi}_{wn} \to_p \bm\varphi_0^*$. Moreover, if Assumption \ref{assum-SpaceTukey}(ii) further holds, then (ii) $\sqrt{n}(\check{\bm\varphi}_{wn}-\bm\varphi_0^*)\rightarrow_d N\left(\bm 0,\Sigma_w^*\right)$; and (iii) $\sqrt{n}(\check{\bm\theta}_{wn}^*(\tau)-\bm\theta(\tau)-B(\tau))\rightarrow_d N\left(\bm 0,g_{\tau}(\bm\varphi_0^*)\Sigma_w^*g_{\tau}^{\prime}(\bm\varphi_0^*)\right)$, where $B(\tau)=g_{\tau}(\bm\varphi_0^*)-\bm\theta(\tau)$ is a systematic bias. 
	\end{thm}

	Theorem \ref{thm-WCQR}(iii) reveals that $\check{\bm\theta}_{wn}^*(\tau)$ is a biased estimator of $\bm\theta(\tau)$ if $g_{\tau}(\bm\varphi_0^*)\neq\bm\theta(\tau)$ i.e., when $Q_{\tau}(y_t|\mathcal{F}_{t-1})$ is misspecified by $q_{t,\tau}(\bm\varphi_0^*)$. Moreover, the systematic bias $B(\tau)$ depends on the bandwidth $h$, which balances the bias and variance of $\check{\bm\theta}_{wn}^*(\tau)$; see Section \ref{remark-bandwidth} for details. However, at the cost of introducing the systematic bias, the proposed CQR method can greatly improve the estimation efficiency  at high quantile levels, as it overcomes the inefficiency due to data scarcity  at tails. Similar to Theorem \ref{thm-WCQE}, we employ the bracketing method in \cite{Pollard1985}  to tackle the non-convexity and non-differentiability of the objective function. However, due to the possible model misspecification, the mixing CLT  is used instead of the  martingale CLT; see Assumption \ref{assum-Process-Mixing}.  
	We will discuss the estimation of the covariance matrix $\Sigma_w^*$ in the Appendix. 
	
	\begin{remark}[Mixing properties]\label{mixing}
		The proof of the mixing property in Assumption \ref{assum-Process-Mixing} is challenging. For a stationary Markovian process, a common approach to proving that it is geometrically $\beta$-mixing and thus $\alpha$-mixing  is to establish its geometric ergodicity \citep{Doukhan1994,Francq_Zakoian2006}. Note that the proposed quantile GARCH process  can be regarded as a random-coefficient ARCH($\infty$) process. However, ARCH($\infty$) processes are not Markovian in general \citep{Fryzlewicz_SubbaRao2011}. Thus, the above approach is not feasible.
		\cite{Fryzlewicz_SubbaRao2011} provides an alternative method to establish mixing properties. By deriving explicit bounds for mixing coefficients using conditional densities of the process, they obtain mixing properties of stationary ARCH($\infty$) processes and show that the bound on the mixing rate depends on the decay rate of ARCH($\infty$) parameters. This method potentially can be applied to the quantile GARCH process. However, it is challenging to derive the conditional density of $y_{k+s}$ given $\{\ldots,y_0,U_1,\ldots,U_{k-1},y_k,\ldots,y_{k+s-1}\}$ due to the random functional coefficients driven by $U_t$. Thus, we leave this for future research.
	\end{remark}

	\subsection{Selection of the bandwidth $h$}\label{remark-bandwidth}
	
	As shown in Theorem \ref{thm-WCQR}(iii), the bandwidth $h$ plays an important role in  balancing the bias and efficiency of the  estimator $\check{\bm\theta}_{wn}^*(\tau)$. 
	In the extreme case that $h=0$, \eqref{CQR} will become a weighted quantile regression at the fixed quantile level $\tau_0$, and $\check{\bm\theta}_{wn}^*(\tau_0)$ will be equivalent to the QR estimator $\widetilde{\bm\theta}_{wn}(\tau_0)$. Then we have $g_{\tau_0}(\bm\varphi_0^*)=\bm\theta(\tau_0)$ and $B(\tau_0)=0$. 
	Although $B(\tau)$ does not have an explicit form with respect to $h$, our simulation studies show that a larger $h$ usually leads to larger biases but smaller variances of $\check{\bm\theta}_{wn}^*(\tau)$ when the true model is misspecified; see Section \ref{Sec-simulation-WCQR} for details. 
	
	In practice, we can treat $h$ as a hyperparameter and  search for $h$ that achieves the best forecasting performance from a grid of values  via cross-validation.  	
	Specifically, we can divide the dataset into training and validation sets, and choose the value of $h$ that minimizes the check loss in the validation set for the target quantile level $\tau_0$: 
	\begin{equation}\label{optimal-bandwidth-check}
		h^{opt}=\argmin_{h\in(0,d)}\sum_{t=n_0+1}^{n_0+n_1}\rho_{\tau} \left(y_t-\widetilde{q}_{t,\tau_0}(\check{\bm\varphi}_{wn}(h))\right),
	\end{equation}
	where $n_0$ and $n_1$ are the sample sizes of the training and validation sets, respectively, $\check{\bm\varphi}_{wn}(h)$ is the CQR estimator calculated by \eqref{CQR} with bandwidth $h$, and $d>0$ determines the range of the grid search. Usually we take $d$ to be a small value such as 0.1 to avoid large biases. 
	The chosen bandwidth $h^{opt}$ will be used to conduct CQR for rolling forecasting of the conditional quantile at time $t=n_0+n_1+i$ for any $i\geq 1$.

	\section{Simulation studies}\label{Sec-simulation}
	
	\subsection{Data generating processes}
	
	This section conducts simulation experiments to examine the finite sample performance of the proposed  estimators and CvM test. The data generating process (DGP) is
	\begin{align}\label{DGP}
		y_t=\omega(U_t)+\alpha_1(U_t)\sum_{j=1}^{\infty}[\beta_1(U_t)]^{j-1}|y_{t-j}|,
	\end{align}
	where $\{U_t\}$ are $i.i.d.$ standard uniform random variables. For evaluation of the QR and CQR estimators, we consider two sets of coefficient functions  as follows: 
	\begin{align} \label{sim1coef1}
		\omega(\tau)=0.1F^{-1}(\tau), \; \alpha_1(\tau)=0.1F^{-1}(\tau), \;  \beta_1(\tau)=0.8,
	\end{align}
	and
	\begin{align}\label{sim1coef2}
		\omega(\tau)=0.1F^{-1}(\tau), \; \alpha_1(\tau)=\tau-0.5+0.1F^{-1}(\tau), \; \beta_1(\tau)=0.3+0.6|\tau-0.5|,
	\end{align}
	where $F(\cdot)$ is the distribution function of the standard normal distribution or Tukey-lambda distribution in \eqref{TukeyLambda} with the shape parameter $\lambda = -0.2$, denoted by $F_N(\cdot)$ and $F_T(\cdot)$ respectively. Note that $F_T$ has heavy Pareto tails and does not have the finite fifth moment \citep{Karian_Dudewicz_Mcdonald1996}.
	For coefficient functions in \eqref{sim1coef2}, the strict stationarity condition \eqref{StationaryCondition11_fractional} with $s=1$ in Theorem \ref{thm-stationarity11} can be verified for $F=F_N$ or $F_T$ by direct calculation or simulating $10^5$ random numbers for $U_t$, respectively.
	Note that the DGP with coefficient functions in \eqref{sim1coef1} is simply the following GARCH($1,1$) process:
	\begin{align*}
		y_t = \varepsilon_t\left(0.1+0.1\sum_{j=1}^{\infty} 0.8^{j-1}|y_{t-j}|\right),
	\end{align*}
	where $\varepsilon_t$ follows the distribution $F$. 
	As a result, the model is correctly specified for the CQR under \eqref{sim1coef1} with $F$ being the Tukey-lambda distribution (i.e. $F=F_T$), whereas it is misspecified under all other settings. 
	Two sample sizes, $n=1000$ and 2000, are considered, and 1000 replications are generated for each sample size. 
	
	In addition, for the CvM test in \eqref{CvM_constant}, we  consider the following coefficient functions:  
	\begin{align}\label{sim1coef3}
		\omega(\tau)=0.1F^{-1}(\tau), \; \alpha_1(\tau)=0.1F^{-1}(\tau), \; \beta_1(\tau)=0.3+d(\tau-0.5)^{2},
	\end{align}
	where $d=0, 1$ or 1.6, and all other settings are the same as those for  \eqref{sim1coef1}. We can similarly verify that the strict stationarity condition holds with  $s=1$ under this setting. Note that the case of $d=0$ corresponds to the size of the test, whereas the case of $d = 1$ or 1.6 corresponds to the power.  
	
	The computation of QR and CQR estimators and the CvM test involves a infinite sum. For computational efficiency, we adopt an exact algorithm based on the fast Fourier transform instead of the standard linear convolution algorithm; see \cite{Nielsen_Noel2021} for details. 
	
	\subsection{Self-weighted QR estimator}\label{Sec-simulation-WCQE}
	
	The first experiment focuses on the self-weighted QR estimator $\widetilde{\bm\theta}_{wn}(\tau)$ in Section \ref{subsec-WCQE}. 
	For the estimation of the asymptotic standard deviation (ASD) of $\widetilde{\bm\theta}_{wn}(\tau)$, we employ the two bandwidths \eqref{bandwidths}. The resulting ASDs with respect to  bandwidths $\ell_B$ and $\ell_{HS}$ are denoted by ASD$_1$ and ASD$_2$, respectively. 
	Tables \ref{tab.estimation.CQE.w2.DGP1} and \ref{tab.estimation.CQE.w2.DGP2} display the biases, empirical standard deviations (ESDs) and ASDs of $\widetilde{\bm\theta}_{wn}(\tau)$ at quantile level $\tau=0.5\%,1\%$ or 5\% for \eqref{sim1coef1} and \eqref{sim1coef2} with $F$ being the standard normal distribution $F_N$ or Tukey-lambda distribution $F_T$, respectively. We have the following findings. First, as the sample size increases, most of the biases, ESDs and ASDs decrease, and the ESDs get closer to the corresponding ASDs.
	Secondly, the ASDs calculated using $\ell_{HS}$ are marginally smaller than those using $\ell_B$ and closer to the ESDs.  Thus, we use the bandwidth $\ell_{HS}$ in the following for stabler performance.
	Thirdly, when $\tau$ is closer to zero, the performance of $\widetilde{\bm\theta}_{wn}(\tau)$ gets worse with larger biases, ESDs and ASDs, which indicates that the self-weighted QR estimator tends to deteriorate as the target quantile becomes more extreme.  
	
	The above results are obtained based on the self-weights in \eqref{choice1_of_selfweights} with $c$ being the 95\% sample quantile of $\{y_t\}_{t=1}^n$.  We have also considered the 90\% sample quantile for the value of $c$, and the above findings are unchanged. In addition,
	simulation results for the unweighted QR estimator are given in the Appendix. It is shown that  the unweighted estimator is less efficient than the self-weighted one when $E|y_t|^3=\infty$.

	\subsection{The CvM test}\label{Sec-simulation-CvM}
	
	The second experiment evaluates the performance of the CvM test in Section \ref{subsec-CvM}. 
	Since we are particularly interested in the behavior of persistence coefficient function $\beta_1(\tau)$ at tails, we consider $\mathcal{T}=[0.7,0.995]$ and $[0.8,0.995]$. 	
	To calculate $S_{n}$ in \eqref{CvM_constant}, we use a grid $\mathcal{T}_n$ with equal cell size $\delta_n=0.005$ in place of $\mathcal{T}$.
	Moreover, $\ell_{HS}$ in \eqref{bandwidths} is employed to calculate $\widetilde{z}_{t}(\tau)$ in the subsampling procedure.
	The rejection rates of $S_n$ at $5\%$ significance level are summarized in Table \ref{tab_const_test}. 
	Firstly, observe that the size is close to the nominal rate when $b_n=\lfloor n^{1/2}\rfloor$. The case with $b_n=\lfloor 0.5n^{1/2}\rfloor$ tends to be undersized, while that with $b_n=\lfloor 2n^{1/2}\rfloor$ tends to be oversized. Secondly, the power generally increases as the sample size $n$ or departure level $d$ increases. Thirdly, a larger subsampling block size $b_n$ or wider interval $\mathcal{T}$ tends to result in a greater power.
	Hence,  we recommend using $b_n=\lfloor n^{1/2}\rfloor$ since it leads to reasonable size and power. For a fixed $\mathcal{T}$, we have also considered other settings for $\mathcal{T}_n$, and the above findings are unchanged. This indicates that the CvM test is not sensitive to the choice of the grid.

	\subsection{Self-weighted CQR estimator}\label{Sec-simulation-WCQR}
	
	In the third experiment, we examine the performance of the proposed CQR method in Section \ref{Sec-CQR} via the transformed estimator $g_{\tau}(\check{\bm\varphi}_{wn})=\check{\bm\theta}_{wn}^*(\tau)=(\check{\omega}_{wn}^*(\tau),\check{\alpha}_{1wn}^*(\tau),\check{\beta}_{1wn}^*(\tau))^{\prime}$. The DGP is preserved from the first experiment. 
	To obtain the weighted CQR estimator $\check{\bm\varphi}_{wn}$ in \eqref{CQR}, we let $\mathcal{T}_h=\{\tau_k: \tau_k=\tau_0+h(k-1)/(K-1)\}_{k=1}^K$, where $K=19$,  $\tau_0=0.5\%,1\%$ or 5\% is the target quantile level, and $h>0$ is the bandwidth.  
	
	To investigate the influence of bandwidth $h$ on the CQR, we obtain the estimator $g_{\tau}(\check{\bm\varphi}_{wn})$ for each $h\in\{0.01,0.02,\ldots,0.10\}$ at quantile level $\tau=0.5\%,1\%$ or $5\%$ for the DGP in \eqref{DGP} with \eqref{sim1coef1} or \eqref{sim1coef2}, $F=F_T$, and  sample size $n=2000$. 
	Figures \ref{fig_CQR_h_DGP1} and \ref{fig_CQR_h_DGP2} illustrate the empirical squared bias, variance and mean squared error (MSE) of $g_{\tau}(\check{\bm\varphi}_{wn})$ versus $h$ for  coefficient functions in  \eqref{sim1coef1} and \eqref{sim1coef2}, respectively.
	Note that the model is correctly specified under coefficient functions in  \eqref{sim1coef1} with $F=F_T$ and misspecified under \eqref{sim1coef2} with $F=F_T$. 
	Figure \ref{fig_CQR_h_DGP1} shows that the squared bias is close to zero, which is because the model is correctly specified. Meanwhile,  as $h$ increases, the variance and MSE get smaller, indicating  the efficiency gain from using  more data for the estimation. On the other hand, Figure \ref{fig_CQR_h_DGP2} shows that a larger $h$ leads to larger biases but smaller variances  under model misspecification.  Consequently, as $h$ increases, the MSE first  decreases and then increases. 
	Moreover, it can be observed that the CQR estimator can have much smaller MSE than the QR estimator (i.e., the case with $h=0$) especially for the high quantiles. This corroborates the usefulness of the CQR for high quantile levels.
	
	Next we verify the asymptotic results of the CQR estimator by focusing on a fixed bandwidth $h=0.1$. 
	The ASD of $g_\tau(\check{\bm\varphi}_{wn})$ is calculated based on $\dot{g}_{\tau}(\check{\bm\varphi}_{wn})\check{\Sigma}_w^*\dot{g}_{\tau}^{\prime}(\check{\bm\varphi}_{wn})$, where $\check{\Sigma}_w^*$ is obtained as in Section \ref{remark-WCQR_cov} of the Appendix. 
	Specifically, to estimate $\Omega_{1w}^*$, the bandwidth $\ell_k$ for quantile level $\tau_k$ is set to $\ell_{HS}$ defined in \eqref{bandwidths} with $\tau$ replaced by $\tau_k$.   
	To obtain the kernel estimator $\check{\Omega}_{0w}^*$ in  \eqref{kernel_estimator}, we consider the QS kernel in \eqref{QS-kernel} with the automatic bandwidth $\widehat{B}_n=1.3221[n\widehat{\alpha}(2)]^{1/5}$, $0.1\widehat{B}_n$ or $10\widehat{B}_n$ for $B_n$, where the latter two choices of $B_n$ correspond to under- or over-smoothing in comparison to $\widehat{B}_n$, respectively. The resulting ASDs with respect to $\widehat{B}_n, 0.1\widehat{B}_n$ and $10\widehat{B}_n$ are denoted as ASD$_{a}$, ASD$_{b}$ and ASD$_{c}$, respectively.  
	Tables \ref{tab.estimation.CQR.DGP1} and \ref{tab.estimation.CQR.DGP2} report the biases, ESDs and ASDs of $g_\tau(\check{\bm\varphi}_{wn})$  for the DGP with coefficient functions in \eqref{sim1coef1} and \eqref{sim1coef2},  respectively.  The  quantile levels $\tau=0.5\%,1\%$ and $5\%$ and distributions $F=F_N$ and $F_T$ are considered.
	
	We first examine the results in Table \ref{tab.estimation.CQR.DGP1}, which corresponds to the DGP with \eqref{sim1coef1} and covers two scenarios: correctly specified  (when $F=F_T$) and missspecified  (when $F=F_N$) models. For both scenarios, we have  three main findings as follows. Firstly, as the sample size increases, most of the biases, ESDs and ASDs become smaller, and the ESDs get closer to the corresponding ASDs. Secondly, as  $\tau$ approaches zero, the biases, ESDs and ASDs of $\check{\omega}_{wn}^*(\tau)$ and $\check{\alpha}_{1wn}^*(\tau)$ get larger, while that of $\check{\beta}_{1wn}^*(\tau)$ is almost unchanged. This is expected since $\check{\omega}_{wn}^*(\tau)$ and $\check{\alpha}_{1wn}^*(\tau)$ are $\tau$-dependent, and their true values have larger absolute values as $\tau$ goes to zero. However, $\check{\beta}_{1wn}^*(\tau)=\check{b}_{1wn}$ is independent of $\tau$.
	Thirdly, the results of ASD$_{a}$, ASD$_{b}$ and ASD$_{c}$ are very similar, which suggests that the kernel estimator in \eqref{kernel_estimator} is insensitive to the selection of bandwidth $B_n$. 
	
	It is also interesting to compare the results under the two scenarios in Table \ref{tab.estimation.CQR.DGP1}. However, it is worth noting that
	the true values of $\omega(\tau)$ and $\alpha_1(\tau)$ for the correctly specified model (i.e., when $F=F_T$) are larger than those for the misspecified model (i.e., when $F=F_N$) in absolute value. As a result,  the absolute biases, ESDs and ASDs of $\widetilde{\omega}_{wn}(\tau)$ and $\widetilde{\alpha}_{1wn}(\tau)$ are much smaller for $F_N$ than that for $F_T$ in  Table \ref{tab.estimation.CQR.DGP1}. On the other hand,  note that the true values of $\beta_1(\tau)$ are the same for $F_N$ and $F_T$. Thus, the comparison of the results for $\widetilde{\beta}_{1wn}(\tau)$  under $F_N$ and $F_T$ can directly reveal the effect of model misspecification. Indeed,  Table \ref{tab.estimation.CQR.DGP1} shows that the absolute biases, ESDs and ASDs of $\widetilde{\beta}_{1wn}(\tau)$ for $F_T$ are much smaller than those for $F_N$. This confirms that the CQR performs better under  correct specification  (i.e.,  $F=F_T$) than misspecification (i.e.,  $F=F_N$).
	
	Note that the above misspecification is only due to the misspecified innovation distribution $F$, whereas the coefficient function (i.e., model structure) is correctly specified via  \eqref{sim1coef1}. By contrast, the DGP with \eqref{sim1coef2} have  a misspecified model structure, which is more severe than the former. As a result, Table \ref{tab.estimation.CQR.DGP2} shares the three main findings from Table \ref{tab.estimation.CQR.DGP1} for the ESDs and ASDs but not for the biases. In particular, most biases do not decrease as the sample size increases. This is consistent with Theorem \ref{thm-WCQR}(iii), which shows that $g_{\tau}(\check{\bm\varphi}_{wn})$ is  in general a biased estimator of $\bm\theta(\tau)$ under model misspecification. It also indicates that the misspecification in the model structure is systematic and has greater impact on the bias than that in the innovation distribution $F$.

	We have also considered other choices of the number of quantile levels $K$ and the kernel function $K(\cdot)$. The above findings are unchanged. To save space, these results are omitted.

	\subsection{Comparison between QR and CQR estimators}\label{Sec-simulation-comparison}
	
	We aim to compare the in-sample and out-of-sample performance of QR and CQR in predicting conditional quantiles. The self-weights $\{w_t\}$ in \eqref{choice1_of_selfweights} are employed for both QR and CQR, and the set $\mathcal{T}_h$ with $K=19$ and $h=0.1$ is used for CQR as in the third experiment. 
	
	For evaluation of the prediction performance, we use $\widetilde{q}_t(\bm\theta(\tau))$ as the true value of the conditional quantile $Q_\tau(y_t|\mathcal{F}_{t-1})$. 
	Based on the QR estimator $\widetilde{\bm\theta}_{wn}(\tau)$ and the transformed CQR estimator $g_{\tau}(\check{\bm\varphi}_{wn})$, $Q_\tau(y_t|\mathcal{F}_{t-1})$ can be predicted by $\widetilde{q}_t(\widetilde{\bm\theta}_{wn}(\tau))$ and $\widetilde{q}_t(g_{\tau}(\check{\bm\varphi}_{wn}))$, respectively.  
	Note that estimates of $Q_\tau(y_t|\mathcal{F}_{t-1})$ for $t=1,\ldots,n$  are  in-sample predictions, and that of $Q_\tau(y_{n+1}|\mathcal{F}_{n})$ is the out-of-sample forecast. 
	We measure the in-sample and out-of-sample prediction performance separately, using the biases and RMSEs of conditional quantile estimates by averaging individual values over all time points and replications as follows:
	\begin{align*}
		\text{Bias}_{In}(\bm\theta_{\tau})&=\dfrac{1}{Mn}\sum_{k=1}^{M}\sum_{t=1}^{n}[\widetilde{q}_t^{(k)}(\bm\theta_{\tau})-\widetilde{q}_t^{(k)}(\bm\theta(\tau))], \\ \text{Bias}_{Out}(\bm\theta_{\tau})&=\dfrac{1}{M}\sum_{k=1}^{M}[\widetilde{q}_{n+1}^{(k)}(\bm\theta_{\tau})-\widetilde{q}_{n+1}^{(k)}(\bm\theta(\tau))], \\
		\text{RMSE}_{In}(\bm\theta_{\tau})&=\left\{\dfrac{1}{Mn}\sum_{k=1}^{M}\sum_{t=1}^{n}[\widetilde{q}_t^{(k)}(\bm\theta_{\tau})-\widetilde{q}_t^{(k)}(\bm\theta(\tau))]^2\right\}^{1/2}, \\ \text{RMSE}_{Out}(\bm\theta_{\tau})&=\left\{\dfrac{1}{M}\sum_{k=1}^{M}[\widetilde{q}_{n+1}^{(k)}(\bm\theta_{\tau})-\widetilde{q}_{n+1}^{(k)}(\bm\theta(\tau))]^2\right\}^{1/2},
	\end{align*}
	where $M=1000$ is the total number of replications, $\widetilde{q}_t^{(k)}(\bm\theta_{\tau})$ represents the conditional quantile estimate at time  $t$ in the $k$th replication, and $\bm\theta_{\tau}$ is  the QR estimator $\widetilde{\bm\theta}_{wn}(\tau)$ or the transformed CQR estimator $g_{\tau}(\check{\bm\varphi}_{wn})$.
	
	Table \ref{tab.qt.forecast.comparison} reports the above measures for the DGP in \eqref{DGP} with coefficient functions in \eqref{sim1coef1} and \eqref{sim1coef2}. Firstly, note that most of the biases and RMSEs decrease as the sample size increases. Secondly, the QR and CQR   perform  similarly for  \eqref{sim1coef1} and \eqref{sim1coef2} with $F=F_N$ in terms of the bias and RMSE. However, when $F=F_T$, obviously the CQR outperforms the QR in biases and RMSEs especially for high quantiles. This confirms that the CQR can be more favorable than the QR at high quantile levels if the data is heavy-tailed, yet can be comparable to the latter if otherwise. This is also consistent with the findings in Figures \ref{fig_CQR_h_DGP1} and \ref{fig_CQR_h_DGP2}. Lastly, although the CQR estimator is biased under model misspecification, the biases of its conditional quantile predictions are very close to or even smaller than those of the QR. This suggests that the CQR can provide satisfactory approximation of conditional quantiles,  possibly owing to the flexibility of the Tukey-lambda distribution.
	
	In the Appendix, we also provide a simulation experiment to investigate the effect of quantile rearrangement on the prediction performance.
	
	\section{An empirical example}\label{Sec-realdata}
	
	This section analyzes  daily log returns of the S\&P500 Index based on the proposed quantile GARCH model. 
	The daily closing prices from July 1, 2015 to December 30, 2021, denoted by $\{p_t\}$, are downloaded from the website of Yahoo Finance. 
	Let $y_t=100\left( \ln p_t-\ln p_{t-1} \right)$ be the log return in percentage, which has $n=1637$ observations in total.
	The time plot of $\{y_t\}$ suggests that the series exhibits  volatility clustering, and it is very volatile at the beginning of 2020 due to COVID-19 pandemic; see  Figure \ref{fig_real_data}.
	Table \ref{table_real_data_stat} displays summary statistics of $\{y_t\}$, where the sample skewness with value $-1.053$ and kurtosis with value $23.721$ indicate that the data are left-skewed and very heavy-tailed. 
	The above findings motivate us to fit $\{y_t\}$ by our proposed quantile GARCH model to capture the conditional heteroscedasticity of the return series and possible asymmetric dynamics over its different quantiles.
	
	We fit a quantile GARCH($1,1$) model to $\{y_t\}$. 
	Since the data are very heavy-tailed, the self-weighted QR estimator in \eqref{PairwiseWCQE} is used to obtain estimates of $\bm\theta(\tau)=(\omega(\tau),\alpha_1(\tau),\beta_1(\tau))^{\prime}$, where the self-weights in \eqref{choice1_of_selfweights} are employed with $c$ being the 95\% sample quantile of $\{y_t\}$. 
	The estimates of $\bm\theta(\tau)$ for $\tau\in (0.7,1)$ together with their 95\% pointwise confidence intervals are plotted against the quantile level in Figure \ref{fig_real_data}. 
	Note that $\bm\theta(\tau)$ of our model corresponds to $\bm\theta_{\tau}=(a_0Q_{\tau}(\varepsilon_t)/(1-b_1),a_1Q_{\tau}(\varepsilon_t),b_1)$ in the linear GARCH($1,1$) model in  \eqref{lgarch11}.   
	To compare the fitted coefficients of our model with those of model \eqref{qlgarch11}, we also provide estimates of $\bm\theta_{\tau}$ using the filtered historical simulation (FHS)  method \citep{Kuester_Mittnik_Paolella2006} based on the  Gaussian quasi-maximum likelihood estimation (QMLE). Specifically, $a_0,a_1$ and $b_1$ are estimated by Gaussian QMLE of the linear GARCH($1,1$) model in \eqref{lgarch11}, and then $Q_{\tau}(\varepsilon_t)$ is estimated by the empirical quantile of resulting residuals $\{\widehat{\varepsilon}_t\}$.
	
	From Figure \ref{coef_comparison_ARCHinfty}, we can see that the confidence intervals of $\omega(\tau)$, $\alpha_1(\tau)$ and $\beta_1(\tau)$ do not include the FHS estimates of $\bm\theta_{\tau}$ for $\tau\in (0.7,0.8), (0.9,1)$ and $(0.9,1)$ respectively. Since the quantile GARCH model includes the linear GARCH model as a special case, this  indicates that the model with constant coefficients fails to capture the asymmetric dynamic structures across different quantiles. In addition, we apply the CvM test in Section \ref{subsec-CvM} to check whether $\beta_1(\tau)$ is constant for  $\tau\in \mathcal{T}_1=[0.700,0.850]$, $\tau\in \mathcal{T}_2=[0.850,0.950]$, $\tau\in \mathcal{T}_3=[0.950,0.980]$, $\tau\in \mathcal{T}_4=[0.980,0.995]$, and $\tau\in \mathcal{T}=[0.700,0.995]=\cup_{i=1}^4\mathcal{T}_i$. The CvM test statistic $S_{n}$ is calculated using a grid $\mathcal{T}_n$ with equal cell size $\delta_n=0.005$. Its critical value is approximated using the proposed subsampling procedure with $b_n=\lfloor n^{1/2}\rfloor$. The $p$-values of $S_{n}$ for $\mathcal{T}_1,\dots, \mathcal{T}_4$, and $\mathcal{T}$ are  $0.585$, $0.054$, $0.555$, $0.017$, and $0.150$, respectively. 
	Therefore, it is likely that $\beta_1(\tau)$ is varying over $[0.850,0.950]$ and $[0.980,0.995]$.

	Since the 5\% VaR is of common interest in practice, we report the fitted quantile GARCH model at $\tau=0.05$ as follows: 
	\begin{equation}\label{fittedmodel}
		\widetilde{Q}_{0.05}(y_t| \mathcal{F}_{t-1})=-0.380_{0.100}-0.341_{0.075}\sum_{j=1}^{\infty}0.790_{0.033}^{j-1}|y_{t-j}|,
	\end{equation}
	where the standard errors are given in the corresponding subscripts of the estimated coefficients. We divide the dataset into a training set ($\mathcal{S}_{\textrm{train}}$) with size $n_0=1000$ and a test set ($\mathcal{S}_{\textrm{test}}$) with size $n-n_0=637$. Then we conduct a rolling forecast procedure at level $\tau=0.05$ (i.e. negative 5\% VaR) with a fixed moving window of size $n_0$  from the forecast origin $t_0=n_0+1$ (June 24, 2019). That is, we first obtain the  one-step-ahead conditional quantile forecast for $t_0$ (i.e., the first time point in $\mathcal{S}_{\textrm{test}}$) based on data from $t=1$ to $t=n_0$, using the formula  $\widetilde{Q}_{0.05}(y_{t_0}|\mathcal{F}_{n_0})=\widetilde{\omega}_{wn_0}(0.05)+\widetilde{\alpha}_{1wn_0}(0.05)\sum_{j=1}^{n_0}[\widetilde{\beta}_{1wn_0}(0.05)]^{j-1}|y_{t_0-j}|$. Then for each $i=1,\dots, n-n_0-1$, we set the forecast origin to $t_0+i$ and conduct the forecast based on data from $t=1+i$ to $t=n_0+i$. 
	These forecasts are displayed in the time plot in Figure \ref{fig_real_data}. It is clear that the VaR forecasts keep in step with the returns closely, and the return falls below the corresponding negative 5\% VaR forecasts occasionally.

	We also thoroughly  compare the forecasting performance of the proposed model with that of existing conditional quantile estimation methods as follows:
	\begin{itemize}
		\item FHS: The FHS method \citep{Kuester_Mittnik_Paolella2006} based on the linear GARCH($1,1$) model in  \eqref{lgarch11}, where the coefficients are estimated by the Gaussian QMLE, and the residual empirical quantiles are used to approximate the innovation quantiles.
		\item XK: The two-step estimation method QGARCH2 of \cite{Xiao_Koenker2009} based on linear GARCH(1,1) model \eqref{linearGARCH}. Specifically, the initial estimates of $\{h_t\}$ are obtained by combining the conditional quantile estimates of sieve ARCH approximation $h_t=\gamma_0+\sum_{j=1}^m\gamma_j|y_{t-j}|$ over multiple quantile levels, $\tau_{k}=k/20$ for $k=1,2,\ldots,19$, via the minimum distance estimation. Here we set $m=3n^{1/4}$ as in their paper. 
		\item Hybrid: The hybrid estimation method proposed in \cite{Zheng_Zhu_Li_Xiao2018} based on Bollerslev's GARCH($1,1$) model in \eqref{garch11} with $x_t=y_t$. 
		\item CAViaR: The indirect GARCH($1,1$)-based CAViaR method in \cite{Engle_Manganelli2004}, where we use the same code and settings for the optimization as in their paper.
	\end{itemize}
	We consider the lower and upper $1\%, 2.5\%$ and $5\%$ quantiles and conduct the above rolling forecast procedure for all competing methods. The forecasting performance is evaluated via the empirical coverage rate (ECR), prediction error (PE), and VaR backtests. The ECR is calculated as the percentage of observations in the test set $\mathcal{S}_{\textrm{test}}$ that fall below the corresponding fitted conditional quantiles. 
	The PE is calculated as follows:  
	\[PE=\dfrac{1}{\sqrt{\tau(1-\tau)/(n-n_0)}}\left|\dfrac{1}{n-n_0}\sum_{t=n_0+1}^nI\{y_t<\widehat{Q}_\tau(y_t | \mathcal{F}_{t-1})\}-\tau\right|,\]
	where $n-n_0$ is the size of $\mathcal{S}_{\textrm{test}}$, and $\widehat{Q}_\tau(y_t | \mathcal{F}_{t-1})$ is the one-step-ahead conditional quantile forecast based on each estimation method. 
	
	We conduct two VaR backtests: the likelihood ratio test for correct conditional coverage (CC) in \cite{Christoffersen1998} and the dynamic quantile (DQ) test in \cite{Engle_Manganelli2004}. 
	The null hypothesis of the CC test is that, conditional on $\mathcal{F}_{t-1}$, $\{H_t\}$ are $i.i.d.$ Bernoulli random variables with the success probability being $\tau$, where $H_t=I(y_t<Q_{\tau}(y_t | \mathcal{F}_{t-1}))$ is the hit series.
	For the DQ test in \cite{Engle_Manganelli2004}, we consider the regression  of $H_t$ on a constant and four lagged hits $H_{t-\ell}$ with $1\leq \ell \leq 4$. The null hypothesis is that the intercept equals to $\tau$ and the regression coefficients are zero. 
	If we fail to reject the null hypotheses of the VaR backtests, then the  forecasting method is satisfactory.   
	Table \ref{tabForecasting1} reports the ECRs, PEs and $p$-values of  VaR backtests for the one-step-ahead forecasts. 
	In terms of  ECRs and backtests, all methods perform reasonably well, since the ECRs  are close to the corresponding nominal levels, and at least one backtest is not rejected at the 5\% significance level. 
	However, it is clear that the proposed QR estimator has the smallest PEs in most cases. 
	
	Furthermore,  we compare the performance of the proposed self-weighted QR and CQR estimators at high quantile levels, including the lower and upper $0.1\%, 0.25\%$ and $0.5\%$ quantiles. 
	For a more accurate evaluation, we enlarge the S\&P500 dataset to cover the period from February 23, 2000 to December 30, 2021, which includes $n=5500$ observations in total.
	Moreover, since  the self-weighted CQR requires a predetermined bandwidth $h$, we divide the dataset into a training set ($\mathcal{S}_{\textrm{train}}$) with size $n_0=1000$, a validation set ($\mathcal{S}_{\textrm{val}}$) with size $n_1=500$, and a test set ($\mathcal{S}_{\textrm{test}}$) with size $n_2=n-n_0-n_1$. We choose the optimal $h$ that minimizes the check loss in \eqref{optimal-bandwidth-check} for $\mathcal{S}_{\textrm{val}}$; see Section \ref{remark-bandwidth} for details. Then based on the chosen $h$, we conduct a moving-window rolling forecast procedure similar to the previous one. The window size is $n_0$, and the forecast origin is $t_0=n_0+n_1+1=1501$. That is, we first obtain the conditional quantile forecast for $t_0$ (i.e., the first time point in $\mathcal{S}_{\textrm{test}}$) based on data from $t=t_0-n_0=501$ to $t=t_0-1=1500$  (i.e., the last 500 observations in $\mathcal{S}_{\textrm{train}}$ and all observations in $\mathcal{S}_{\textrm{val}}$). We repeat this procedure by  advancing the forecast origin and moving window  until the end of $\mathcal{S}_{\textrm{test}}$ is reached.
	Table \ref{tabForecasting2} displays the results for the proposed QR, CQR and other competing methods. 
	Notably, the CQR method has the smallest PE and the most accurate ECR at almost all quantile levels, while the  QR method is generally competitive among the other methods. 
	In summary, for the S\&P 500 dataset, the proposed quantile GARCH model has superior forecasting performance than the original GARCH model, and the proposed CQR  estimator outperforms the QR estimator at high quantile levels.

	Finally, to remedy the quantile crossing problem, we have further conducted the quantile rearrangement \citep{Chernozhukov2010} for the proposed QR method. There are only inconsequential changes to  Tables \ref{tabForecasting1} and \ref{tabForecasting2}, while all main findings summarized earlier remain the same. In addition, for Figure \ref{coef_comparison_ARCHinfty}, we can also  rearrange the self-weighted QR estimates $\{\widetilde{\omega}_{wn}(\tau_k)\}_{k=1}^K$ and $\{\widetilde{\alpha}_{1wn}(\tau_k)\}_{k=1}^K$ to ensure the monotonicity of the curves.  After the  rearrangement, the curves for  $\omega(\cdot)$ and $\alpha_1(\cdot)$  become smoother than those in  Figure \ref{coef_comparison_ARCHinfty}. The corresponding confidence intervals are slightly narrower than the original ones; see Section \ref{Sec-rearrangement} of the Appendix for details.

	\section{Conclusion and discussion}\label{Sec-conclusion}
	
	This paper proposes  the quantile GARCH model, a new conditional heteroskedastic model whose coefficients are functions of a standard uniform random variable. 
	A sufficient condition for the strict stationarity of this model is derived. To estimate the unknown coefficient functions without any moment restriction on the data, we develop the self-weighted QR and CQR methods. By efficiently borrowing information from intermediate quantile levels via a  flexible parametric approximation, the CQR method is more favorable than the QR at high quantile levels. 
	Our empirical analysis shows that the proposed approach can provide more accurate conditional quantile forecasts at high or even extreme quantile levels than existing  ones.
	
	The proposed approach can be improved and extended in the following directions. 
	Firstly, the estimation of the asymptotic covariance matrices for the QR and CQR estimator are complicated  due to the unknown conditional density function. As an alternative to the kernel density estimation, an easy-to-use bootstrap method such as the block bootstrap and random-weight bootstrap may be developed, and asymptotically valid bootstrap inference for the estimated coefficient functions and conditional quantiles can be further studied. Secondly, it is worth investigating whether it is possible to construct a debiased CQR estimator that is provably no less efficient than the proposed biased estimator at high quantile levels.
	Thirdly, the expected shortfall, defined as the expectation of the loss that exceeds the VaR, is another important risk measure. It is also of interest to forecast the ES based on the proposed quantile GARCH model.  
	Lastly, the parametric  method to model the tails based on the flexible Tukey-lambda distribution is  efficient and computationally simple. It can be generalized to other high quantile estimation problems for various data settings. 
	
%
%
	
	\bibliography{quantile}

\begin{thebibliography}{}

\bibitem[Amemiya, 1985]{Amemiya1985}
Amemiya, T. (1985).
\newblock {\em Advanced Econometrics}.
\newblock Harvard University Press.

\bibitem[Andrews, 1991]{Andrews1991}
Andrews, D. W.~K. (1991).
\newblock Heteroskedasticity and autocorrelation consistent covariance matrix
  estimation.
\newblock {\em Econometrica}, 59:817--858.

\bibitem[Artzner et~al., 1999]{Artzner_Delbaen_Eber_Heath1999}
Artzner, P., Delbaen, F., Eber, J.~M., and Heath, D. (1999).
\newblock Coherent measures of risk.
\newblock {\em Mathematical Finance}, 9:203--228.

\bibitem[Baur et~al., 2012]{Baur_Dimpfl_Jung2012}
Baur, D.~G., Dimpfl, T., and Jung, R.~C. (2012).
\newblock Stock return autocorrelations revisited: a quantile regression
  approach.
\newblock {\em Journal of Empirical Finance}, 19:254--265.

\bibitem[Berkes et~al., 2003]{BerkesHorvarthKokoszka2003}
Berkes, I., Horv\'{a}th, L., and Kokoszka, P. (2003).
\newblock {GARCH processes: structure and estimation}.
\newblock {\em Bernoulli}, 9:201--227.

\bibitem[Billingsley, 1961]{Billingsley1961}
Billingsley, P. (1961).
\newblock {The Lindeberg-L\'{e}vy Theorem for Martingales}.
\newblock {\em Proceedings of the American Mathematical Society}, 12:788--792.

\bibitem[Billingsley, 1999]{Billingsley1999}
Billingsley, P. (1999).
\newblock {\em Convergence of Probability Measures}.
\newblock John Wiley \& Sons.

\bibitem[Bollerslev, 1986]{Bollerslev1986}
Bollerslev, T. (1986).
\newblock Generalized autoregressive conditional heteroscedasticity.
\newblock {\em Journal of Econometrics}, 31:307--327.

\bibitem[Cai and Stander, 2008]{Cai_Stander2008}
Cai, Y. and Stander, J. (2008).
\newblock Quantile self-exciting threshold autoregressive time series models.
\newblock {\em Journal of Time Series Analysis}, 29:186--202.

\bibitem[Chernozhukov, 2005]{Chernozhukov2005}
Chernozhukov, V. (2005).
\newblock Extremal quantile regression.
\newblock {\em The Annals of Statistics}, 33:806--839.

\bibitem[Chernozhukov et~al., 2009]{Chernozhukov2009}
Chernozhukov, V., Fern\'{a}ndez-Val, I., and Galichon, A. (2009).
\newblock Improving point and interval estimators of monotone functions by
  rearrangement.
\newblock {\em Biometrika}, 96:559--575.

\bibitem[Chernozhukov et~al., 2010]{Chernozhukov2010}
Chernozhukov, V., Fern\'{a}ndez-Val, I., and Galichon, A. (2010).
\newblock Quantile and probability curves without crossing.
\newblock {\em Econometrica}, 78:1093--1125.

\bibitem[Chernozhukov and Hansen, 2006]{Chernozhukov2006}
Chernozhukov, V. and Hansen, C. (2006).
\newblock Instrumental quantile regression inference for structural and
  treatment effect models.
\newblock {\em Journal of Econometrics}, 132:491--525.

\bibitem[Christoffersen, 1998]{Christoffersen1998}
Christoffersen, P. (1998).
\newblock Evaluating interval forecasts.
\newblock {\em International Economic Review}, 39:841--862.

\bibitem[Davydov et~al., 1998]{Davydov1998}
Davydov, Y.~A., Lifshits, M.~A., and Smorodina, N.~V. (1998).
\newblock {\em Local properties of distributions of stochastic functionals}.
\newblock American Mathematical Society.

\bibitem[Douc et~al., 2008]{Douc_Roueff_Soulier2008}
Douc, R., Roueff, F., and Soulier, P. (2008).
\newblock {On the existence of some ARCH($\infty$) processes}.
\newblock {\em Stochastic Processes and their Applications}, 118:755--761.

\bibitem[Doukhan, 1994]{Doukhan1994}
Doukhan, P. (1994).
\newblock {\em Mixing: Properties and Examples}.
\newblock Lecture Notes in Statistics 85, Berlin: Springer.

\bibitem[Durbin, 1973]{Durbin1973}
Durbin, J. (1973).
\newblock Weak convergence of the sample distribution function when parameters
  are estimated.
\newblock {\em The Annals of Statistics}, 1:279--290.

\bibitem[Engle, 1982]{Engle1982}
Engle, R.~F. (1982).
\newblock Autoregressive conditional heteroscedasticity with estimates of the
  variance of u.k. inflation.
\newblock {\em Econometrica}, 50:987--1007.

\bibitem[Engle and Manganelli, 2004]{Engle_Manganelli2004}
Engle, R.~F. and Manganelli, S. (2004).
\newblock {CAViaR: conditional autoregressive value at risk by regression
  quantiles}.
\newblock {\em Journal of Business and Economic Statistics}, 22:367--381.

\bibitem[Fan and Yao, 2003]{Fan_Yao2003}
Fan, J. and Yao, Q. (2003).
\newblock {\em Nonlinear Time Series: Nonparametric and Parametric Methods}.
\newblock Springer, New York.

\bibitem[Ferreira, 2011]{Ferreira2011}
Ferreira, M.~S. (2011).
\newblock Capturing asymmetry in real exchange rate with quantile
  autoregression.
\newblock {\em Applied Economics}, 43:327--340.

\bibitem[Francq and Zakoian, 2006]{Francq_Zakoian2006}
Francq, C. and Zakoian, J.-M. (2006).
\newblock Mixing properties of a general class of \textsc{GARCH}(1,1) models
  without moment assumptions on the observed process.
\newblock {\em Econometric Theory}, 22:815--834.

\bibitem[Francq and Zakoian, 2010]{Francq_Zakoian2010}
Francq, C. and Zakoian, J.-M. (2010).
\newblock {\em GARCH Models: Structure, Statistical Inference and Financial
  Applications}.
\newblock John Wiley \& Sons, Chichester, UK.

\bibitem[Francq and Zakoian, 2015]{Francq_Zakoian2015}
Francq, C. and Zakoian, J.-M. (2015).
\newblock Risk-parameter estimation in volatility models.
\newblock {\em Journal of Econometrics}, 184:158--174.

\bibitem[Fryzlewicz and Subba~Rao, 2011]{Fryzlewicz_SubbaRao2011}
Fryzlewicz, P. and Subba~Rao, S. (2011).
\newblock {Mixing properties of ARCH and time-varying ARCH processes}.
\newblock {\em Bernoulli}, 17:320--346.

\bibitem[Galvao et~al., 2011]{Galvao2011}
Galvao, A.~F., Montes-Rojas, G., and Olmo, J. (2011).
\newblock Threshold quantile autoregressive models.
\newblock {\em Journal of Time Series Analysis}, 32:253--267.

\bibitem[Gilchrist, 2000]{Gilchrist2000}
Gilchrist, W.~G. (2000).
\newblock {\em Statistical Modelling with Quantile Functions}.
\newblock CRC Press.

\bibitem[Giraitis et~al., 2000]{Giraitis_Kokoszka_Leipus_2000}
Giraitis, L., Kokoszka, P., and Leipus, R. (2000).
\newblock {Stationary ARCH models: dependence structure and central limit
  theorem}.
\newblock {\em Econometric Theory}, 16:3--22.

\bibitem[He et~al., 2020]{He_Hou_Peng_Shen2020}
He, Y., Hou, Y., Peng, L., and Shen, H. (2020).
\newblock Inference for conditional value-at-risk of a predictive regression.
\newblock {\em The Annals of Statistics}, 48:3442--3464.

\bibitem[Hill, 1975]{Hill1975simple}
Hill, B.~M. (1975).
\newblock A simple general approach to inference about the tail of a
  distribution.
\newblock {\em The Annals of Statistics}, \:1163--1174.

\bibitem[Huber, 1967]{Huber1967}
Huber, P. (1967).
\newblock The behavior of maximum likelihood estimates under nonstandard
  conditions.
\newblock {\em Proceedings of the fifth Berkeley symposium on mathematical
  statistics and probability}, 1:221--233.

\bibitem[Joiner and Rosenblatt, 1971]{Joiner_Rosenblatt1971}
Joiner, B.~L. and Rosenblatt, J.~R. (1971).
\newblock Some properties of the range in samples from tukey's symmetric lambda
  distributions.
\newblock {\em Journal of the American Statistical Association}, 66:394--399.

\bibitem[Karian et~al., 1996]{Karian_Dudewicz_Mcdonald1996}
Karian, Z.~A., Dudewicz, E.~J., and Mcdonald, P. (1996).
\newblock The extended generalized lambda distribution system for fitting
  distributions to data: history, completion of theory, tables, applications,
  the ``final word'' on moment fits.
\newblock {\em Communications in Statistics-Simulation and Computation},
  25:611--642.

\bibitem[Knight, 1998]{Knight1998}
Knight, K. (1998).
\newblock Limiting distributions for \text{L$_1$} regression estimators under
  general conditions.
\newblock {\em The Annals of Statistics}, 26:755--770.

\bibitem[Koenker, 2005]{Koenker2005}
Koenker, R. (2005).
\newblock {\em Quantile regression}.
\newblock Cambridge University Press, Cambridge.

\bibitem[Koenker and Xiao, 2006]{Koenker_Xiao2006}
Koenker, R. and Xiao, Z. (2006).
\newblock Quantile autoregression.
\newblock {\em Journal of the American Statistical Association}, 101:980--990.

\bibitem[Koenker and Zhao, 1996]{Koenker_Zhao1996}
Koenker, R. and Zhao, Q. (1996).
\newblock {Conditional quantile estimation and inference for \textsc{ARCH}
  models}.
\newblock {\em Econometric Theory}, 12:793--813.

\bibitem[Kuester et~al., 2006]{Kuester_Mittnik_Paolella2006}
Kuester, K., Mittnik, S., and Paolella, M.~S. (2006).
\newblock {Value-at-risk prediction: A comparison of alternative strategies}.
\newblock {\em Journal of Financial Econometrics}, 4:53--89.

\bibitem[Lee and Noh, 2013]{Lee_Noh2013}
Lee, S. and Noh, J. (2013).
\newblock Quantile regression estimator for \textsc{GARCH} models.
\newblock {\em Scandinavian Journal of Statistics}, 40:2--20.

\bibitem[Li and Wang, 2019]{Li_Wang2019}
Li, D. and Wang, H.~J. (2019).
\newblock Extreme quantile estimation for autoregressive models.
\newblock {\em Journal of Bussiness and Economic Statistics}, 37:661--670.

\bibitem[Ling, 2005]{Ling2005}
Ling, S. (2005).
\newblock Self-weighted least absolute deviation estimation for infinite
  variance autoregressive models.
\newblock {\em Journal of the Royal Statistical Society: Series B},
  67:381--393.

\bibitem[Ling and McAleer, 2003]{Ling_McAleer2003}
Ling, S. and McAleer, M. (2003).
\newblock Asymptotic theory for a vector \textsc{ARMA--GARCH} model.
\newblock {\em Econometric Theory}, 19:280--310.

\bibitem[Nelson, 1991]{nelson1991conditional}
Nelson, D.~B. (1991).
\newblock Conditional heteroskedasticity in asset returns: a new approach.
\newblock {\em Econometrica}, 59:347--370.

\bibitem[Newey, 1991]{Newey1991}
Newey, W.~K. (1991).
\newblock Uniform convergence in probability and stochastic equicontinuity.
\newblock {\em Econometrica}, 59:1161--1167.

\bibitem[Nielsen and No{\"e}l, 2021]{Nielsen_Noel2021}
Nielsen, M.~{\O}. and No{\"e}l, A.~L. (2021).
\newblock {To infinity and beyond: Efficient computation of ARCH($\infty$)
  models}.
\newblock {\em Journal of Time Series Analysis}, 42:338--354.

\bibitem[Phillips, 2015]{Phillips2015}
Phillips, P. C.~B. (2015).
\newblock {Halbert White Jr. memorial JFEC lecture: Pitfalls and possibilities
  in predictive regression}.
\newblock {\em Journal of Financial Econometrics}, 13:521--555.

\bibitem[Pollard, 1985]{Pollard1985}
Pollard, D. (1985).
\newblock New ways to prove central limit theorems.
\newblock {\em Econometric Theory}, 1:295--314.

\bibitem[Royer, 2023]{Royer2022}
Royer, J. (2023).
\newblock Conditional asymmetry in \textsc{P}ower \textsc{ARCH}($\infty$)
  models.
\newblock {\em Journal of Econometrics}, 234:178--204.

\bibitem[Shao, 2011]{Shao2011bootstrap}
Shao, X. (2011).
\newblock A bootstrap-assisted spectral test of white noise under unknown
  dependence.
\newblock {\em Journal of Econometrics}, 162:213--224.

\bibitem[Taylor, 2008]{Taylor2008}
Taylor, S.~J. (2008).
\newblock {\em Modelling financial time series}.
\newblock World Scientific, New York.

\bibitem[Trapani, 2016]{Trapani2016testing}
Trapani, L. (2016).
\newblock Testing for (in)finite moments.
\newblock {\em Journal of Econometrics}, 191:57--68.

\bibitem[Wang et~al., 2022]{Wang_Zhu_Li_Li2022}
Wang, G., Zhu, K., Li, G., and Li, W.~K. (2022).
\newblock Hybrid quantile estimation for asymmetric power \textsc{GARCH}
  models.
\newblock {\em Journal of Econometrics}, 227:264--284.

\bibitem[Wang et~al., 2012]{Wang_Li_He2012}
Wang, H.~J., Li, D., and He, X. (2012).
\newblock Estimation of high conditional quantiles for heavy-tailed
  distributions.
\newblock {\em Journal of the American Statistical Association},
  107:1453--1464.

\bibitem[White and Domowitz, 1984]{White_Domowitz1984}
White, H. and Domowitz, I. (1984).
\newblock Nonlinear regression with dependent observations.
\newblock {\em Econometrica}, 52:143--162.

\bibitem[Wu and Xiao, 2002]{Wu_Xiao2002}
Wu, G. and Xiao, Z. (2002).
\newblock An analysis of risk measures.
\newblock {\em Journal of Risk}, 4:53--75.

\bibitem[Wu and Pourahmadi, 2009]{Wu_Pourahmadi2009}
Wu, W. and Pourahmadi, M. (2009).
\newblock Banding sample autocovariance matrices of stationary processes.
\newblock {\em Statistica Sinica}, 19:1755--1768.

\bibitem[Xiao and Koenker, 2009]{Xiao_Koenker2009}
Xiao, Z. and Koenker, R. (2009).
\newblock Conditional quantile estimation for generalized autoregressive
  conditional heteroscedasticity models.
\newblock {\em Journal of the American Statistical Association},
  104:1696--1712.

\bibitem[Zaffaroni, 2004]{Zaffaroni_2004}
Zaffaroni, P. (2004).
\newblock {Stationarity and memory of ARCH($\infty$) models}.
\newblock {\em Econometric Theory}, 20:147--160.

\bibitem[Zakoian, 1994]{zakoian1994threshold}
Zakoian, J.~M. (1994).
\newblock Threshold heteroskedastic models.
\newblock {\em Journal of Economic Dynamics and Control}, 18:931--955.

\bibitem[Zheng et~al., 2018]{Zheng_Zhu_Li_Xiao2018}
Zheng, Y., Zhu, Q., Li, G., and Xiao, Z. (2018).
\newblock Hybrid quantile regression estimation for time series models with
  conditional heteroscedasticity.
\newblock {\em Journal of the Royal Statistical Society: Series B},
  80:975--993.

\bibitem[Zhu et~al., 1997]{zhu_Byrd_Lu_Nocedal1997}
Zhu, C., Byrd, R.~H., Lu, P., and Nocedal, J. (1997).
\newblock Algorithm 778: \textsc{L-BFGS-B}: Fortran subroutines for large-scale
  bound-constrained optimization.
\newblock {\em ACM Transactions on mathematical software (TOMS)}, 23:550--560.

\bibitem[Zhu and Ling, 2011]{Zhu_Ling2011}
Zhu, K. and Ling, S. (2011).
\newblock Global self-weighted and local quasi-maximum exponential likelihood
  estimators for \textsc{ARMA}-\textsc{GARCH/IGARCH} models.
\newblock {\em The Annals of Statistics}, 39:2131--2163.

\bibitem[Zhu and Li, 2022]{Zhu_Li2022}
Zhu, Q. and Li, G. (2022).
\newblock Quantile double autoregression.
\newblock {\em Econometric Theory}, 38:793--839.

\bibitem[Zhu et~al., 2021]{Zhu_Li_Xiao2021QRGARCHX}
Zhu, Q., Li, G., and Xiao, Z. (2021).
\newblock {Quantile estimation of regression models with GARCH-X errors}.
\newblock {\em Statistica Sinica}, 31:1261--1284.

\bibitem[Zhu et~al., 2018]{Zhu_Zheng_Li2018}
Zhu, Q., Zheng, Y., and Li, G. (2018).
\newblock Linear double autoregression.
\newblock {\em Journal of Econometrics}, 207:162--174.

\bibitem[Zou and Yuan, 2008]{Zou_Yuan2008}
Zou, H. and Yuan, M. (2008).
\newblock Composite quantile regression and the oracle model selection theory.
\newblock {\em The Annals of Statistics}, 36(3):1108--1126.

\end{thebibliography}
	
	\newpage
	\begin{table}[htp]
		\caption{Biases, ESDs and ASDs of the self-weighted QR estimator $\widetilde{\bm\theta}_{wn}(\tau)$ at quantile level $\tau=0.5\%,1\%$ or $5\%$ for DGP \eqref{DGP} with Setting \eqref{sim1coef1}. ASD$_{1}$ and ASD$_{2}$ correspond to the bandwidths $\ell_{B}$ and $\ell_{HS}$, respectively. $F$ is the standard normal distribution $F_N$ or Tukey-lambda distribution $F_T$.}
		\label{tab.estimation.CQE.w2.DGP1}
		\centering
		\begin{tabular}{lrrrrrrrrrrrrrr}
			\hline
			&&   && \multicolumn{5}{c}{$F=F_{N}$} && \multicolumn{5}{c}{$F=F_{T}$}\\
			\cline{5-9}\cline{11-15}
			&& $n$ &&True & Bias & ESD & ASD$_{\text{1}}$ & ASD$_{\text{2}}$&& True & Bias & ESD & ASD$_{\text{1}}$ & ASD$_{\text{2}}$\\ 
			\hline
			&&&&\multicolumn{11}{c}{$\tau=0.5\%$}\\
			$\omega$
			&& 1000&&-0.258 & -0.002 & 0.073 & 0.099 & 0.066 &  & -0.942 & -0.391 & 1.011 & 1.573 & 1.207 \\ 
			&& 2000&&-0.258 & -0.002 & 0.058 & 0.066 & 0.049 &  & -0.942 & -0.204 & 0.721 & 1.055 & 0.787 \\ 
			$\alpha_1$
			&& 1000&&-0.258 & -0.031 & 0.151 & 0.212 & 0.151 &  & -0.942 & -0.159 & 0.609 & 0.808 & 0.598 \\ 
			&& 2000&&-0.258 & -0.027 & 0.128 & 0.144 & 0.109 &  & -0.942 & -0.104 & 0.421 & 0.530 & 0.376 \\ 
			$\beta_1$
			&& 1000&&0.800 & -0.052 & 0.141 & 0.257 & 0.169 &  & 0.800 & -0.043 & 0.114 & 0.169 & 0.118 \\ 
			&& 2000&&0.800 & -0.043 & 0.129 & 0.159 & 0.114 &  & 0.800 & -0.027 & 0.082 & 0.103 & 0.071 \\ 
			\hline
			&&&&\multicolumn{11}{c}{$\tau=1\%$}\\
			$\omega$
			&& 1000&&-0.233 & -0.007 & 0.062 & 0.069 & 0.052 &  & -0.755 & -0.270 & 0.726 & 0.911 & 0.712 \\ 
			&& 2000&&-0.233 & -0.004 & 0.049 & 0.050 & 0.037 &  & -0.755 & -0.178 & 0.510 & 0.629 & 0.495 \\ 
			$\alpha_1$
			&& 1000&&-0.233 & -0.023 & 0.135 & 0.160 & 0.122 &  & -0.755 & -0.095 & 0.377 & 0.472 & 0.332 \\ 
			&& 2000&&-0.233 & -0.016 & 0.098 & 0.111 & 0.084 &  & -0.755 & -0.060 & 0.271 & 0.306 & 0.234 \\ 
			$\beta_1$
			&& 1000&&0.800 & -0.056 & 0.145 & 0.232 & 0.169 &  & 0.800 & -0.033 & 0.093 & 0.118 & 0.084 \\ 
			&& 2000&&0.800 & -0.042 & 0.127 & 0.140 & 0.101 &  & 0.800 & -0.020 & 0.068 & 0.075 & 0.057 \\ 
			\hline
			&&&&\multicolumn{11}{c}{$\tau=5\%$}\\
			$\omega$
			&& 1000&&-0.164 & -0.008 & 0.038 & 0.040 & 0.033 &  & -0.405 & -0.130 & 0.315 & 0.305 & 0.257 \\ 
			&& 2000&&-0.164 & -0.004 & 0.030 & 0.030 & 0.027 &  & -0.405 & -0.063 & 0.202 & 0.218 & 0.185 \\ 
			$\alpha_1$
			&& 1000&&-0.164 & -0.015 & 0.085 & 0.090 & 0.077 &  & -0.405 & -0.029 & 0.135 & 0.144 & 0.118 \\ 
			&& 2000&&-0.164 & -0.008 & 0.060 & 0.063 & 0.057 &  & -0.405 & -0.016 & 0.090 & 0.098 & 0.083 \\ 
			$\beta_1$
			&& 1000&&0.800 & -0.063 & 0.156 & 0.178 & 0.150 &  & 0.800 & -0.022 & 0.065 & 0.069 & 0.055 \\ 
			&& 2000&&0.800 & -0.033 & 0.109 & 0.105 & 0.093 &  & 0.800 & -0.011 & 0.042 & 0.046 & 0.038 \\ 
			\hline
		\end{tabular}
	\end{table}

	\begin{table}[htp]
		\caption{Biases, ESDs and ASDs of the self-weighted QR estimator $\widetilde{\bm\theta}_{wn}(\tau)$ at quantile level $\tau=0.5\%,1\%$ or $5\%$ for DGP \eqref{DGP} with Setting \eqref{sim1coef2}. ASD$_{1}$ and ASD$_{2}$ correspond to the bandwidths $\ell_{B}$ and $\ell_{HS}$, respectively. $F$ is the standard normal distribution $F_N$ or Tukey-lambda distribution $F_T$.}
		\label{tab.estimation.CQE.w2.DGP2}
		\centering
		\begin{tabular}{lrrrrrrrrrrrrrr}
			\hline
			&&   && \multicolumn{5}{c}{$F=F_{N}$} && \multicolumn{5}{c}{$F=F_{T}$}\\
			\cline{5-9}\cline{11-15}
			&& $n$ &&True & Bias & ESD & ASD$_{\text{1}}$ & ASD$_{\text{2}}$&& True & Bias & ESD & ASD$_{\text{1}}$ & ASD$_{\text{2}}$\\ 
			\hline
			&&&&\multicolumn{11}{c}{$\tau=0.5\%$}\\
			$\omega$
			&& 1000&&-0.258 & -0.016 & 0.058 & 0.063 & 0.042 &  & -0.942 & -0.142 & 0.496 & 0.689 & 0.485 \\ 
			&& 2000&&-0.258 & -0.007 & 0.040 & 0.038 & 0.028 &  & -0.942 & -0.100 & 0.365 & 0.440 & 0.321 \\ 
			$\alpha_1$
			&& 1000&&-0.753 & -0.002 & 0.120 & 0.145 & 0.103 &  & -1.437 & -0.053 & 0.476 & 0.621 & 0.477 \\ 
			&& 2000&&-0.753 & -0.006 & 0.088 & 0.096 & 0.074 &  & -1.437 & -0.033 & 0.346 & 0.431 & 0.306 \\ 
			$\beta_1$
			&& 1000&&0.597 & -0.024 & 0.079 & 0.086 & 0.061 &  & 0.597 & -0.028 & 0.112 & 0.149 & 0.109 \\ 
			&& 2000&&0.597 & -0.013 & 0.053 & 0.053 & 0.040 &  & 0.597 & -0.019 & 0.085 & 0.100 & 0.070 \\ 
			\hline
			&&&&\multicolumn{11}{c}{$\tau=1\%$}\\
			$\omega$
			&& 1000&&-0.233 & -0.012 & 0.046 & 0.043 & 0.033 &  & -0.755 & -0.110 & 0.334 & 0.396 & 0.288 \\
			&& 2000&&-0.233 & -0.006 & 0.032 & 0.031 & 0.024 &  & -0.755 & -0.058 & 0.247 & 0.278 & 0.211 \\ 
			$\alpha_1$
			&& 1000&&-0.723 & -0.006 & 0.106 & 0.114 & 0.089 &  & -1.245 & -0.041 & 0.330 & 0.402 & 0.285 \\ 
			&& 2000&&-0.723 & -0.003 & 0.077 & 0.083 & 0.067 &  & -1.245 & -0.018 & 0.240 & 0.266 & 0.202 \\ 
			$\beta_1$
			&& 1000&&0.594 & -0.021 & 0.070 & 0.069 & 0.053 &  & 0.594 & -0.025 & 0.096 & 0.110 & 0.078 \\ 
			&& 2000&&0.594 & -0.010 & 0.049 & 0.048 & 0.038 &  & 0.594 & -0.011 & 0.068 & 0.071 & 0.053 \\ 
			\hline
			&&&&\multicolumn{11}{c}{$\tau=5\%$}\\
			$\omega$
			&& 1000&&-0.164 & -0.006 & 0.033 & 0.032 & 0.029 &  & -0.405 & -0.034 & 0.147 & 0.157 & 0.130 \\ 
			&& 2000&&-0.164 & -0.002 & 0.022 & 0.023 & 0.021 &  & -0.405 & -0.013 & 0.100 & 0.109 & 0.096 \\ 
			$\alpha_1$
			&& 1000&&-0.614 & -0.000 & 0.093 & 0.096 & 0.087 &  & -0.855 & -0.009 & 0.150 & 0.159 & 0.136 \\ 
			&& 2000&&-0.614 & -0.003 & 0.066 & 0.068 & 0.063 &  & -0.855 & -0.006 & 0.105 & 0.110 & 0.098 \\ 
			$\beta_1$
			&& 1000&&0.570 & -0.012 & 0.071 & 0.073 & 0.065 &  & 0.570 & -0.011 & 0.069 & 0.072 & 0.061 \\ 
			&& 2000&&0.570 & -0.006 & 0.049 & 0.051 & 0.046 &  & 0.570 & -0.005 & 0.047 & 0.050 & 0.044 \\ 
			\hline
		\end{tabular}
	\end{table} 
	
	\begin{table}
		\caption{\label{tab_const_test} Rejection rates of the CvM test at the 5\% significance level for $\mathcal{T}=[0.7,0.995]$ and $[0.8,0.995]$, where $b_1$, $b_2$ and $b_3$ correspond to $\lfloor cn^{1/2}\rfloor$ with $c=0.5, 1$ and 2, respectively. $F$ is the standard normal distribution $F_N$ or Tukey-lambda distribution $F_T$.}
		\centering
		\begin{tabular}{cllcccccccc}
			\hline
			&&&& \multicolumn{3}{c}{$F=F_N$} & & \multicolumn{3}{c}{$F=F_T$}\\
			\cline{5-7}\cline{9-11}
			$\mathcal{T}$& $n$&$d$&& $b_1$ & $b_2$ & $b_3$& & $b_1$ & $b_2$ & $b_3$ \\ 
			\hline
			&&   0 && 0.045 & 0.058 & 0.084 &  & 0.028 & 0.041 & 0.056 \\
			&1000&   1 && 0.101 & 0.116 & 0.140 & & 0.098 & 0.122 & 0.167 \\
			$[0.7,0.995]$&&   1.6 && 0.236 & 0.278 & 0.332 & & 0.259 & 0.325 & 0.403 \\ 
			&&   0 && 0.047 & 0.055 & 0.064 &  & 0.034 & 0.049 & 0.062 \\
			&2000&   1 && 0.190 & 0.214 & 0.236 &  & 0.240 & 0.274 & 0.334 \\ 
			&&   1.6 && 0.571 & 0.597 & 0.656 &  & 0.689 & 0.739 & 0.784 \\ 
			\hline
			&&   0 &&  0.036 & 0.046 & 0.071 &  & 0.026 & 0.040 & 0.054 \\ 
			&1000&   1 && 0.071 & 0.093 & 0.124 &  & 0.061 & 0.073 & 0.121 \\
			$[0.8,0.995]$&&   1.6 && 0.169 & 0.223 & 0.284 &  & 0.147 & 0.191 & 0.274 \\
			&&   0 && 0.028 & 0.037 & 0.056 &  & 0.027 & 0.038 & 0.055 \\  
			&2000&   1 && 0.143 & 0.161 & 0.202 &  & 0.131 & 0.173 & 0.213 \\
			&&   1.6 && 0.481 & 0.558 & 0.601 &  & 0.473 & 0.530 & 0.610 \\
			\hline
		\end{tabular}
	\end{table}
	
	\begin{table}[htp]
		\caption{Biases, ESDs and ASDs of the transformed CQR estimator  $\check{\bm\theta}_{wn}^*(\tau)$ with bandwidth $h=0.1$, at quantile level $\tau=0.5\%,1\%$ or $5\%$ for DGP \eqref{DGP} with Setting \eqref{sim1coef1}. ASD$_{a}$, ASD$_{b}$ and ASD$_{c}$ correspond to the optimal, under-smoothing and over-smoothing bandwidths $\widehat{B}_n, 0.1\widehat{B}_n$ and $10\widehat{B}_n$, respectively. $F$ is the standard normal distribution $F_N$ or Tukey-lambda distribution $F_T$.}
		\label{tab.estimation.CQR.DGP1}
		\renewcommand\arraystretch{0.8}
		\centering
		\begin{tabular}{L{0.1cm}R{0cm}R{0.55cm}R{0cm}rrrrrrR{0cm}rrrrrr}
			\hline
			&&   && \multicolumn{6}{c}{$F=F_{N}$} && \multicolumn{6}{c}{$F=F_{T}$}\\
			\cline{5-10}\cline{12-17}
			&& $n$ && True & Bias & ESD & ASD$_{a}$ & ASD$_{b}$ & ASD$_{c}$ && True & Bias & ESD & ASD$_{a}$ & ASD$_{b}$ & ASD$_{c}$\\ 
			\hline
			&&&&\multicolumn{13}{c}{$\tau=0.5\%$}\\
			$\omega$
			&& 1000&&-0.258 &-0.009 & 0.054 & 0.061 & 0.061 & 0.060 &  &-0.942 & -0.331 & 0.697 & 0.642 & 0.640 & 0.639 \\ 
			&& 2000&&-0.258 &-0.004 & 0.041 & 0.045 & 0.045 & 0.045 &  &-0.942 & -0.151 & 0.441 & 0.437 & 0.437 & 0.436 \\ 
			$\alpha_1$
			&& 1000&&-0.258 &-0.020 & 0.122 & 0.127 & 0.127 & 0.125 &  &-0.942 & -0.047 & 0.303 & 0.315 & 0.316 & 0.307 \\ 
			&& 2000&&-0.258 &-0.008 & 0.085 & 0.088 & 0.088 & 0.087 &  &-0.942 & -0.023 & 0.198 & 0.208 & 0.208 & 0.206 \\ 
			$\beta_1$
			&& 1000&&0.800 &-0.059 & 0.144 & 0.153 & 0.153 & 0.150 &  &0.800 & -0.020 & 0.055 & 0.058 & 0.058 & 0.057 \\ 
			&& 2000&&0.800 &-0.029 & 0.098 & 0.095 & 0.095 & 0.095 &  &0.800 & -0.010 & 0.036 & 0.038 & 0.038 & 0.037 \\ 
			\hline
			&&&&\multicolumn{13}{c}{$\tau=1\%$}\\
			$\omega$
			&& 1000&&-0.233 &-0.008 & 0.048 & 0.053 & 0.053 & 0.053 &  &-0.755 & -0.256 & 0.542 & 0.499 & 0.498 & 0.496 \\ 
			&& 2000&&-0.233 &-0.004 & 0.037 & 0.041 & 0.041 & 0.041 &  &-0.755 & -0.117 & 0.348 & 0.343 & 0.343 & 0.342 \\ 
			$\alpha_1$
			&& 1000&&-0.233 &-0.016 & 0.110 & 0.112 & 0.112 & 0.111 &  &-0.755 & -0.037 & 0.229 & 0.238 & 0.239 & 0.232 \\ 
			&& 2000&&-0.233 &-0.008 & 0.078 & 0.080 & 0.080 & 0.079 &  &-0.755 & -0.019 & 0.152 & 0.159 & 0.159 & 0.157 \\ 
			$\beta_1$
			&& 1000&&0.800 &-0.054 & 0.140 & 0.134 & 0.134 & 0.133 &  &0.800 & -0.019 & 0.055 & 0.057 & 0.057 & 0.056 \\ 
			&& 2000&&0.800 &-0.029 & 0.100 & 0.093 & 0.093 & 0.092 &  &0.800 & -0.009 & 0.036 & 0.037 & 0.038 & 0.037 \\ 
			\hline
			&&&&\multicolumn{13}{c}{$\tau=5\%$}\\
			$\omega$
			&& 1000&&-0.164 &-0.007 & 0.036 & 0.053 & 0.053 & 0.052 &  &-0.405 & -0.114 & 0.281 & 0.249 & 0.249 & 0.246 \\ 
			&& 2000&&-0.164 &-0.004 & 0.028 & 0.030 & 0.030 & 0.030 &  &-0.405 & -0.051 & 0.182 & 0.175 & 0.175 & 0.174 \\ 
			$\alpha_1$
			&& 1000&&-0.164 &-0.015 & 0.086 & 0.121 & 0.119 & 0.117 &  &-0.405 & -0.016 & 0.112 & 0.114 & 0.114 & 0.110 \\ 
			&& 2000&&-0.164 &-0.007 & 0.060 & 0.060 & 0.060 & 0.059 &  &-0.405 & -0.009 & 0.078 & 0.078 & 0.078 & 0.077 \\ 
			$\beta_1$
			&& 1000&&0.800 &-0.063 & 0.160 & 0.304 & 0.298 & 0.292 &  &0.800 & -0.016 & 0.055 & 0.054 & 0.054 & 0.052 \\ 
			&& 2000&&0.800 &-0.033 & 0.112 & 0.100 & 0.100 & 0.099 &  &0.800 & -0.008 & 0.035 & 0.036 & 0.036 & 0.035 \\ 
			\hline
		\end{tabular}
	\end{table}
	
	\begin{table}[htp]
		\caption{Biases, ESDs and ASDs of the transformed CQR estimator $\check{\bm\theta}_{wn}^*(\tau)$ with bandwidth $h=0.1$, at quantile level $\tau=0.5\%,1\%$ or $5\%$ for DGP \eqref{DGP} with Setting \eqref{sim1coef2}. ASD$_{a}$, ASD$_{b}$ and ASD$_{c}$ correspond to the optimal, under-smoothing, and over-smoothing bandwidths $\widehat{B}_n, 0.1\widehat{B}_n$ and $10\widehat{B}_n$, respectively. $F$ is the standard normal distribution $F_N$ or Tukey-lambda distribution $F_T$.}
		\label{tab.estimation.CQR.DGP2}
		\renewcommand\arraystretch{0.8}
		\centering
		\begin{tabular}{L{0.1cm}R{0cm}R{0.55cm}R{0cm}rrrrrrR{0cm}rrrrrr}
			\hline
			&&   && \multicolumn{6}{c}{$F=F_{N}$} && \multicolumn{6}{c}{$F=F_{T}$}\\
			\cline{5-10}\cline{12-17}
			&& $n$ && True & Bias & ESD & ASD$_{a}$ & ASD$_{b}$ & ASD$_{c}$ && True & Bias & ESD & ASD$_{a}$ & ASD$_{b}$ & ASD$_{c}$\\ 
			\hline
			&&&&\multicolumn{13}{c}{$\tau=0.5\%$}\\
			$\omega$
			&& 1000&&-0.258&0.017 & 0.043 & 0.040 & 0.040 & 0.040 &  &-0.942& 0.126 & 0.260 & 0.244 & 0.243 & 0.240 \\ 
			&& 2000&&-0.258&0.022 & 0.029 & 0.028 & 0.028 & 0.028 &  &-0.942& 0.162 & 0.178 & 0.169 & 0.169 & 0.167 \\ 
			$\alpha_1$
			&& 1000&&-0.753&-0.115 & 0.122 & 0.115 & 0.115 & 0.114 &  &-1.437& -0.174 & 0.295 & 0.277 & 0.277 & 0.270 \\ 
			&& 2000&&-0.753&-0.119 & 0.085 & 0.081 & 0.081 & 0.081 &  &-1.437& -0.177 & 0.202 & 0.192 & 0.192 & 0.190 \\ 
			$\beta_1$
			&& 1000&&0.597&-0.045 & 0.067 & 0.062 & 0.062 & 0.061 &  & 0.597& -0.042 & 0.063 & 0.059 & 0.059 & 0.058 \\ 
			&& 2000&&0.597&-0.038 & 0.045 & 0.043 & 0.043 & 0.042 &  & 0.597& -0.037 & 0.042 & 0.041 & 0.041 & 0.040 \\ 
			\hline
			&&&&\multicolumn{13}{c}{$\tau=1\%$}\\
			$\omega$
			&& 1000&&-0.233&0.009 & 0.040 & 0.037 & 0.037 & 0.037 &  &-0.755& 0.062 & 0.217 & 0.202 & 0.202 & 0.199 \\ 
			&& 2000&&-0.233&0.014 & 0.027 & 0.026 & 0.026 & 0.026 &  &-0.755& 0.093 & 0.148 & 0.140 & 0.140 & 0.139 \\ 
			$\alpha_1$
			&& 1000&&-0.723&-0.094 & 0.114 & 0.110 & 0.110 & 0.109 &  &-1.245& -0.146 & 0.239 & 0.228 & 0.228 & 0.222 \\ 
			&& 2000&&-0.723&-0.096 & 0.080 & 0.078 & 0.078 & 0.077 &  &-1.245& -0.147 & 0.163 & 0.158 & 0.158 & 0.156 \\ 
			$\beta_1$
			&& 1000&& 0.594&-0.044 & 0.069 & 0.063 & 0.063 & 0.062 &  & 0.594& -0.041 & 0.063 & 0.059 & 0.059 & 0.058 \\ 
			&& 2000&& 0.594&-0.037 & 0.046 & 0.044 & 0.044 & 0.043 &  & 0.594& -0.036 & 0.043 & 0.041 & 0.041 & 0.040 \\ 
			\hline
			&&&&\multicolumn{13}{c}{$\tau=5\%$}\\
			$\omega$
			&& 1000&&-0.164&-0.003 & 0.034 & 0.032 & 0.032 & 0.032 &  &-0.405& -0.009 & 0.138 & 0.125 & 0.125 & 0.123 \\ 
			&& 2000&&-0.164&0.000 & 0.023 & 0.023 & 0.023 & 0.023 &  &-0.405& 0.007 & 0.092 & 0.087 & 0.088 & 0.087 \\ 
			$\alpha_1$
			&& 1000&&-0.614&-0.047 & 0.104 & 0.105 & 0.105 & 0.103 &  &-0.855& -0.070 & 0.145 & 0.147 & 0.148 & 0.143 \\ 
			&& 2000&&-0.614&-0.049 & 0.073 & 0.074 & 0.074 & 0.074 &  &-0.855& -0.072 & 0.102 & 0.104 & 0.104 & 0.102 \\ 
			$\beta_1$
			&& 1000&& 0.570&-0.043 & 0.085 & 0.076 & 0.076 & 0.075 &  & 0.570& -0.038 & 0.069 & 0.063 & 0.063 & 0.061 \\ 
			&& 2000&& 0.570&-0.034 & 0.055 & 0.053 & 0.053 & 0.052 &  & 0.570& -0.033 & 0.046 & 0.043 & 0.043 & 0.043 \\ 
			\hline
		\end{tabular}
	\end{table}
	
	\begin{figure}[htp]
		\centering
		\includegraphics[width=6.4in]{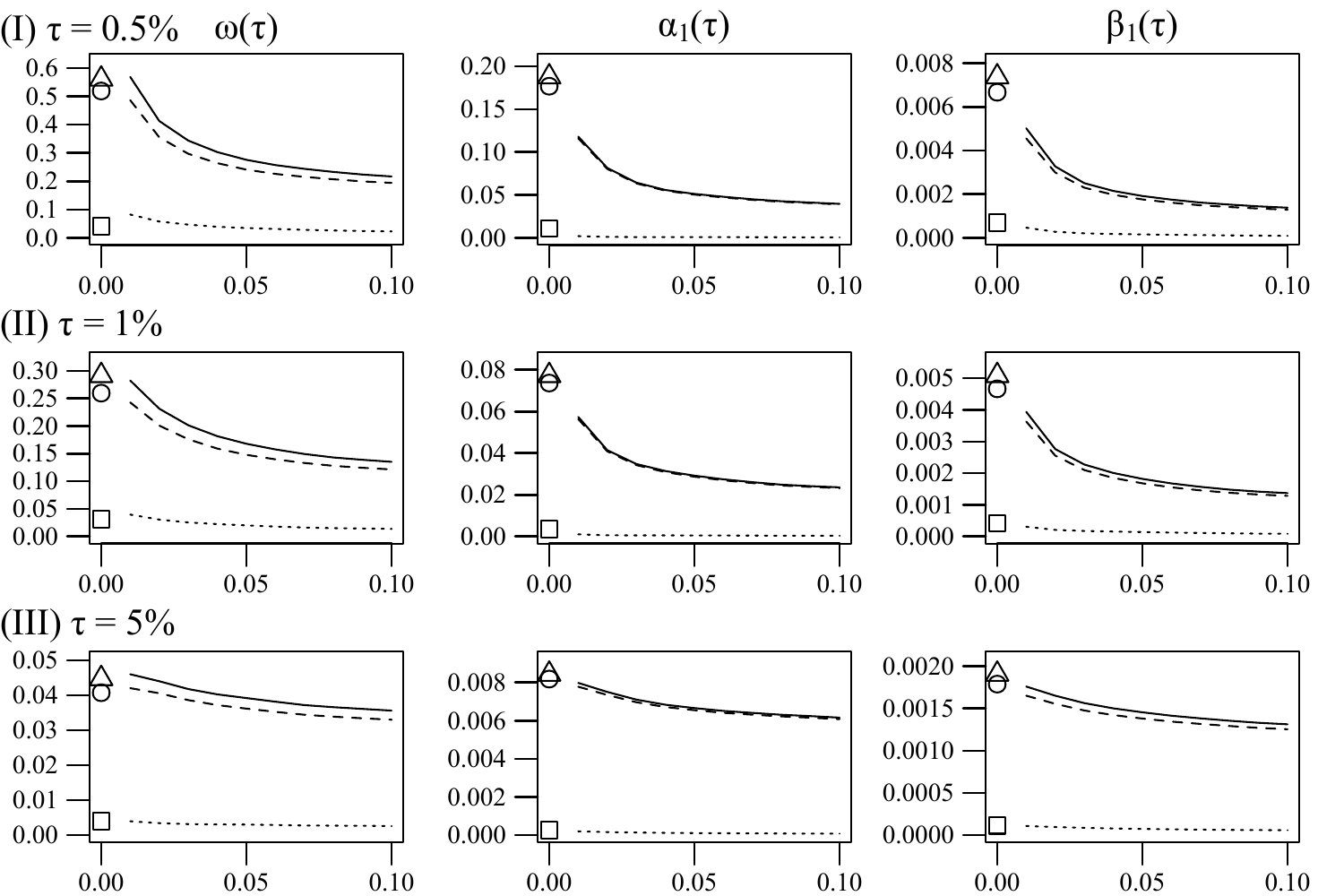}
		\caption{\label{fig_CQR_h_DGP1} Empirical squared bias (dotted line), variance (dashed line) and MSE (solid line) of the transformed CQR estimator $\check{\bm\theta}_{wn}^*(\tau)$ versus the bandwidth $h$ at quantile level $\tau=0.5\%,1\%$ or $5\%$ for DGP \eqref{DGP} with Setting \eqref{sim1coef1} and $F$ being the Tukey-lambda distribution $F_T$. Empirical squared bias (square), variance (circle) and MSE (triangle) of the QR estimator are also labeled at $h=0$ for comparison.}
	\end{figure}

	\begin{figure}[htp]
		\centering
		\includegraphics[width=6.4in]{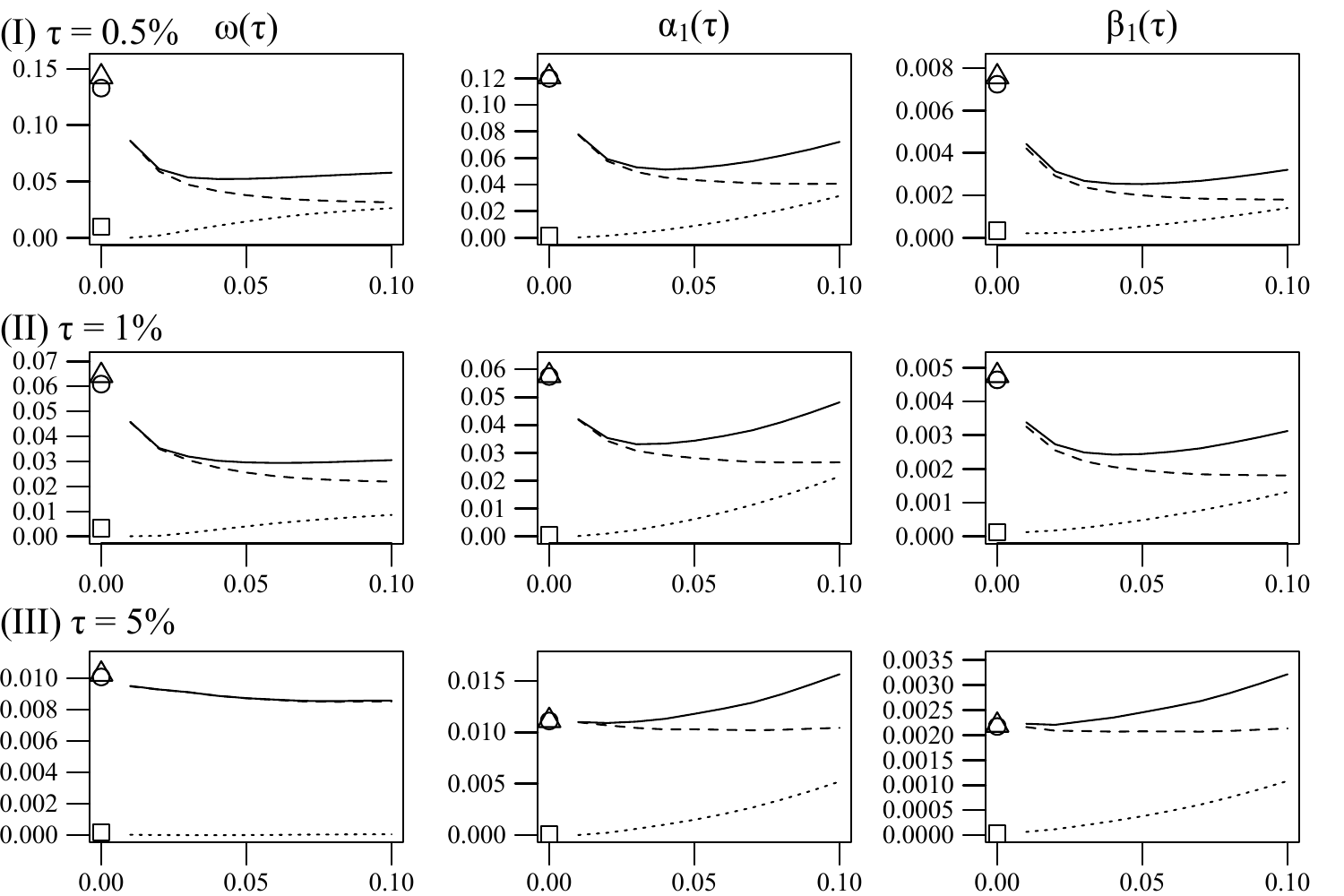}
		\caption{\label{fig_CQR_h_DGP2} Empirical squared bias (dotted line), variance (dashed line) and MSE (solid line) of the transformed CQR estimator $\check{\bm\theta}_{wn}^*(\tau)$ versus the bandwidth $h$ at quantile level $\tau=0.5\%,1\%$ or $5\%$ for DGP \eqref{DGP} with Setting \eqref{sim1coef2} and $F$ being the Tukey-lambda distribution $F_T$. Empirical squared bias (square), variance (circle) and MSE (triangle) of the QR estimator are also labeled at $h=0$ for comparison.}
	\end{figure}

	\begin{table}[htp]
		\caption{\label{tab.qt.forecast.comparison}Biases and RMSEs for conditional quantile estimates of the QR and CQR with bandwidth $h=0.1$, at quantile level $\tau=0.5\%,1\%$ or $5\%$ for DGP \eqref{DGP} with Settings \eqref{sim1coef1} and \eqref{sim1coef2}. $F$ is the standard normal distribution $F_N$ or Tukey-lambda distribution $F_T$.}		
		\renewcommand\arraystretch{0.8}
		\centering
		\begin{tabular}{llllllrrrrrrrrrrr}
			\hline
			&&   &&	&& \multicolumn{5}{c}{DGP1} && \multicolumn{5}{c}{DGP2}\\
			\cline{7-11}\cline{13-17}	
			&&&&&&	\multicolumn{2}{c}{Bias}	&&	\multicolumn{2}{c}{RMSE}	&&	\multicolumn{2}{c}{Bias}	&&	\multicolumn{2}{c}{RMSE}\\	
			\cline{7-8}\cline{10-11}\cline{13-14}\cline{16-17}
			$F$	&&	$n$	&&	Method	&&	In &	Out	&&	In &	Out&&	In &	Out&&	In &	Out\\
			\hline
			&&&&&&\multicolumn{11}{c}{$\tau=0.5\%$}\\
			$F_N$	
			&& 1000&&QR&  & 0.000 & -0.001 &  & 0.039 & 0.038 &  & 0.003 & 0.001 &  & 0.057 & 0.062 \\ 
			&& 1000&&CQR&  & 0.006 & 0.005 &  & 0.034 & 0.034 &  & -0.004 & -0.005 &  & 0.060 & 0.062 \\ 
			&& 2000&&QR&  & 0.000 & 0.000 &  & 0.029 & 0.029 &  & 0.001 & 0.000 &  & 0.040 & 0.038 \\ 
			&& 2000&&CQR&  & 0.003 & 0.003 &  & 0.024 & 0.023 &  & -0.007 & -0.005 &  & 0.049 & 0.053 \\ 
			$F_T$	
			&& 1000&&QR&  & -0.423 & -0.544 &  & 14.818 & 10.154 &  & -0.007 & 0.358 &  & 24.762 & 13.228 \\ 
			&& 1000&&CQR&  & -0.265 & -0.101 &  & 9.607 & 3.810 &  & -0.052 & 0.169 &  & 8.168 & 6.090 \\ 
			&& 2000&&QR&  & -0.240 & -0.181 &  & 8.028 & 6.216 &  & 0.020 & 0.038 &  & 16.072 & 1.978 \\ 
			&& 2000&&CQR&  & -0.089 & -0.082 &  & 5.500 & 2.782 &  & -0.073 & -0.014 &  & 4.051 & 2.115 \\ 
			\hline
			&&&&&&\multicolumn{11}{c}{$\tau=1\%$}\\
			$F_N$	
			&& 1000&&QR&  & 0.000 & -0.001 &  & 0.031 & 0.032 &  & 0.001 & 0.001 &  & 0.048 & 0.050 \\ 
			&& 1000&&CQR&  & 0.004 & 0.003 &  & 0.029 & 0.028 &  & -0.004 & -0.004 &  & 0.053 & 0.053 \\ 
			&& 2000&&QR&  & 0.000 & -0.000 &  & 0.023 & 0.022 &  & 0.000 & 0.001 &  & 0.034 & 0.034 \\ 
			&& 2000&&CQR&  & 0.002 & 0.001 &  & 0.020 & 0.019 &  & -0.005 & -0.004 &  & 0.042 & 0.046 \\ 
			$F_T$	
			&& 1000&&QR&  & -0.338 & -0.401 &  & 6.027 & 7.031 &  & -0.022 & 0.052 &  & 17.539 & 6.121 \\ 
			&& 1000&&CQR&  & -0.235 & -0.096 &  & 6.139 & 2.938 &  & -0.105 & 0.059 &  & 7.412 & 3.842 \\ 
			&& 2000&&QR&  & -0.157 & -0.061 &  & 4.591 & 2.476 &  & -0.013 & 0.048 &  & 9.131 & 1.705 \\ 
			&& 2000&&CQR&  & -0.067 & -0.059 &  & 4.096 & 2.003 &  & -0.107 & -0.064 &  & 3.838 & 1.690 \\
			\hline
			&&&&&&\multicolumn{11}{c}{$\tau=5\%$}\\
			$F_N$	
			&& 1000&&QR&  & -0.001 & -0.001 &  & 0.019 & 0.019 &  & -0.000 & 0.001 &  & 0.038 & 0.040 \\ 
			&& 1000&&CQR&  & 0.001 & 0.001 &  & 0.020 & 0.019 &  & -0.000 & 0.001 &  & 0.042 & 0.043 \\ 
			&& 2000&&QR&  & -0.000 & -0.001 &  & 0.014 & 0.014 &  & -0.000 & -0.001 &  & 0.027 & 0.031 \\ 
			&& 2000&&CQR&  & 0.000 & -0.000 &  & 0.014 & 0.014 &  & -0.002 & -0.001 &  & 0.031 & 0.033 \\ 
			$F_T$	
			&& 1000&&QR&  & -0.151 & -0.019 &  & 2.454 & 1.652 &  & -0.042 & 0.052 &  & 4.947 & 1.792 \\ 
			&& 1000&&CQR&  & -0.079 & 0.035 &  & 1.922 & 1.853 &  & -0.052 & 0.007 &  & 7.302 & 1.268 \\ 
			&& 2000&&QR&  & -0.057 & -0.052 &  & 1.604 & 0.925 &  & -0.029 & -0.020 &  & 2.204 & 0.913 \\ 
			&& 2000&&CQR&  & -0.038 & -0.015 &  & 1.387 & 0.871 &  & -0.060 & -0.035 &  & 3.092 & 0.889 \\ 
			\hline
		\end{tabular}
	\end{table}
	
	\begin{figure}[htp]
		\centering
		\includegraphics[width=6.4in]{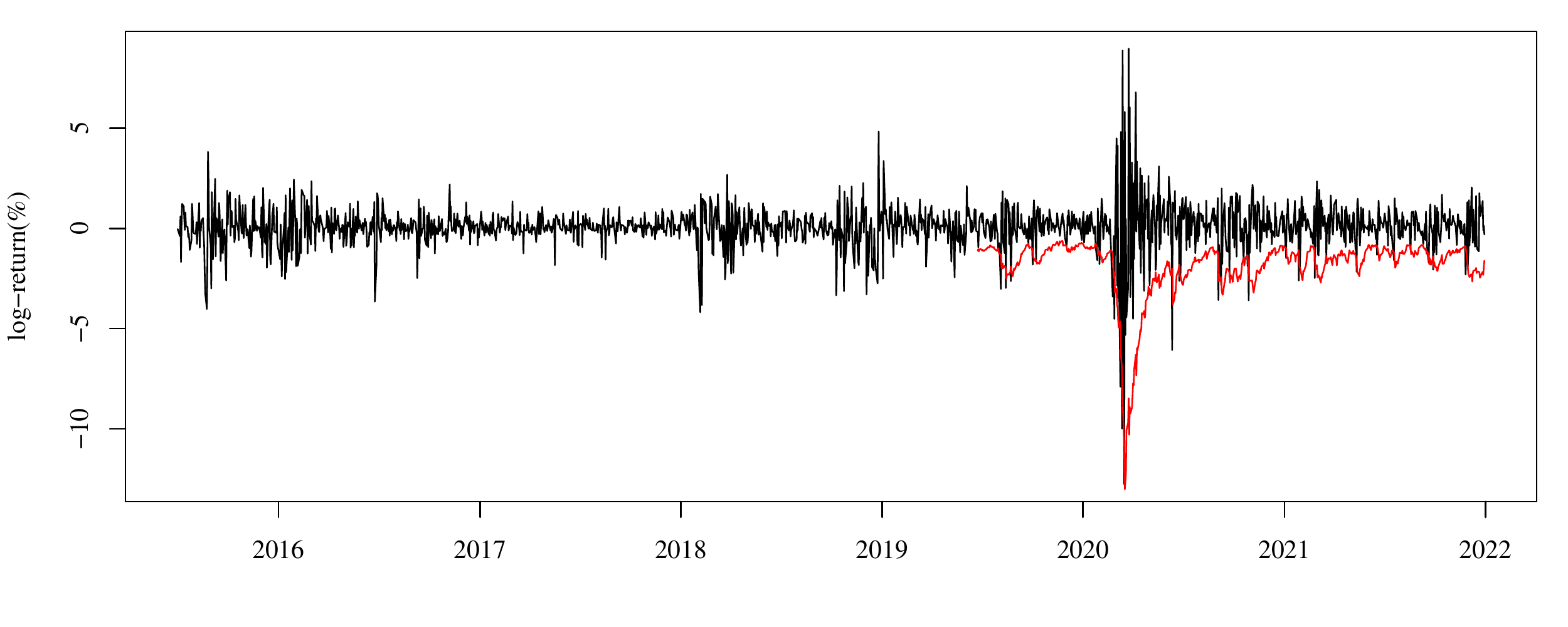}
		\caption{\label{fig_real_data}  Time plot for daily log returns in percentage (black line) of S\&P500 Index from July 2, 2015 to December 30, 2021, with negative 5\% VaR forecasts (red line) from June 24, 2019 to December 30, 2021.}
	\end{figure}
	
	\begin{table}
		\caption{\label{table_real_data_stat} Summary statistics for S\&P500 returns.}
		\begin{center}
			\begin{tabular}{ccccccc}
				\hline
				Mean& Median & Std.Dev. & Skewness & Kurtosis & Min & Max\\
				\hline
				0.051 & 0.074 & 1.161 & -1.053 & 23.721 &-12.765 & 8.968\\
				\hline
			\end{tabular}
		\end{center}
	\end{table}
	
	\begin{figure}[htp]
		\centering
		\includegraphics[width=6.0in]{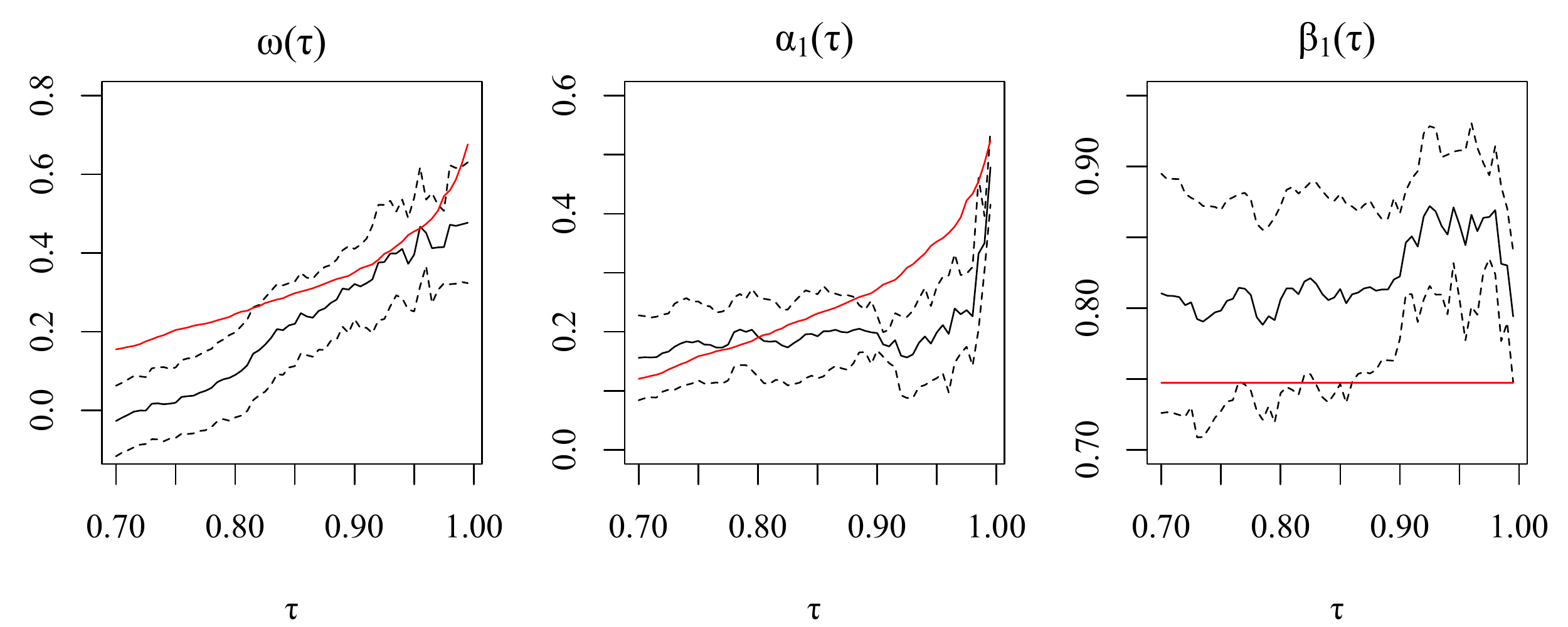}
		\caption{\label{coef_comparison_ARCHinfty}  Self-weighted QR estimates of $\bm\theta(\tau)=(\omega(\tau),\alpha_1(\tau),\beta_1(\tau))^{\prime}$ (black solid), together with their 95\% confidence intervals (black dotted) at $\tau_k=k/200$ with $140\leq k \leq 199$, and estimates of $\bm\theta_{\tau}=(a_0Q_{\tau}(\varepsilon_t)/(1-b_1),a_1Q_{\tau}(\varepsilon_t),b_1)$ (red solid) for the linear ARCH($\infty$) model in \eqref{lgarch11} using the FHS method. }
	\end{figure}
	
	\begin{table}
		\centering
		\caption{\label{tabForecasting1} Empirical coverage rates (ECRs) in percentage, prediction errors (PEs) and $p$-values for correct conditional coverage (CC) and the dynamic quantile (DQ) tests for five estimation methods at lower and upper $1\%, 2.5\%, 5\%$ quantile levels.
			The ECR closest to the nominal level $\tau$ and the smallest PE are marked in bold.}
		\begin{tabular}{llllR{1.2cm}R{1.2cm}R{1.2cm}R{1.2cm}R{1.2cm}}
			\hline
			$\tau$ &&	&&	\multicolumn{1}{c}{QR}	&	\multicolumn{1}{c}{FHS}&	\multicolumn{1}{c}{XK}	&	\multicolumn{1}{c}{Hybrid}	&	\multicolumn{1}{c}{CAViaR}	\\
			\hline
			$1\%$	
			&&	ECR 	&& 1.26 & \textbf{1.10} & 1.26 & 1.26 & 1.26\\
			&&	PE		&& 0.65 & \textbf{0.25} & 0.65 & 0.65 & 0.65\\
			&&	CC	test&& 0.74 & 0.90 & 0.74 & 0.74 & 0.74\\
			&&	DQ	test&& 0.96 & 0.99 & 0.05 & 0.96 & 0.96\\
			\hline
			$2.5\%$
			&&	ECR  	&&  \textbf{2.98} & 3.30 & 3.14 & 3.14 & 3.45\\
			&&	PE		&&  \textbf{0.78} & 1.29 & 1.03 & 1.03 & 1.54\\
			&&	CC	test&&  0.42 & 0.44 & 0.32 & 0.32 & 0.33\\
			&&	DQ	test&&  0.75 & 0.12 & 0.26 & 0.72 & 0.72\\
			\hline
			$5\%$		
			&&	ECR 	&&  6.12 & 6.12 & 6.12 & \textbf{5.65} & 5.97\\
			&&	PE		&&  1.30 & 1.30 & 1.30 & \textbf{0.75} & 1.12\\
			&&	CC	test&&  0.42 & 0.27 & 0.27 & 0.32 & 0.30\\
			&&	DQ	test&&  0.01 & 0.05 & 0.00 & 0.36 & 0.58\\
			\hline
			$95\%$		
			&&	ECR 	&&  94.51 & 95.92 & 92.94 & 94.19 & \textbf{95.45}\\
			&&	PE		&&  0.57 & 1.06 & 2.39 & 0.94 & \textbf{0.52}\\
			&&	CC	test&&  0.11 & 0.18 & 0.06 & 0.07 & 0.22\\
			&&	DQ	test&& 0.62 & 0.70 & 0.08 & 0.44 & 0.67\\
			\hline
			$97.5\%$	
			&&	ECR 	&& \textbf{97.65} & 97.96 & 96.86 & 97.80 & \textbf{97.65}\\
			&&	PE		&& \textbf{0.23} & 0.74 & 1.03 & 0.49 & \textbf{0.23}\\
			&&	CC	test&& 0.68 & 0.57 & 0.55 & 0.64 & 0.68\\
			&&	DQ	test&& 0.68 & 0.72 & 0.64 & 0.79 & 0.90\\
			\hline
			$99\%$		
			&&	ECR 	&& \textbf{99.06} & 98.90 & 99.22 & 98.90 & 98.74\\
			&&	PE		&& \textbf{0.15} & 0.25 & 0.55 & 0.25 & 0.65\\
			&&	CC	test&& 0.93 & 0.90 & 0.82 & 0.90 & 0.74\\
			&&	DQ	test&& 1.00 & 0.99 & 0.99 & 0.99 & 0.96\\
			\hline
		\end{tabular}
	\end{table}

	\begin{table}
		\centering
		\caption{\label{tabForecasting2} Empirical coverage rates (ECRs) in percentage and prediction errors (PEs) for six estimation methods at lower and upper $0.1\%, 0.25\%, 0.5\%$  quantile levels. CQR represents the composite quantile regression with the optimal $h$ by minimizing the check loss in \eqref{optimal-bandwidth-check} for the validation set.
			The ECR closest to the nominal level $\tau$ and the smallest PE are marked in bold.}
		\begin{tabular}{llllR{1.2cm}R{1.2cm}R{1.2cm}R{1.2cm}R{1.2cm}R{1.2cm}}
			\hline
			$\tau$&&	&&	\multicolumn{1}{c}{CQR}	&	\multicolumn{1}{c}{QR}&	\multicolumn{1}{c}{FHS} &	\multicolumn{1}{c}{XK}	&	\multicolumn{1}{c}{Hybrid}	&	\multicolumn{1}{c}{CAViaR}	\\
			\hline
			$0.1\%$	
			&&	ECR 	&& \textbf{0.15} & 0.27 & 0.32 & 0.70 & 0.62 & 0.45\\
			&&	PE		&& \textbf{1.00} & 3.50 & 4.50 & 12.01 & 10.51 & 7.00\\
			\hline
			$0.25\%$	
			&&	ECR		&& \textbf{0.35} & 0.55 & 0.52 & 0.90 & 0.80 & 0.62\\
			&&	PE		&& \textbf{1.27} & 3.80 & 3.48 & 8.23 & 6.97 & 4.75\\
			\hline
			$0.5\%$	
			&&	ECR		&& \textbf{0.75} & 0.90 & 0.78 & 1.23 & 1.12 & 0.88\\
			&&	PE		&& \textbf{2.24} & 3.59 & 2.47 & 6.50 & 5.60 & 3.36\\
			\hline
			$99.5\%$	
			&&	ECR		&& 99.48 & 99.42 & 99.40 & 99.33 & \textbf{99.47} & 99.35\\
			&&	PE		&& \textbf{0.22} & 0.67 & 0.90 & 1.57 & \textbf{0.22} & 1.34\\
			\hline
			$99.75\%$	
			&&	ECR		&& \textbf{99.70} & 99.60 & \textbf{99.70} & 99.45 & 99.62 & 99.62\\
			&&	PE		&& \textbf{0.63} & 1.90 & \textbf{0.63} & 3.80 & 1.58 & 1.58\\
			\hline
			$99.9\%$	
			&&	ECR		&& \textbf{99.85}  & 99.78 & 99.83 & 99.60 & 99.80 & 99.72\\
			&&	PE		&& \textbf{1.00} & 2.50 & 1.50 & 6.00 & 2.00 & 3.50\\  
			\hline
		\end{tabular}
	\end{table}
	
	\clearpage
	\appendix
	\setcounter{figure}{0}
	\setcounter{table}{0}
	\renewcommand{\theassum}{A.\arabic{assum}}
	\renewcommand{\thethm}{A.\arabic{thm}}
	\renewcommand{\thelemma}{A.\arabic{lemma}}
	\renewcommand{\thefigure}{A.\arabic{figure}}
	\renewcommand{\thetable}{A.\arabic{table}}
	
	\section*{Appendix}
	
	This appendix presents the generalization of results for the quantile GARCH($1,1$) model to the quantile GARCH($p,q$) model and the estimation of $\Sigma_w^*$ in Theorem \ref{thm-WCQR}. 
	It also provides notation, technical details for Theorems \ref{thm-stationarity11}--\ref{thm-WCQR} as well as Corollary \ref{thm-WCQE}, and introduces Lemmas \ref{lem00}--\ref{lem-Tightness2} which give some preliminary results for proving Theorems \ref{thm-WCQE-uniform-consistency}--\ref{thm-WCQE-weak-convergence} and Corollary \ref{thm-WCQE}, and Lemmas \ref{lem0-Tukey}--\ref{lem3-Tukey} for Theorem \ref{thm-WCQR}. 
	Moreover, additional results for simulation and empirical analysis are also included in this appendix. 
	Throughout the appendix, the notation $C$ is a generic constant which may take different values in different locations, and $\rho \in (0,1)$ is a generic constant which may take different values at its different occurrences. 

	\section{General results for quantile GARCH($p,q$) models}\label{Sec-QGARCHpq}
	
	\subsection{Proposed quantile GARCH$(p,q)$ model}
	
	This section extends the quantile GARCH$($1,1$)$ model and the methods in Sections \ref{subsec-WCQE} and \ref{Sec-CQR} to the general GARCH$(p,q)$ setting. The linear GARCH($p, q$) model \citep{Taylor2008} is given by
	\begin{equation*}
		y_t= h_t\varepsilon_t, \hspace{5mm} h_t=a_0+\sum_{i=1}^qa_i|y_{t-i}|+\sum_{j=1}^pb_j h_{t-j},
	\end{equation*}
	where $a_0>0$, $a_i\geq 0$ for $1\leq i \leq q$, $b_j\geq 0$ for $1\leq j \leq p$, and the innovations $\{\varepsilon_t\}$ are $i.i.d.$ random variables with mean zero and variance one. If $\{y_t\}$ is strictly stationary, then $\sum_{j=1}^{p}b_j<1$, and the process has the linear ARCH($\infty$) representation,
	\begin{equation}\label{lgarch1}
		y_t=\varepsilon_t\left[\omega+\sum_{j=1}^{\infty}\gamma_j(a_1, \dots, a_q, b_1, \dots, b_p)|y_{t-j}|\right],
	\end{equation}
	where $\omega=a_0(1-\sum_{j=1}^{p}b_j)^{-1}$, the functions $\gamma_j(\cdot)$'s are defined on $\mathbb{R}^q\times D$ with $D=\{(d_1, \dots, d_p)^\prime\in \mathbb{R}^p: \sum_{j=1}^{p}d_j<1, \min_{1\leq j \leq p}d_j\geq 0\}$, such that for any $(c_1, \dots, c_q)^\prime \in\mathbb{R}^q$ and $(d_1, \dots, d_p)^\prime\in D$, it holds that
	\begin{equation}\label{gamma}
		\sum_{j=1}^{\infty}\gamma_j(c_1, \dots, c_q, d_1, \dots, d_p)z^j=\frac{\sum_{i=1}^{q}{c_i z^i}}{1-\sum_{j=1}^{p}d_jz^j}, \quad |z| \leq 1.
	\end{equation}
	
	Motivated by \eqref{lgarch1}, we define the quantile GARCH($p,q$) model as follows:
	\begin{equation}\label{qgarch4}
		Q_\tau(y_t|\mathcal{F}_{t-1})= \omega(\tau)+ \sum_{j=1}^\infty \gamma_j (\alpha_1(\tau), \dots, \alpha_q(\tau), \beta_1(\tau), \dots, \beta_p(\tau))|y_{t-j}|,
	\end{equation}
	or equivalently,
	\begin{equation}\label{qgarch}
		y_t= \omega(U_t) + \sum_{j=1}^\infty \gamma_j (\alpha_1(U_t), \dots, \alpha_q(U_t), \beta_1(U_t), \dots, \beta_p(U_t))|y_{t-j}|,
	\end{equation}
	where $\omega: (0, 1)\rightarrow\mathbb{R}$ and $\alpha_i: (0, 1)\rightarrow\mathbb{R}$ are unknown monotonic increasing functions, and $\beta_k: (0, 1)\rightarrow [0, 1)$ is a non-negative real-valued function, for $1\leq i\leq q$ and $1\leq k\leq p$, with $\sum_{j=1}^{p}\beta_j(\cdot)< 1$.
	In particular, the quantile GARCH($1,1$) model has the form of \eqref{qgarch11}.
	
	The quantile GARCH($p,q$) model in \eqref{qgarch4} or \eqref{qgarch} also requires condition \eqref{location11} for its identifiability. By \eqref{gamma} and Lemma 2.1 of \cite{BerkesHorvarthKokoszka2003}, we can verify that condition \eqref{location11} holds if and only if
	\begin{equation}\label{loc2}
		\omega(0.5)=\alpha_1(0.5)=\cdots=\alpha_q(0.5)=0.
	\end{equation}
	Hence, we impose \eqref{loc2} for the quantile GARCH($p,q$) model. As in Section \ref{sec:model}, we refrain from imposing  any monotonicity constraint on $\beta_j(\cdot)$'s  to avoid restricting the flexibility of the functions. Nonetheless,  the monotonicity of the right side of \eqref{qgarch4} in $\tau$ is guaranteed if  $\beta_1(\cdot), \ldots, \beta_p(\cdot)$ are monotonic decreasing on $(0, 0.5)$ and monotonic increasing on $(0.5, 1)$; see Remark \ref{remark-monotonicity}.
	
	Moreover, we can write the quantile GARCH($p,q$) model in the form of
	\begin{align*}
		y_t &= \text{sgn}(U_t-0.5)|y_t|,\\
		|y_t|&= |\omega(U_t) |+ \sum_{j=1}^\infty \gamma_j (|\alpha_1(U_t)|, \dots, |\alpha_q(U_t)|, \beta_1(U_t), \dots, \beta_p(U_t)) |y_{t-j}|.
	\end{align*}
	Then the quantile GARCH($p,q$) model is equivalent to 
	\begin{align}\label{equ-qgarch-general}
		y_t = \text{sgn}(U_t-0.5)|y_t|, \quad |y_t|= \phi_{0, t} + \sum_{j=1}^{\infty}\phi_{j,t}|y_{t-j}|, \quad j\geq1,
	\end{align}
	where $\phi_{0, t}=|\omega(U_t)|$ and $\phi_{j, t}=\gamma_j (|\alpha_1(U_t)|, \dots, |\alpha_q(U_t)|, \beta_1(U_t), \dots, \beta_p(U_t))$ for $j\geq1$. 
	This enables us to establish a sufficient condition for the existence of a strictly stationary solution of the quantile GARCH($p,q$) model in the following theorem.
	
	\begin{thm}
		Suppose condition \eqref{loc2} holds. If there exists $s\in(0,1]$ such that 
		\begin{equation*}
			E(\phi_{0, t}^s)<\infty \quad\text{and}\quad \sum_{j=1}^{\infty}E(\phi_{j,t}^s)<1,
		\end{equation*}
		or $s>1$ such that 
		\begin{equation*}
			E(\phi_{0, t}^s)<\infty \quad\text{and}\quad \sum_{j=1}^{\infty}[E(\phi_{j,t}^s)]^{1/s}<1,
		\end{equation*}
		then there exists a strictly stationary solution of the quantile GARCH($p,q$) equations in \eqref{equ-qgarch-general}, and the process $\{y_t\}$ defined by
		\begin{equation}\label{solution}
			y_t=\textup{sgn}(U_t-0.5)\left(\phi_{0,t}+\sum_{\ell=1}^{\infty}\sum_{j_1, \dots, j_\ell=1}^{\infty}	\phi_{0,  t-j_1-\cdots-j_\ell}\phi_{j_1, t}\phi_{j_2, t-j_1}\cdots\phi_{j_\ell, t-j_1-\cdots-j_{\ell-1}}\right)
		\end{equation}
		is the unique strictly stationary and $\mathcal{F}_{t}^{U}$-measurable solution to \eqref{equ-qgarch-general} such that $E|y_t|^s<\infty$, where $\mathcal{F}_{t}^{U}$ is the $\sigma$-field generated by $\{U_{t}, U_{t-1}, \dots\}$.
	\end{thm}	
	
	\subsection{Proposed estimation methods}
	
	The proposed QR and CQR estimators can be extended to the quantile GARCH($p,q$) model \eqref{qgarch} with minor adjustments in notations and assumptions. 
	
	For the QR method,  denote the parameter vector of model \eqref{qgarch4} by $\bm\theta=(\omega, \bm\vartheta^{\prime})^{\prime}=(\omega, \bm\alpha^{\prime},\bm\beta^{\prime})^{\prime}$, where $\bm\alpha=(\alpha_{1},\ldots,\alpha_{q})^{\prime}$, $\bm\beta=(\beta_{1},\ldots,\beta_{p})^{\prime}$, and the parameter space is $\Theta \subset \mathbb{R}^{q+1}\times [0,1)^p$.
	Then the conditional quantile functions $q_t(\bm\theta)$ and $\widetilde{q}_t(\bm\theta)$ are given by
	\[
	q_t(\bm\theta) =\omega + \sum_{j=1}^\infty \gamma_j (\bm\vartheta)|y_{t-j}| \quad\text{and}\quad 
	\widetilde{q}_t(\bm\theta)=\omega + \sum_{j=1}^{t-1} \gamma_j (\bm\vartheta)|y_{t-j}|.
	\]
	Accordingly, the self-weighted QR estimator for quantile GARCH($p,q$) model can be defined as $\widetilde{\bm\theta}_{wn}(\tau)$ in \eqref{PairwiseWCQE}. 
	Let the true value of the parameter vector  be $\bm\theta(\tau)=(\omega(\tau), \bm\vartheta^{\prime}(\tau))^{\prime}=(\omega(\tau), \bm\alpha^{\prime}(\tau), \bm\beta^{\prime}(\tau))^{\prime}$, where $\bm\alpha(\tau)=(\alpha_{1}(\tau),\ldots,\alpha_{q}(\tau))^{\prime}$ and  $\bm\beta(\tau)=(\beta_{1}(\tau),\ldots,\beta_{p}(\tau))^{\prime}$.
	Denote the first derivative of $q_t(\bm\theta)$ by $\dot{q}_t(\bm\theta)=(1,\sum_{j=1}^{\infty}\dot{\gamma}_j^{\prime}(\bm\vartheta)|y_{t-j}|)^{\prime}$. 
	
	\begin{assum}\label{assum-FunctionDerivatives}
		For each $j\geq 1$, $\gamma_j(\cdot)$'s are twice differentiable functions, with derivatives of first and second orders,  $\dot{\gamma}_j(\cdot)$ and $\ddot{\gamma}_j(\cdot)$, satisfying that  
		(i) $\sup_{\|\bm \nu\|\leq r}|\gamma_j(\bm \nu+\bm\vartheta(\tau))|\leq c_1\rho^j$; 
		(ii) $\sup_{\|\bm \nu\|\leq r}\|\dot{\gamma}_j(\bm \nu+\bm\vartheta(\tau))\|\leq c_2\rho^j$; 
		(iii) $\sup_{\|\bm \nu\|\leq r}\|\ddot{\gamma}_j(\bm \nu+\bm\vartheta(\tau))\|\leq c_3\rho^j$ for some constants $c_1, c_2, c_3>0$, 
		where $\bm \nu\in \mathbb{R}^{p+q}$, $r>0$ is a fixed small value, and $0< \rho <1$. 
	\end{assum}
	
	\begin{thm}
		For $\{y_t\}$ generated by model \eqref{qgarch} under condition \eqref{loc2}, suppose $E|y_t|^s<\infty$ for some $s\in (0,1)$ and $\Sigma_{w}(\tau,\tau)$ is positive definite. 
		If Assumptions \ref{assum-Process}, \ref{assum-Space}(i), \ref{assum-ConditionalDensity}, \ref{assum-RandomWeight} and \ref{assum-FunctionDerivatives} hold, then as $n\rightarrow\infty$, we have (i) $\widetilde{\bm\theta}_{wn}(\tau) \to_p \bm\theta(\tau)$;  (ii) $\sqrt{n}(\widetilde{\bm\theta}_{wn}(\tau)-\bm\theta(\tau))\rightarrow_d N\left(\bm 0,\Sigma_{w}(\tau,\tau)\right)$  if Assumption \ref{assum-Space}(ii) is further satisfied. 
	\end{thm}

	For the CQR method, let
	$\bm\varphi=(\bm\phi^{\prime}, \lambda)^{\prime}=(a_0,\bm\psi^{\prime},\lambda)^{\prime}=(a_0, a_1,\ldots, a_q, b_1, \ldots, b_p, \lambda)^{\prime}$ be the parameter vector of the linear GARCH($p,q$) model in  \eqref{lgarch1} with the innovation $\varepsilon_t$ following the Tukey-lambda distribution in \eqref{TukeyLambda}. 
	Let $\Phi \subset (0,\infty)\times[0,\infty)^{q}\times [0,1)^p\times \Lambda$ be the parameter space of $\bm\varphi$. 
	The conditional quantile functions $q_{t,\tau}(\bm\varphi)$ and $\widetilde{q}_{t,\tau}(\bm\varphi)$ are defined as in Section \ref{subsec-CQR} with $h_t(\bm\phi)=a_0(1-\sum_{j=1}^{p}b_j)^{-1}+\sum_{j=1}^{\infty}\gamma_j(\bm\psi)|y_{t-j}|$ and $\widetilde{h}_t(\bm\phi)=a_0(1-\sum_{j=1}^{p}b_j)^{-1}+\sum_{j=1}^{t-1}\gamma_j(\bm\psi)|y_{t-j}|$, respectively. 
	Then the self-weighted CQR estimator for the quantile GARCH($p,q$) model is given by $\check{\bm\varphi}_{wn}$ in \eqref{CQR}, and the transformed CQR estimator is $\check{\bm\theta}_{wn}^*(\tau)=g_{\tau}(\check{\bm\varphi}_{wn})$, where $g_{\tau}(\cdot): \mathbb{R}^{p+q+2}\to\mathbb{R}^{p+q+1}$ is   the measurable  function defined as $g_{\tau}(\bm\varphi)=(a_0Q_{\tau}(\lambda)/(1-\sum_{j=1}^pb_j),a_1Q_{\tau}(\lambda),\ldots,a_qQ_{\tau}(\lambda),b_1,\ldots,b_p)^{\prime}$. 
	Let $\bm\varphi_0^*=(\bm\phi_0^{\prime}, \lambda_0)^{\prime}=(a_{00},\bm\psi_0^{\prime},\lambda_0)^{\prime}=(a_{00}, a_{10}, \ldots, a_{q0}, b_{10}, \ldots,b_{p0}, \lambda_0)^{\prime}$ be the pseudo-true parameter defined as in \eqref{pseudo-true}. Define the first derivative of $q_{t,\tau}(\bm\varphi)$ as $\dot{q}_{t,\tau}(\bm\varphi)=(Q_{\tau}(\lambda)\dot{h}_t^{\prime}(\bm\phi),\dot{Q}_{\tau}(\lambda)h_t(\bm\phi))^{\prime}$, where $\dot{h}_t(\bm\phi)$ and $\dot{Q}_{\tau}(\lambda)$ are the first derivatives of $h_t(\bm\phi)$ and $Q_{\tau}(\lambda)$, respectively.

	\begin{assum}\label{assum-FunctionDerivativesTukey}
		For each $j\geq 1$, $\gamma_j(\cdot)$'s are twice differentiable functions, with derivatives of first and second orders,  $\dot{\gamma}_j(\cdot)$ and $\ddot{\gamma}_j(\cdot)$, satisfying that  
		(i) $\sup_{\|\bm \nu\|\leq r}|\gamma_j(\bm \nu+\bm\psi_0)|\leq c_1\rho^j$; 
		(ii) $\sup_{\|\bm \nu\|\leq r}\|\dot{\gamma}_j(\bm \nu+\bm\psi_0)\|\leq c_2\rho^j$; 
		(iii) $\sup_{\|\bm \nu\|\leq r}\|\ddot{\gamma}_j(\bm \nu+\bm\psi_0)\|\leq c_3\rho^j$ for some constants $c_1, c_2, c_3>0$, 
		where $\bm \nu\in \mathbb{R}^{p+q}$, $r>0$ is a fixed small value, and $0< \rho <1$. 
	\end{assum}
	
	\begin{thm}
		For $\{y_t\}$ generated by model \eqref{qgarch} under condition \eqref{loc2}, suppose $E|y_t|^s<\infty$ for some $s\in (0,1)$ and $\Sigma_w^*$ is positive definite. 
		If Assumptions \ref{assum-ConditionalDensity}, \ref{assum-RandomWeight}, \ref{assum-Process-Mixing}, \ref{assum-SpaceTukey}(i) and \ref{assum-FunctionDerivativesTukey} hold, then as $n\rightarrow\infty$, we have (i) $\check{\bm\varphi}_{wn} \to_p \bm\varphi_0^*$. Moreover, if Assumption \ref{assum-SpaceTukey}(ii) further holds, then (ii) $\sqrt{n}(\check{\bm\varphi}_{wn}-\bm\varphi_0^*)\rightarrow_d N\left(\bm 0,\Sigma_w^*\right)$; and (iii) $\sqrt{n}(\check{\bm\theta}_{wn}^*(\tau)-\bm\theta(\tau)-B(\tau))\rightarrow_d N\left(\bm 0,g_{\tau}(\bm\varphi_0^*)\Sigma_w^*g_{\tau}^{\prime}(\bm\varphi_0^*)\right)$, where $B(\tau)=g_{\tau}(\bm\varphi_0^*)-\bm\theta(\tau)$.
	\end{thm}

	\section{Estimation of $\Sigma_w^*$ in Theorem \ref{thm-WCQR}}\label{remark-WCQR_cov}

	To approximate the asymptotic variance $\Sigma_w^*$ in Theorem \ref{thm-WCQR}, it suffices to estimate $\Omega_{0w}^*$ and $\Omega_{1w}^*$. Note that $\Omega_{11}^*$ and $\Omega_{12}^*$ can be consistently estimated by the sample averages:
	\begin{align*}
		\check{\Omega}_{11}^*&=\dfrac{1}{n}\sum_{t=1}^n\sum_{k=1}^Kw_t\ddot{\widetilde{q}}_{t,\tau_k}(\check{\bm\varphi}_{wn})\psi_{\tau_k}(y_t-\widetilde{q}_{t,\tau_k}(\check{\bm\varphi}_{wn})) \quad \text{and} \\
		\check{\Omega}_{12}^*&=\dfrac{1}{n}\sum_{t=1}^n\sum_{k=1}^Kw_t\check{f}_{t-1}(\widetilde{q}_{t,\tau_k}(\check{\bm\varphi}_{wn}))\dot{\widetilde{q}}_{t,\tau_k}(\check{\bm\varphi}_{wn})\dot{\widetilde{q}}_{t,\tau_k}^{\prime}(\check{\bm\varphi}_{wn}),
	\end{align*}
	where $\check{f}_{t-1}(\widetilde{q}_{t,\tau_k}(\check{\bm\varphi}_{wn}))=2\ell_k[\widetilde{q}_{t,\tau_k+\ell_k}(\check{\bm\varphi}_{wn})-\widetilde{q}_{t,\tau_k-\ell_k}(\check{\bm\varphi}_{wn})]^{-1}$, with $\ell_k$ representing the bandwidth defined as in \eqref{bandwidths} for quantile level $\tau_k$, and $\dot{\widetilde{q}}_{t,\tau}(\check{\bm\varphi}_{wn})$ and $\ddot{\widetilde{q}}_{t,\tau}(\check{\bm\varphi}_{wn})$ are obtained from $\dot{q}_{t,\tau}(\check{\bm\varphi}_{wn})$ and $\ddot{q}_{t,\tau}(\check{\bm\varphi}_{wn})$ by setting the initial values $y_t=0$ for $t\leq 0$, respectively. 
	Then $\Omega_{1w}^*$ can be consistently estimated by $\check{\Omega}_{1w}^*=\check{\Omega}_{12}^*-\check{\Omega}_{11}^*$.

	As the population covariance matrix of $n^{-1/2}\sum_{t=1}^n\bm X_t$, the matrix $\Omega_{0w}^*$ cannot be consistently estimated by the corresponding sample covariance matrix \citep{Wu_Pourahmadi2009}. 
	Alternatively, we adopt the following kernel estimator of spectral density matrix \citep{Andrews1991}:
	\begin{align}\label{kernel_estimator}
		\check{\Omega}_{0w}^*=\dfrac{n}{n-d}\sum_{\ell=-n+1}^{n-1}K\left(\dfrac{\ell}{B_n}\right)\check{\Gamma}(\ell),
	\end{align} 
	where $n/(n-d)$ with $d=4$ is a small sample degrees of freedom adjustment to offset the effect of estimating $\bm\varphi_0^*\in\mathbb{R}^d$ using $\check{\bm\varphi}_{wn}$, $\check{\Gamma}(\ell)=I(\ell\geq 0)n^{-1}\sum_{t=\ell+1}^n\check{\bm X}_t\check{\bm X}_{t-\ell}^{\prime}+I(\ell< 0)n^{-1}\sum_{t=-\ell+1}^n\check{\bm X}_{t+\ell}\check{\bm X}_{t}^{\prime}$ with $\check{\bm X}_t=\sum_{k=1}^K w_t\dot{\widetilde{q}}_{t,\tau_k}(\check{\bm\varphi}_{wn})\psi_{\tau_k}(y_t-\widetilde{q}_{t,\tau_k}(\check{\bm\varphi}_{wn}))$, $B_n$ is a bandwidth, and $K(\cdot): \mathbb{R}\to [-1,1]$ is a real-valued kernel function satisfying
	\begin{align}\label{assum-kernel}
		K(0)=1,~ K(x)=K(-x),~\int_{-\infty}^{\infty}K^2(x)dx<\infty,~\text{and}~K(\cdot)~\text{is continuous}.
	\end{align} 
	Under Assumption \ref{assum-Process-Mixing}, if $B_n\to\infty$, $B_n^2/n\to 0$, $E(\|\bm X_t\|^{2\delta})<\infty$, and $\sum_{n=1}^{\infty}n^2[\alpha(n)]^{1-2/\delta}$ for some $\delta>2$, \cite{Andrews1991} showed that $\check{\Omega}_{0w}^*\to_p \Omega_{0w}^*$ as $n\to\infty$. 
	As a result, the asymptotic covariance $\Sigma_w^*$ can be estimated by $\check{\Sigma}_w^*=\check{\Omega}_{1w}^{*-1}\check{\Omega}_{0w}^*\check{\Omega}_{1w}^{*-1}$. 
	
	Many kernel functions satisfy \eqref{assum-kernel}, such as the Bartlett, Parzen, Tukey-Hanning and quadratic spectral (QS) kernels. \cite{Andrews1991} showed that under some regular conditions the QS kernel is optimal with respect to the asymptotic truncated mean squared error (MSE) among the aforementioned kernels. Therefore, we employ the QS kernel defined as follows:
	\begin{align}\label{QS-kernel}
		K(x)=\dfrac{25}{12\pi^2x^2}\left[\dfrac{\sin(6\pi x/5)}{6\pi x/5}-\cos(6\pi x/5)\right].
	\end{align}
	
	It remains to choose the bandwidth $B_n$ for $\check{\Omega}_{0w}^*$ in \eqref{kernel_estimator}. \cite{Andrews1991} introduced the automatic bandwidth for the QS kernel as $\widehat{B}_n=1.3221[n\widehat{\alpha}(2)]^{1/5}$, where $\widehat{\alpha}(2)$ is calculated using some approximating parametric models for each element of $\bm X_t$ or $\bm X_t$ as a whole. For simplicity, we fit AR(1) models for $\{X_{it}\}~(i=1,\ldots,4)$ and obtain the estimates $(\widehat{\rho}_i,\widehat{\sigma}_i^2)$ for the AR coefficient and innovation variance $(\rho_i,\sigma_i^2)$, respectively. Then $\widehat{\alpha}(2)$ can be calculated as
	\[\widehat{\alpha}(2)=\sum_{i=1}^{4}\iota_i\dfrac{4\widehat{\rho}_i^2\widehat{\sigma}_i^4}{(1-\widehat{\rho}_i)^8}\bigg/\sum_{i=1}^{4}\iota_i\dfrac{\widehat{\sigma}_i^4}{(1-\widehat{\rho}_i)^4},\] 
	where $\iota_i$'s are the weights assigned to the diagonal elements of $\check{\Omega}_{0w}^*$, and the usual choice of $\iota_i$ is one for $i=1,\ldots,4$.
	
	\begin{remark}[Estimation of $\Sigma_w^*$ under correct model specification]\label{remark-WCQR_cov_correct}
		If $Q_{\tau}(y_t|\mathcal{F}_{t-1})$ is correctly specified by $q_{t,\tau}(\bm\varphi_0^*)$ (i.e., $g_{\tau}(\bm\varphi_0^*)=\bm\theta(\tau)$) for each $\tau\in\mathcal{T}_h$,  then $E[\psi_{\tau}(y_t-q_{t,\tau}(\bm\varphi_0^*))|\mathcal{F}_{t-1}]=0$, and the martingale CLT can be used to establish the asymptotic normality of $\check{\bm\varphi}_{wn}$. 
		In this case, $\Sigma_w^*$ can be largely simplified since $\Omega_{1w}^*=\Omega_{12}^*$ and  $\Omega_{0w}^*=\sum_{k=1}^K\sum_{k^{\prime}=1}^K\Psi_{k,k^{\prime}}E[w_t^2\dot{q}_{t,\tau_k}(\bm\varphi_0^*)\dot{q}_{t,\tau_{k^{\prime}}}^{\prime}(\bm\varphi_0^*)]$ with $\Psi_{k,k^{\prime}}=\min\{\tau_{k},\tau_{k^{\prime}}\}(1-\max\{\tau_{k},\tau_{k^{\prime}}\})$. As a result, $\Omega_{0w}^*$ can be estimated by its sample average as for $\Omega_{0w}(\tau)$ in Section \ref{subsec-WCQE}, and a consistent estimator $\check{\Sigma}_w^*=\check{\Omega}_{12}^{*-1}\check{\Omega}_{0w}^*\check{\Omega}_{12}^{*-1}$ can be constructed for $\Sigma_w^*$.    
	\end{remark}

	\section{Notation}
	
	\subsection{$\dot{q}_t(\bm\theta)$ and $\ddot{q}_{t}(\bm\theta)$ for quantile GARCH$(1,1)$ model}
	
	Recall that $q_t(\bm\theta) =\omega + \alpha_{1}\sum_{j=1}^\infty \beta_{1}^{j-1}|y_{t-j}|$, then its first and second derivatives are 
	\begin{align}\label{1st-derivative-qt}
		\dot{q}_t(\bm\theta)=\left(1,\sum_{j=1}^{\infty}\beta^{j-1}_{1}|y_{t-j}|,\alpha_{1}\sum_{j=2}^{\infty}(j-1)\beta^{j-2}_{1}|y_{t-j}|\right)^{\prime},
	\end{align} 
	and 
	\begin{align}\label{2nd-derivative-qt}
		\ddot{q}_{t}(\bm\theta)=\begin{pmatrix}
			0 & 0 & 0 \\ 
			0 & 0 & \sum\limits_{j=2}^{\infty}(j-1)\beta^{j-2}_{1}|y_{t-j}| \\
			0 & \sum\limits_{j=2}^{\infty}(j-1)\beta^{j-2}_{1}|y_{t-j}|  & \alpha_{1}\sum\limits_{j=3}^{\infty}(j-1)(j-2)\beta^{j-3}_{1}|y_{t-j}| 
		\end{pmatrix}.
	\end{align}
	
	\subsection{$\dot{q}_{t,\tau}(\bm\varphi)$ and $\ddot{q}_{t,\tau}(\bm\varphi)$ for quantile GARCH$(1,1)$ model}

	Recall that 
	\begin{align}\label{qt-Tukey}
		q_{t,\tau}(\bm\varphi) &= Q_{\tau}(\lambda) \left(\frac{a_0}{1-b_1}+a_1\sum_{j=1}^{\infty} b_1^{j-1}|y_{t-j}|\right):= Q_{\tau}(\lambda)h_t(\bm\phi).
	\end{align} 
	Then its first and second derivatives are as follows
	\begin{align}\label{1st-derivative-qt-Tukey}
		\dot{q}_{t,\tau}(\bm\varphi)&=(Q_{\tau}(\lambda)\dot{h}_t^{\prime}(\bm\phi),\dot{Q}_{\tau}(\lambda)h_t(\bm\phi))^{\prime}
	\end{align} 
	and
	\begin{align}\label{2nd-derivative-qt-Tukey}
		\ddot{q}_{t,\tau}(\bm\varphi)=\begin{pmatrix}
			Q_{\tau}(\lambda)\ddot{h}_t(\bm\phi) & \dot{Q}_{\tau}(\lambda)\dot{h}_t(\bm\phi) \\ 
			\dot{Q}_{\tau}(\lambda)\dot{h}_t^{\prime}(\bm\phi) & \ddot{Q}_{\tau}(\lambda)h_t(\bm\phi)
		\end{pmatrix},
	\end{align}
	where $\dot{Q}_{\tau}(\lambda)$ and $\dot{h}_t(\bm\phi)$ ($\ddot{Q}_{\tau}(\lambda)$ and $\ddot{h}_t(\bm\phi)$) are the first (second) derivatives of $Q_{\tau}(\lambda)$ and $h_t(\bm\phi)$, respectively. Specifically, they are defined as follows
	\begin{align*}
		\dot{Q}_{\tau}(\lambda)&=\lambda^{-2}\{\tau^{\lambda}(\lambda\ln\tau-1)-(1-\tau)^{\lambda}[\lambda\ln(1-\tau)-1]\}, \\
		\ddot{Q}_{\tau}(\lambda)&=\lambda^{-3}\{\tau^{\lambda}[(\lambda\ln\tau-1)^2+1]-(1-\tau)^{\lambda}[(\lambda\ln(1-\tau)-1)^2+1]\}, \\
		\dot{h}_t(\bm\phi)&=\left(\dfrac{1}{1-b_1},\sum_{j=1}^{\infty}b_1^{j-1}|y_{t-j}|,\dfrac{a_0}{(1-b_1)^2}+a_1\sum_{j=2}^{\infty}(j-1)b_1^{j-2}|y_{t-j}|\right)^{\prime}, \\
		\ddot{h}_t(\bm\phi)&=\begin{pmatrix}
			0 & 0 & \dfrac{1}{(1-b_1)^2} \\ 
			0 & 0 & \sum\limits_{j=2}^{\infty}(j-1)b_1^{j-2}|y_{t-j}| \\
			\dfrac{1}{(1-b_1)^2} & \sum\limits_{j=2}^{\infty}(j-1)b_1^{j-2}|y_{t-j}|  & \dfrac{2a_0}{(1-b_1)^3}+a_1\sum\limits_{j=3}^{\infty}(j-1)(j-2)b_1^{j-3}|y_{t-j}| 
		\end{pmatrix}.
	\end{align*}

	\section{Technical Proofs}
	
	\subsection{Proof of Theorem \ref{thm-stationarity11}}
	
	\begin{proof}[Proof of Theorem \ref{thm-stationarity11}]
		Let $\{X_t\}$ be a sequence of random variables with
		\begin{equation}\label{aeq1}
			X_t=\phi_{0,t}+\sum_{\ell=1}^{\infty}\sum_{j_1, \dots, j_\ell=1}^{\infty}	\phi_{0,  t-j_1-\cdots-j_\ell}\phi_{j_1, t}\phi_{j_2, t-j_1}\cdots\phi_{j_\ell, t-j_1-\cdots-j_{\ell-1}}
		\end{equation}
		taking values in $[0,\infty]$. 
		
		(i) We first consider the case with $s\in(0,1]$. For any $s\in(0,1]$, using the inequality $(x+y)^s\leq x^s+y^s$ for $x, y\geq 0$, we have
		\[X_t^s \leq \phi_{0,t}^s+\sum_{\ell=1}^{\infty}\sum_{j_1, \dots, j_\ell=1}^{\infty}	\phi^s_{0,  t-j_1-\cdots-j_\ell}\phi^s_{j_1, t}\phi^s_{j_2, t-j_1}\cdots\phi^s_{j_\ell, t-j_1-\cdots-j_{\ell-1}}.\]
		Denote $A_s=\sum_{j=1}^{\infty} E(\phi^s_{j, t})$. Observe that the $\phi_{j, t}$'s in every summand on the right side of the above inequality are independent, where $j\geq0$. Thus it follows that
		\begin{align}\label{aeq2}
			E(X_t^s)&\leq E(\phi_{0,t}^s)+\sum_{\ell=1}^{\infty}\sum_{j_1, \dots, j_\ell=1}^{\infty}	E(\phi^s_{0,  t-j_1-\cdots-j_\ell})E(\phi^s_{j_1, t})E(\phi^s_{j_2, t-j_1})\cdots E(\phi^s_{j_\ell, t-j_1-\cdots-j_{\ell-1}})\notag\\
			&=	E(\phi^s_{0,t}) [1+\sum_{\ell=1}^{\infty} A_s^\ell] = \frac{E(\phi^s_{0,t})}{1-A_s} < \infty,
		\end{align}
		where we used the condition in \eqref{StationaryCondition11_fractional}.  Consequently, $\{X_t\}$ is a sequence of almost surely finite random variables. With all the summands being non-negative, we can write
		\begin{align*}
			\sum_{j=1}^{\infty}\phi_{j,t}X_{t-j}&=\sum_{j_0=1}^{\infty}\phi_{0,t-j_0}\phi_{j_0, t}\\
			&\hspace{5mm}+\sum_{j_0=1}^{\infty}\sum_{\ell=1}^{\infty}\sum_{j_1, \dots, j_\ell=1}^{\infty}\phi_{0,t-j_0-j_1-\cdots-j_\ell}\phi_{j_0,t}\phi_{j_1, t-j_0}\cdots\phi_{j_\ell, t-j_0-j_1-\cdots-j_{\ell-1}}\\
			&=\sum_{\ell=1}^{\infty}\sum_{j_1, \dots, j_\ell=1}^{\infty}	\phi_{0,  t-j_1-\cdots-j_\ell}\phi_{j_1, t}\phi_{j_2, t-j_1}\cdots\phi_{j_\ell, t-j_1-\cdots-j_{\ell-1}}.
		\end{align*}
		Comparing this with \eqref{aeq1}, we have that $\{X_t\}$ satisfies the recursive equation
		\[ X_t= \phi_{0,t}+	\sum_{j=1}^{\infty}\phi_{j,t}X_{t-j}.\]
		Hence the existence of a strictly stationary solution to \eqref{equ-qgarch11} is proved by setting $y_t=\text{sgn}(U_t-0.5)X_t$.  In addition,  $E|y_t|^s=E(X_t^s)<\infty$ by \eqref{aeq2}.
		
		Now suppose that $\{y_t\}$ is a strictly stationary and causal solution to the model in \eqref{equ-qgarch11}.  Then, for any $m\in\mathbb{N}$, by successively substituting the $|y_{t-j}|$'s in the second equation of \eqref{equ-qgarch11} $m$ times, we have
		\[|y_t|=\phi_{0,t}+\sum_{\ell=1}^{m}\sum_{j_1,\dots, j_\ell=1}^{\infty}\phi_{0,  t-j_1-\cdots-j_\ell}\phi_{j_1, t}\phi_{j_2, t-j_1}\cdots\phi_{j_\ell, t-j_1-\cdots-j_{\ell-1}}+R_{t,m},\]
		where
		\[R_{t,m}=\sum_{j_1,\dots, j_{m+1}=1}^{\infty}\phi_{j_1, t}\phi_{j_2, t-j_1}\cdots\phi_{j_{m+1}, t-j_1-\cdots-j_m}|y_{t-j_1-\cdots-j_{m+1}}|.
		\]
		By the causality of $\{y_t\}$, the $\phi_{j_1, t}, \phi_{j_2, t-j_1}, \dots\phi_{j_{m+1}, t-j_1-\cdots-j_m}$ and $|y_{t-j_1-\cdots-j_{m+1}}|$ in every summand on the right side of the above expression are independent.  As a result,
		\begin{align*}
			E(R_{t,m}^s)&\leq\sum_{j_1,\dots, j_{m+1}=1}^{\infty}
			E(\phi_{j_1, t}^s)E(\phi_{j_2, t-j_1}^s)\cdots E(\phi_{j_{m+1}, t-j_1-\cdots-j_m}^s)E|y_{t-j_1-\cdots-j_{m+1}}|^s\\
			& =A_s^{m+1}E|y_t|^s,
		\end{align*}
		which implies $E(\sum_{m=1}^{\infty}R_{t,m}^s)= \sum_{m=1}^{\infty}E(R_{t,m}^s) <\infty$, since $0\leq A_s<1$ and $E|y_t|^s<\infty$.  It follows that, as $m\to \infty$, $R_{t,m}\rightarrow0$ a.s., and thus $|y_t|=X_t$ a.s. Finally, since $y_t=\text{sgn}(U_t-0.5)|y_t|$, we have $y_t=\text{sgn}(U_t-0.5)X_t$ a.s.
		
		(ii) We next consider the case with $s\in\{2,3,4,\ldots\}$, where we only need to show $E(X_t^s)<\infty$ since the remainder of the proof is the same as for $s\in(0,1]$ in (i). 
		
		By Minkowski inequality, for $s\geq 1$, we have
		\begin{equation}\label{EXts}
			\|X_t\|_s \leq \|\phi_{0,t}\|_s + \left\|\sum_{\ell=1}^{\infty}\sum_{j_1, \dots, j_\ell=1}^{\infty}	\phi_{0,  t-j_1-\cdots-j_\ell}\phi_{j_1, t}\phi_{j_2, t-j_1}\cdots\phi_{j_\ell, t-j_1-\cdots-j_{\ell-1}}\right\|_s.
		\end{equation}
		Since $E(\phi_{0, t}^s)<\infty$ by the condition in \eqref{StationaryCondition11_integer}, to show $E(X_t^s)<\infty$, it suffices to show that $E[(\sum_{\ell=1}^{\infty}\sum_{j_1, \dots, j_\ell=1}^{\infty}\phi_{0, t-j_1-\cdots-j_\ell}\phi_{j_1, t}\phi_{j_2, t-j_1}\cdots\phi_{j_\ell, t-j_1-\cdots-j_{\ell-1}})^s]<\infty$.
		Consider the case with $s=2$ for illustration. 
		By H\"{o}lder's inequality and the independence of $\phi_{j, t}$'s, it holds that
		\begin{align*}
			& E\left[\left(\sum_{\ell=1}^{\infty}\sum_{j_1, \dots, j_\ell=1}^{\infty}\phi_{0, t-j_1-\cdots-j_\ell}\phi_{j_1, t}\phi_{j_2, t-j_1}\cdots\phi_{j_\ell, t-j_1-\cdots-j_{\ell-1}}\right)^2\right] \\
			=& \sum_{\ell=1}^{\infty}\sum_{j_1, \dots, j_\ell=1}^{\infty} \sum_{k=1}^{\infty}\sum_{i_1, \dots, i_k=1}^{\infty}E\bigg[\left(\phi_{0, t-j_1-\cdots-j_\ell}\phi_{j_1, t}\phi_{j_2, t-j_1}\cdots\phi_{j_\ell, t-j_1-\cdots-j_{\ell-1}}\right) \\
			&\hspace{45mm} \times \left(\phi_{0, t-i_1-\cdots-i_k}\phi_{i_1, t}\phi_{i_2, t-i_1}\cdots\phi_{i_k, t-i_1-\cdots-i_{k-1}}\right) \bigg] \\
			\leq &\sum_{\ell=1}^{\infty}\sum_{j_1, \dots, j_\ell=1}^{\infty} \sum_{k=1}^{\infty}\sum_{i_1, \dots, i_k=1}^{\infty}\left[E\left(\phi_{0, t-j_1-\cdots-j_\ell}\phi_{j_1, t}\phi_{j_2, t-j_1}\cdots\phi_{j_\ell, t-j_1-\cdots-j_{\ell-1}}\right)^2\right]^{1/2} \\
			&\hspace{45mm} \times \left[E\left(\phi_{0, t-i_1-\cdots-i_k}\phi_{i_1, t}\phi_{i_2, t-i_1}\cdots\phi_{i_k, t-i_1-\cdots-i_{k-1}}\right)^2\right]^{1/2} \\
			=&E(\phi^2_{0, t})\sum_{\ell=1}^{\infty}\sum_{j_1, \dots, j_\ell=1}^{\infty} \sum_{k=1}^{\infty}\sum_{i_1, \dots, i_k=1}^{\infty}\left[\prod_{h=1}^{\ell}E(\phi^2_{j_h, t})\right]^{1/2}\left[\prod_{m=1}^{k}E(\phi^2_{i_m, t})\right]^{1/2} \\
			= & E(\phi^2_{0, t})\left\{\sum_{\ell=1}^{\infty}\sum_{j_1, \dots, j_\ell=1}^{\infty}\left[\prod_{h=1}^{\ell}E(\phi^2_{j_h, t})\right]^{1/2}\right\}
			\left\{\sum_{k=1}^{\infty}\sum_{i_1, \dots, i_k=1}^{\infty}\left[\prod_{m=1}^{k}E(\phi^2_{i_m, t})\right]^{1/2}\right\} \\
			= & E(\phi^2_{0, t}) \left(\sum_{\ell=1}^{\infty}B_{2,1/2}^{\ell} \right) \left(\sum_{k=1}^{\infty}B_{2,1/2}^{k} \right)\\
			< & \dfrac{E(\phi^2_{0, t})}{(1-B_{2,1/2})^2}<\infty, 
		\end{align*} 
		where we used the conditions $E(\phi^s_{0, t})<\infty$ and $B_{s,1/s}=\sum_{j=1}^{\infty}[E(\phi^s_{j, t})]^{1/s}<1$ in \eqref{StationaryCondition11_integer}. Hence, $E(X_t^s)<\infty$ holds for $s=2$. For the cases with $s\geq 3$, we can similarly show that $E(X_t^s)<\infty$ if \eqref{StationaryCondition11_integer} holds.
		The proof of this theorem is complete. 
		
	\end{proof}
	
	\subsection{Proofs of Theorems \ref{thm-WCQE-uniform-consistency}--\ref{thm-WCQE-weak-convergence} and Corollary \ref{thm-WCQE}}
	\begin{proof}[Proof of Theorem \ref{thm-WCQE-uniform-consistency}]
		Recall that $q_t(\bm\theta) =\omega + \alpha_{1}\sum_{j=1}^\infty \beta_{1}^{j-1}|y_{t-j}|$ and $\widetilde{q}_t(\bm\theta)=\omega + \alpha_{1}\sum_{j=1}^{t-1} \beta_{1}^{j-1}|y_{t-j}|$, where $\bm\theta=(\omega, \alpha_{1},\beta_{1})^{\prime}$. 
		Define $L(\bm\theta,\tau)=E[w_t\ell_t(\bm\theta,\tau)]$, $L_n(\bm\theta,\tau)=n^{-1}\sum_{t=1}^{n}w_t\ell_t(\bm\theta,\tau)$ and $\widetilde{L}_n(\bm\theta,\tau)=n^{-1}\sum_{t=1}^{n}w_t\widetilde{\ell}_t(\bm\theta,\tau)$, where $\ell_t(\bm\theta,\tau)=\rho_{\tau}(y_t-q_t(\bm\theta))$ and $\widetilde{\ell}_t(\bm\theta,\tau)=\rho_{\tau}(y_t-\widetilde{q}_t(\bm\theta))$. 
		To show the uniform consistency, we first verify the following claims:
		\begin{itemize}
			\item[(i)] $\sup\limits_{\tau\in\mathcal{T}}\sup\limits_{\Theta}|\widetilde{L}_n(\bm\theta,\tau)-L_n(\bm\theta,\tau)|=o_p(1)$;
			\item[(ii)] $E\left(\sup\limits_{\tau\in\mathcal{T}}\sup\limits_{\Theta}w_t|\ell_t(\bm\theta,\tau)|\right)<\infty$;
			\item[(iii)] $L(\bm\theta,\tau)$ has a unique minimum at $\bm\theta(\tau)$.
		\end{itemize}
		We first prove Claim (i). By the Lipschitz continuity of $\rho_{\tau}(\cdot)$, strict stationarity and ergodicity of $y_t$ by Assumption \ref{assum-Process}, Lemma \ref{lem0}(i) and $E(w_{t}\varsigma_{\rho})<\infty$ by Assumption \ref{assum-RandomWeight}, it holds that 
		\begin{align*}
			\sup_{\tau\in\mathcal{T}}\sup_{\Theta}|\widetilde{L}_n(\bm\theta,\tau)-L_n(\bm\theta,\tau)|
			&\leq \dfrac{1}{n}\sum_{t=1}^{n}w_t\sup_{\tau\in\mathcal{T}}\sup_{\Theta}|\rho_{\tau}(y_t-\widetilde{q}_t(\bm\theta))-\rho_{\tau}(y_t-q_t(\bm\theta))| \\
			&\leq\dfrac{C}{n}\sum_{t=1}^{n}w_t\sup_{\tau\in\mathcal{T}}\sup_{\Theta}|q_t(\bm\theta)-\widetilde{q}_t(\bm\theta)| \\
			&\leq \dfrac{C}{n}\sum_{t=1}^{n}\rho^{t}w_t\varsigma_{\rho}=o_p(1),
		\end{align*} 
		where $\varsigma_{\rho}=\sum_{s=0}^{\infty}\rho^{s}|y_{-s}|$.
		
		We next prove Claim (ii). By Assumption \ref{assum-Space}, there exist constant $0<\overline{c}<\infty$ and $0<\rho<1$ such that $\max\{|\omega|, |\alpha_{1}|\} \leq \overline{c}$ and $0<\beta_{1}\leq\rho$. By the fact that $|\rho_{\tau}(x)|\leq |x|$, and $E(w_t)<\infty$ and $E(w_{t}|y_{t-j}|^3)<\infty$ for all $j\geq 1$ by Assumption \ref{assum-RandomWeight}, we have 
		\begin{align*}
			E[\sup_{\tau\in\mathcal{T}}\sup_{\Theta}w_t|\ell_t(\bm\theta,\tau)|]&\leq E[w_t|y_t|]+E[w_t\sup_{\tau\in\mathcal{T}}\sup_{\Theta}|q_t(\bm\theta)|] \\
			&\leq E(w_t|y_t|) + \overline{c}E\left[w_{t}\left(1+\sum_{j=1}^{\infty}\rho^{j-1}|y_{t-j}|\right)\right] < \infty.
		\end{align*}
		Hence, (ii) is verified. 
		
		We consider Claim (iii). For $x\neq 0$, it holds that
		\begin{align}\label{identity}
			\rho_{\tau}(x-y)-\rho_{\tau}(x)&=-y\psi_{\tau}(x)+y\int_{0}^{1}[I(x\leq ys)-I(x\leq 0)]ds \nonumber \\
			&=-y\psi_{\tau}(x)+(x-y)[I(0>x>y)-I(0<x<y)],
		\end{align}
		where $\psi_{\tau}(x)=\tau-I(x<0)$; see \cite{Knight1998}. Let $\nu_t(\bm\theta,\tau)=q_t(\bm\theta)-q_t(\bm\theta(\tau))$ and $\eta_{t,\tau}=y_{t}-q_t(\bm\theta(\tau))$. By \eqref{identity}, it follows that
		\begin{align*}
			&\ell_t(\bm\theta,\tau)-\ell_t(\bm\theta(\tau),\tau) \\
			=&-\nu_t(\bm\theta,\tau)\psi_{\tau}(\eta_{t,\tau})+[\eta_{t,\tau}-\nu_t(\bm\theta,\tau)]\left[I(0>\eta_{t,\tau}>\nu_t(\bm\theta,\tau))-I(0<\eta_{t,\tau}<\nu_t(\bm\theta,\tau))\right].
		\end{align*}
		This together with $E[\psi_{\tau}(\eta_{t,\tau})|\mathcal{F}_{t-1}]=0$, implies that
		\begin{align*}
			& L(\bm\theta,\tau)-L(\bm\theta(\tau),\tau) \\ 		
			=& E\left\{w_t[\eta_{t,\tau}-\nu_t(\bm\theta,\tau)]\left[I(0>\eta_{t,\tau}>\nu_t(\bm\theta,\tau))-I(0<\eta_{t,\tau}<\nu_t(\bm\theta,\tau))\right]\right\}  \geq 0.
		\end{align*}
		Since $f_{t-1}(x)$ is continuous at a neighborhood of $q_t(\bm\theta(\tau))$ by Assumption \ref{assum-ConditionalDensity} and $\{w_t\}$ are nonnegative random weights, then the above equality holds if and only if $\nu_t(\bm\theta,\tau)=0$ with probability one for $t \in \mathbb{Z}$. This together with $q_t(\bm\theta) =\omega + \alpha_{1}\sum_{j=1}^\infty \beta_{1}^{j-1}|y_{t-j}|$, implies that 
		\[\omega-\omega(\tau) =\sum_{j=1}^{\infty} \left[\alpha_1(\tau)\beta_1^{j-1}(\tau)-\alpha_{1}\beta_{1}^{j-1}\right]|y_{t-j}|.\]
		Note that $y_{t-1}$ is independent of all the others given $\mathcal{F}_{t-2}$. As a result, we have $\omega=\omega(\tau)$ and $\alpha_{1}=\alpha_{1}(\tau)$, thus $\beta_{1}=\beta_1(\tau)$ follows.
		Therefore, $\bm\theta=\bm\theta(\tau)$ and (iii) is verified. 
		
		Note that $L(\bm\theta,\tau)$ is a measurable function of $y_t$ in Euclidean space for each $(\bm\theta,\tau)\in\Theta\times\mathcal{T}$, and $L(\bm\theta,\tau)$ is a continuous function of $(\bm\theta,\tau)\in\Theta\times\mathcal{T}$ for each $y_t$. Then by Theorem 3.1 of \cite{Ling_McAleer2003}, together with Claim (ii) and the strict stationarity and ergodicity of $\{y_t\}$ under Assumption \ref{assum-Process}, we have
		\begin{align*}
			\sup\limits_{\tau\in\mathcal{T}}\sup\limits_{\bm\theta\in\Theta}|L_n(\bm\theta,\tau)-L(\bm\theta,\tau)|=o_p(1).
		\end{align*}
		This together with Claim (i), implies that
		\begin{align}\label{loss-uniform-consistency}
			\sup\limits_{\tau\in\mathcal{T}}\sup\limits_{\bm\theta\in\Theta}|\widetilde{L}_n(\bm\theta,\tau)-L(\bm\theta,\tau)|=o_p(1).
		\end{align}	
		We next verify the uniform consistency by extending the standard consistency argument in \cite{Amemiya1985}; see also Lemma B.1 of \cite{Chernozhukov2006}.
		
		For any $c>0$, with probability tending to 1 uniformly in $\epsilon\geq c$ and uniformly in $\tau\in\mathcal{T}$: by $\widetilde{\bm\theta}_{wn}(\tau)=\argmin_{\bm\theta\in\Theta}\widetilde{L}_n(\bm\theta,\tau)$, it follows that
		\begin{align}\label{Claim-a}
			\widetilde{L}_n(\widetilde{\bm\theta}_{wn}(\tau),\tau) \leq \widetilde{L}_n(\bm\theta(\tau),\tau) + \epsilon/3,
		\end{align}
		and by \eqref{loss-uniform-consistency}, it holds that
		\begin{align}\label{Claim-b}
			L(\widetilde{\bm\theta}_{wn}(\tau),\tau) < \widetilde{L}_n(\widetilde{\bm\theta}_{wn}(\tau),\tau) + \epsilon/3 \;\;\text{and}\;\; \widetilde{L}_n(\bm\theta(\tau),\tau) < L(\bm\theta(\tau),\tau) + \epsilon/3.
		\end{align} 
		Combining \eqref{Claim-a}--\eqref{Claim-b}, with probability tending to 1, we have
		\begin{align}\label{loss-inequality}
			L(\widetilde{\bm\theta}_{wn}(\tau),\tau) < \widetilde{L}_n(\widetilde{\bm\theta}_{wn}(\tau),\tau) + \epsilon/3 \leq \widetilde{L}_n(\bm\theta(\tau),\tau) + 2\epsilon/3 < L(\bm\theta(\tau),\tau) + \epsilon.
		\end{align}
		
		Pick any $\delta>0$, let $\{B_{\delta}(\tau),\tau\in\mathcal{T}\}$ be a collection of balls with radius $\delta>0$, each centered at $\bm\theta(\tau)$. Then $B^c\equiv\Theta/B_{\delta}(\tau)$ is compact, and thus $\inf_{\bm\theta\in B^c}L(\bm\theta,\tau)$ exists. 
		Denote $\epsilon=\inf_{\tau\in\mathcal{T}}\left[\inf_{\bm\theta\in B^c}L(\bm\theta,\tau)-L(\bm\theta(\tau),\tau)\right]$. Since $\bm\theta(\tau)=\argmin_{\bm\theta\in\Theta}L(\bm\theta,\tau)$ is unique by Claim (iii), then $\epsilon>0$. For any $\epsilon>0$, we can pick $c>0$ such that $\text{Pr}(\epsilon\geq c)>1-\epsilon$. 
		Together with \eqref{loss-inequality}, it follows that with probability becoming greater than $1-\epsilon$, uniformly in $\tau\in\mathcal{T}$: 
		\begin{align*}
			L(\widetilde{\bm\theta}_{wn}(\tau),\tau) & < L(\bm\theta(\tau),\tau) + \inf_{\tau\in\mathcal{T}}\left[\inf_{\bm\theta\in B^c}L(\bm\theta,\tau)-L(\bm\theta(\tau),\tau)\right]  < \inf_{\bm\theta\in B^c}L(\bm\theta,\tau). 
		\end{align*}
		Thus with probability becoming greater than $1-\epsilon$, $\sup_{\tau\in\mathcal{T}}\|\widetilde{\bm\theta}_{wn}(\tau)-\bm\theta(\tau)\|\leq \delta$.
		By the arbitrariness of $\epsilon$, it implies that $\sup_{\tau\in\mathcal{T}}\|\widetilde{\bm\theta}_{wn}(\tau)-\bm\theta(\tau)\|\leq \delta$ with probability tending to 1.  
		The proof of this theorem is complete.	
	\end{proof}
	
	\begin{proof}[Proof of Corollary \ref{thm-WCQE}]
		For $\bm u\in \mathbb{R}^{3}$, define $\widetilde{H}_n(\bm u)=n[\widetilde{L}_n(\bm\theta(\tau)+\bm u)-\widetilde{L}_n(\bm\theta(\tau))]$, where $\widetilde{L}_n(\bm\theta)=n^{-1}\sum_{t=1}^{n}w_t\rho_{\tau}(y_t-\widetilde{q}_t(\bm\theta))$. Denote $\widetilde{\bm u}_n=\widetilde{\bm\theta}_{wn}(\tau)-\bm\theta(\tau)$. By Theorem \ref{thm-WCQE-uniform-consistency}, it holds that $\widetilde{\bm u}_n=o_p(1)$.	
		Note that $\widetilde{\bm u}_n$ is the minimizer of $\widetilde{H}_n(\bm u)$, since $\widetilde{\bm\theta}_{wn}(\tau)$ minimizes $\widetilde{L}_n(\bm\theta)$. Define $J=\Omega_{1w}(\tau)/2$, where $\Omega_{1w}(\tau)=E\left[f_{t-1}(F_{t-1}^{-1}(\tau))w_{t}\dot{q}_t(\bm\theta(\tau))\dot{q}_t^{\prime}(\bm\theta(\tau))\right]$. By the ergodic theorem and Assumptions \ref{assum-Process}--\ref{assum-RandomWeight}, we have $J_n=J+o_p(1)$, where $J_n$ is defined as in Lemma \ref{lem2}.	
		Moreover, from Lemmas \ref{lem2}--\ref{lem3}, it follows that
		\begin{align}\label{Hn}
			\widetilde{H}_n(\widetilde{\bm u}_n)=&-\sqrt{n}\widetilde{\bm u}_n^{\prime}\bm T_n+\sqrt{n}\widetilde{\bm u}_n^{\prime}J\sqrt{n}\widetilde{\bm u}_n+o_p(\sqrt{n}\|\widetilde{\bm u}_n\|+n\|\widetilde{\bm u}_n\|^2) \\
			\geq &-\sqrt{n}\|\widetilde{\bm u}_n\|[\|\bm T_n\|+o_p(1)]+n\|\widetilde{\bm u}_n\|^2[\lambda_{\min}+o_p(1)], \nonumber
		\end{align}
		where $\lambda_{\min}$ is the smallest eigenvalue of $J$, and $\bm T_n=n^{-1/2}\sum_{t=1}^{n}w_t\dot{q}_t(\bm\theta(\tau))\psi_{\tau}(\eta_{t,\tau})$. 
		Note that $E(\psi_{\tau}(\eta_{t,\tau})|\mathcal{F}_{t-1})=0$ and the positive definite matrix $\Omega_{0w}(\tau)=E\left[w_t^2\dot{q}_t(\bm\theta(\tau))\dot{q}_t^{\prime}(\bm\theta(\tau))\right]<\infty$ by Assumptions \ref{assum-Space} and \ref{assum-RandomWeight}. Then by the Lindeberg–L\'{e}vy theorem for martingales \citep{Billingsley1961} and the Cram\'{e}r-Wold device, together with the stationarity and ergodicity of $y_t$ by Assumption \ref{assum-Process}, $\bm T_n$ converges in distribution to a normal random variable with mean zero and variance matrix $\tau(1-\tau)\Omega_{0w}(\tau)$ as $n\rightarrow \infty$.
		
		Note that $\lambda_{\min}>0$ as $\Omega_{1w}(\tau)$ is positive definite, and $\|\bm T_n\|<\infty$ by Assumptions \ref{assum-Process}, \ref{assum-Space} and \ref{assum-RandomWeight}. Since $\widetilde{H}_n(\widetilde{\bm u}_n)\leq 0$, then we have 
		\begin{align}\label{root-n}
			\sqrt{n}\|\widetilde{\bm u}_n\|\leq [\lambda_{\min}+o_p(1)]^{-1}[\|\bm T_n\|+o_p(1)]=O_p(1).
		\end{align}
		This together with Theorem \ref{thm-WCQE-uniform-consistency} verifies the $\sqrt{n}$-consistency of $\widetilde{\bm\theta}_{wn}(\tau)$, i.e. $\sqrt{n}(\widetilde{\bm\theta}_{wn}(\tau)-\bm\theta(\tau))=O_p(1)$. 
		
		Let $\sqrt{n}\bm u_n^*=J^{-1}\bm T_n/2=\Omega_{1w}^{-1}(\tau)\bm T_n$, then we have
		\begin{align*}
			\sqrt{n}\bm u_n^*\rightarrow N\left(\bm 0,\tau(1-\tau)\Omega_{1w}^{-1}(\tau)\Omega_{0w}(\tau)\Omega_{1w}^{-1}(\tau)\right)
		\end{align*}
		in distribution as $n\rightarrow \infty$. Therefore, it suffices to show that $\sqrt{n}\bm u_n^*-\sqrt{n}\widetilde{\bm u}_n=o_p(1)$.
		By \eqref{Hn} and \eqref{root-n}, we have
		\begin{align*}
			\widetilde{H}_n(\widetilde{\bm u}_n)=&-\sqrt{n}\widetilde{\bm u}_n^{\prime}\bm T_n+\sqrt{n}\widetilde{\bm u}_n^{\prime}J\sqrt{n}\widetilde{\bm u}_n+o_p(1) \\
			=&-2\sqrt{n}\widetilde{\bm u}_n^{\prime}J\sqrt{n}\bm u_n^*+\sqrt{n}\widetilde{\bm u}_n^{\prime}J\sqrt{n}\widetilde{\bm u}_n+o_p(1),
		\end{align*}
		and
		\begin{align*}
			\widetilde{H}_n(\bm u_n^*)=&-\sqrt{n}\bm u_n^{*\prime}\bm T_n+\sqrt{n}\bm u_n^{*\prime}J\sqrt{n}\bm u_n^*+o_p(1)
			=-\sqrt{n}\bm u_n^{*\prime}J\sqrt{n}\bm u_n^*+o_p(1).
		\end{align*}
		It follows that
		\begin{align}\label{normality}
			\widetilde{H}_n(\widetilde{\bm u}_n)-\widetilde{H}_n(\bm u_n^*)=&(\sqrt{n}\widetilde{\bm u}_n-\sqrt{n}\bm u_n^*)^{\prime}J(\sqrt{n}\widetilde{\bm u}_n-\sqrt{n}\bm u_n^*)+o_p(1)  \nonumber \\
			\geq & \lambda_{\min}\|\sqrt{n}\widetilde{\bm u}_n-\sqrt{n}\bm u_n^*\|^2+o_p(1).
		\end{align}
		Since $\widetilde{H}_n(\widetilde{\bm u}_n)-\widetilde{H}_n(\bm u_n^*)=n[\widetilde{L}_n(\bm\theta(\tau)+\widetilde{\bm u}_n)-\widetilde{L}_n(\bm\theta(\tau)+\bm u_n^*)]\leq 0$ a.s., then \eqref{normality} implies that $\|\sqrt{n}\widetilde{\bm u}_n-\sqrt{n}\bm u_n^*\|=o_p(1)$. We verify the asymptotic normality of $\widetilde{\bm\theta}_{wn}(\tau)$, and the proof is hence accomplished.
	\end{proof}
	
	\begin{proof}[Proof of Theorem \ref{thm-WCQE-weak-convergence}]
		Recall that $\widetilde{L}_n(\bm\theta,\tau)=n^{-1}\sum_{t=1}^{n}w_t\rho_{\tau}(y_t-\widetilde{q}_t(\bm\theta))$ and $L_n(\bm\theta,\tau)=n^{-1}\sum_{t=1}^{n}w_t\rho_{\tau}(y_t-q_t(\bm\theta))$. 
		Let $\psi_{\tau}(x)=\tau-I(x<0)$ and $\eta_{t,\tau}=y_t-q_t(\bm\theta(\tau))$. Denote 
		\[\bm T_n(\tau)=\dfrac{1}{\sqrt{n}}\sum_{t=1}^{n}w_t\dot{q}_t(\bm\theta(\tau))\psi_{\tau}(\eta_{t,\tau}) \hspace{2mm}\text{and}\hspace{2mm} J_n(\tau)=\dfrac{1}{2n}\sum_{t=1}^{n}f_{t-1}(F_{t-1}^{-1}(\tau))w_t\dot{q}_t(\bm\theta(\tau))\dot{q}_t^{\prime}(\bm\theta(\tau)).\]
		Note that we have established the uniform consistency and finite dimensional convergence of $\widetilde{\bm\theta}_{wn}(\tau)$ in Theorem \ref{thm-WCQE-uniform-consistency} and Corollary \ref{thm-WCQE}, respectively. 
		By Corollary 2.2 of \cite{Newey1991}, to show the weak convergence of $\widetilde{\bm\theta}_{wn}(\tau)$, we need to prove the stochastic equicontinuity. 
		As a result, it suffices to verify the following claims:
		\begin{itemize}
			\item[(1)] If $E|y_t|^s<\infty$ for some $0<s\leq 1$ and Assumptions \ref{assum-Process}--\ref{assum-Tightness} hold, then for any sequence of random
			variables $\bm u_n\equiv\bm u_n(\tau)$ such that $\sup_{\tau\in\mathcal{T}}\bm u_n(\tau)=o_p(1)$,  
			\[
			n[\widetilde{L}_n(\bm u_n+\bm\theta(\tau),\tau)-\widetilde{L}_n(\bm\theta(\tau),\tau)]-n[L_n(\bm u_n+\bm\theta(\tau),\tau)-L_n(\bm\theta(\tau),\tau)]=o_p(\sqrt{n}\|\bm u_n\|+n\|\bm u_n\|^2),
			\]
			where the remainder term is uniform in $\tau\in\mathcal{T}$. 
			\item[(2)] If $E|y_t|^s<\infty$ for some $0<s\leq 1$ and Assumptions \ref{assum-Process}--\ref{assum-Tightness} hold, then for any sequence of random
			variables $\bm u_n\equiv\bm u_n(\tau)$ such that $\sup_{\tau\in\mathcal{T}}\bm u_n(\tau)=o_p(1)$,
			\[
			n[L_n(\bm u_n+\bm\theta(\tau),\tau)-L_n(\bm\theta(\tau),\tau)]=-\sqrt{n}\bm u_n^{\prime}\bm T_n(\tau)+\sqrt{n}\bm u_n^{\prime}J_n(\tau)\sqrt{n}\bm u_n +o_p(\sqrt{n}\|\bm u_n\|+n\|\bm u_n\|^2),
			\]
			where the remainder term is uniform in $\tau\in\mathcal{T}$.
			\item[(3)] If $E|y_t|^s<\infty$ for some $0<s\leq 1$ and Assumptions \ref{assum-Process}--\ref{assum-Tightness} hold, then as $n\rightarrow\infty$, $\bm T_n(\bm{\cdot}) \rightsquigarrow \mathbb{G}_0(\bm{\cdot}) \:\: \text{in} \:\: (\ell^{\infty}(\mathcal{T}))^3$, where $\mathbb{G}_0(\bm{\cdot})$ is a zero mean Gaussian process with covariance kernel $(\min\{\tau_1, \tau_2\}-\tau_1\tau_2)\Omega_{0w}(\tau_1, \tau_2)$ with $\Omega_{0w}(\tau_1, \tau_2)=E\left[w_t^2\dot{q}_t(\bm\theta(\tau_1))\dot{q}_t^{\prime}(\bm\theta(\tau_2))\right]$.
		\end{itemize}
		
		For Claim (1), it extends the pointwise result in Lemma \ref{lem3} to the uniform version. For $\bm\theta=(\omega,\alpha_{1},\beta_{1})^{\prime}\in\Theta$, note that $|\alpha_{1}|\leq\overline{c}<\infty$ and $0<\beta_{1}\leq\rho<1$ for $\tau\in\mathcal{T}$ by Assumption \ref{assum-Space}. 
		Then Lemma \ref{lem0} holds for $\tau\in\mathcal{T}$, that is, under Assumption \ref{assum-Space}, we have 
		\begin{align}\label{lem0-uniform}
			&\sup_{\tau\in\mathcal{T}}\sup_{\Theta}|q_t(\bm\theta)-\widetilde{q}_t(\bm\theta)|
			\leq C\rho^{t}\varsigma_{\rho} \;\text{and}\; \nonumber\\
			&\sup_{\tau\in\mathcal{T}}\sup_{\Theta}\|\dot{q}_t(\bm\theta)-\dot{\widetilde{q}}_t(\bm\theta)\|
			\leq C\rho^{t}(\varsigma_{\rho}+t\varsigma_{\rho}+\xi_{\rho}),
		\end{align}
		where $\varsigma_{\rho}=\sum_{s=0}^{\infty}\rho^{s}|y_{-s}|$ and $\xi_{\rho}=\sum_{s=0}^{\infty}s\rho^{s}|y_{-s}|$ with the constant $\rho\in (0,1)$.
		As a result, \eqref{initialHnrep}--\eqref{Pi4tilde} in the proof of Lemma \ref{lem-Tightness2} hold uniformly in $\tau\in\mathcal{T}$.

		For Claim (2), we generalize the notation in Lemma \ref{lem2} and decompose $n[L_n(\bm u+\bm\theta(\tau))-L_n(\bm\theta(\tau))]$ as follows:
		\begin{align}\label{Gnrep-process}
			& n[L_n(\bm u+\bm\theta(\tau))-L_n(\bm\theta(\tau))] \nonumber\\
			= & -\sqrt{n}\bm u^{\prime}\bm T_n(\tau)-\sqrt{n}\bm u^{\prime}R_{1n}(\bm u^*,\tau)\sqrt{n}\bm u + \sum_{i=2}^5R_{in}(\bm u,\tau),
		\end{align}
		where $\bm T_n(\tau)=n^{-1/2}\sum_{t=1}^{n}w_t\dot{q}_t(\bm\theta(\tau))\psi_{\tau}(y_t-q_t(\bm\theta(\tau)))$, 
		\begin{align*}
			R_{1n}(\bm u, \tau)&=\dfrac{1}{2n}\sum_{t=1}^{n}w_t\ddot{q}_t(\bm u+\bm\theta(\tau))\psi_{\tau}(y_t-q_t(\bm\theta(\tau))), \\
			R_{2n}(\bm u, \tau)&=\bm u^{\prime}\sum_{t=1}^{n}w_t\dot{q}_t(\bm\theta(\tau))E[\xi_{1t}(\bm u, \tau)|\mathcal{F}_{t-1}], \\
			R_{3n}(\bm u, \tau)&=\bm u^{\prime}\sum_{t=1}^{n}w_t\dot{q}_t(\bm\theta(\tau))E[\xi_{2t}(\bm u, \tau)|\mathcal{F}_{t-1}], \\
			R_{4n}(\bm u, \tau)&=\bm u^{\prime}\sum_{t=1}^{n}w_t\dot{q}_t(\bm\theta(\tau))\{\xi_{t}(\bm u, \tau)-E[\xi_{t}(\bm u, \tau)|\mathcal{F}_{t-1}]\} \hspace{2mm}\text{and}\hspace{2mm} \\
			R_{5n}(\bm u, \tau)&=\dfrac{\bm u^{\prime}}{2}\sum_{t=1}^{n}w_t\ddot{q}_t(\bm u^*+\bm\theta(\tau))\xi_{t}(\bm u, \tau)\bm u
		\end{align*}
		with $\psi_{\tau}(x)=\tau-I(x<0)$, $\bm u^*$ between $\bm 0$ and $\bm u$, $\nu_t(\bm u, \tau)=q_t(\bm u+\bm\theta(\tau))-q_t(\bm\theta(\tau))$, 
		\begin{align*}
			\xi_{t}(\bm u, \tau)&=\int_{0}^{1}\left[I(y_t\leq F_{t-1}^{-1}(\tau)+\nu_t(\bm u, \tau)s)-I(y_t\leq F_{t-1}^{-1}(\tau))\right]ds  \\
			\xi_{1t}(\bm u, \tau)&=\int_{0}^{1}\left[I(y_{t}\leq F_{t-1}^{-1}(\tau)+\bm u^{\prime}\dot{q}_t(\bm\theta(\tau))s)-I(y_t\leq F_{t-1}^{-1}(\tau))\right]ds \hspace{2mm}\text{and}\hspace{2mm} \\
			\xi_{2t}(\bm u, \tau)&=\int_{0}^{1}\left[I(y_{t}\leq F_{t-1}^{-1}(\tau)+\nu_t(\bm u, \tau)s)-I(y_{t}\leq F_{t-1}^{-1}(\tau)+\bm u^{\prime}\dot{q}_t(\bm\theta(\tau))s)\right]ds.
		\end{align*}
		Given the pointwise result in Lemma \ref{lem2}, to establish Claim (2) it suffices to show the stochastic equicontinuity related to $R_{in}(\bm u, \tau)$ for $i=1,\ldots,5$.  
		
		For $R_{1n}(\bm u, \tau)$, by Lemma \ref{lem00}(ii) and the fact that $|\psi_{\tau}(\eta_{t,\tau})|\leq 1$, we have
		\[E\left[\sup_{\tau\in\mathcal{T}}\sup_{\|\bm u\|\leq \eta}\left\|w_t \ddot{q}_t(\bm u+\bm\theta(\tau))\psi_{\tau}(\eta_{t,\tau})\right\|\right]\leq CE\left[w_t\sup_{\bm\theta\in\Theta}\left\|\ddot{q}_t(\bm\theta)\right\|\right]<\infty.\]
		Moreover, $E\left[w_t \ddot{q}_t(\bm u+\bm\theta(\tau))\psi_{\tau}(\eta_{t,\tau})\right]=0$ by iterated-expectation and the fact that $E[\psi_{\tau}(\eta_{t,\tau})| \mathcal{F}_{t-1}]=0$. 	
		Since $\{y_t\}$ is strictly stationary and ergodic under Assumption \ref{assum-Process}, then by Theorem 3.1 in \cite{Ling_McAleer2003}, we can show that
		\begin{align*}
			\sup_{\tau\in\mathcal{T}}\sup_{\|\bm u\|\leq \eta}\|R_{1n}(\bm u, \tau)\|=o_p(1).
		\end{align*}
		This together with $\sup_{\tau\in\mathcal{T}}\bm u_n(\tau)=o_p(1)$, implies that
		\begin{align}\label{R1n-process}
			\sup_{\tau\in\mathcal{T}}\|R_{1n}(\bm u_n(\tau), \tau)\|=o_p(1).
		\end{align}
		
		We next focus on $R_{2n}(\bm u, \tau)$. By Taylor expansion, we have
		\begin{align*}
			R_{2n}(\bm u, \tau)=\sqrt{n}\bm u^{\prime}J_n(\tau)\sqrt{n}\bm u+\sqrt{n}\bm u^{\prime}\Pi_{1n}(\bm u, \tau)\sqrt{n}\bm u,
		\end{align*}
		where $J_n(\tau)=(2n)^{-1}\sum_{t=1}^{n}f_{t-1}(F_{t-1}^{-1}(\tau))w_t\dot{q}_t(\bm\theta(\tau))\dot{q}_t^{\prime}(\bm\theta(\tau))$ and
		\[\Pi_{1n}(\bm u, \tau)=\dfrac{1}{n}\sum_{t=1}^{n}w_t\dot{q}_t(\bm\theta(\tau))\dot{q}_t^{\prime}(\bm\theta(\tau))\int_{0}^{1}[f_{t-1}(F_{t-1}^{-1}(\tau)+\bm u^{\prime}\dot{q}_t(\bm\theta(\tau))s^*)-f_{t-1}(F_{t-1}^{-1}(\tau))]sds.\]
		By Taylor expansion and $\sup_{x}|\dot{f}_{t-1}(x)|<\infty$ under Assumption \ref{assum-ConditionalDensity}, for any $\eta>0$, we have
		\begin{align*}
			\sup_{\tau\in\mathcal{T}}\sup_{\|\bm u\|\leq \eta}\|\Pi_{1n}(\bm u, \tau)\|&\leq \dfrac{1}{n}\sum_{t=1}^{n}\sup_{\tau\in\mathcal{T}}\sup_{\|\bm u\|\leq \eta} \|w_t\dot{q}_t(\bm\theta(\tau))\dot{q}_t^{\prime}(\bm\theta(\tau))\sup_{x}|\dot{f}_{t-1}(x)|\bm u^{\prime}\dot{q}_t(\bm\theta(\tau))\| \\
			&\leq C\eta \cdot \dfrac{1}{n}\sum_{t=1}^{n}w_t\sup_{\tau\in\mathcal{T}}\|\dot{q}_t(\bm\theta(\tau))\|^3.
		\end{align*}
		Then by Assumption \ref{assum-Space}(ii) and Lemma \ref{lem00}(i), it holds that
		\begin{align*}
			E\left(\sup_{\tau\in\mathcal{T}}\sup_{\|\bm u\|\leq \eta}\|\Pi_{1n}(\bm u, \tau)\|\right)\leq  C\eta E\left(w_t\sup_{\Theta}\|\dot{q}_t(\bm\theta)\|^3\right) \leq  C\eta
		\end{align*}
		tends to $0$ as $\eta \to 0$. Similar to \eqref{epsdelta1} and \eqref{epsdelta2}, for $\sup_{\tau\in\mathcal{T}}\bm u_n(\tau)=o_p(1)$, we can show that $\sup_{\tau\in\mathcal{T}}\|\Pi_{1n}(\bm u_n, \tau)\|=o_p(1)$. 
		It follows that
		\begin{align}\label{R2n-process}
			\sup_{\tau\in\mathcal{T}}[R_{2n}(\bm u_n, \tau)-\sqrt{n}\bm u_n^{\prime}J_n(\tau)\sqrt{n}\bm u_n]=o_p(n\|\bm u_n\|^2).
		\end{align}
		
		For $R_{3n}(\bm u, \tau)$, by Taylor expansion, the Cauchy-Schwarz inequality and the strict stationarity and ergodicity of $y_t$ under Assumption \ref{assum-Process}, together with \eqref{xi2}, Assumption \ref{assum-Space}(ii), Lemma \ref{lem00} and $\sup_{x}f_{t-1}(x)<\infty$ by Assumption \ref{assum-ConditionalDensity}, for any $\eta>0$, we have
		\begin{align*}
			E\left(\sup_{\tau\in\mathcal{T}}\sup_{\|\bm u\|\leq \eta}\dfrac{|R_{3n}(\bm u)|}{n\|\bm u\|^2}\right)
			\leq & \dfrac{\eta}{n}\sum_{t=1}^{n}E\left\{w_t\sup_{\tau\in\mathcal{T}}\left\|\dot{q}_t(\bm\theta(\tau))\right\|\dfrac{1}{2}\sup_{x}f_{t-1}(x)\sup_{\bm\theta\in\Theta}\left\|\ddot{q}_t(\bm\theta)\right\| \right\} \\
			\leq &  C\eta E\left\{\sqrt{w_t}\sup_{\bm\theta\in\Theta}\left\|\dot{q}_t(\bm\theta)\right\|\cdot\sqrt{w_t}\sup_{\bm\theta\in\Theta}\left\|\ddot{q}_t(\bm\theta)\right\| \right\} \\
			\leq & C\eta \left[E\left(w_t\sup_{\bm\theta\in\Theta}\left\|\dot{q}_t(\bm\theta)\right\|^2\right)\right]^{1/2}\left[E\left(w_t\sup_{\bm\theta\in\Theta}\left\|\ddot{q}_t(\bm\theta)\right\|^2\right)\right]^{1/2}
		\end{align*}
		tends to $0$ as $\eta \to 0$. Similar to \eqref{epsdelta1} and \eqref{epsdelta2}, we can show that
		\begin{align}\label{R3n-process}
			\sup_{\tau\in\mathcal{T}}|R_{3n}(\bm u_n,\tau)|=o_p(n\|\bm u_n\|^2).
		\end{align}
		
		For $R_{4n}(\bm u, \tau)$ and $R_{5n}(\bm u, \tau)$, by Lemma \ref{lem-Tightness2} and $\sup_{\tau\in\mathcal{T}}\bm u_n(\tau)=o_p(1)$, we have
		\begin{align}\label{Rn45-process}
			\sup_{\tau\in\mathcal{T}}R_{4n}(\bm u_n, \tau)=o_p(\sqrt{n}\|\bm u_n\|+n\|\bm u_n\|^2) \;\;\text{and}\;\; 
			\sup_{\tau\in\mathcal{T}}R_{5n}(\bm u_n, \tau)=o_p(n\|\bm u_n\|^2).
		\end{align}
		Combining \eqref{Gnrep-process}--\eqref{Rn45-process}, it follows that Claim (2) holds.

		Finally, we consider Claim (3). 
		By the Lindeberg–L\'{e}vy theorem for martingales \citep{Billingsley1961} and the Cram\'{e}r-Wold device, together with the stationarity and ergodicity of $y_t$ by Assumption \ref{assum-Process} and $E(w_t)<\infty$ and $E(w_{t}|y_{t-j}|^3)<\infty$ for all $j\geq 1$ by Assumption \ref{assum-RandomWeight}, the finite-dimensional convergence of $\bm T_n(\tau)$ has been established in the proof of Corollary \ref{thm-WCQE}, that is, $\bm T_n(\tau)\to_d N(\bm 0,\tau(1-\tau)\Omega_{0w}(\tau,\tau))$ as $n\rightarrow \infty$, where $\Omega_{0w}(\tau_1, \tau_2)=E\left[w_t^2\dot{q}_t(\bm\theta(\tau_1))\dot{q}_t^{\prime}(\bm\theta(\tau_2))\right]$. It suffices to verify the stochastic equicontinuity of $\bm e_n^{\prime}\bm T_n(\tau)$ in $\ell^{\infty}(\mathcal{T})$ with $\bm e_n\in\mathbb{R}^3$ being an arbitrary vector. Without loss of generality, we will assume that $\bm e_n$ is a sequence of vectors with $\|\bm e_n\|=1$. 
		It holds that
		\begin{align}\label{Tn-Tightness-decompose}
			\bm e_n^{\prime}\bm T_n(\tau_2)-\bm e_n^{\prime}\bm T_n(\tau_1) = \dfrac{1}{\sqrt{n}}\sum_{t=1}^{n}(\bm e_n^{\prime}\bm a_t+\bm e_n^{\prime}\bm b_t),
		\end{align}
		where $\bm a_t=w_t[\dot{q}_t(\bm\theta(\tau_2))-\dot{q}_t(\bm\theta(\tau_1))] \psi_{\tau_1}(y_t-q_t(\bm\theta(\tau_1)))$ and $\bm b_t=w_t\dot{q}_t(\bm\theta(\tau_2))[c_t-E(c_t|\mathcal{F}_{t-1})]$ with $c_t\equiv c_t(\tau_1,\tau_2)=I(y_t<q_t(\bm\theta(\tau_1)))-I(y_t<q_t(\bm\theta(\tau_2)))$. 
		By Lemma \ref{lem-Tightness1}(ii), the fact that $|\psi_{\tau}(x)|<1$, the strict stationarity and ergodicity of $y_t$ under Assumption \ref{assum-Process}, and $E(w_t^2\Delta_{\rho,t}^2)<\infty$ under Assumption \ref{assum-RandomWeight} with $\Delta_{\rho,t}=1+\sum_{j=1}^\infty\rho^{j-1}|y_{t-j}|+\sum_{j=2}^\infty(j-1)\rho^{j-2}|y_{t-j}|+\sum_{j=3}^\infty(j-1)(j-2)\rho^{j-3}|y_{t-j}|+\sum_{j=4}^\infty(j-1)(j-2)(j-3)\rho^{j-4}|y_{t-j}|$, we have
		\begin{align}\label{at-2nd-moment}
			E[(\bm e_n^{\prime}\bm a_t)^2] \leq E(\|\bm e_n\|^2w_t^2\|\dot{q}_t(\bm\theta(\tau_2))-\dot{q}_t(\bm\theta(\tau_1))\|^2) \leq C(\tau_2-\tau_1)^2.  
		\end{align}
		Moreover, note that $I(X<a)-I(X<b)=I(b< X < a)-I(b> X > a)$ and $E\{[I(X<a)-I(X<b)]^2\}=E[I(b< X < a)+I(b> X > a)]=|\text{Pr}(X<a)-\text{Pr}(X<b)|$. These together with $\var(X)\leq E(X^2)$, imply that
		\begin{align}\label{ct-2nd-moment}
			E\{[c_t-E(c_t|\mathcal{F}_{t-1})]^2|\mathcal{F}_{t-1}\} \leq  
			|F_{t-1}(q_t(\bm\theta(\tau_2)))-F_{t-1}(q_t(\bm\theta(\tau_1)))|=|\tau_2-\tau_1|.
		\end{align}
		Then by iterative-expectation, together with $E(w_t^2\|\dot{q}_t(\bm\theta(\tau_2))\|^2)<\infty$ implied by Lemma \ref{lem00}(i), we have
		\begin{align}\label{bt-2nd-moment}
			E[(\bm e_n^{\prime}\bm b_t)^2] \leq E\left\{\|\bm e_n\|^2w_t^2\|\dot{q}_t(\bm\theta(\tau_2))\|^2 E\{[c_t-E(c_t|\mathcal{F}_{t-1})]^2|\mathcal{F}_{t-1}\}\right\} \leq C|\tau_2-\tau_1|.
		\end{align}  
		Thus, by the Cauchy-Schwarz inequality and $E(\bm a_t|\mathcal{F}_{t-1})=E(\bm b_t|\mathcal{F}_{t-1})=\bm 0$, together with \eqref{at-2nd-moment} and \eqref{bt-2nd-moment}, it can be verified that
		\begin{align*}
			E[\bm e_n^{\prime}\bm T_n(\tau_2)-\bm e_n^{\prime}\bm T_n(\tau_1)]^2=& E[(\bm e_n^{\prime}\bm a_t)^2] + E[(\bm e_n^{\prime}\bm b_t)^2] + 2E[(\bm e_n^{\prime}\bm a_t)(\bm e_n^{\prime}\bm b_t)] \\
			\leq & 4\{E[(\bm e_n^{\prime}\bm a_t)^2]E[(\bm e_n^{\prime}\bm b_t)^2]\}^{1/2} \leq C|\tau_2-\tau_1|^{3/2}.
		\end{align*}
		Therefore, for any $\tau_1 \leq \tau \leq \tau_2$ and $\bm e_n\in\mathbb{R}^3$, by the Cauchy-Schwarz inequality, we have
		\begin{align*}
			& E\{|\bm e_n^{\prime}\bm T_n(\tau)-\bm e_n^{\prime}\bm T_n(\tau_1)||\bm e_n^{\prime}\bm T_n(\tau)-\bm e_n^{\prime}\bm T_n(\tau_2)|\} \\
			\leq & \left\{E[\bm e_n^{\prime}\bm T_n(\tau)-\bm e_n^{\prime}\bm T_n(\tau_1)]^2\right\}^{1/2}\left\{E[\bm e_n^{\prime}\bm T_n(\tau)-\bm e_n^{\prime}\bm T_n(\tau_2)]^2\right\}^{1/2} \leq C|\tau_2-\tau_1|^{3/2}.
		\end{align*}
		This proves the asymptotic tightness. 
		Then by Theorem 13.5 of \cite{Billingsley1999}, we establish the weak convergence of $\bm T_n(\tau)$ in $(\ell^{\infty}(\mathcal{T}))^3$. 
		
		By the ergodic theorem and Assumptions \ref{assum-Process}--\ref{assum-RandomWeight}, we have $J_n(\tau)=\Omega_{1w}(\tau)/2+o_p(1)$ uniformly on $\mathcal{T}$, where $\Omega_{1w}(\tau)=E\left[f_{t-1}(F_{t-1}^{-1}(\tau))w_{t}\dot{q}_t(\bm\theta(\tau))\dot{q}_t^{\prime}(\bm\theta(\tau))\right]$. 
		Moreover, $\Omega_{1w}(\tau)$ is Lipschitz continuous on $\mathcal{T}$ by Assumptions \ref{assum-Process}--\ref{assum-Tightness}. This together with the positive definiteness of $\Omega_{1w}(\tau)$, implies that $\Omega_{1w}^{-1}(\tau)$ is continuous on $\mathcal{T}$, and thus $\Omega_{1w}^{-1}(\tau)\bm T_n(\tau)$ is tight on $\mathcal{T}$ (Theorem 7.3 of \cite{Billingsley1999}).  
		Similar to the proof of Corollary \ref{thm-WCQE}, then by Claims (1)--(3), we can show that
		\begin{align}\label{weak-convergence-representation}
			\sqrt{n}(\widetilde{\bm\theta}_{wn}(\cdot)-\bm\theta(\cdot)) = \Omega_{1w}^{-1}(\cdot)\bm T_n(\cdot) + o_p(1)
			\rightsquigarrow \mathbb{G}(\bm{\cdot}) \:\: \text{in} \:\: (\ell^{\infty}(\mathcal{T}))^3,
		\end{align}
		where the remainder term is uniform in $\tau\in\mathcal{T}$, and $\mathbb{G}(\bm{\cdot})$ is a zero mean Gaussian process with covariance kernel $\Sigma(\tau_1, \tau_2)$.
	\end{proof}
	
	\subsection{Proof of Corollary \ref{thm-test1}}
	
	\begin{proof}[Proof of Corollary \ref{thm-test1}]
		Recall that $v_{n}(\tau)=R\widetilde{\bm\theta}_{wn}(\tau)-\widetilde{r}=R[\widetilde{\bm\theta}_{wn}(\tau)-\int_{\mathcal{T}}\widetilde{\bm\theta}_{wn}(\tau)d\tau]$ and $v_{0}(\tau)=R[\mathbb{G}(\tau)-\int_{\mathcal{T}}\mathbb{G}(\tau)d\tau]$. 
		By continuous mapping theorem and \eqref{weak-convergence-representation} by Theorem \ref{thm-WCQE-weak-convergence}, under the null hypothesis $H_0$ that $R\bm\theta(\tau)=r$, we have
		\begin{align*}
			\sqrt{n}v_{n}(\tau)&=R\sqrt{n}(\widetilde{\bm\theta}_{wn}(\tau)-\bm\theta(\tau))-\sqrt{n}(\widetilde{r}-r)+\sqrt{n}(R\bm\theta(\tau)-r) \\
			&\overset{H_0}{=}R\sqrt{n}(\widetilde{\bm\theta}_{wn}(\tau)-\bm\theta(\tau))-R\int_{\mathcal{T}}\sqrt{n}(\widetilde{\bm\theta}_{wn}(\tau)-\bm\theta(\tau))d\tau 
			\overset{H_0}{\to} v_{0}(\tau).
		\end{align*}
		Since the covariance function of $v_{0}(\bm{\cdot})$ is nondegenerate by assumption, then by Theorem 11.1 in \cite{Davydov1998}, the distribution of functionals $S=\int_{\mathcal{T}}v_0^2(\tau)d\tau$ is absolutely continuous on (0,$\infty$). As a result, by continuous mapping theorem, it follows that $S_{n}=n\int_{\mathcal{T}}v_{n}^2(\tau)d\tau \to_d S\equiv\int_{\mathcal{T}} v_{0}^2(\tau)d\tau$ as $n\to\infty$. 
	\end{proof}

	\subsection{Proof of Theorem \ref{thm-WCQR}}
	\begin{proof}[Proof of Theorem \ref{thm-WCQR}]
		Recall that $q_{t,\tau}(\bm\varphi) =Q_{\tau}(\lambda)h_t(\bm\phi)$ and $\widetilde{q}_{t,\tau}(\bm\varphi) =Q_{\tau}(\lambda)\widetilde{h}_t(\bm\phi)$, where $\bm\varphi=(\bm\phi^{\prime}, \lambda)^{\prime}=(a_0, a_1, b_1, \lambda)^{\prime}$, $h_t(\bm\phi)=a_0(1-b_1)^{-1}+a_1\sum_{j=1}^{\infty} b_1^{j-1}|y_{t-j}|$ and $\widetilde{h}_t(\bm\phi)=a_0(1-b_1)^{-1}+a_1\sum_{j=1}^{t-1} b_1^{j-1}|y_{t-j}|$.  
		Define $L^*(\bm\varphi)=E[w_t\ell_t^*(\bm\varphi)]$, $L_n^*(\bm\varphi)=n^{-1}\sum_{t=1}^{n}w_t\ell_t^*(\bm\varphi)$ and $\widetilde{L}_n^*(\bm\varphi)=n^{-1}\sum_{t=1}^{n}w_t\widetilde{\ell}_t^*(\bm\varphi)$, where 
		$\ell_t^*(\bm\varphi)=\sum_{k=1}^K\rho_{\tau_k}(y_t-q_{t,\tau_k}(\bm\varphi))$ and $\widetilde{\ell}_t^*(\bm\varphi)=\sum_{k=1}^K\rho_{\tau_k}(y_t-\widetilde{q}_{t,\tau_k}(\bm\varphi))$. 
		
		\vspace{5mm}
		\noindent(I) Proof of (i) $\check{\bm\varphi}_{wn} \to_p \bm\varphi_0^*$
		
		To show the consistency, we first verify the following claims:
		\begin{itemize}
			\item[(1)] $\sup\limits_{\Phi}|\widetilde{L}_n^*(\bm\varphi)-L_n^*(\bm\varphi)|=o_p(1)$;
			\item[(2)] $E[\sup\limits_{\Phi}w_t|\ell_t^*(\bm\varphi)|]<\infty$;
			\item[(3)] $L^*(\bm\varphi)$ has a unique minimum at $\bm\varphi_0^*$.
		\end{itemize}
		For Claim (1), by the Lipschitz continuity of $\rho_{\tau}(\cdot)$, strict stationarity and ergodicity of $\{y_t\}$ by Assumption \ref{assum-Process-Mixing}, $E(w_{t}\varsigma_{\rho})<\infty$ by Assumption \ref{assum-RandomWeight} and Lemma \ref{lem0-Tukey}(i), we have
		\begin{align*}
			\sup_{\Phi}|\widetilde{L}_n^*(\bm\varphi)-L_n^*(\bm\varphi)|
			&\leq \dfrac{1}{n}\sum_{t=1}^{n}\sum_{k=1}^K\sup_{\Phi}w_t|\rho_{\tau_k}(y_t-\widetilde{q}_{t,\tau_k}(\bm\varphi))-\rho_{\tau_k}(y_t-q_{t,\tau_k}(\bm\varphi))| \\
			&\leq\dfrac{C}{n}\sum_{t=1}^{n}\sum_{k=1}^K\sup_{\Phi}w_t|q_{t,\tau_k}(\bm\varphi)-\widetilde{q}_{t,\tau_k}(\bm\varphi)| \\
			&\leq \dfrac{CK}{n}\sum_{t=1}^{n}\rho^{t}w_t\varsigma_{\rho}=o_p(1).
		\end{align*} 
		
		We next show Claim (2). Since $\max\{a_0, a_1\} \leq \overline{c}<\infty$ and $0<b_{1}\leq\rho<1$ by Assumption \ref{assum-SpaceTukey}, $E(w_t)<\infty$ and $E(w_{t}|y_{t-j}|^3)<\infty$ for all $j\geq 1$ by Assumption \ref{assum-RandomWeight}, and by the fact that $|\rho_{\tau}(x)|\leq |x|$ and Assumption \ref{assum-Process-Mixing}, it holds that  
		\begin{align*}
			E[\sup_{\Phi}w_t|\ell_t^*(\bm\varphi)|]&\leq \sum_{k=1}^K\left\{E[\sup_{\Phi}w_t|y_t|]+E[\sup_{\Phi}|Q_{\tau_k}(\lambda)|w_th_t(\bm\phi)]\right\} \\
			&\leq KE(w_t|y_t|) + CK\overline{c}E\left[w_{t}\left(1+\sum_{j=1}^{\infty}\rho^{j-1}|y_{t-j}|\right)\right] < \infty.
		\end{align*}
		Hence, Claim (2) is verified. 
		For Claim (3), it holds by definition of $\bm\varphi_0^*$ in \eqref{pseudo-true} and Assumption \ref{assum-SpaceTukey}(i). 
		
		Note that $L^*(\bm\varphi)$ is a measurable function of $y_t$ in Euclidean space for each $\bm\varphi\in\Phi$, and $L(\bm\varphi)$ is a continuous function of $\bm\varphi\in\Phi$ for each $y_t$. Then by Theorem 3.1 of \cite{Ling_McAleer2003}, together with Claim (1) and the strict stationarity and ergodicity of $\{y_t\}$ implied by Assumption \ref{assum-Process-Mixing}, we have
		\begin{align*}
			\sup_{\bm\varphi\in\Phi}|L_n^*(\bm\varphi)-L^*(\bm\varphi)|=o_p(1).
		\end{align*}
		This together with Claim (1), implies that
		\begin{align}\label{Tukeyloss-consistency}
			\sup_{\bm\varphi\in\Phi}|\widetilde{L}_n^*(\bm\varphi)-L^*(\bm\varphi)|=o_p(1).
		\end{align}	
		We next verify the consistency using the standard consistency argument in \cite{Amemiya1985}. 
		For any $c>0$, with probability tending to 1 uniformly in $\epsilon\geq c$, we have: (a) by $\check{\bm\varphi}_{wn}=\argmin_{\bm\varphi\in\Phi}\widetilde{L}_n^*(\bm\varphi)$,
		\begin{align}\label{TukeyClaim-a}
			\widetilde{L}_n^*(\check{\bm\varphi}_{wn}) \leq \widetilde{L}_n^*(\bm\varphi_0^*) + \epsilon/3;
		\end{align}
		and (b) by \eqref{Tukeyloss-consistency},
		\begin{align}\label{TukeyClaim-b}
			L^*(\check{\bm\varphi}_{wn}) < \widetilde{L}_n^*(\check{\bm\varphi}_{wn}) + \epsilon/3 \;\;\text{and}\;\; \widetilde{L}_n^*(\bm\varphi_0^*) < L^*(\bm\varphi_0^*) + \epsilon/3.
		\end{align} 
		Combining \eqref{Claim-a}--\eqref{Claim-b}, with probability tending to 1, we have
		\begin{align}\label{Tukeyloss-inequality}
			L^*(\check{\bm\varphi}_{wn}) < \widetilde{L}_n^*(\check{\bm\varphi}_{wn}) + \epsilon/3 \leq \widetilde{L}_n^*(\bm\varphi_0^*) + 2\epsilon/3 < L^*(\bm\varphi_0^*) + \epsilon.
		\end{align}
		
		Pick any $\delta>0$, let $B_{\delta}$ be a ball centered at $\bm\varphi_0^*$ with radius $\delta>0$. Then $B^c\equiv\Phi/B_{\delta}$ is compact, and thus $\inf_{\bm\varphi\in B^c}L^*(\bm\varphi)$ exists. 
		Denote $\epsilon=\inf_{\bm\varphi\in B^c}L^*(\bm\varphi)-L^*(\bm\varphi_0^*)$. Since $\bm\varphi_0^*=\argmin_{\bm\varphi\in\Phi}L^*(\bm\varphi)$ is unique by Claim (3), then $\epsilon>0$. For any $\epsilon>0$, we can pick $c>0$ such that $\text{Pr}(\epsilon\geq c)>1-\epsilon$. 
		Together with \eqref{Tukeyloss-inequality}, it follows that with probability becoming greater than $1-\epsilon$: 
		\begin{align*}
			L^*(\check{\bm\varphi}_{wn}) & < L^*(\bm\varphi_0^*) + \inf_{\bm\varphi\in B^c}L^*(\bm\varphi)-L^*(\bm\varphi_0^*) = \inf_{\bm\varphi\in B^c}L^*(\bm\varphi). 
		\end{align*}
		Thus with probability becoming greater than $1-\epsilon$, $\|\check{\bm\varphi}_{wn}-\bm\varphi_0^*\|\leq \delta$.
		By the arbitrariness of $\epsilon$, it implies that $\|\check{\bm\varphi}_{wn}-\bm\varphi_0^*\|\leq \delta$ with probability tending to 1.  
		The proof of (i) $\check{\bm\varphi}_{wn} \to_p \bm\varphi_0^*$ is complete.

		\vspace{5mm}
		\noindent(II) Proof of (ii) $\sqrt{n}(\check{\bm\varphi}_{wn}-\bm\varphi_0^*)\rightarrow_d N\left(\bm 0,\Sigma_w^*\right)$
		
		Recall that $\widetilde{L}_n^*(\bm\varphi)=n^{-1}\sum_{k=1}^{K}\sum_{t=1}^{n}w_t\rho_{\tau_k}(y_t-\widetilde{q}_{t,\tau_k}(\bm\varphi))$. For $\bm u\in \mathbb{R}^{4}$, define $\widetilde{H}_n^*(\bm u)=n[\widetilde{L}_n^*(\bm\varphi_0^*+\bm u)-\widetilde{L}_n^*(\bm\varphi_0^*)]$. Denote $\check{\bm u}_n=\check{\bm\varphi}_{wn}-\bm\varphi_0^*$. By the consistency of $\check{\bm\varphi}_{wn}$, it holds that $\check{\bm u}_n=o_p(1)$.	
		Note that $\check{\bm u}_n$ is the minimizer of $\widetilde{H}_n^*(\bm u)$, since $\check{\bm\varphi}_{wn}$ minimizes $\widetilde{L}_n^*(\bm\varphi)$. 	
		This together with Lemmas \ref{lem2-Tukey}--\ref{lem3-Tukey}, implies that
		\begin{align}\label{Hn-Tukey}
			\widetilde{H}_n^*(\check{\bm u}_n)=&-\sqrt{n}\check{\bm u}_n^{\prime}\bm T_n^*+\sqrt{n}\check{\bm u}_n^{\prime}J^*\sqrt{n}\check{\bm u}_n+o_p(\sqrt{n}\|\check{\bm u}_n\|+n\|\check{\bm u}_n\|^2) \\
			\geq &-\sqrt{n}\|\check{\bm u}_n\|[\|\bm T_n^*\|+o_p(1)]+n\|\check{\bm u}_n\|^2[\lambda_{\min}+o_p(1)], \nonumber
		\end{align}
		where $\lambda_{\min}$ is the smallest eigenvalue of $J^*=\Omega_{1w}^*/2$ with $\Omega_{1w}^*$ defined before Theorem \ref{thm-WCQR}, and $\bm T_n^*=n^{-1/2}\sum_{k=1}^{K}\sum_{t=1}^{n}w_t\dot{q}_{t,\tau_k}(\bm\varphi_0^*)\psi_{\tau_k}(e_{t,\tau_k}^*)$ with $e_{t,\tau_k}^*=y_t-q_{t,\tau_k}(\bm\varphi_0^*)$. Denote $\bm X_t=\sum_{k=1}^K w_t\dot{q}_{t,\tau_k}(\bm\varphi_0^*)\psi_{\tau_k}(y_t-q_{t,\tau_k}(\bm\varphi_0^*))$, then $\bm T_n^*=n^{-1/2}\sum_{t=1}^{n}\bm X_t$. 
		By definition of $\bm\varphi_0^*$ in \eqref{pseudo-true}, we have $E(\bm X_t)=\bm 0$. Moreover, by Lemma 2.1 of \cite{White_Domowitz1984} and Assumption \ref{assum-Process-Mixing}, for any nonzero vector $\bm c\in\mathbb{R}^4$, we can show that $\bm c^{\prime}\bm X_t$ is also a strictly stationary and $\alpha$-mixing sequence with the mixing coefficient $\alpha(n)$ satisfying $\sum_{n\geq 1}[\alpha(n)]^{1-2/\delta}<\infty$ for some $\delta>2$. 
		As a result, by central limit theorem for $\alpha$-mixing process given in Theorem 2.21 of \cite{Fan_Yao2003} 
		and the Cram\'{e}r-Wold device, $\bm T_n^*$ converges in distribution to a normal random variable with mean zero and variance matrix $\Omega_{0w}^*=E(\bm X_t\bm X_t^{\prime})+n^{-1}\sum_{t\neq s}^nE(\bm X_t\bm X_s^{\prime})$ or equivalently $\Omega_{0w}^*=E(\bm X_t\bm X_t^{\prime})+\sum_{\ell=1}^{\infty}[E(\bm X_t\bm X_{t-\ell}^{\prime})+E(\bm X_{t-\ell}\bm X_t^{\prime})]$ as $n\rightarrow \infty$.
		
		Note that $\lambda_{\min}>0$ as $\Omega_{1w}^*=2J^*$ is positive definite, and $\|\bm T_n^*\|<\infty$ by Assumptions \ref{assum-RandomWeight}--\ref{assum-SpaceTukey}. Since $\widetilde{H}_n^*(\check{\bm u}_n)\leq 0$, then we have 
		\begin{align}\label{root-n-Tukey}
			\sqrt{n}\|\check{\bm u}_n\|\leq [\lambda_{\min}+o_p(1)]^{-1}[\|\bm T_n^*\|+o_p(1)]=O_p(1).
		\end{align}
		This together with the consistency of $\check{\bm\varphi}_{wn}$, verifies the $\sqrt{n}$-consistency of $\check{\bm\varphi}_{wn}$, i.e. $\sqrt{n}(\check{\bm\varphi}_{wn}-\bm\varphi_0^*)=O_p(1)$. 
		
		Let $\sqrt{n}\bm u_n^*=J^{*-1}\bm T_n^*/2=\Omega_{1w}^{*-1}\bm T_n^*$, then we have
		\begin{align*}
			\sqrt{n}\bm u_n^*\rightarrow N\left(\bm 0,\Sigma_w^*\right)
		\end{align*}
		in distribution as $n\rightarrow \infty$, where $\Sigma_w^*=\Omega_{1w}^{*-1}\Omega_{0w}^*\Omega_{1w}^{*-1}$. Therefore, it suffices to show that $\sqrt{n}\bm u_n^*-\sqrt{n}\check{\bm u}_n=o_p(1)$.
		By \eqref{Hn-Tukey} and \eqref{root-n-Tukey}, we have
		\begin{align*}
			\widetilde{H}_n^*(\check{\bm u}_n)=&-\sqrt{n}\check{\bm u}_n^{\prime}\bm T_n^*+\sqrt{n}\check{\bm u}_n^{\prime}J^*\sqrt{n}\check{\bm u}_n+o_p(1) \\
			=&-2\sqrt{n}\check{\bm u}_n^{\prime}J^*\sqrt{n}\bm u_n^*+\sqrt{n}\check{\bm u}_n^{\prime}J^*\sqrt{n}\check{\bm u}_n+o_p(1),
		\end{align*}
		and
		\begin{align*}
			\widetilde{H}_n^*(\bm u_n^*)=&-\sqrt{n}\bm u_n^{*\prime}\bm T_n^*+\sqrt{n}\bm u_n^{*\prime}J^*\sqrt{n}\bm u_n^*+o_p(1)
			=-\sqrt{n}\bm u_n^{*\prime}J^*\sqrt{n}\bm u_n^*+o_p(1).
		\end{align*}
		It follows that
		\begin{align}\label{normality-Tukey}
			\widetilde{H}_n^*(\check{\bm u}_n)-\widetilde{H}_n^*(\bm u_n^*)=&(\sqrt{n}\check{\bm u}_n-\sqrt{n}\bm u_n^*)^{\prime}J^*(\sqrt{n}\check{\bm u}_n-\sqrt{n}\bm u_n^*)+o_p(1)  \nonumber \\
			\geq & \lambda_{\min}\|\sqrt{n}\check{\bm u}_n-\sqrt{n}\bm u_n^*\|^2+o_p(1).
		\end{align}
		Since $\widetilde{H}_n^*(\check{\bm u}_n)-\widetilde{H}_n^*(\bm u_n^*)=n[\widetilde{L}_n^*(\bm\varphi_0^*+\check{\bm u}_n)-\widetilde{L}_n^*(\bm\varphi_0^*+\bm u_n^*)]\leq 0$ a.s., then \eqref{normality-Tukey} implies that $\|\sqrt{n}\check{\bm u}_n-\sqrt{n}\bm u_n^*\|=o_p(1)$. We verify the asymptotic normality of $\check{\bm\varphi}_{wn}$.

		\vspace{5mm}
		\noindent(III) Proof of (iii) $\sqrt{n}(\check{\bm\theta}_{wn}^*(\tau)-\bm\theta(\tau)-B(\tau))\rightarrow_d N\left(\bm 0,g_{\tau}(\bm\varphi_0^*)\Sigma_w^*g_{\tau}^{\prime}(\bm\varphi_0^*)\right)$
		
		Based on the proof of Theorem \ref{thm-WCQR}(i)--(ii), it follows that
		\begin{align*}
			\sqrt{n}(\check{\bm\varphi}_{wn}-\bm\varphi_0^*)=\Omega_{1w}^{*-1}\dfrac{1}{\sqrt{n}}\sum_{k=1}^{K}\sum_{t=1}^{n}w_t\dot{q}_{t,\tau_k}(\bm\varphi_0^*)\psi_{\tau_k}(e_{t,\tau_k}^*) + o_p(1),
		\end{align*}
		where $e_{t,\tau_k}^*=y_t-q_{t,\tau_k}(\bm\varphi_0^*)$. By Delta method and the $\sqrt{n}$-consistency of $\check{\bm\varphi}_{wn}$, we have
		\begin{align}\label{transformedWCQR-rep}
			\sqrt{n}[g_{\tau}(\check{\bm\varphi}_{wn})-g_{\tau}(\bm\varphi_0^*)] 
			&= \dot{g}_{\tau}(\bm\varphi_0^*)\sqrt{n}(\check{\bm\varphi}_{wn}-\bm\varphi_0^*) + o_p(1) \nonumber\\
			&= \dot{g}_{\tau}(\bm\varphi_0^*)\Omega_{1w}^{*-1}\dfrac{1}{\sqrt{n}}\sum_{k=1}^{K}\sum_{t=1}^{n}w_t\dot{q}_{t,\tau_k}(\bm\varphi_0^*)\psi_{\tau_k}(e_{t,\tau_k}^*) + o_p(1) \nonumber\\
			&\rightarrow_d N\left(\bm 0,g_{\tau}(\bm\varphi_0^*)\Sigma_w^*g_{\tau}^{\prime}(\bm\varphi_0^*)\right), 
		\end{align}
		where $g_{\tau}(\cdot): \mathbb{R}^4\to\mathbb{R}^3$ is a measurable transformation function on $\check{\bm\varphi}_{wn}$ such that $g_{\tau}(\check{\bm\varphi}_{wn})=\check{\bm\theta}_{wn}^*(\tau)=(\check{a}_{0wn}Q_{\tau}(\check{\lambda}_{wn})/(1-\check{b}_{1wn}),\check{a}_{1wn}Q_{\tau}(\check{\lambda}_{wn}),\check{b}_{1wn})^{\prime}$, and 
		\[
		\dot{g}_{\tau}(\bm\varphi_0^*)=\begin{pmatrix}
			\dfrac{Q_{\tau}(\lambda_0)}{1-b_{10}}  &  0  &  \dfrac{Q_{\tau}(\lambda_0)a_{00}}{(1-b_{10})^2}  &  \dfrac{\dot{Q}_{\tau}(\lambda_0)a_{00}}{1-b_{10}} \\
			0  &  Q_{\tau}(\lambda_0)  &  0  &  \dot{Q}_{\tau}(\lambda_0)a_{10}  \\
			0  &  0  &  1  &  0
		\end{pmatrix}.
		\]
		The proof of this theorem is hence accomplished.
	\end{proof}

	\subsection{Lemmas for Corollary \ref{thm-WCQE} and Theorems \ref{thm-WCQE-uniform-consistency}--\ref{thm-WCQE-weak-convergence}}
	
	This section provides seven preliminary lemmas with proofs, where Lemma \ref{lem00} gives basic results for all lemmas and theorems, Lemma \ref{lem0} is used to handle initial values, Lemmas \ref{lem1}--\ref{lem3} are used to prove Corollary \ref{thm-WCQE}, and Lemmas \ref{lem-Tightness1}--\ref{lem-Tightness2} are basic results to show Theorem \ref{thm-WCQE-weak-convergence}. 
	Specifically, Lemma \ref{lem1} verifies the stochastic differentiability condition defined by \cite{Pollard1985}, and the bracketing method in \cite{Pollard1985} is used for their proofs. 
	Lemma \ref{lem2} is used to obtain the $\sqrt{n}$-consistency and asymptotic normality of $\widetilde{\bm\theta}_{wn}(\tau)$, and its proof needs Lemma \ref{lem1}. 
	Based on Lemma \ref{lem0}, Lemma \ref{lem3} will be used to handle initial values in establishing asymptotic normality.
	Lemma \ref{lem-Tightness1} provides basic results to verify the stochastic equicontinuity, and Lemma \ref{lem-Tightness2} is used to establish the weak convergence of $\widetilde{\bm\theta}_{wn}(\tau)$.
	
	\begin{lemma}\label{lem00}
		If $E|y_t|^s<\infty$ for some $0<s\leq 1$ and Assumptions \ref{assum-Process}, \ref{assum-Space} and \ref{assum-RandomWeight} hold, then we have
		\[(i)~E(w_t\sup_{\Theta}\|\dot{q}_t(\bm\theta)\|^{\kappa}) < \infty \;\text{for}\; \kappa=1, 2, 3; \quad (ii)~E(w_t\sup_{\Theta}\|\ddot{q}_t(\bm\theta)\|^{\kappa}) < \infty \;\text{for}\; \kappa=1, 2.\]
	\end{lemma}	
	
	\begin{lemma}\label{lem0}
		Let $\varsigma_{\rho}=\sum_{s=0}^{\infty}\rho^{s}|y_{-s}|$ and $\xi_{\rho}=\sum_{s=0}^{\infty}s\rho^{s}|y_{-s}|$ be positive random variables depending on a constant $\rho\in (0,1)$. If Assumption \ref{assum-Space} holds, then we have
		\begin{itemize}
			\item[(i)] $\sup_{\Theta}|q_t(\bm\theta)-\widetilde{q}_t(\bm\theta)|
			\leq C\rho^{t}\varsigma_{\rho}$;
			\item[(ii)]
			$\sup_{\Theta}\|\dot{q}_t(\bm\theta)-\dot{\widetilde{q}}_t(\bm\theta)\|
			\leq C\rho^{t}(\varsigma_{\rho}+t\varsigma_{\rho}+\xi_{\rho})$. 
		\end{itemize}
	\end{lemma}	
	
	\begin{lemma}\label{lem1}
		Under Assumptions \ref{assum-Process}--\ref{assum-RandomWeight}, then for any sequence of random variables $\bm u_n$ such that $\bm u_n=o_p(1)$, if $E|y_t|^s<\infty$ for some $0<s\leq 1$, then it holds that	
		\begin{align*}
			\zeta_n(\bm u_n)=o_p(\sqrt{n}\|\bm u_n\|+n\|\bm u_n\|^2),
		\end{align*}
		where $\zeta_n(\bm u)=\bm u^{\prime}\sum_{t=1}^{n}w_t\dot{q}_t(\bm\theta(\tau))\left\{\xi_{t}(\bm u)-E[\xi_{t}(\bm u)|\mathcal{F}_{t-1}]\right\}$ with 
		\begin{align*}
			\xi_{t}(\bm u)&=\int_{0}^{1}\left[I(y_t\leq F_{t-1}^{-1}(\tau)+\nu_t(\bm u)s)-I(y_t\leq F_{t-1}^{-1}(\tau))\right]ds 
		\end{align*}
		and $\nu_t(\bm u)=q_t(\bm u+\bm\theta(\tau))-q_t(\bm\theta(\tau))$. 
	\end{lemma}
	
	\begin{lemma}\label{lem2}
		If $E|y_t|^s<\infty$ for some $0<s\leq 1$ and Assumptions \ref{assum-Process}--\ref{assum-RandomWeight} hold, then for any sequence of random variables $\bm u_n$ such that $\bm u_n=o_p(1)$, we have
		\begin{align*}
			n[L_n(\bm u_n+\bm\theta(\tau))-L_n(\bm\theta(\tau))]=&-\sqrt{n}\bm u_n^{\prime}\bm T_n+\sqrt{n}\bm u_n^{\prime}J_n\sqrt{n}\bm u_n +o_p(\sqrt{n}\|\bm u_n\|+n\|\bm u_n\|^2),
		\end{align*}
		where $L_n(\bm\theta)=n^{-1}\sum_{t=1}^{n}w_t\rho_{\tau}(y_t-q_t(\bm\theta))$, and	
		\[\bm T_n=\dfrac{1}{\sqrt{n}}\sum_{t=1}^{n}w_t\dot{q}_t(\bm\theta(\tau))\psi_{\tau}(\eta_{t,\tau}) \hspace{2mm}\text{and}\hspace{2mm} J_n=\dfrac{1}{2n}\sum_{t=1}^{n}f_{t-1}(F_{t-1}^{-1}(\tau))w_t\dot{q}_t(\bm\theta(\tau))\dot{q}_t^{\prime}(\bm\theta(\tau))\]
		with $\psi_{\tau}(x)=\tau-I(x<0)$ and $\eta_{t,\tau}=y_t-q_t(\bm\theta(\tau))$.		
	\end{lemma}
	
	\begin{lemma}\label{lem3}
		If $E|y_t|^s<\infty$ for some $0<s\leq 1$ and Assumptions \ref{assum-Process}--\ref{assum-RandomWeight} hold, then for any sequence of random variables $\bm u_n$ such that $\bm u_n=o_p(1)$, we have	
		\begin{align*}
			n[\widetilde{L}_n(\bm u_n+\bm\theta(\tau))-\widetilde{L}_n(\bm\theta(\tau))]-n[L_n(\bm u_n+\bm\theta(\tau))-L_n(\bm\theta(\tau))]=o_p(\sqrt{n}\|\bm u_n\|+n\|\bm u_n\|^2),
		\end{align*}
		where $\widetilde{L}_n(\bm\theta)=n^{-1}\sum_{t=1}^{n}w_t\rho_{\tau}(y_t-\widetilde{q}_t(\bm\theta))$ and $L_n(\bm\theta)=n^{-1}\sum_{t=1}^{n}w_t\rho_{\tau}(y_t-q_t(\bm\theta))$.
	\end{lemma}
	
	\begin{lemma}\label{lem-Tightness1}
		Let $\Delta_{\rho,t}=1+\sum_{j=1}^\infty\rho^{j-1}|y_{t-j}|+\sum_{j=2}^\infty(j-1)\rho^{j-2}|y_{t-j}|+\sum_{j=3}^\infty(j-1)(j-2)\rho^{j-3}|y_{t-j}|+\sum_{j=4}^\infty(j-1)(j-2)(j-3)\rho^{j-4}|y_{t-j}|$ be positive random variables depending on a constant $\rho\in (0,1)$. If Assumptions \ref{assum-Space} and \ref{assum-Tightness} hold, for $\tau_1,\tau_2\in\mathcal{T}$ and $\bm u\in\Lambda$ with $\Lambda=\{\bm u\in\mathbb{R}^3: \bm u+\bm\theta(\tau) \in \Theta\}$, then we have  
		\begin{itemize}
			\item[(i)] $|q_t(\bm\theta(\tau_2))-q_t(\bm\theta(\tau_1))|
			\leq C|\tau_2-\tau_1|\Delta_{\rho,t}$;
			\item[(ii)]
			$\|\dot{q}_t(\bm u+\bm\theta(\tau_2))-\dot{q}_t(\bm u+\bm\theta(\tau_1))\| \leq C|\tau_2-\tau_1|\Delta_{\rho,t}$;  
			\item[(iii)] $\|\ddot{q}_t(\bm u+\bm\theta(\tau_2))-\ddot{q}_t(\bm u+\bm\theta(\tau_1))\| \leq C|\tau_2-\tau_1|\Delta_{\rho,t}$.
		\end{itemize}
	\end{lemma}
	
	\begin{lemma}\label{lem-Tightness2}
		Under Assumptions \ref{assum-Process}--\ref{assum-Tightness}, for any $\eta>0$, we have 
		\[ \sup_{\tau\in\mathcal{T}}\sup_{\|\bm u\|\leq \eta}\dfrac{|R_{4n}(\bm u,\tau)|}{\sqrt{n}\|\bm u\|+n\|\bm u\|^2}=o_p(1) \quad\text{and}\quad  \sup_{\tau\in\mathcal{T}}\sup_{\|\bm u\|\leq \eta}\dfrac{|R_{5n}(\bm u,\tau)|}{n\|\bm u\|^2}=o_p(1),\]
		where $R_{4n}(\bm u, \tau)$ and $R_{5n}(\bm u, \tau)$ are defined in the proof of Theorem \ref{thm-WCQE-weak-convergence}.
	\end{lemma}
	\begin{proof}[Proof of Lemma \ref{lem00}]
		Recall $\dot{q}_t(\bm\theta)$ in \eqref{1st-derivative-qt} and $\ddot{q}_t(\bm\theta)$ in \eqref{2nd-derivative-qt}, where $\bm\theta=(\omega, \alpha_{1},\beta_{1})^{\prime}$. For any $s\in(0,1]$, using the inequality $(x+y)^s\leq x^s+y^s$ for $x, y\geq 0$, we have
		\begin{align*}
			\|\dot{q}_t(\bm\theta)\| \leq & 1+\sum_{j=1}^{\infty}\beta^{j-1}_{1}|y_{t-j}|+|\alpha_{1}|\sum_{j=2}^{\infty}(j-1)\beta^{j-2}_{1}|y_{t-j}| \quad\text{and} \\
			\|\ddot{q}_t(\bm\theta)\| \leq & 2\sum\limits_{j=2}^{\infty}(j-1)\beta^{j-2}_{1}|y_{t-j}| + |\alpha_{1}|\sum\limits_{j=3}^{\infty}(j-1)(j-2)\beta^{j-3}_{1}|y_{t-j}|.
		\end{align*}
		For $\kappa=1,2,3$, denote $M_{\kappa}=\max_{j}\{E(w_{t}|y_{t-j}|^{\kappa})\}$, then it follows that $M_{\kappa}<\infty$ since $E(w_t)<\infty$ and $E(w_{t}|y_{t-j}|^3)<\infty$ for all $j\geq 1$ under Assumption \ref{assum-RandomWeight}.
		
		For $\kappa=1$, by the strict stationarity and ergodicity of $y_t$ under Assumption \ref{assum-Process}, $\max\{|\omega|,|\alpha_1|\}<\overline{c}<\infty$ and $\beta_1\leq\rho<1$ by Assumption \ref{assum-Space}, and $E(w_t)<\infty$ and $E(w_{t}|y_{t-j}|^3)<\infty$ for all $j\geq 1$ by Assumption \ref{assum-RandomWeight}, it holds that	
		\begin{align*}
			E\left(w_t\sup_{\Theta}\|\dot{q}_t(\bm\theta)\|\right)
			\leq & E(w_t) + M_1\sum_{j=1}^{\infty} \rho^{j-1}+\overline{c}M_1\sum_{j=2}^{\infty}(j-1)\rho^{j-2} <\infty \;\;\text{and} \\
			E\left(w_t\sup_{\Theta}\|\ddot{q}_t(\bm\theta)\|\right) 
			\leq & 2M_1\sum_{j=2}^{\infty}(j-1)\rho^{j-2} + \overline{c}M_1\sum_{j=3}^{\infty}(j-1)(j-2)\rho^{j-3}<\infty.
		\end{align*}
		Thus, (i) and (ii) hold for $\kappa=1$.
		
		For $\kappa=2$, under Assumptions \ref{assum-Process} and \ref{assum-RandomWeight}, by the Cauchy-Schwarz inequality, we have
		\[E\left[w_t\left(\sum_{j=1}^{\infty} \rho^{j-1}|y_{t-j}|\right)^2\right]\leq \sum_{i=1}^{\infty}\sum_{j=1}^{\infty}\rho^{i+j-2}\left[E\left(w_t|y_{t-i}|^2\right)\right]^{1/2}\left[E\left(w_t|y_{t-j}|^2\right)\right]^{1/2}\leq \dfrac{M_2}{(1-\rho)^2}<\infty,\]
		\[E\left[w_t\left(\sum_{j=1}^{\infty} (j-1)\rho^{j-2}|y_{t-j}|\right)^2\right] \leq M_2\left(\sum_{j=1}^{\infty}(j-1)\rho^{j-2}\right)^2<\infty \;\;\text{and}\]
		\[E\left[w_t\left(\sum_{j=1}^{\infty} (j-1)(j-2)\rho^{j-3}|y_{t-j}|\right)^2\right] \leq M_2\left(\sum_{j=1}^{\infty}(j-1)(j-2)\rho^{j-3}\right)^2<\infty.\]
		Then under Assumptions \ref{assum-Process}, \ref{assum-Space} and \ref{assum-RandomWeight}, by $(a+b)^2=a^2+2ab+b^2$ and the Cauchy-Schwarz inequality, we can show that
		\begin{align*}
			& E\left(w_t\sup_{\Theta}\|\dot{q}_t(\bm\theta)\|^2\right) \\
			\leq & E(w_t)+E\left[w_t\left(\sum_{j=1}^{\infty} \rho^{j-1}|y_{t-j}|\right)^2\right] + \overline{c}^2E\left[w_t\left(\sum_{j=1}^{\infty} (j-1)\rho^{j-2}|y_{t-j}|\right)^2\right]  \\
			& + 2\overline{c}\sum_{j=1}^{\infty} \rho^{j-1}E(w_t|y_{t-j}|) + 2\overline{c}\sum_{j=1}^{\infty}(j-1)\rho^{j-2}E(w_t|y_{t-j}|) \\
			& + 2\overline{c}\left[Ew_t\left(\sum_{j=1}^{\infty} \rho^{j-1}|y_{t-j}|\right)^2\right]^{\frac{1}{2}}\left[Ew_t\left(\sum_{j=1}^{\infty} (j-1)\rho^{j-2}|y_{t-j}|\right)^2\right]^{\frac{1}{2}} <\infty,
		\end{align*}
		and
		\begin{align*}
			& E\left(w_t\sup_{\Theta}\|\ddot{q}_t(\bm\theta)\|^2\right) \\
			\leq & 4E\left[w_t\left(\sum_{j=1}^{\infty} (j-1)\rho^{j-2}|y_{t-j}|\right)^2\right] 
			+ \overline{c}^2E\left[w_t\left(\sum_{j=1}^{\infty} (j-1)(j-2)\rho^{j-3}|y_{t-j}|\right)^2\right] \\
			&+ 4\overline{c}\left[Ew_t\left(\sum_{j=1}^{\infty} (j-1)\rho^{j-2}|y_{t-j}|\right)^2\right]^{\frac{1}{2}}\left[Ew_t\left(\sum_{j=1}^{\infty} (j-1)(j-2)\rho^{j-3}|y_{t-j}|\right)^2\right]^{\frac{1}{2}} <\infty.
		\end{align*}
		Hence, (i) and (ii) hold for $\kappa=2$.
		
		For $\kappa=3$, under Assumptions \ref{assum-Process} and \ref{assum-RandomWeight}, by H\"{o}lder's inequality, we have
		\begin{align*}
			& E\left[w_t\left(\sum_{j=1}^{\infty} \rho^{j-1}|y_{t-j}|\right)^3\right] \\ 
			\leq & \sum_{i=1}^{\infty}\sum_{j=1}^{\infty}\sum_{k=1}^{\infty}\rho^{i+j+k-3}\left[E\left(w_t|y_{t-i}|^3\right)\right]^{\frac{1}{3}}\left[E\left(w_t|y_{t-j}|^3\right)\right]^{\frac{1}{3}} \left[E\left(w_t|y_{t-k}|^3\right)\right]^{\frac{1}{3}} 
			\leq  \dfrac{M_3}{(1-\rho)^3}<\infty, \\
			& E\left[w_t\left(\sum_{j=1}^{\infty} (j-1)\rho^{j-2}|y_{t-j}|\right)^3\right] \leq M_3\left(\sum_{j=1}^{\infty}(j-1)\rho^{j-2}\right)^3<\infty \;\;\text{and} \\
			& E\left[w_t\left(\sum_{j=1}^{\infty} (j-1)(j-2)\rho^{j-3}|y_{t-j}|\right)^3\right] \leq M_3\left(\sum_{j=1}^{\infty}(j-1)(j-2)\rho^{j-3}\right)^3<\infty.
		\end{align*}
		Then similar to the proof for $\kappa=2$, under Assumptions \ref{assum-Process}, \ref{assum-Space} and \ref{assum-RandomWeight}, by $(a+b)^3=a^3+3a^2b+3ab^2+b^3$ and H\"{o}lder's inequality, we can show that (i) holds for $\kappa=3$. 
		The proof of this lemma is complete.
	\end{proof}
	\begin{proof}[Proof of Lemma \ref{lem0}]
		Recall that for $\bm\theta=(\omega, \alpha_{1},\beta_{1})^{\prime}$, $q_t(\bm\theta) =\omega + \alpha_{1}\sum_{j=1}^\infty \beta_{1}^{j-1}|y_{t-j}|$, $\widetilde{q}_t(\bm\theta)=\omega + \alpha_{1}\sum_{j=1}^{t-1} \beta_{1}^{j-1}|y_{t-j}|$,  
		$\dot{q}_t(\bm\theta)=(1,\sum_{j=1}^{\infty}\beta_{1}^{j-1}|y_{t-j}|,\alpha_{1}\sum_{j=2}^{\infty}(j-1)\beta_{1}^{j-2}|y_{t-j}|)^{\prime}$ 
		and $\dot{\widetilde{q}}_t(\bm\theta)=(1,\sum_{j=1}^{t-1}\beta^{j-1}_{1}|y_{t-j}|,\alpha_{1}\sum_{j=2}^{t-1}(j-1)\beta^{j-2}_{1}|y_{t-j}|)^{\prime}$. 
		It follows that
		\begin{align*}
			q_t(\bm\theta)-\widetilde{q}_t(\bm\theta)&=\alpha_{1}\sum_{j=t}^{\infty} \beta_{1}^{j-1}|y_{t-j}| \quad\text{and}\quad \\
			\dot{q}_t(\bm\theta)-\dot{\widetilde{q}}_t(\bm\theta)&=\left(0,\sum_{j=t}^{\infty}\beta^{j-1}_{1}|y_{t-j}|,\alpha_{1}\sum_{j=t}^{\infty}(j-1)\beta^{j-2}_{1}|y_{t-j}|\right)^{\prime}.
		\end{align*}
		Since $|\alpha_{1}|\leq\overline{c}<\infty$ and $0<\beta_{1}\leq\rho<1$ by Assumption \ref{assum-Space}, it holds that 
		\begin{align*}
			\sup_{\Theta}|q_t(\bm\theta)-\widetilde{q}_t(\bm\theta)|
			&\leq |\alpha_{1}|\sum_{j=t}^{\infty} \beta_{1}^{j-1}|y_{t-j}|
			\leq \overline{c}\rho^{t-1}\sum_{s=0}^{\infty}\rho^{s}|y_{-s}|\leq C\rho^{t}\varsigma_{\rho}, 
		\end{align*}
		and
		\begin{align*}
			\sup_{\Theta}\|\dot{q}_t(\bm\theta)-\dot{\widetilde{q}}_t(\bm\theta)\| &\leq \sup_{\Theta}\left[\sum_{j=t}^{\infty}\beta^{j-1}_{1}|y_{t-j}|+|\alpha_{1}|\sum_{j=t}^{\infty}(j-1)\beta^{j-2}_{1}|y_{t-j}|\right] \\
			&\leq \rho^{t-1}\varsigma_{\rho}+\overline{c}t\rho^{t-2}\varsigma_{\rho}+\overline{c}\rho^{t-1}\sum_{s=0}^{\infty}(s-1)\rho^{s-1}|y_{-s}|\\
			&\leq C\rho^{t}(\varsigma_{\rho}+t\varsigma_{\rho}+\xi_{\rho}),
		\end{align*}
		where $\varsigma_{\rho}=\sum_{s=0}^{\infty}\rho^{s}|y_{-s}|$ and $\xi_{\rho}=\sum_{s=0}^{\infty}s\rho^{s}|y_{-s}|$.	
		The proof of this lemma is complete.	
	\end{proof}
	
	\begin{proof}[Proof of Lemma \ref{lem1}]
		Recall that $\bm\theta=(\omega, \alpha_{1},\beta_{1})^{\prime}$ and its true parameter vector $\bm\theta(\tau)=(\omega(\tau), \alpha_1(\tau), \beta_1(\tau))^{\prime}$. For $\bm u\in\mathbb{R}^{d}$ with $d=3$, note that
		\begin{align*}
			|\zeta_n(\bm u)| \leq \sqrt{n}\|\bm u\|\sum_{j=1}^{3}\left|\dfrac{1}{\sqrt{n}}\sum_{t=1}^{n}m_{t,j}\left\{\xi_{t}(\bm u)-E[\xi_{t}(\bm u)|\mathcal{F}_{t-1}]\right\}\right|,
		\end{align*}
		where $m_{t,j}=w_t\partial q_t(\bm\theta(\tau))/\partial \theta_{j}$ with $\theta_{j}$ being the $j$th element of $\bm\theta$. For $1\leq j\leq d$, define $g_t=\max_{j}\{m_{t,j},0\}$ or $g_t=\max_j\{-m_{t,j},0\}$. Let $\varrho_t(\bm u)=g_t\xi_{t}(\bm u)$ and define
		\begin{align*}
			D_n(\bm u)=\dfrac{1}{\sqrt{n}}\sum_{t=1}^{n}\left\{\varrho_t(\bm u)-E\left[\varrho_t(\bm u)|\mathcal{F}_{t-1}\right]\right\}.
		\end{align*}
		To establish Lemma \ref{lem1}, it suffices to show that, for any $\delta>0$,
		\begin{align}\label{diffcond}
			\sup_{\|\bm u\|\leq \delta}\dfrac{|D_n(\bm u)|}{1+\sqrt{n}\|\bm u\|}=o_p(1).
		\end{align}
		
		We follow the method in Lemma 4 of \cite{Pollard1985} to verify \eqref{diffcond}.
		Let $\mathfrak{F}=\{\varrho_t(\bm u): \|\bm u\|\leq \delta\}$ be a collection of functions indexed by $\bm u$. First, we verify that $\mathfrak{F}$ satisfies the bracketing condition defined on page 304 of \cite{Pollard1985}. 	
		Let $B_{r}(\bm v)$ be an open neighborhood of $\bm v$ with radius $r>0$, and define a constant $C_0$ to be selected later. For any $\epsilon>0$ and $0< r\leq \delta$, there exists a sequence of small cubes $\{B_{\epsilon r/C_0}(\bm u_{i})\}_{i=1}^{K(\epsilon)}$ to cover $B_r(\bm 0)$, where $K(\epsilon)$ is an integer less than $C\epsilon^{-d}$, and the constant $C$ is not depending on $\epsilon$ and $r$; see \cite{Huber1967}, page 227.
		Denote $V_i(r)=B_{\epsilon r/C_0}(\bm u_{i})\bigcap B_r(\bm0)$, and let $U_1(r)=V_1(r)$ and $U_i(r)=V_i(r)-\bigcup_{j=1}^{i-1}V_j(r)$ for $i\geq 2$. Note that $\{U_i(r)\}_{i=1}^{K(\epsilon)}$ is a partition of $B_r(\bm0)$.
		For each $\bm u_i\in U_i(r)$ with $1\leq i \leq K(\epsilon)$, define the following bracketing functions	
		\begin{align*}
			\varrho_t^L(\bm u_i)&= g_t\int_{0}^{1}\left[I\left(y_t\leq F_{t-1}^{-1}(\tau)+\nu_t(\bm u_i)s-\dfrac{\epsilon r}{C_0}\|\dot{q}_t(\bm\theta(\tau))\|\right)-I(y_t\leq F_{t-1}^{-1}(\tau))\right]ds, \\
			\varrho_t^U(\bm u_i)&= g_t\int_{0}^{1}\left[I\left(y_t\leq F_{t-1}^{-1}(\tau)+\nu_t(\bm u_i)s+\dfrac{\epsilon r}{C_0}\|\dot{q}_t(\bm\theta(\tau))\|\right)-I(y_t\leq F_{t-1}^{-1}(\tau))\right]ds.
		\end{align*}
		Since the indicator function $I(\cdot)$ is non-decreasing and $g_t\geq 0$, for any $\bm u \in U_i(r)$, we have
		\begin{align}\label{brac1}
			\varrho_t^L(\bm u_i)\leq \varrho_t(\bm u)\leq \varrho_t^U(\bm u_i).
		\end{align}
		Furthermore, by Taylor expansion, it holds that
		\begin{align}\label{uncond}
			E\left[\varrho_t^U(\bm u_i)-\varrho_t^L(\bm u_i)|\mathcal{F}_{t-1}\right]\leq \dfrac{\epsilon r}{C_0}\cdot2\sup_{x}f_{t-1}(x) w_t\left\|\dot{q}_t(\bm\theta(\tau))\right\|^2.
		\end{align}
		Denote $\aleph_t=2\sup_{x}f_{t-1}(x)w_t\left\|\dot{q}_t(\bm\theta(\tau))\right\|^2$. By Assumption \ref{assum-ConditionalDensity}, we have $\sup_{x}f_{t-1}(x)<\infty$. Choose $C_0=E(\aleph_t)$. Then by iterated-expectation, Assumption \ref{assum-Space}(ii) and Lemma \ref{lem00}(i), it follows that
		\begin{align*}
			E\left[\varrho_t^U(\bm u_i)-\varrho_t^L(\bm u_i)\right]=E\left\{E\left[\varrho_t^U(\bm u_i)-\varrho_t^L(\bm u_i)|\mathcal{F}_{t-1}\right]\right\}\leq \epsilon r.
		\end{align*}	
		This together with \eqref{brac1}, implies that the family $\mathfrak{F}$ satisfies the bracketing condition.	
		
		Put $r_k=2^{-k}\delta$. Let $B(k)=B_{r_k}(\bm0)$ and $A(k)$ be the annulus $B(k)\setminus B(k+1)$. From the bracketing condition, for fixed $\epsilon>0$, there is a partition $U_1(r_k), U_2(r_k), \ldots, U_{K(\epsilon)}(r_k)$ of $B(k)$. First, consider the upper tail case. For $\bm u \in U_i(r_k)$, by \eqref{uncond}, it holds that
		\begin{align}\label{upper}
			D_n(\bm u) \leq &\dfrac{1}{\sqrt{n}}\sum_{t=1}^{n}\left\{\varrho_t^U(\bm u_i)-E\left[\varrho_t^U(\bm u_i)|\mathcal{F}_{t-1}\right]\right\}+\dfrac{1}{\sqrt{n}}\sum_{t=1}^{n}E\left[\varrho_t^U(\bm u_i)-\varrho_t^L(\bm u_i)|\mathcal{F}_{t-1}\right] \nonumber \\
			\leq & D_n^U(\bm u_i)+\sqrt{n}\epsilon r_k\dfrac{1}{nC_0}\sum_{t=1}^{n}\aleph_t,
		\end{align}	
		where \[D_n^U(\bm u_i)=\dfrac{1}{\sqrt{n}}\sum_{t=1}^{n}\left\{\varrho_t^U(\bm u_i)-E\left[\varrho_t^U(\bm u_i)|\mathcal{F}_{t-1}\right]\right\}.\]
		Define the event
		\begin{align*}
			E_n=\left\{\omega: \dfrac{1}{nC_0}\sum_{t=1}^{n}\aleph_t(\omega) < 2 \right\}.
		\end{align*}
		
		For $\bm u \in A(k)$, $1+\sqrt{n}\|\bm u\|>\sqrt{n}r_{k+1}=\sqrt{n}r_{k}/2$. Then by \eqref{upper} and the Chebyshev's inequality, we have
		\begin{align}\label{Ak0}
			\text{Pr}\left(\sup_{\bm u \in A(k)}\dfrac{D_n(\bm u)}{1+\sqrt{n}\|\bm u\|}>6\epsilon, E_n\right)
			\leq & \text{Pr}\left(\max_{1 \leq i \leq K(\epsilon)}\sup_{\bm u \in U_i(r_k) \cap A(k)}D_n(\bm u)>3\sqrt{n}\epsilon r_k, E_n\right) \nonumber\\
			\leq & K(\epsilon)\max_{1 \leq i \leq K(\epsilon)}\text{Pr}\left(D_n^U(\bm u_i)>\sqrt{n}\epsilon r_k\right) \nonumber\\
			\leq & K(\epsilon)\max_{1 \leq i \leq K(\epsilon)}\dfrac{E\{[D_n^U(\bm u_i)]^2\}}{n\epsilon^2 r_k^2}.
		\end{align}
		Moreover, by iterated-expectation, Taylor expansion and the H\"{o}lder inequality, together with $\|\bm u_i\|\leq r_k$ for $\bm u_i \in U_i(r_k)$, we have
		\begin{align*}	
			& E\left\{[\varrho_t^U(\bm u_i)]^2\right\}
			= E\left\{E\left\{[\varrho_t^U(\bm u_i)]^2|\mathcal{F}_{t-1}\right\}\right\} \nonumber \\
			\leq & 2E\left\{g_t^2\left| \int_{0}^{1}\left[F_{t-1}\left(F_{t-1}^{-1}(\tau)+\nu_t(\bm u_i)s+\dfrac{\epsilon r_k}{C_0}\|\dot{q}_t(\bm\theta(\tau))\|\right)-F_{t-1}\left(F_{t-1}^{-1}(\tau)\right)\right]ds\right|\right\} \nonumber\\
			\leq & C\sup_{x}f_{t-1}(x)r_kE\left\{w_t^2\left[\left\|\dot{q}_t(\bm\theta(\tau))\right\|^3+\left\|\dot{q}_t(\bm\theta(\tau))\right\|^2\sup_{\bm\theta^*\in\Theta}\|\dot{q}_t(\bm\theta^*)\|\right]\right\} \\
			\leq & C\sup_{x}f_{t-1}(x)r_k\left\{E\left(w_t^2\left\|\dot{q}_t(\bm\theta(\tau))\right\|^3\right)+\left[E\left(w_t^2\left\|\dot{q}_t(\bm\theta(\tau))\right\|^3\right)\right]^{2/3}\left[E\left(w_t^2\sup_{\bm\theta^*\in\Theta}\left\|\dot{q}_t(\bm\theta^*)\right\|^3\right)\right]^{1/3}\right\} \\ &	:= \Upsilon(r_k),
		\end{align*}
		where $\bm\theta^*=\bm u_i^*+\bm\theta(\tau)$ with $\bm u_i^*$ between $\bm 0$ and $\bm u_i$. 
		This, together with $\sup_{x}f_{t-1}(x)<\infty$ by Assumption \ref{assum-ConditionalDensity}, Lemma \ref{lem00}(i) and the fact that $\varrho_t^U(\bm u_i)-E[\varrho_t^U(\bm u_i)|\mathcal{F}_{t-1}]$ is a martingale difference sequence, implies that
		\begin{align}\label{ED}
			E\{[D_n^U(\bm u_i)]^2\}&=\dfrac{1}{n}\sum_{t=1}^{n}E\{\{\varrho_t^U(\bm u_i)-E[\varrho_t^U(\bm u_i)|\mathcal{F}_{t-1}]\}^2\} \nonumber \\
			&\leq \dfrac{1}{n}\sum_{t=1}^{n}E\{[\varrho_t^U(\bm u_i)]^2\} \leq\Upsilon(r_k)<\infty.
		\end{align}
		Combining \eqref{Ak0} and \eqref{ED}, we have
		\begin{align*}
			\text{Pr}\left(\sup_{\bm u \in A(k)}\dfrac{D_n(\bm u)}{1+\sqrt{n}\|\bm u\|}>6\epsilon, E_n\right)
			\leq \dfrac{K(\epsilon)\Upsilon(r_k)}{n\epsilon^2r_k^2}.
		\end{align*}
		Similar to the proof of the upper tail case, we can obtain the same bound for the lower tail case. Therefore,
		\begin{align}\label{Ak}
			& \text{Pr}\left(\sup_{\bm u \in A(k)}\dfrac{|D_n(\bm u)|}{1+\sqrt{n}\|\bm u\|}>6\epsilon, E_n\right) \leq \dfrac{2K(\epsilon)\Upsilon(r_k)}{n\epsilon^2r_k^2}.
		\end{align}
		
		Note that $\Upsilon(r_k)\to 0$ as $k\to \infty$, we can choose $k_{\epsilon}$ such that $2K(\epsilon)\Upsilon(r_k)/(\epsilon^2\delta^2)<\epsilon$ for $k\geq k_{\epsilon}$. Let $k_n$ be the integer such that $n^{-1/2}\delta \leq r_{k_n} \leq 2n^{-1/2}\delta$, and split $B_{\delta}(\bm 0)$ into two events $B:=B(k_n+1)$ and $B^c:=B(0)-B(k_n+1)$. Note that $B^c=\bigcup_{k=0}^{k_n}A(k)$ and $\Upsilon(r_k)$ is bounded. Then by \eqref{Ak}, it holds that
		\begin{align}\label{Bc}
			\text{Pr}\left(\sup_{\bm u \in B^c}\dfrac{|D_n(\bm u)|}{1+\sqrt{n}\|\bm u\|}>6\epsilon\right)
			\leq & \sum_{k=0}^{k_n}\text{Pr}\left(\sup_{\bm u \in A(k)}\dfrac{|D_n(\bm u)|}{1+\sqrt{n}\|\bm u\|}>6\epsilon, E_n\right) + \text{Pr}(E_n^c)\nonumber \\
			\leq & \dfrac{1}{n}\sum_{k=0}^{k_{\epsilon}-1}\dfrac{CK(\epsilon)}{\epsilon^2\delta^2}2^{2k}+ \dfrac{\epsilon}{n}\sum_{k=k_{\epsilon}}^{k_n}2^{2k}+ \text{Pr}(E_n^c) \nonumber \\
			\leq & O\left(\dfrac{1}{n}\right) + 4\epsilon + \text{Pr}(E_n^c).
		\end{align}
		
		Furthermore, for $\bm u \in B$, we have $1+\sqrt{n}\|\bm u\|\geq 1$ and $r_{k_n+1}\leq n^{-1/2}\delta<n^{-1/2}$. Similar to the proof of \eqref{Ak0} and \eqref{ED}, we can show that
		\begin{align*}
			\text{Pr}\left(\sup_{\bm u \in B}\dfrac{D_n(\bm u)}{1+\sqrt{n}\|\bm u\|}>3\epsilon, E_n\right) \leq \text{Pr}\left(\max_{1 \leq i \leq K(\epsilon)}D_n^U(\bm u_i)>\epsilon, E_n\right) \leq \dfrac{K(\epsilon)\Upsilon(r_{k_n+1})}{\epsilon^2}.
		\end{align*}
		We can obtain the same bound for the lower tail. Therefore, we have
		\begin{align}\label{B}
			\text{Pr}\left(\sup_{\bm u \in B}\dfrac{|D_n(\bm u)|}{1+\sqrt{n}\|\bm u\|}>3\epsilon\right)
			= &  \text{Pr}\left(\sup_{\bm u \in B}\dfrac{|D_n(\bm u)|}{1+\sqrt{n}\|\bm u\|}>3\epsilon, E_n\right)+ \text{Pr}(E_n^c) \nonumber \\
			\leq & \dfrac{2K(\epsilon)\Upsilon(r_{k_n+1})}{\epsilon^2} + \text{Pr}(E_n^c).
		\end{align}
		Note that $\Upsilon(r_{k_n+1})\to 0$ as $n\to \infty$. Moreover, by the ergodic theorem, $\text{Pr}(E_n)\rightarrow 1$ and thus $\text{Pr}(E_n^c)\rightarrow 0$ as $n\rightarrow \infty$. \eqref{B} together with \eqref{Bc} asserts \eqref{diffcond}. The proof of this lemma is accomplished.	
	\end{proof}
	
	\begin{proof}[Proof of Lemma \ref{lem2}] 
		Recall that $L_n(\bm\theta)=n^{-1}\sum_{t=1}^{n}w_t\rho_{\tau}(y_t-q_t(\bm\theta))$ and $q_t(\bm\theta(\tau))=F_{t-1}^{-1}(\tau)$. 
		Let $\xi_t(\bm u)=\int_{0}^{1}\left[I(y_{t}\leq F_{t-1}^{-1}(\tau)+\nu_t(\bm u)s)-I(y_{t}\leq F_{t-1}^{-1}(\tau))\right]ds$ with $\nu_t(\bm u)=q_t(\bm u+\bm\theta(\tau))-q_t(\bm\theta(\tau))$. 
		By the Knight identity \eqref{identity}, it can be verified that
		\begin{align}\label{Gnrep}
			n[L_n(\bm u+\bm\theta(\tau))-L_n(\bm\theta(\tau))] =&\sum_{t=1}^{n}w_t\left[\rho_{\tau}\left(\eta_{t,\tau}-\nu_t(\bm u)\right)-\rho_{\tau}\left(\eta_{t,\tau}\right)\right] \nonumber \\
			=& K_{1n}(\bm u)+K_{2n}(\bm u),
		\end{align}	
		where $\bm u\in\Lambda\equiv\{\bm u\in\mathbb{R}^3: \bm u+\bm\theta(\tau) \in \Theta\}$, $\eta_{t,\tau}=y_t-q_t(\bm\theta(\tau))$,
		\begin{align*}
			K_{1n}(\bm u)=-\sum_{t=1}^{n}w_t\nu_t(\bm u)\psi_{\tau}(\eta_{t,\tau}) \hspace{2mm}\text{and}\hspace{2mm}  K_{2n}(\bm u)=\sum_{t=1}^{n}w_t\nu_t(\bm u)\xi_t(\bm u).
		\end{align*}
		By Taylor expansion, we have $\nu_t(\bm u)=q_{1t}(\bm u)+q_{2t}(\bm u)$, where $q_{1t}(\bm u)=\bm u^{\prime}\dot{q}_t(\bm\theta(\tau))$ and $q_{2t}(\bm u)=\bm u^{\prime}\ddot{q}_t(\bm u^*+\bm\theta(\tau))\bm u/2$ for $\bm u^*$ between $\bm u$ and $\bm 0$.
		Then it follows that
		\begin{align}\label{K1rep}
			K_{1n}(\bm u)&=-\sum_{t=1}^{n}w_tq_{1t}(\bm u)\psi_{\tau}(\eta_{t,\tau})-\sum_{t=1}^{n}w_tq_{2t}(\bm u)\psi_{\tau}(\eta_{t,\tau}) \nonumber \\
			&=-\sqrt{n}\bm u^{\prime}\bm T_n-\sqrt{n}\bm u^{\prime}R_{1n}(\bm u^*)\sqrt{n}\bm u,
		\end{align}
		where
		\[\bm T_n=\dfrac{1}{\sqrt{n}}\sum_{t=1}^{n}w_t\dot{q}_t(\bm\theta(\tau))\psi_{\tau}(\eta_{t,\tau}) \hspace{2mm}\text{and}\hspace{2mm} R_{1n}(\bm u^*)=\dfrac{1}{2n}\sum_{t=1}^{n}w_t\ddot{q}_t(\bm u^*+\bm\theta(\tau))\psi_{\tau}(\eta_{t,\tau}).\]
		By Lemma \ref{lem00}(ii) and the fact that $|\psi_{\tau}(\eta_{t,\tau})|\leq 1$, we have
		\[E\left[\sup_{\bm u^*\in\Lambda}\left\|w_t \ddot{q}_t(\bm u^*+\bm\theta(\tau))\psi_{\tau}(\eta_{t,\tau})\right\|\right]\leq CE\left[\sup_{\bm\theta^*\in\Theta}\left\|w_t \ddot{q}_t(\bm\theta^*)\right\|\right]<\infty.\]
		Moreover, by iterated-expectation and the fact that $E[\psi_{\tau}(\eta_{t,\tau}) | \mathcal{F}_{t-1}]=0$, it follows that
		\[E\left[w_t \ddot{q}_t(\bm u^*+\bm\theta(\tau))\psi_{\tau}(\eta_{t,\tau})\right]=0.\]
		Then by Theorem 3.1 in \cite{Ling_McAleer2003} and Assumption \ref{assum-Process}, we can show that
		\begin{align}\label{R1n}
			\sup_{\bm u^*\in\Lambda}\|R_{1n}(\bm u^*+\bm\theta(\tau))\|=o_p(1).
		\end{align}
		This together with \eqref{K1rep}, implies that
		\begin{align}\label{Kn1}
			K_{1n}(\bm u_n)=-\sqrt{n}\bm u_n^{\prime}\bm T_n+o_p(n\|\bm u_n\|^2).
		\end{align}
		Denote $\xi_t(\bm u)=\xi_{1t}(\bm u)+\xi_{2t}(\bm u)$, where
		\begin{align*}
			\xi_{1t}(\bm u)&=\int_{0}^{1}\left[I(y_{t}\leq F_{t-1}^{-1}(\tau)+q_{1t}(\bm u)s)-I(y_{t}\leq F_{t-1}^{-1}(\tau))\right]ds \hspace{2mm}\text{and}\hspace{2mm} \\
			\xi_{2t}(\bm u)&=\int_{0}^{1}\left[I(y_{t}\leq F_{t-1}^{-1}(\tau)+\nu_t(\bm u)s)-I(y_{t}\leq F_{t-1}^{-1}(\tau)+q_{1t}(\bm u)s)\right]ds.
		\end{align*}
		Then for $K_{2n}(\bm u)$, by Taylor expansion, it holds that
		\begin{align}\label{K2rep}
			K_{2n}(\bm u)=R_{2n}(\bm u)+R_{3n}(\bm u)+R_{4n}(\bm u)+R_{5n}(\bm u),
		\end{align}
		where
		\begin{align*}
			R_{2n}(\bm u)&=\bm u^{\prime}\sum_{t=1}^{n}w_t\dot{q}_t(\bm\theta(\tau))E[\xi_{1t}(\bm u)|\mathcal{F}_{t-1}], \\
			R_{3n}(\bm u)&=\bm u^{\prime}\sum_{t=1}^{n}w_t\dot{q}_t(\bm\theta(\tau))E[\xi_{2t}(\bm u)|\mathcal{F}_{t-1}], \\
			R_{4n}(\bm u)&=\bm u^{\prime}\sum_{t=1}^{n}w_t\dot{q}_t(\bm\theta(\tau))\left\{\xi_{t}(\bm u)-E[\xi_{t}(\bm u)|\mathcal{F}_{t-1}]\right\} \hspace{2mm}\text{and}\hspace{2mm} \\
			R_{5n}(\bm u)&=\dfrac{\bm u^{\prime}}{2}\sum_{t=1}^{n}w_t\ddot{q}_t(\bm\theta^*)\xi_{t}(\bm u)\bm u.
		\end{align*}
		Note that
		\begin{align}\label{xi1}
			E[\xi_{1t}(\bm u)|\mathcal{F}_{t-1}]=\int_{0}^{1}[F_{t-1}(F_{t-1}^{-1}(\tau)+q_{1t}(\bm u)s)-F_{t-1}(F_{t-1}^{-1}(\tau))]ds.
		\end{align}
		Then by Taylor expansion, together with Assumption \ref{assum-ConditionalDensity}, it follows that
		\begin{align*}
			E[\xi_{1t}(\bm u)|\mathcal{F}_{t-1}]=&\dfrac{1}{2}f_{t-1}(F_{t-1}^{-1}(\tau))q_{1t}(\bm u) \\
			&+q_{1t}(\bm u)\int_{0}^{1}[f_{t-1}(F_{t-1}^{-1}(\tau)+q_{1t}(\bm u)s^*)-f_{t-1}(F_{t-1}^{-1}(\tau))]sds,
		\end{align*}
		where $s^*$ is between 0 and $s$. Therefore, it follows that
		\begin{align}\label{R2nrep}
			R_{2n}(\bm u)=\sqrt{n}\bm u^{\prime}J_n\sqrt{n}\bm u+\sqrt{n}\bm u^{\prime}\Pi_{1n}(\bm u)\sqrt{n}\bm u,
		\end{align}
		where
		$J_n=(2n)^{-1}\sum_{t=1}^{n}f_{t-1}(F_{t-1}^{-1}(\tau))w_t\dot{q}_t(\bm\theta(\tau))\dot{q}_t^{\prime}(\bm\theta(\tau))$ and
		\begin{align*}
			\Pi_{1n}(\bm u)&=\dfrac{1}{n}\sum_{t=1}^{n}w_t\dot{q}_t(\bm\theta(\tau))\dot{q}_t^{\prime}(\bm\theta(\tau))\int_{0}^{1}[f_{t-1}(F_{t-1}^{-1}(\tau)+q_{1t}(\bm u)s^*)-f_{t-1}(F_{t-1}^{-1}(\tau))]sds.
		\end{align*}
		By Taylor expansion, together with Assumption \ref{assum-Space}(ii), $\sup_{x}|\dot{f}_{t-1}(x)|<\infty$ by Assumption \ref{assum-ConditionalDensity} and Lemma \ref{lem00}(i), for any $\eta>0$, it holds that
		\begin{align*}
			E\left(\sup_{\|\bm u\|\leq \eta}\|\Pi_{1n}(\bm u)\|\right)&\leq \dfrac{1}{n}\sum_{t=1}^{n}E\left[\sup_{\|\bm u\|\leq \eta} \|w_t\dot{q}_t(\bm\theta(\tau))\dot{q}_t^{\prime}(\bm\theta(\tau))\sup_{x}|\dot{f}_{t-1}(x)|\bm u^{\prime}\dot{q}_t(\bm\theta(\tau))\|\right] \\
			&\leq C\eta\sup_{x}|\dot{f}_{t-1}(x)| E[w_t\|\dot{q}_t(\bm\theta(\tau))\|^3]
		\end{align*}
		tends to $0$ as $\eta \to 0$.
		Therefore, by Markov’s theorem, for any $\epsilon$, $\delta>0$, there exists $\eta_0=\eta_0(\epsilon)>0$ such that
		\begin{align}\label{epsdelta1}
			\text{Pr}\left(\sup_{\|\bm u\|\leq \eta_0}\|\Pi_{1n}(\bm u)\|> \delta\right)<\dfrac{\epsilon}{2}
		\end{align}
		for all $n\geq 1$. Since $\bm u_n=o_p(1)$, it follows that
		\begin{align}\label{epsdelta2}
			\text{Pr}\left(\|\bm u_n\|> \eta_0\right)<\dfrac{\epsilon}{2}
		\end{align}
		as $n$ is large enough. From \eqref{epsdelta1} and \eqref{epsdelta2}, we have
		\begin{align*}
			\text{Pr}\left(\|\Pi_{1n}(\bm u_n)\|> \delta\right)&\leq \text{Pr}\left(\|\Pi_{1n}(\bm u_n)\|> \delta, \|\bm u_n\|\leq \eta_0\right)+\text{Pr}\left(\|\bm u_n\|> \eta_0\right) \\
			&\leq \text{Pr}\left(\sup_{\|\bm u\|\leq \eta_0}\|\Pi_{1n}(\bm u)\|> \delta\right)+\dfrac{\epsilon}{2}<\epsilon
		\end{align*}
		as $n$ is large enough. Thus $\Pi_{1n}(\bm u_n)=o_p(1)$. This together with \eqref{R2nrep}, implies that
		\begin{align}\label{R2n}
			R_{2n}(\bm u_n)=\sqrt{n}\bm u_n^{\prime}J_n\sqrt{n}\bm u_n+o_p(n\|\bm u_n\|^2).
		\end{align}
		
		Note that
		\begin{align}\label{xi2}
			E[\xi_{2t}(\bm u)|\mathcal{F}_{t-1}]=\int_{0}^{1}\left[F_{t-1}(F_{t-1}^{-1}(\tau)+\nu_t(\bm u)s)-F_{t-1}(F_{t-1}^{-1}(\tau)+q_{1t}(\bm u)s)\right]ds.
		\end{align}
		Then by Taylor expansion, the Cauchy-Schwarz inequality and the strict stationarity and ergodicity of $y_t$ under Assumption \ref{assum-Process}, together with Assumption \ref{assum-Space}(ii), $\sup_{x}f_{t-1}(x)<\infty$ by Assumption \ref{assum-ConditionalDensity} and Lemma \ref{lem00}, for any $\eta>0$, it holds that
		\begin{align*}
			E\bigg(\sup_{\|\bm u\|\leq \eta}\dfrac{|R_{3n}(\bm u)|}{n\|\bm u\|^2}\bigg)
			\leq & \dfrac{\eta}{n}\sum_{t=1}^{n}E\left\{w_t\left\|\dot{q}_t(\bm\theta(\tau))\right\|\dfrac{1}{2}\sup_{x}f_{t-1}(x)\sup_{\bm\theta\in\Theta}\left\|\ddot{q}_t(\bm\theta)\right\| \right\} \\
			\leq &  C\eta E\left\{\left\|\sqrt{w_t}\dot{q}_t(\bm\theta(\tau))\right\|\sup_{\bm\theta\in\Theta}\left\|\sqrt{w_t}\ddot{q}_t(\bm\theta)\right\| \right\} \\
			\leq & C\eta \left[E\left(w_t\left\|\dot{q}_t(\bm\theta(\tau))\right\|^2\right)\right]^{1/2}\left[E\left(\sup_{\bm\theta\in\Theta}w_t\left\|\ddot{q}_t(\bm\theta)\right\|^2\right)\right]^{1/2}
		\end{align*}
		tends to $0$ as $\eta \to 0$. Similar to \eqref{epsdelta1} and \eqref{epsdelta2}, we can show that
		\begin{align}\label{R3n}
			R_{3n}(\bm u_n)=o_p(n\|\bm u_n\|^2).
		\end{align}
		
		For $R_{4n}(\bm u)$, by Lemma \ref{lem1}, it holds that
		\begin{align}\label{R4n}
			R_{4n}(\bm u_n)=o_p(\sqrt{n}\|\bm u_n\|+n\|\bm u_n\|^2).
		\end{align}	
		
		Finally, we consider $R_{5n}(\bm u)$. Since $I(x\leq a)-I(x\leq b)=I(b\leq x \leq a)-I(b\geq x \geq a)$ and $\nu_t(\bm u)=\bm u^{\prime}\dot{q}_t(\bm\theta^{\star})$ with $\bm\theta^{\star}$ between $\bm\theta(\tau)$ and $\bm u+\bm\theta(\tau)$ by Taylor expansion, we have
		\begin{align}\label{sup-xi}
			\sup_{\|\bm u\|\leq \eta}|\xi_{t}(\bm u)| 
			\leq & \int_{0}^{1}\sup_{\|\bm u\|\leq \eta}\left|I(F_{t-1}^{-1}(\tau)\leq y_{t}\leq F_{t-1}^{-1}(\tau)+\nu_t(\bm u)s)\right|ds \nonumber \\ 
			& + \int_{0}^{1}\sup_{\|\bm u\|\leq \eta}\left|I(F_{t-1}^{-1}(\tau) \geq y_{t}\geq F_{t-1}^{-1}(\tau)+\nu_t(\bm u)s)\right|ds \nonumber\\
			\leq & I\left(F_{t-1}^{-1}(\tau)\leq y_{t}\leq F_{t-1}^{-1}(\tau)+ \eta\sup_{\bm\theta^{\star}\in\Theta}\left\|\dot{q}_t(\bm\theta^{\star})\right\|\right) \nonumber\\ 
			& + I\left(F_{t-1}^{-1}(\tau)\geq y_{t}\geq F_{t-1}^{-1}(\tau)- \eta\sup_{\bm\theta^{\star}\in\Theta}\left\|\dot{q}_t(\bm\theta^{\star})\right\|\right).
		\end{align}
		Then by iterated-expectation, the Cauchy-Schwarz inequality and the strict stationarity and ergodicity of $y_t$ under Assumption \ref{assum-Process}, together with $\sup_{x}f_{t-1}(x)<\infty$ by Assumption \ref{assum-ConditionalDensity} and Lemma \ref{lem00}, for any $\eta>0$, it follows that
		\begin{align*}
			E\left(\sup_{\|\bm u\|\leq \eta}\dfrac{|R_{5n}(\bm u)|}{n\|\bm u\|^2}\right)  
			\leq & \dfrac{1}{2n}\sum_{t=1}^{n}E\left[w_t\sup_{\bm\theta^*\in\Theta}\left\|\ddot{q}_t(\bm\theta^*)\right\|E\left(\sup_{\|\bm u\|\leq \eta}|\xi_{t}(\bm u)| | \mathcal{F}_{t-1}\right)\right] \\
			\leq & \eta\sup_{x}f_{t-1}(x) E\left[w_t\sup_{\bm\theta^*\in\Theta}\left\|\ddot{q}_t(\bm\theta^*)\right\|\sup_{\bm\theta^{\star}\in\Theta}\left\|\dot{q}_t(\bm\theta^{\star})\right\| \right] \\
			\leq & C\eta \left[E\left(\sup_{\bm\theta^*\in\Theta}w_t\left\|\ddot{q}_t(\bm\theta^*)\right\|^2\right)\right]^{1/2}\left[E\left(\sup_{\bm\theta^{\star}\in\Theta}w_t\left\|\dot{q}_t(\bm\theta^{\star})\right\|^2\right)\right]^{1/2} 
		\end{align*}
		tends to $0$ as $\eta \to 0$. Similar to \eqref{epsdelta1} and \eqref{epsdelta2}, we can show that
		\begin{align}\label{R5n}
			R_{5n}(\bm u_n)=o_p(n\|\bm u_n\|^2).
		\end{align}
		From \eqref{K2rep}, \eqref{R2n}, \eqref{R3n}, \eqref{R4n} and \eqref{R5n}, we have
		\begin{align}\label{Kn2}
			K_{2n}(\bm u_n)=\sqrt{n}\bm u_n^{\prime}J_n\sqrt{n}\bm u_n+o_p(\sqrt{n}\|\bm u_n\|+n\|\bm u_n\|^2).
		\end{align}
		In view of \eqref{Gnrep}, \eqref{Kn1} and \eqref{Kn2}, we accomplish the proof of this lemma.
	\end{proof}
	
	\begin{proof}[Proof of Lemma \ref{lem3}]
		Recall that $\eta_{t,\tau}=y_t-q_t(\bm\theta(\tau))$, $\nu_t(\bm u)=q_t(\bm u+\bm\theta(\tau))-q_t(\bm\theta(\tau))$ and $\xi_t(\bm u)=\int_{0}^{1}[I(y_{t}\leq F_{t-1}^{-1}(\tau)+\nu_t(\bm u)s)-I(y_{t}\leq F_{t-1}^{-1}(\tau))]ds$ with $F_{t-1}^{-1}(\tau)=q_t(\bm\theta(\tau))$. 
		Let $\widetilde{\eta}_{t,\tau}=y_t-\widetilde{q}_t(\bm\theta(\tau))$, $\widetilde{\nu}_t(\bm u)=\widetilde{q}_t(\bm u+\bm\theta(\tau))-\widetilde{q}_t(\bm\theta(\tau))$
		and 
		$\widetilde{\xi}_t(\bm u)=\int_{0}^{1}[I(y_{t}\leq \widetilde{q}_t(\bm\theta(\tau))+\widetilde{\nu}_t(\bm u)s)-I(y_{t}\leq \widetilde{q}_t(\bm\theta(\tau)))]ds$. 
		Similar to \eqref{Gnrep}, by the Knight identity \eqref{identity}, we can verify that
		\begin{align}\label{initialHnrep}
			&n[\widetilde{L}_n(\bm u+\bm\theta(\tau))-\widetilde{L}_n(\bm\theta(\tau))]-n[L_n(\bm u+\bm\theta(\tau))-L_n(\bm\theta(\tau))] \nonumber\\
			=& \sum_{t=1}^n w_t \left\{[-\widetilde{\nu}_t(\bm u)\psi_{\tau}(\widetilde{\eta}_{t,\tau})+\widetilde{\nu}_t(\bm u)\widetilde{\xi}_t(\bm u)] - [-\nu_t(\bm u)\psi_{\tau}(\eta_{t,\tau})+\nu_t(\bm u)\xi_t(\bm u)]\right\} \nonumber\\
			=& \widetilde{A}_{1n}(\bm u)+\widetilde{A}_{2n}(\bm u)+\widetilde{A}_{3n}(\bm u)+\widetilde{A}_{4n}(\bm u),
		\end{align}
		where $\bm u\in\Lambda\equiv\{\bm u\in\mathbb{R}^3: \bm u+\bm\theta(\tau) \in \Theta\}$,
		\begin{align*}
			\widetilde{A}_{1n}(\bm u)&= \sum_{t=1}^n w_t [\nu_t(\bm u)-\widetilde{\nu}_t(\bm u)]\psi_{\tau}(\widetilde{\eta}_{t,\tau}),\quad \widetilde{A}_{2n}(\bm u)=\sum_{t=1}^n w_t [\psi_{\tau}(\eta_{t,\tau})-\psi_{\tau}(\widetilde{\eta}_{t,\tau})]\nu_t(\bm u),\\
			\widetilde{A}_{3n}(\bm u)&= \sum_{t=1}^n w_t [\widetilde{\nu}_t(\bm u)-\nu_t(\bm u)]\widetilde{\xi}_t(\bm u) \;\;\text{and}\;\; \widetilde{A}_{4n}(\bm u)=\sum_{t=1}^n w_t [\widetilde{\xi}_t(\bm u)-\xi_t(\bm u)] \nu_t(\bm u).
		\end{align*}
		
		We first consider $\widetilde{A}_{1n}(\bm u)$. 
		Since $|\psi_{\tau}(\cdot)|\leq 1$, $\{y_t\}$ is strictly stationary and ergodic by Assumption \ref{assum-Process} and $E(w_t)<\infty$ and $E(w_{t}|y_{t-j}|^3)<\infty$ for all $j\geq 1$ by Assumption \ref{assum-RandomWeight}, then by Taylor expansion and Lemma \ref{lem0}(ii), we have
		\begin{align*}
			\sup_{\bm u\in\Lambda}\dfrac{|\widetilde{A}_{1n}(\bm u)|}{\sqrt{n}\|\bm u\|} 
			&\leq \dfrac{1}{\sqrt{n}}\sum_{t=1}^n w_t \sup_{\bm u\in\Lambda}\dfrac{|\nu_t(\bm u)-\widetilde{\nu}_t(\bm u)|}{\|\bm u\|}|\psi_{\tau}(\widetilde{\eta}_{t,\tau})| \\
			&\leq \dfrac{1}{\sqrt{n}}\sum_{t=1}^n w_t \sup_{\Theta}\|\dot{q}_t(\bm\theta^*)-\dot{\widetilde{q}}_t(\bm\theta^*)\| \\
			&\leq \dfrac{C}{\sqrt{n}}\sum_{t=1}^n \rho^t w_t(\varsigma_{\rho}+\xi_{\rho})+ \dfrac{C}{\sqrt{n}}\sum_{t=1}^n t\rho^t w_t\varsigma_{\rho}=o_p(1),
		\end{align*}
		where $\bm\theta^*$ is between $\bm\theta$ and $\bm\theta(\tau)$, $\varsigma_{\rho}=\sum_{s=0}^{\infty}\rho^{s}|y_{-s}|$ and $\xi_{\rho}=\sum_{s=0}^{\infty}s\rho^{s}|y_{-s}|$. Therefore, for $\bm u_n=o_p(1)$, it holds that
		\begin{align}\label{Pi1tilde}
			\widetilde{A}_{1n}(\bm u_n)=o_p(\sqrt{n}\|\bm u_n\|).
		\end{align}

		We next consider $\widetilde{A}_{2n}(\bm u)$. 
		Using $I(x < a)-I(x < b)=I(0 < x-b < a-b)-I(0> x-b > a-b)$ and $\psi_{\tau}(\eta_{t,\tau})-\psi_{\tau}(\widetilde{\eta}_{t,\tau})=I(y_t<\widetilde{q}_t(\bm\theta(\tau)))-I(y_t<q_t(\bm\theta(\tau)))$, we have 
		\begin{align*}
			E[|\psi_{\tau}(\eta_{t,\tau})-\psi_{\tau}(\widetilde{\eta}_{t,\tau})| |\mathcal{F}_{t-1}] 
			\leq & E\left[I(0< y_t-q_t(\bm\theta(\tau))< |\widetilde{q}_t(\bm\theta(\tau))-q_t(\bm\theta(\tau))|) |\mathcal{F}_{t-1}\right] \\
			&+ E\left[I(0> y_t-q_t(\bm\theta(\tau))> -|\widetilde{q}_t(\bm\theta(\tau))-q_t(\bm\theta(\tau))|)|\mathcal{F}_{t-1}\right] \\
			\leq & F_{t-1}\left(q_t(\bm\theta(\tau))+|\widetilde{q}_t(\bm\theta(\tau))-q_t(\bm\theta(\tau))|\right) \\
			&-F_{t-1}\left(q_t(\bm\theta(\tau))-|\widetilde{q}_t(\bm\theta(\tau))-q_t(\bm\theta(\tau))|\right).
		\end{align*} 
		Then by iterative-expectation and Cauchy-Schwarz inequality, together with $\nu_t(\bm u)=\bm u^{\prime}\dot{q}_t(\bm\theta^*)$ by Taylor expansion, Lemma \ref{lem00}(i), Lemma \ref{lem0}(i), $\sup_{x}f_{t-1}(x)<\infty$ by Assumption \ref{assum-ConditionalDensity} and $E(w_t)<\infty$ and $E(w_{t}|y_{t-j}|^3)<\infty$ for all $j\geq 1$ by Assumption \ref{assum-RandomWeight}, it holds that
		\begin{align*}
			E\sup_{\bm u\in\Lambda}\dfrac{|\widetilde{A}_{2n}(\bm u)|}{\sqrt{n}\|\bm u\|} 
			&\leq \dfrac{1}{\sqrt{n}}\sum_{t=1}^n E\left\{ w_t\sup_{\Theta}\|\dot{q}_t(\bm\theta^*)\|\cdot E[|\psi_{\tau}(\eta_{t,\tau})-\psi_{\tau}(\widetilde{\eta}_{t,\tau})| |\mathcal{F}_{t-1}] \right\} \\
			&\leq 2C\sup_{x}f_{t-1}(x)\dfrac{1}{\sqrt{n}}\sum_{t=1}^n \rho^t E\left\{w_t\sup_{\Theta}\|\dot{q}_t(\bm\theta^*)\| \varsigma_{\rho} \right\} \\
			&\leq \dfrac{C}{\sqrt{n}}\sum_{t=1}^n \rho^t \cdot \left[E\left(w_t\sup_{\Theta}\|\dot{q}_t(\bm\theta^*)\|^2\right)\right]^{1/2}\cdot \left[E(w_t\varsigma_{\rho}^2)\right]^{1/2}=o(1).
		\end{align*}
		As a result, for $\bm u_n=o_p(1)$, it follows that
		\begin{align}\label{Pi2tilde}
			\widetilde{A}_{2n}(\bm u_n)=o_p(\sqrt{n}\|\bm u_n\|).
		\end{align}

		For $\widetilde{A}_{3n}(\bm u)$, since $|\widetilde{\xi}_t(\bm u)|<2$, similar to the proof of $\widetilde{A}_{1n}(\bm u)$, for $\bm u_n=o_p(1)$, we have
		\begin{align}\label{Pi3tilde}
			\widetilde{A}_{3n}(\bm u_n)=o_p(\sqrt{n}\|\bm u_n\|).
		\end{align}

		Finally, we consider $\widetilde{A}_{4n}(\bm u)$. Denote $\widetilde{c}_t=I(y_t\leq\widetilde{q}_t(\bm\theta(\tau)))-I(y_t\leq q_t(\bm\theta(\tau)))$ and $\widetilde{d}_t= \int_{0}^1 \delta_t(s) ds$ with $\widetilde{\delta}_t(s)=I(y_t\leq\widetilde{q}_t(\bm\theta(\tau))+\widetilde{\nu}_t(\bm u)s)-I(y_t\leq q_t(\bm\theta(\tau))+\nu_t(\bm u)s)$. 
		Using $I(X\leq a)-I(X\leq b)=I(b\leq X \leq a)-I(b\geq X \geq a)$, it holds that
		\begin{align*}
			|\widetilde{c}_t| \leq & I\left(|y_t-q_t(\bm\theta(\tau))| \leq |\widetilde{q}_t(\bm\theta(\tau))-q_t(\bm\theta(\tau))|\right) \quad \text{and} \\
			\sup_{\bm u\in\Lambda}|\widetilde{\delta}_t(s)| \leq & I\left(|y_t-q_t(\bm\theta(\tau))-\nu_t(\bm u)s| \leq |\widetilde{q}_t(\bm\theta(\tau))-q_t(\bm\theta(\tau))|+\sup_{\bm u\in\Lambda}|\widetilde{\nu}_t(\bm u)-\nu_t(\bm u)|s \right).
		\end{align*}
		Then by Taylor expansion, together with $\sup_{x}f_{t-1}(x)<\infty$ under Assumption \ref{assum-ConditionalDensity} and Lemma \ref{lem0}, we have 
		\begin{align*}
			E\left(|\widetilde{c}_t| |\mathcal{F}_{t-1}\right) 
			\leq & 2\sup_{x}f_{t-1}(x) |\widetilde{q}_t(\bm\theta(\tau))-q_t(\bm\theta(\tau))| \leq C\rho^t \varsigma_{\rho} \quad \text{and} \\
			E\left(\sup_{\bm u\in\Lambda}|\widetilde{\delta}_t(s)| |\mathcal{F}_{t-1}\right) 
			\leq & 2\sup_{x}f_{t-1}(x) \left(|\widetilde{q}_t(\bm\theta(\tau))-q_t(\bm\theta(\tau))|+\sup_{\bm u\in\Lambda}|\widetilde{\nu}_t(\bm u)-\nu_t(\bm u)|\right) \\
			\leq & C\rho^t [\varsigma_{\rho} + \|\bm u\|(\varsigma_{\rho}+t\varsigma_{\rho}+\xi_{\rho})].
		\end{align*}
		These together with $\widetilde{\xi}_t(\bm u)-\xi_t(\bm u)=\widetilde{d}_t-\widetilde{c}_t$, imply that
		\begin{align*}
			E\left(\sup_{\bm u\in\Lambda}|\widetilde{\xi}_t(\bm u)-\xi_t(\bm u)| |\mathcal{F}_{t-1}\right) \leq C\rho^t \varsigma_{\rho} + C\|\bm u\|\rho^t(\varsigma_{\rho}+t\varsigma_{\rho}+\xi_{\rho}).
		\end{align*}
		As a result, by iterative-expectation and Cauchy-Schwarz inequality, together with $\nu_t(\bm u)=\bm u^{\prime}\dot{q}_t(\bm u^*+\bm\theta(\tau))$ by Taylor expansion, Lemma \ref{lem00}(i) and $E(w_t)<\infty$ and $E(w_{t}|y_{t-j}|^3)<\infty$ for all $j\geq 1$ by Assumption \ref{assum-RandomWeight}, we have
		\begin{align*}
			E\sup_{\bm u\in\Lambda}\dfrac{|\widetilde{A}_{4n}(\bm u)|}{\sqrt{n}\|\bm u\|+n\|\bm u\|^2} 
			\leq & \sum_{t=1}^n E\left\{w_t E\left(\sup_{\bm u\in\Lambda}\dfrac{|\widetilde{\xi}_t(\bm u)-\xi_t(\bm u)|}{\sqrt{n}+n\|\bm u\|} |\mathcal{F}_{t-1}\right)\sup_{\bm u\in\Lambda}\dfrac{|\nu_t(\bm u)|}{\|\bm u\|} \right\} \\
			\leq & \dfrac{C}{\sqrt{n}}\sum_{t=1}^n \rho^t \cdot \left[E\left(w_t\sup_{\Theta}\|\dot{q}_t(\bm\theta^*)\|^2\right)\right]^{1/2}\cdot \left[E(w_t\varsigma_{\rho}^2)\right]^{1/2} \\
			& + \dfrac{C}{n} \sum_{t=1}^n \rho^t \cdot \left[E\left(w_t\sup_{\Theta}\|\dot{q}_t(\bm\theta^*)\|^2\right)\right]^{1/2}\cdot \left[E(w_t\varsigma_{\rho}^2)\right]^{1/2} \\
			& + \dfrac{C}{n} \sum_{t=1}^n t\rho^t \cdot \left[E\left(w_t\sup_{\Theta}\|\dot{q}_t(\bm\theta^*)\|^2\right)\right]^{1/2}\cdot \left[E(w_t\varsigma_{\rho}^2)\right]^{1/2} \\
			& + \dfrac{C}{n} \sum_{t=1}^n \rho^t \cdot \left[E\left(w_t\sup_{\Theta}\|\dot{q}_t(\bm\theta^*)\|^2\right)\right]^{1/2}\cdot \left[E(w_t\xi_{\rho}^2)\right]^{1/2} = o(1).
		\end{align*}
		Hence, for $\bm u_n=o_p(1)$, it follows that
		\begin{align}\label{Pi4tilde}
			\widetilde{A}_{4n}(\bm u_n)=o_p(\sqrt{n}\|\bm u_n\|+n\|\bm u_n\|^2).
		\end{align}
		Combining \eqref{initialHnrep}--\eqref{Pi4tilde}, we accomplish the proof of this lemma.
	\end{proof}

	\begin{proof}[Proof of Lemma \ref{lem-Tightness1}]
		Recall that $q_t(\bm\theta(\tau))=\omega(\tau) + \alpha_1(\tau)\sum_{j=1}^\infty \beta_1(\tau)^{j-1}|y_{t-j}|$. By the Lipschitz continuous conditions in Assumption \ref{assum-Tightness} and Taylor expansion, we have
		\begin{align}\label{Lipschitz}
			\left|\omega(\tau_2)-\omega(\tau_1)\right|&\leq C|\tau_2-\tau_1|, \; \left|\alpha_1(\tau_2)-\alpha_1(\tau_1)\right|\leq C|\tau_2-\tau_1| \quad\text{and}\quad \nonumber\\
			\left|\beta_1^{j}(\tau_2)-\beta_1^{j}(\tau_1)\right|&\leq Cj(\beta_1^*)^{j-1}|\tau_2-\tau_1|,		
		\end{align} 
		where $\beta_1^*$ is between $\beta_1(\tau_1)$ and $\beta_1(\tau_2)$. 
		These together with $|\alpha_1(\tau)|<\overline{c}<\infty$ and $\beta_1(\tau), \beta_1^*\leq\rho<1$ by Assumption \ref{assum-Space}, imply that for $\tau_1,\tau_2\in\mathcal{T}$,
		\begin{align*}
			&|q_t(\bm\theta(\tau_2))-q_t(\bm\theta(\tau_1))| \\
			= & \left|\omega(\tau_2)-\omega(\tau_1) + \sum_{j=1}^\infty [\alpha_1(\tau_2)\beta_1^{j-1}(\tau_2)-\alpha_1(\tau_1)\beta_1^{j-1}(\tau_1)]|y_{t-j}|\right| \\
			\leq & \left|\omega(\tau_2)-\omega(\tau_1)\right| + \left|\alpha_1(\tau_2)-\alpha_1(\tau_1)\right|\sum_{j=1}^\infty\beta_1^{j-1}(\tau_2)|y_{t-j}| + |\alpha_1(\tau_1)|\sum_{j=1}^\infty \left|\beta_1^{j-1}(\tau_2)-\beta_1^{j-1}(\tau_1)\right| |y_{t-j}| \\
			\leq & C|\tau_2-\tau_1| \left(1 + \sum_{j=1}^\infty\rho^{j-1}|y_{t-j}| + \overline{c}\sum_{j=2}^\infty(j-1)\rho^{j-2}|y_{t-j}| \right)
			\leq  C|\tau_2-\tau_1|\Delta_{\rho,t},		
		\end{align*}
		where $\Delta_{\rho,t}=1+\sum_{j=1}^\infty\rho^{j-1}|y_{t-j}|+\sum_{j=2}^\infty(j-1)\rho^{j-2}|y_{t-j}|+\sum_{j=3}^\infty(j-1)(j-2)\rho^{j-3}|y_{t-j}|+\sum_{j=4}^\infty(j-1)(j-2)(j-3)\rho^{j-4}|y_{t-j}|$. 
		Therefore, (i) holds. 
		
		For (ii), recall the first derivative of $q_t(\bm\theta)$ defined in \eqref{1st-derivative-qt}. It holds that
		\begin{align*}
			&\dot{q}_t(\bm u+\bm\theta(\tau_2))-\dot{q}_t(\bm u+\bm\theta(\tau_1)) \\
			=&\bigg(0,\sum_{j=1}^{\infty}\{[u_{3}+\beta_1(\tau_2)]^{j-1}-[u_{3}+\beta_1(\tau_1)]^{j-1}\}|y_{t-j}|, \\ & \sum_{j=2}^{\infty}(j-1)\{[u_{2}+\alpha_1(\tau_2)][u_{3}+\beta_1(\tau_2)]^{j-2}-[u_{2}+\alpha_1(\tau_1)][u_{3}+\beta_1(\tau_1)]^{j-2}\}|y_{t-j}|\bigg)^{\prime}. 	
		\end{align*}
		By Taylor expansion, we have 
		\begin{align}\label{beta}
			[u_{3}+\beta_1(\tau_2)]^{j}-[u_{3}+\beta_1(\tau_1)]^{j}=j(\beta_1^*)^{j-1}[\beta_1(\tau_2)-\beta_1(\tau_1)],
		\end{align}
		where $\beta_1^*$ is between $u_{3}+\beta_1(\tau_1)$ and $u_{3}+\beta_1(\tau_2)$. Moreover, it follows that
		\begin{align*}
			&[u_{2}+\alpha_1(\tau_2)][u_{3}+\beta_1(\tau_2)]^{j}-[u_{2}+\alpha_1(\tau_1)][u_{3}+\beta_1(\tau_1)]^{j} \\
			=& [\alpha_1(\tau_2)-\alpha_1(\tau_1)][u_{3}+\beta_1(\tau_2)]^{j} 
			+ j(\beta_1^*)^{j-1}[u_{2}+\alpha_1(\tau_1)][\beta_1(\tau_2)-\beta_1(\tau_1)].
		\end{align*}   
		These together with \eqref{Lipschitz}, $|\alpha_1(\tau)|<\overline{c}<\infty$ and $\beta_1(\tau), \beta_1^*\leq\rho<1$ by Assumption \ref{assum-Space}, for $\tau_1,\tau_2\in\mathcal{T}$ and $\bm u=(u_{1},u_{2},u_{3})^{\prime}\in \Lambda$ such that $\bm u+\bm\theta(\tau_i) \in \Theta$ with $i=1,2$, imply that 
		\begin{align*}
			&\|\dot{q}_t(\bm u+\bm\theta(\tau_2))-\dot{q}_t(\bm u+\bm\theta(\tau_1))\| \\
			\leq &  \sum_{j=1}^{\infty}|[u_{3}+\beta_1(\tau_2)]^{j-1}-[u_{3}+\beta_1(\tau_1)]^{j-1}||y_{t-j}| \\
			& + \sum_{j=2}^{\infty}(j-1)|[u_{2}+\alpha_1(\tau_2)][u_{3}+\beta_1(\tau_2)]^{j-2}-[u_{2}+\alpha_1(\tau_1)][u_{3}+\beta_1(\tau_1)]^{j-2}||y_{t-j}|\\
			\leq & [|\alpha_1(\tau_2)-\alpha_1(\tau_1)|+|\beta_1(\tau_2)-\beta_1(\tau_1)|] \sum_{j=2}^\infty(j-1)\rho^{j-2}|y_{t-j}| \\
			& + \overline{c}|\beta_1(\tau_2)-\beta_1(\tau_1)|\sum_{j=3}^\infty(j-1)(j-2)\rho^{j-3}|y_{t-j}| \\
			\leq & C|\tau_2-\tau_1|\Delta_{\rho,t}.
		\end{align*}
		Hence, (ii) is verified. 
		
		To show (iii), recall the second derivative of $q_t(\bm\theta)$ defined in \eqref{2nd-derivative-qt}. 
		For $\tau_1,\tau_2\in\mathcal{T}$ and $\bm u=(u_{1},u_{2},u_{3})^{\prime}\in \Lambda$ such that $\bm u+\bm\theta(\tau) \in \Theta$, by $\max\{|\alpha_1(\tau)|, |u_{2}+\alpha_1(\tau)|\}<\overline{c}<\infty$ and $\beta_1(\tau), \beta_1^*,u_{3}+\beta_1(\tau)\leq\rho<1$ for $\bm u+\bm\theta(\tau) \in \Theta$ under Assumption \ref{assum-Space}, together with \eqref{Lipschitz}--\eqref{beta}, we have
		\begin{align*}
			&\|\ddot{q}_t(\bm u+\bm\theta(\tau_2))-\ddot{q}_t(\bm u+\bm\theta(\tau_1))\| \\
			\leq & 2\sum_{j=2}^{\infty}(j-1)\sup_{\Theta}|[u_{3}+\beta_1(\tau_2)]^{j-2}-[u_{3}+\beta_1(\tau_1)]^{j-2}||y_{t-j}| \\
			& + \sup_{\Theta}\left|\alpha_1(\tau_2)-\alpha_1(\tau_1)\right|\sum_{j=3}^\infty (j-1)(j-2)\sup_{\Theta}[u_{3}+\beta_1(\tau_1)]^{j-3}|y_{t-j}| \\
			& + \sup_{\Theta}|u_{2}+\alpha_1(\tau_2)| \sum_{j=3}^\infty (j-1)(j-2)\sup_{\Theta}|[u_{3}+\beta_1(\tau_2)]^{j-3}-[u_{3}+\beta_1(\tau_1)]^{j-3}||y_{t-j}| \\
			\leq & C|\tau_2-\tau_1| \left(3\sum_{j=3}^\infty(j-1)(j-2)\rho^{j-3}|y_{t-j}| + \overline{c}\sum_{j=4}^\infty(j-1)(j-2)(j-3)\rho^{j-4}|y_{t-j}|\right) \\
			\leq & C|\tau_2-\tau_1|\Delta_{\rho,t}.
		\end{align*}
		As a result, (iii) holds. The proof of this lemma is complete. 
	\end{proof}

	\begin{proof}[Proof of Lemma \ref{lem-Tightness2}] 
		For any $\eta>0$, \eqref{diffcond} in Lemma \ref{lem1} and the proof for \eqref{R5n} in Lemma \ref{lem2} imply that
		\begin{align}\label{R4R5}
			\sup_{\|\bm u\|\leq \eta}\dfrac{|R_{4n}(\bm u,\tau)|}{\sqrt{n}\|\bm u\|+n\|\bm u\|^2}=o_p(1) \quad\text{and}\quad \sup_{\|\bm u\|\leq \eta}\dfrac{|R_{5n}(\bm u,\tau)|}{n\|\bm u\|^2}=o_p(1).
		\end{align}
		By Corollary 2.2 of \cite{Newey1991}, to show Lemma \ref{lem-Tightness2}, it remains to establish the stochastic equicontinuity of $\sup_{\|\bm u\|\leq \eta}|R_{4n}(\bm u,\tau)|/(\sqrt{n}\|\bm u\|+n\|\bm u\|^2)$ and $\sup_{\|\bm u\|\leq \eta}|R_{5n}(\bm u,\tau)|/(n\|\bm u\|^2)$.

		We first consider the stochastic equicontinuity of $\sup_{\|\bm u\|\leq \eta}|R_{4n}(\bm u,\tau)|/(\sqrt{n}\|\bm u\|+n\|\bm u\|^2)$. 
		Denote $R_{4n}(\bm u,\tau)=\bm u^{\prime}\sum_{t=1}^{n}w_t\dot{q}_t(\bm\theta(\tau))\bar{\xi}_{t}(\bm u, \tau)$, where $\bar{\xi}_{t}(\bm u, \tau)=\xi_{t}(\bm u, \tau)-E[\xi_{t}(\bm u, \tau)|\mathcal{F}_{t-1}]$,
		\[\xi_{t}(\bm u, \tau)=\int_{0}^{1}\left[I(y_t\leq F_{t-1}^{-1}(\tau)+\nu_t(\bm u, \tau)s)-I(y_t\leq F_{t-1}^{-1}(\tau))\right]ds,\] 
		with $\nu_t(\bm u, \tau)=q_t(\bm u+\bm\theta(\tau))-q_t(\bm\theta(\tau))$.
		Recall that $c_t\equiv c_t(\tau_1,\tau_2)=I(y_t<q_t(\bm\theta(\tau_1)))-I(y_t<q_t(\bm\theta(\tau_2)))$. Let $d_t(\bm u)\equiv d_t(\bm u, \tau_1,\tau_2)= \int_{0}^1 \delta_t(\bm u, s) ds$, where $\delta_t(\bm u, s)\equiv \delta_t(\bm u, \tau_1,\tau_2, s)=I(y_t\leq Q_t(\bm u, \tau_2, s))-I(y_t\leq Q_t(\bm u, \tau_1, s))$ with $Q_t(\bm u, \tau, s)=F_{t-1}^{-1}(\tau)+\nu_t(\bm u, \tau)s=F_{t-1}^{-1}(\tau) + \bm u^{\prime}\dot{q}_t(\bm u^*+\bm\theta(\tau))s$ for $\bm u^*$ between $\bm 0$ and $\bm u$.  
		Note that $\xi_{t}(\bm u, \tau_2)-\xi_{t}(\bm u, \tau_1)=c_t + d_t$ and $|\sup_x|g_1(x)|-\sup_x|g_2(x)||\leq \sup_x||g_1(x)|-|g_2(x)||$ for functions $g_1$ and $g_2$. Then it holds that
		\begin{align}\label{R4-Tightness-decompose}
			\left|\sup_{\|\bm u\|\leq \eta}\dfrac{|R_{4n}(\bm u,\tau_2)|}{\sqrt{n}\|\bm u\|+n\|\bm u\|^2} - \sup_{\|\bm u\|\leq \eta}\dfrac{|R_{4n}(\bm u,\tau_1)|}{\sqrt{n}\|\bm u\|+n\|\bm u\|^2}\right| \leq \sum_{i=1}^3 \sup_{\|\bm u\|\leq \eta}\dfrac{|R_{4i}(\bm u,\tau_1,\tau_2)|}{\sqrt{n}\|\bm u\|},
		\end{align}
		where 
		\begin{align*}
			R_{41}(\bm u,\tau_1,\tau_2) &= \bm u^{\prime}\sum_{t=1}^{n}w_t[\dot{q}_t(\bm\theta(\tau_2))-\dot{q}_t(\bm\theta(\tau_1))]\bar{\xi}_{t}(\bm u, \tau_2), \\ 
			R_{42}(\bm u,\tau_1,\tau_2) &= \bm u^{\prime}\sum_{t=1}^{n}w_t\dot{q}_t(\bm\theta(\tau_1))[c_t-E(c_t|\mathcal{F}_{t-1})] \;\text{and}\\
			R_{43}(\bm u,\tau_1,\tau_2) &= \bm u^{\prime}\sum_{t=1}^{n}w_t\dot{q}_t(\bm\theta(\tau_1))[d_t(\bm u)-E(d_t(\bm u)|\mathcal{F}_{t-1})].
		\end{align*} 
		Using $I(X\leq a)-I(X\leq b)=I(b\leq X \leq a)-I(b\geq X \geq a)$, similar to the proof of \eqref{sup-xi}, we can show that
		\begin{align*}
			\sup_{\|\bm u\|\leq \eta}\xi_{t}^2(\bm u, \tau) \leq &  I\left(F_{t-1}^{-1}(\tau)-\eta\sup_{\Theta}\|\dot{q}_t(\bm\theta)\| \leq y_t\leq F_{t-1}^{-1}(\tau)+\eta\sup_{\Theta}\|\dot{q}_t(\bm\theta)\|\right). 
		\end{align*}
		Then by Taylor expansion, it follows that
		\[E\left(\sup_{\|\bm u\|\leq \eta}\xi_{t}^2(\bm u, \tau)|\mathcal{F}_{t-1}\right) \leq 2\eta\sup_{\Theta}\|\dot{q}_t(\bm\theta)\|.\]
		Note that $E(\bar{\xi}_{t}(\bm u, \tau)|\mathcal{F}_{t-1})=0$ and $\var(n^{-1/2}\sum_{t=1}^{n}M_t)=\var(M_t)$ for a martingale difference sequence $\{M_t\}$. Then by iterative-expectation and the Cauchy-Schwarz inequality, together with Lemma \ref{lem-Tightness1}(ii), Lemma \ref{lem00}(i) and $E(w_t\Delta_{\rho,t}^2)<\infty$ under Assumptions \ref{assum-Space} and \ref{assum-RandomWeight}, it can be verified that
		\begin{align*}
			& \var\left(\sup_{\|\bm u\|\leq \eta}\dfrac{|R_{41}(\bm u,\tau_1,\tau_2)|}{\sqrt{n}\|\bm u\|}\right) \\
			=&\var\left(w_t\|\dot{q}_t(\bm\theta(\tau_2))-\dot{q}_t(\bm\theta(\tau_1))\|\sup_{\|\bm u\|\leq \eta}|\bar{\xi}_{t}(\bm u, \tau_2)|\right) \\
			\leq & E\left[w_t^2\|\dot{q}_t(\bm\theta(\tau_2))-\dot{q}_t(\bm\theta(\tau_1))\|^2E\left(\sup_{\|\bm u\|\leq \eta}\xi_{t}^2(\bm u, \tau_2)|\mathcal{F}_{t-1}\right)\right] \\
			\leq & C|\tau_2-\tau_1|^2\eta E\left(w_t^2\Delta_{\rho,t}^2\sup_{\Theta}\|\dot{q}_t(\bm\theta)\|\right)\leq C|\tau_2-\tau_1|^2.
		\end{align*}
		This implies that
		\begin{align}\label{R41}
			\sup_{\|\bm u\|\leq \eta}\dfrac{|R_{41}(\bm u,\tau_1,\tau_2)|}{\sqrt{n}\|\bm u\|}
			= O_p(1)|\tau_2-\tau_1|. 
		\end{align} 
		We next consider $R_{42}(\bm u,\tau_1,\tau_2)$. 
		By Lemma \ref{lem-Tightness1}, we can show that
		\begin{align*}
			& \sup_{\|\bm u\|\leq \eta}|Q_t(\bm u, \tau_2, s)-Q_t(\bm u, \tau_1, s)| \\
			\leq & |F_{t-1}^{-1}(\tau_2)-F_{t-1}^{-1}(\tau_1)| +\sup_{\|\bm u\|\leq \eta}\|\bm u\|\|\dot{q}_t(\bm u^*+\bm\theta(\tau_2))-\dot{q}_t(\bm u^*+\bm\theta(\tau_1))\| \\
			\leq & C|\tau_2-\tau_1|(1+\eta)\Delta_{\rho,t}. 
		\end{align*}
		Thus, by Taylor expansion and $\sup_{x}f_{t-1}(x)<\infty$ under Assumption \ref{assum-ConditionalDensity}, together with $I(x \leq a)-I(x \leq b)=I(0 \leq x-b \leq a-b)-I(0\geq x-b \geq a-b)$, we have
		\begin{align*}
			& E\left(\sup_{\|\bm u\|\leq \eta}\delta_t^2(\bm u, s)|\mathcal{F}_{t-1}\right) \\
			\leq & F_{t-1}\left(Q_t(\bm u, \tau_1, s)+\sup_{\|\bm u\|\leq \eta}|Q_t(\bm u, \tau_2, s)-Q_t(\bm u, \tau_1, s)|\right) \\
			&-F_{t-1}\left(Q_t(\bm u, \tau_1, s)-\sup_{\|\bm u\|\leq \eta}|Q_t(\bm u, \tau_2, s)-Q_t(\bm u, \tau_1, s)|\right) \\
			\leq & 2\sup_{x}f_{t-1}(x)\sup_{\|\bm u\|\leq \eta}|Q_t(\bm u, \tau_2, s)-Q_t(\bm u, \tau_1, s)| 
			\leq C|\tau_2-\tau_1|\Delta_{\rho,t}.
		\end{align*}
		Therefore, by the Cauchy-Schwarz inequality, it holds that
		\begin{align}\label{dt-2nd-moment}
			& E\left(\sup_{\|\bm u\|\leq \eta}d_t^2(\bm u)|\mathcal{F}_{t-1}\right) \nonumber\\
			= & E\left(\int_{0}^1\int_{0}^1\sup_{\|\bm u\|\leq \eta}\delta_t(\bm u, s_1)\sup_{\|\bm u\|\leq \eta}\delta_t(\bm u, s_2)ds_1ds_2  |\mathcal{F}_{t-1}\right) \nonumber\\
			\leq & \int_{0}^1\int_{0}^1 \{E[\sup_{\|\bm u\|\leq \eta}\delta_t^2(\bm u, s_1)|\mathcal{F}_{t-1}]\}^{1/2}\{E[\sup_{\|\bm u\|\leq \eta}\delta_t^2(\bm u, s_2)|\mathcal{F}_{t-1}]\}^{1/2}ds_1ds_2 \nonumber\\
			\leq & C|\tau_2-\tau_1|\Delta_{\rho,t}.
		\end{align}  
		Similar to the proof for $\var(\sup_{\|\bm u\|\leq \eta}|R_{41}(\bm u,\tau_1,\tau_2)|/(\sqrt{n}\|\bm u\|))$, by iterative-expectation and the Cauchy-Schwarz inequality, together with \eqref{ct-2nd-moment}, Lemma \ref{lem00}(i) and $E(w_t\Delta_{\rho,t}^2)<\infty$ under Assumptions \ref{assum-Space} and \ref{assum-RandomWeight}, it can be verified that
		\[\var\left(\sup_{\|\bm u\|\leq \eta}\dfrac{|R_{42}(\bm u,\tau_1,\tau_2)|}{\sqrt{n}\|\bm u\|}\right) \nonumber\\
		\leq  E\left[w_t^2\|\dot{q}_t(\bm\theta(\tau_2))\|^2 E\{[c_t-E(c_t|\mathcal{F}_{t-1})]^2|\mathcal{F}_{t-1}\}\right] 
		\leq  C|\tau_2-\tau_1|,\]
		and 
		\begin{align*}
			\var\left(\sup_{\|\bm u\|\leq \eta}\dfrac{|R_{43}(\bm u,\tau_1,\tau_2)|}{\sqrt{n}\|\bm u\|}\right)  
			\leq & E\left[w_t^2\|\dot{q}_t(\bm\theta(\tau_2))\|^2 E\{\sup_{\|\bm u\|\leq \eta}[d_t(\bm u)-E(d_t(\bm u)|\mathcal{F}_{t-1})]^2|\mathcal{F}_{t-1}\}\right] \\ 
			\leq & C|\tau_2-\tau_1|E\left[w_t^2\|\dot{q}_t(\bm\theta(\tau_2))\|^2\Delta_{\rho,t}\right] 
			\leq  C|\tau_2-\tau_1|. 
		\end{align*}
		Therefore, it holds that
		\begin{align}\label{R42}
			\sup_{\|\bm u\|\leq \eta}\dfrac{|R_{42}(\bm u,\tau_1,\tau_2)|}{\sqrt{n}\|\bm u\|}= 
			O_p(1)|\tau_2-\tau_1|^{1/2},
		\end{align} 
		and
		\begin{align}\label{R43}
			\sup_{\|\bm u\|\leq \eta}\dfrac{|R_{43}(\bm u,\tau_1,\tau_2)|}{\sqrt{n}\|\bm u\|}= 
			O_p(1)|\tau_2-\tau_1|^{1/2}.
		\end{align}
		Combining \eqref{R4-Tightness-decompose}--\eqref{R43}, the stochastic equicontinuity of $\sup_{\|\bm u\|\leq \eta}|R_{4n}(\bm u,\tau)|/(\sqrt{n}\|\bm u\|+n\|\bm u\|^2)$ follows.

		Next, we consider the stochastic equicontinuity of $\sup_{\|\bm u\|\leq \eta}|R_{5n}(\bm u,\tau)|/(n\|\bm u\|^2)$. It can be verified that
		\begin{align}\label{R5-Tightness-decompose}
			\left|\sup_{\|\bm u\|\leq \eta}\dfrac{|R_{5n}(\bm u,\tau_2)|}{n\|\bm u\|^2} - \sup_{\|\bm u\|\leq \eta}\dfrac{|R_{5n}(\bm u,\tau_1)|}{n\|\bm u\|^2}\right| \leq \sum_{i=1}^3 \sup_{\|\bm u\|\leq \eta}\dfrac{|R_{5i}(\bm u,\tau_1,\tau_2)|}{n\|\bm u\|^2},
		\end{align}
		where 
		\begin{align*}
			R_{51}(\bm u,\tau_1,\tau_2) &= \dfrac{\bm u^{\prime}}{2}\sum_{t=1}^{n}w_t[\ddot{q}_t(\bm u^*+\bm\theta(\tau_2))-\ddot{q}_t(\bm u^*+\bm\theta(\tau_1))] \xi_{t}(\bm u, \tau_2)\bm u, \\
			R_{52}(\bm u,\tau_1,\tau_2) &= \dfrac{\bm u^{\prime}}{2}\sum_{t=1}^{n}w_t\ddot{q}_t(\bm u^*+\bm\theta(\tau_1))c_t\bm u \quad\text{and}\\
			R_{53}(\bm u,\tau_1,\tau_2) &= \dfrac{\bm u^{\prime}}{2}\sum_{t=1}^{n}w_t\ddot{q}_t(\bm u^*+\bm\theta(\tau_1))d_t(\bm u)\bm u.
		\end{align*}
		By Lemma \ref{lem-Tightness1}(iii), the fact that $|\xi_{t}(\bm u, \tau)|<1$, the strict stationarity and ergodicity of $y_t$ under Assumption \ref{assum-Process} and $E(w_t\Delta_{\rho,t})<\infty$ under Assumptions \ref{assum-Space} and \ref{assum-RandomWeight}, it holds that
		\begin{align}\label{R51}
			\sup_{\|\bm u\|\leq \eta}\dfrac{|R_{51}(\bm u,\tau_1,\tau_2)|}{n\|\bm u\|^2} 
			&\leq \dfrac{1}{2n}\sum_{t=1}^{n}w_t\sup_{\Theta}\|\ddot{q}_t(\bm u^*+\bm\theta(\tau_2))-\ddot{q}_t(\bm u^*+\bm\theta(\tau_1))\| \nonumber \\
			&\leq  C|\tau_2-\tau_1|\dfrac{1}{n}\sum_{t=1}^{n}w_t\Delta_{\rho,t} = O_p(1)|\tau_2-\tau_1|.
		\end{align} 
		For $R_{52}(\bm u,\tau_1,\tau_2)$ and $R_{53}(\bm u,\tau_1,\tau_2)$, by iterative-expectation and the Cauchy-Schwarz inequality, the strict stationarity and ergodicity of $y_t$ under Assumption \ref{assum-Process}, Lemma \ref{lem00}(ii) and $E(w_t\Delta_{\rho,t})<\infty$ under Assumptions \ref{assum-Space} and \ref{assum-RandomWeight}, together with $E(c_t^2|\mathcal{F}_{t-1})=|\tau_2-\tau_1|$ by \eqref{ct-2nd-moment} and \eqref{dt-2nd-moment}, we have
		\begin{align*}
			\var\left(\sup_{\|\bm u\|\leq \eta}\dfrac{|R_{52}(\bm u,\tau_1,\tau_2)|}{n\|\bm u\|^2} \right) \leq & E\left(w_t^2\sup_{\Theta}\|\ddot{q}_t(\bm\theta)\|^2 E(c_t^2|\mathcal{F}_{t-1})\right) \leq C|\tau_2-\tau_1| 
		\end{align*}
		and 
		\begin{align*}
			\var\left(\sup_{\|\bm u\|\leq \eta}\dfrac{|R_{53}(\bm u,\tau_1,\tau_2)|}{n\|\bm u\|^2}\right) \leq & E\left(w_t^2\sup_{\Theta}\|\ddot{q}_t(\bm\theta)\|^2 E(d_t^2(\bm u)|\mathcal{F}_{t-1})\right) \leq C|\tau_2-\tau_1|. 
		\end{align*}
		Then we have 
		\begin{align}\label{R52_53}
			\dfrac{|R_{52}(\bm u,\tau_1,\tau_2)|}{n\|\bm u\|^2} =O_p(1)|\tau_2-\tau_1|^{1/2} \;\text{and}\;\dfrac{|R_{53}(\bm u,\tau_1,\tau_2)|}{n\|\bm u\|^2} =O_p(1)|\tau_2-\tau_1|^{1/2}.
		\end{align}
		Combining \eqref{R5-Tightness-decompose}--\eqref{R52_53}, the stochastic equicontinuity of $\sup_{\|\bm u\|\leq \eta}|R_{5n}(\bm u,\tau)|/(n\|\bm u\|^2)$ follows.
		We complete the proof of this lemma. 
	\end{proof}

	\subsection{Lemmas for Theorem \ref{thm-WCQR}}
	
	This section provides four preliminary lemmas with proofs. Specifically, Lemma \ref{lem0-Tukey} is used to handle initial values. Lemmas \ref{lem1-Tukey} verifies the stochastic differentiability condition defined by \cite{Pollard1985}, and the bracketing method in \cite{Pollard1985} is used for its proof. 
	Lemmas \ref{lem2-Tukey} and \ref{lem3-Tukey} are used to obtain the $\sqrt{n}$-consistency and asymptotic normality of $\check{\bm\varphi}_{wn}$, and their proofs need Lemmas \ref{lem1-Tukey} and \ref{lem0-Tukey}, respectively.

	\begin{lemma}\label{lem0-Tukey}
		Let $\varsigma_{\rho}=\sum_{s=0}^{\infty}\rho^{s}|y_{-s}|$ and $\xi_{\rho}=\sum_{s=0}^{\infty}s\rho^{s}|y_{-s}|$ be positive random variables depending on a constant $\rho\in (0,1)$. If Assumption \ref{assum-SpaceTukey}(i) holds, for $\tau\in\mathcal{T}_h=[\tau_0,\tau_0+h]\subset (0,0.5)$ or $\tau\in\mathcal{T}_h=[\tau_0-h,\tau_0]\subset (0.5,1)$ with $h>0$, then we have
		\begin{itemize}
			\item[(i)] $\sup_{\Phi}|q_{t,\tau}(\bm\varphi)-\widetilde{q}_{t,\tau}(\bm\varphi)|
			\leq C\rho^{t}\varsigma_{\rho}$;
			\item[(ii)]
			$\sup_{\Phi}\|\dot{q}_{t,\tau}(\bm\varphi)-\dot{\widetilde{q}}_{t,\tau}(\bm\varphi)\|
			\leq C\rho^{t}(\varsigma_{\rho}+t\varsigma_{\rho}+\xi_{\rho})$. 
		\end{itemize}
	\end{lemma}	
	
	\begin{lemma}\label{lem1-Tukey}
		Under Assumptions \ref{assum-ConditionalDensity}, \ref{assum-RandomWeight}, \ref{assum-Process-Mixing} and \ref{assum-SpaceTukey}, if $E|y_t|^s<\infty$ for some $0<s\leq 1$, then for any sequence of random variables $\bm u_n$ such that $\bm u_n=o_p(1)$, it holds that	
		\begin{align*}
			\zeta_n^*(\bm u_n)=o_p(\sqrt{n}\|\bm u_n\|+n\|\bm u_n\|^2),
		\end{align*}
		where $\zeta_n^*(\bm u)=\bm u^{\prime}\sum_{k=1}^{K}\sum_{t=1}^{n}w_t\dot{q}_{t,\tau_k}(\bm\varphi_0^*)\left\{\xi_{t,\tau_k}(\bm u)-E[\xi_{t,\tau_k}(\bm u)|\mathcal{F}_{t-1}]\right\}$ with 
		\begin{align*}
			\xi_{t,\tau_k}(\bm u)&=\int_{0}^{1}\left[I(y_t\leq q_{t,\tau_k}(\bm\varphi_0^*)+\nu_{t,\tau}(\bm u)s)-I(y_t\leq q_{t,\tau_k}(\bm\varphi_0^*))\right]ds
		\end{align*}
		and $\nu_{t,\tau}(\bm u)=q_{t,\tau}(\bm\varphi_0^*+\bm u)-q_{t,\tau}(\bm\varphi_0^*)$.
	\end{lemma}
	
	\begin{lemma}\label{lem2-Tukey}
		If $E|y_t|^s<\infty$ for some $0<s\leq 1$ and Assumptions \ref{assum-ConditionalDensity}, \ref{assum-RandomWeight}, \ref{assum-Process-Mixing} and \ref{assum-SpaceTukey} hold, then for any sequence of random variables $\bm u_n$ such that $\bm u_n=o_p(1)$, we have
		\begin{align*}
			n[L_n^*(\bm u_n+\bm\varphi_0^*)-L_n^*(\bm\varphi_0^*)]=&-\sqrt{n}\bm u_n^{\prime}\bm T_n^*+\sqrt{n}\bm u_n^{\prime}J^*\sqrt{n}\bm u_n +o_p(\sqrt{n}\|\bm u_n\|+n\|\bm u_n\|^2),
		\end{align*}
		where $L_n^*(\bm\varphi)=n^{-1}\sum_{k=1}^{K}\sum_{t=1}^{n}w_t\rho_{\tau_k}(y_t-q_{t,\tau_k}(\bm\varphi))$, $\bm T_n^*=n^{-1/2}\sum_{k=1}^{K}\sum_{t=1}^{n}w_t\dot{q}_{t,\tau_k}(\bm\varphi_0^*)\psi_{\tau_k}(e_{t,\tau_k}^*)$ and $J^*=J_{2}^*-J_{1}^*$ with $e_{t,\tau_k}^*=y_t-q_{t,\tau_k}(\bm\varphi_0^*)$, $J_{1}^*=\Omega_{11}^*/2=\sum_{k=1}^{K}E[w_t\ddot{q}_{t,\tau_k}(\bm\varphi_0^*)\psi_{\tau_k}(e_{t,\tau_k}^*)]/2$ and $J_{2}^*=\Omega_{12}^*/2=\sum_{k=1}^{K}E[w_t\dot{q}_{t,\tau_k}(\bm\varphi_0^*)\dot{q}_{t,\tau_k}^{\prime}(\bm\varphi_0^*)f_{t-1}(q_{t,\tau_k}(\bm\varphi_0^*))]/2$.
	\end{lemma}
	
	\begin{lemma}\label{lem3-Tukey}
		If $E|y_t|^s<\infty$ for some $0<s\leq 1$ and Assumptions \ref{assum-ConditionalDensity}, \ref{assum-RandomWeight}, \ref{assum-Process-Mixing} and \ref{assum-SpaceTukey} hold, then for any sequence of random variables $\bm u_n$ such that $\bm u_n=o_p(1)$, we have	
		\begin{align*}
			n[\widetilde{L}_n^*(\bm u_n+\bm\varphi_0^*)-\widetilde{L}_n^*(\bm\varphi_0^*)]-n[L_n^*(\bm u_n+\bm\varphi_0^*)-L_n^*(\bm\varphi_0^*)]=o_p(\sqrt{n}\|\bm u_n\|+n\|\bm u_n\|^2),
		\end{align*}
		where $\widetilde{L}_n^*(\bm\varphi)=n^{-1}\sum_{k=1}^{K}\sum_{t=1}^{n}w_t\rho_{\tau_k}(y_t-\widetilde{q}_{t,\tau_k}(\bm\varphi))$ and $L_n^*(\bm\varphi)=n^{-1}\sum_{k=1}^{K}\sum_{t=1}^{n}w_t\rho_{\tau_k}(y_t-q_{t,\tau_k}(\bm\varphi))$.
	\end{lemma}
	
	\begin{proof}[Proof of Lemma \ref{lem0-Tukey}]
		For $\bm\varphi=(\bm\phi^{\prime}, \lambda)^{\prime}=(a_0, a_1, b_1, \lambda)^{\prime}$, recall that
		\begin{align*}
			q_{t,\tau}(\bm\varphi) &= Q_{\tau}(\lambda) \left(\frac{a_0}{1-b_1}+a_1\sum_{j=1}^{\infty} b_1^{j-1}|y_{t-j}|\right):= Q_{\tau}(\lambda)h_t(\bm\phi), \\
			\widetilde{q}_{t,\tau}(\bm\varphi) &= Q_{\tau}(\lambda) \left(\frac{a_0}{1-b_1}+a_1\sum_{j=1}^{t-1} b_1^{j-1}|y_{t-j}|\right):= Q_{\tau}(\lambda)\widetilde{h}_t(\bm\phi),
		\end{align*} 
		$\dot{q}_{t,\tau}(\bm\varphi)=(Q_{\tau}(\lambda)\dot{h}_t^{\prime}(\bm\phi),\dot{Q}_{\tau}(\lambda)h_t(\bm\phi))^{\prime}$ and $\dot{\widetilde{q}}_{t,\tau}(\bm\varphi)=(Q_{\tau}(\lambda)\dot{\widetilde{h}}_t^{\prime}(\bm\phi),\dot{Q}_{\tau}(\lambda)\widetilde{h}_t(\bm\phi))^{\prime}$, where $\dot{Q}_{\tau}(\lambda)=\lambda^{-2}\{\tau^{\lambda}(\lambda\ln\tau-1)-(1-\tau)^{\lambda}[\lambda\ln(1-\tau)-1]\}$,  
		\begin{align*}
			\dot{h}_t(\bm\phi)&=\left(\dfrac{1}{1-b_1},\sum_{j=1}^{\infty}b_1^{j-1}|y_{t-j}|,\dfrac{a_0}{(1-b_1)^2}+a_1\sum_{j=2}^{\infty}(j-1)b_1^{j-2}|y_{t-j}|\right)^{\prime}  \;\;\text{and}\\
			\dot{\widetilde{h}}_t(\bm\phi)&=\left(\dfrac{1}{1-b_1},\sum_{j=1}^{t-1}b_1^{j-1}|y_{t-j}|,\dfrac{a_0}{(1-b_1)^2}+a_1\sum_{j=2}^{t-1}(j-1)b_1^{j-2}|y_{t-j}|\right)^{\prime}.
		\end{align*}
		It follows that $q_{t,\tau}(\bm\varphi)-\widetilde{q}_{t,\tau}(\bm\varphi)=Q_{\tau}(\lambda)a_{1}\sum_{j=t}^{\infty} b_{1}^{j-1}|y_{t-j}|$ and
		\[\dot{q}_{t,\tau}(\bm\varphi)-\dot{\widetilde{q}}_{t,\tau}(\bm\varphi)=\left(0,Q_{\tau}(\lambda)\sum_{j=t}^{\infty}b^{j-1}_{1}|y_{t-j}|,Q_{\tau}(\lambda)a_{1}\sum_{j=t}^{\infty}(j-1)b^{j-2}_{1}|y_{t-j}|,\dot{Q}_{\tau}(\lambda)a_1\sum_{j=t}^{\infty}b^{j-1}_{1}|y_{t-j}|\right)^{\prime}.\] 
		Since $\lambda\geq \underline{c}>0$, $a_{1}\leq\overline{c}<\infty$ and $0<b_{1}\leq\rho<1$ by Assumption \ref{assum-SpaceTukey}, for $\tau \in \mathcal{T}_h$ such that $Q_{\tau}(\lambda)$ and $\dot{Q}_{\tau}(\lambda)$ are bounded, it holds that 
		\begin{align*}
			\sup_{\Phi}|q_{t,\tau}(\bm\varphi)-\widetilde{q}_{t,\tau}(\bm\varphi)|
			\leq & |Q_{\tau}(\lambda)|a_{1}\sum_{j=t}^{\infty} b_{1}^{j-1}|y_{t-j}|
			\leq C\overline{c}\rho^{t-1}\sum_{s=0}^{\infty}\rho^{s}|y_{-s}|\leq C\rho^{t}\varsigma_{\rho} \;\;\text{and}\\
			\sup_{\Phi}\|\dot{q}_{t,\tau}(\bm\varphi)-\dot{\widetilde{q}}_{t,\tau}(\bm\varphi)\| \leq& |Q_{\tau}(\lambda)|\sup_{\Phi}\left[\sum_{j=t}^{\infty}b^{j-1}_{1}|y_{t-j}|+a_{1}\sum_{j=t}^{\infty}(j-1)b^{j-2}_{1}|y_{t-j}|\right] \\
			&+ |\dot{Q}_{\tau}(\lambda)|\sup_{\Phi}a_1 \sum_{j=t}^{\infty}b^{j-1}_{1}|y_{t-j}| \\
			\leq& C\left[\rho^{t-1}\varsigma_{\rho}+\overline{c}t\rho^{t-2}\varsigma_{\rho}+\overline{c}\rho^{t-1}\sum_{s=0}^{\infty}(s-1)\rho^{s-1}|y_{-s}|\right]+C\overline{c}\rho^{t-1}\varsigma_{\rho}\\
			\leq& C\rho^{t}(\varsigma_{\rho}+t\varsigma_{\rho}+\xi_{\rho}),
		\end{align*}
		where $\varsigma_{\rho}=\sum_{s=0}^{\infty}\rho^{s}|y_{-s}|$ and $\xi_{\rho}=\sum_{s=0}^{\infty}s\rho^{s}|y_{-s}|$.	
		The proof of this lemma is complete.	
	\end{proof}

	\begin{proof}[Proof of Lemma \ref{lem1-Tukey}]
		Recall that $\bm\varphi=(\bm\phi^{\prime}, \lambda)^{\prime}=(a_0, a_1, b_1, \lambda)^{\prime}$ and its true parameter vector $\bm\varphi_0^*=(\bm\phi_0^{\prime}, \lambda_0)^{\prime}=(a_{00}, a_{10}, b_{10}, \lambda_0)^{\prime}$. For $\bm u\in\mathbb{R}^{d}$ with $d=4$, note that
		\begin{align*}
			|\zeta_n^*(\bm u)| \leq \sqrt{n}\|\bm u\|\sum_{k=1}^{K}\sum_{j=1}^{d}\left|\dfrac{1}{\sqrt{n}}\sum_{t=1}^{n}m_{t,\tau_k,j}\left\{\xi_{t,\tau_k}(\bm u)-E[\xi_{t,\tau_k}(\bm u)|\mathcal{F}_{t-1}]\right\}\right|,
		\end{align*}
		where $m_{t,\tau_k,j}=w_t\partial q_{t,\tau_k}(\bm\varphi_0^*)/\partial \theta_{j}$ with $\theta_{j}$ being the $j$th element of $\bm\varphi$. For $1\leq j\leq d$ and $\tau \in \mathcal{T}_1 \subset [0,0.5)$ or $\tau \in \mathcal{T}_2 \subset (0.5,1]$, define $g_{t,\tau}=\max_{j}\{m_{t,\tau,j},0\}$ or $g_{t,\tau}=\max_j\{-m_{t,\tau,j},0\}$. Let $\varrho_{t,\tau}(\bm u)=g_{t,\tau}\xi_{t,\tau}(\bm u)$ and define
		\begin{align*}
			D_{n,\tau}(\bm u)=\dfrac{1}{\sqrt{n}}\sum_{t=1}^{n}\left\{\varrho_{t,\tau}(\bm u)-E\left[\varrho_{t,\tau}(\bm u)|\mathcal{F}_{t-1}\right]\right\}.
		\end{align*}
		To establish Lemma \ref{lem1-Tukey}, it suffices to show that, for any $\delta>0$,
		\begin{align}\label{diffcond-Tukey}
			\sup_{\|\bm u\|\leq \delta}\dfrac{|D_{n,\tau}(\bm u)|}{1+\sqrt{n}\|\bm u\|}=o_p(1).
		\end{align}
		
		We follow the method in Lemma 4 of \cite{Pollard1985} to verify \eqref{diffcond-Tukey}.
		Let $\mathfrak{F}_{\tau}=\{\varrho_{t,\tau}(\bm u): \|\bm u\|\leq \delta\}$ be a collection of functions indexed by $\bm u$. First, we verify that $\mathfrak{F}_{\tau}$ satisfies the bracketing condition defined on page 304 of \cite{Pollard1985}. 	
		Let $B_{r}(\bm v)$ be an open neighborhood of $\bm v$ with radius $r>0$, and define a constant $C_0$ to be selected later. For any $\epsilon>0$ and $0< r\leq \delta$, there exists a sequence of small cubes $\{B_{\epsilon r/C_0}(\bm u_{i})\}_{i=1}^{K(\epsilon)}$ to cover $B_r(\bm 0)$, where $K(\epsilon)$ is an integer less than $C\epsilon^{-d}$, and the constant $C$ is not depending on $\epsilon$ and $r$; see \cite{Huber1967}, page 227.
		Denote $V_i(r)=B_{\epsilon r/C_0}(\bm u_{i})\bigcap B_r(\bm0)$, and let $U_1(r)=V_1(r)$ and $U_i(r)=V_i(r)-\bigcup_{j=1}^{i-1}V_j(r)$ for $i\geq 2$. Note that $\{U_i(r)\}_{i=1}^{K(\epsilon)}$ is a partition of $B_r(\bm0)$.
		For each $\bm u_i\in U_i(r)$ with $1\leq i \leq K(\epsilon)$, define the following bracketing functions	
		\begin{align*}
			\varrho_{t,\tau}^L(\bm u_i)&= g_{t,\tau}\int_{0}^{1}\left[I\left(y_t\leq q_{t,\tau}(\bm\varphi_0^*)+\nu_{t,\tau}(\bm u_i)s-\dfrac{\epsilon r}{C_0}\|\dot{q}_{t,\tau}(\bm\varphi_0^*)\|\right)-I(y_t\leq q_{t,\tau}(\bm\varphi_0^*))\right]ds, \\
			\varrho_{t,\tau}^U(\bm u_i)&= g_{t,\tau}\int_{0}^{1}\left[I\left(y_t\leq q_{t,\tau}(\bm\varphi_0^*)+\nu_{t,\tau}(\bm u_i)s+\dfrac{\epsilon r}{C_0}\|\dot{q}_{t,\tau}(\bm\varphi_0^*)\|\right)-I(y_t\leq q_{t,\tau}(\bm\varphi_0^*))\right]ds.
		\end{align*}
		Since $I(\cdot)$ is non-decreasing and $g_{t,\tau}\geq 0$, for any $\bm u \in U_i(r)$, we have
		\begin{align}\label{brac1-Tukey}
			\varrho_{t,\tau}^L(\bm u_i)\leq \varrho_{t,\tau}(\bm u)\leq \varrho_{t,\tau}^U(\bm u_i).
		\end{align}
		Furthermore, by Taylor expansion, it holds that
		\begin{align}\label{uncond-Tukey}
			E\left[\varrho_{t,\tau}^U(\bm u_i)-\varrho_{t,\tau}^L(\bm u_i)|\mathcal{F}_{t-1}\right]\leq \dfrac{\epsilon r}{C_0}\cdot2\sup_{x}f_{t-1}(x) w_t\left\|\dot{q}_{t,\tau}(\bm\varphi_0^*)\right\|^2.
		\end{align}
		Denote $\aleph_{t,\tau}=2\sup_{x}f_{t-1}(x)w_t\left\|\dot{q}_{t,\tau}(\bm\varphi_0^*)\right\|^2$. By Assumption \ref{assum-ConditionalDensity}, we have $\sup_{x}f_{t-1}(x)<\infty$. Choose $C_0=E(\aleph_{t,\tau})$. Then by iterated-expectation and Assumption \ref{assum-RandomWeight}, it follows that
		\begin{align*}
			E\left[\varrho_{t,\tau}^U(\bm u_i)-\varrho_{t,\tau}^L(\bm u_i)\right]=E\left\{E\left[\varrho_{t,\tau}^U(\bm\varphi_i)-\varrho_{t,\tau}^L(\bm\varphi_i)|\mathcal{F}_{t-1}\right]\right\}\leq \epsilon r.
		\end{align*}	
		This together with \eqref{brac1-Tukey}, implies that the family $\mathfrak{F}_{\tau}$ satisfies the bracketing condition.	
		
		Put $r_k=2^{-k}\delta$. Let $B(k)=B_{r_k}(\bm0)$ and $A(k)$ be the annulus $B(k)\setminus B(k+1)$. From the bracketing condition, for fixed $\epsilon>0$, there is a partition $U_1(r_k), U_2(r_k), \ldots, U_{K(\epsilon)}(r_k)$ of $B(k)$. First, consider the upper tail case. For $\bm u \in U_i(r_k)$, by \eqref{uncond-Tukey}, it holds that
		\begin{align}\label{upper-Tukey}
			D_{n,\tau}(\bm u) \leq &\dfrac{1}{\sqrt{n}}\sum_{t=1}^{n}\left\{\varrho_{t,\tau}^U(\bm u_i)-E\left[\varrho_{t,\tau}^U(\bm u_i)|\mathcal{F}_{t-1}\right]\right\}+\dfrac{1}{\sqrt{n}}\sum_{t=1}^{n}E\left[\varrho_{t,\tau}^U(\bm u_i)-\varrho_{t,\tau}^L(\bm u_i)|\mathcal{F}_{t-1}\right] \nonumber \\
			\leq & D_{n,\tau}^U(\bm u_i)+\sqrt{n}\epsilon r_k\dfrac{1}{nC_0}\sum_{t=1}^{n}\aleph_{t,\tau},
		\end{align}	
		where \[D_{n,\tau}^U(\bm u_i)=\dfrac{1}{\sqrt{n}}\sum_{t=1}^{n}\left\{\varrho_{t,\tau}^U(\bm u_i)-E\left[\varrho_{t,\tau}^U(\bm u_i)|\mathcal{F}_{t-1}\right]\right\}.\]
		Define the event
		\begin{align*}
			E_n=\left\{\omega: \dfrac{1}{nC_0}\sum_{t=1}^{n}\aleph_{t,\tau}(\omega) < 2 \right\}.
		\end{align*}
		
		For $\bm u \in A(k)$, $1+\sqrt{n}\|\bm u\|>\sqrt{n}r_{k+1}=\sqrt{n}r_{k}/2$. Then by \eqref{upper-Tukey} and the Chebyshev's inequality, we have
		\begin{align}\label{Ak0-Tukey}
			\text{Pr}\left(\sup_{\bm u \in A(k)}\dfrac{D_{n,\tau}(\bm u)}{1+\sqrt{n}\|\bm u\|}>6\epsilon, E_n\right)
			\leq & \text{Pr}\left(\max_{1 \leq i \leq K(\epsilon)}\sup_{\bm u \in U_i(r_k) \cap A(k)}D_{n,\tau}(\bm u)>3\sqrt{n}\epsilon r_k, E_n\right) \nonumber\\
			\leq & K(\epsilon)\max_{1 \leq i \leq K(\epsilon)}\text{Pr}\left(D_{n,\tau}^U(\bm u_i)>\sqrt{n}\epsilon r_k\right) \nonumber\\
			\leq & K(\epsilon)\max_{1 \leq i \leq K(\epsilon)}\dfrac{E\{[D_{n,\tau}^U(\bm u_i)]^2\}}{n\epsilon^2 r_k^2}.
		\end{align}
		Moreover, by iterated-expectation, Taylor expansion and the H\"{o}lder's inequality, together with $\|\bm u_i\|\leq r_k$ for $\bm u_i \in U_i(r_k)$, we have
		\begin{align*}	
			& E\left\{[\varrho_{t,\tau}^U(\bm u_i)]^2\right\}
			= E\left\{E\left\{[\varrho_{t,\tau}^U(\bm u_i)]^2|\mathcal{F}_{t-1}\right\}\right\} \nonumber \\
			\leq & 2E\left\{g_{t,\tau}^2\left| \int_{0}^{1}\left[F_{t-1}\left(q_{t,\tau}(\bm\varphi_0^*)+\nu_{t,\tau}(\bm u_i)s-\dfrac{\epsilon r}{C_0}\|\dot{q}_{t,\tau}(\bm\varphi_0^*)\|\right)-F_{t-1}\left(q_{t,\tau}(\bm\varphi_0^*)\right)\right]ds\right|\right\} \nonumber\\
			\leq & Cr_k\sup_{x}f_{t-1}(x)E\left[w_t^2\left\|\dot{q}_{t,\tau}(\bm\varphi_0^*)\right\|^3+w_t^2\left\|\dot{q}_{t,\tau}(\bm\varphi_0^*)\right\|^2\sup_{\bm\varphi^{\dagger}\in\Phi}\left\|\dot{q}_{t,\tau}(\bm\varphi^{\dagger})\right\|\right] \\
			\leq & Cr_k\sup_{x}f_{t-1}(x)\left\{E\left(w_t^2\left\|\dot{q}_{t,\tau}(\bm\varphi_0^*)\right\|^3\right)+\left[E\left(w_t^2\left\|\dot{q}_{t,\tau}(\bm\varphi_0^*)\right\|^3\right)\right]^{2/3}\left[E\left(w_t^2\sup_{\bm\varphi^{\dagger}\in\Phi}\left\|\dot{q}_{t,\tau}(\bm\varphi^{\dagger})\right\|^3\right)\right]^{1/3}\right\} \\ &	:= \Upsilon_{\tau}(r_k),
		\end{align*}
		where $\bm\varphi^{\dagger}$ is between $\bm\varphi_0^*$ and $\bm u_i+\bm\varphi_0^*$. 
		This, together with \eqref{1st-derivative-qt-Tukey}, $\sup_{x}f_{t-1}(x)<\infty$ by Assumption \ref{assum-ConditionalDensity}, $E(w_t)<\infty$ and $E(w_{t}|y_{t-j}|^3)<\infty$ for all $j\geq 1$ by Assumption \ref{assum-RandomWeight}, strictly stationarity and $\alpha$-mixing property of $\{y_t\}$ under Assumption \ref{assum-Process-Mixing}, $\max\{a_{00},a_{10},a_0^{\dagger},a_1^{\dagger}\}<\overline{c}<\infty$ and $b_{10},b_1^{\dagger}\leq\rho<1$ by Assumption \ref{assum-SpaceTukey}, and the fact that $\varrho_{t,\tau}^U(\bm u_i)-E[\varrho_{t,\tau}^U(\bm u_i)|\mathcal{F}_{t-1}]$ is a martingale difference sequence, implies that
		\begin{align}\label{ED-Tukey}
			E\{[D_{n,\tau}^U(\bm u_i)]^2\}&=\dfrac{1}{n}\sum_{t=1}^{n}E\{\{\varrho_{t,\tau}^U(\bm u_i)-E[\varrho_{t,\tau}^U(\bm u_i)|\mathcal{F}_{t-1}]\}^2\} \nonumber \\
			&\leq \dfrac{1}{n}\sum_{t=1}^{n}E\{[\varrho_{t,\tau}^U(\bm u_i)]^2\} \leq\Upsilon_{\tau}(r_k)<\infty. 
		\end{align}
		Combining \eqref{Ak0-Tukey} and \eqref{ED-Tukey}, we have
		\begin{align*}
			\text{Pr}\left(\sup_{\bm u \in A(k)}\dfrac{D_{n,\tau}(\bm u)}{1+\sqrt{n}\|\bm u\|}>6\epsilon, E_n\right)
			\leq \dfrac{K(\epsilon)\Upsilon_{\tau}(r_k)}{n\epsilon^2r_k^2}.
		\end{align*}
		Similar to the proof of the upper tail case, we can obtain the same bound for the lower tail case. Therefore,
		\begin{align}\label{Ak-Tukey}
			& \text{Pr}\left(\sup_{\bm u \in A(k)}\dfrac{|D_{n,\tau}(\bm u)|}{1+\sqrt{n}\|\bm u\|}>6\epsilon, E_n\right) \leq \dfrac{2K(\epsilon)\Upsilon_{\tau}(r_k)}{n\epsilon^2r_k^2}.
		\end{align}
		
		Note that $\Upsilon_{\tau}(r_k)\to 0$ as $k\to \infty$, we can choose $k_{\epsilon}$ such that $2K(\epsilon)\Upsilon_{\tau}(r_k)/(\epsilon^2\delta^2)<\epsilon$ for $k\geq k_{\epsilon}$. Let $k_n$ be the integer such that $n^{-1/2}\delta \leq r_{k_n} \leq 2n^{-1/2}\delta$, and split $B_{\delta}(\bm 0)$ into two events $B:=B(k_n+1)$ and $B^c:=B(0)-B(k_n+1)$. Note that $B^c=\bigcup_{k=0}^{k_n}A(k)$ and $\Upsilon_{\tau}(r_k)$ is bounded. Then by \eqref{Ak-Tukey}, it holds that
		\begin{align}\label{Bc-Tukey}
			\text{Pr}\left(\sup_{\bm u \in B^c}\dfrac{|D_{n,\tau}(\bm u)|}{1+\sqrt{n}\|\bm u\|}>6\epsilon\right)
			\leq & \sum_{k=0}^{k_n}\text{Pr}\left(\sup_{\bm u \in A(k)}\dfrac{|D_{n,\tau}(\bm u)|}{1+\sqrt{n}\|\bm u\|}>6\epsilon, E_n\right) + \text{Pr}(E_n^c)\nonumber \\
			\leq & \dfrac{1}{n}\sum_{k=0}^{k_{\epsilon}-1}\dfrac{CK(\epsilon)}{\epsilon^2\delta^2}2^{2k}+ \dfrac{\epsilon}{n}\sum_{k=k_{\epsilon}}^{k_n}2^{2k}+ \text{Pr}(E_n^c) \nonumber \\
			\leq & O\left(\dfrac{1}{n}\right) + 4\epsilon + \text{Pr}(E_n^c).
		\end{align}
		
		Furthermore, for $\bm u \in B$, we have $1+\sqrt{n}\|\bm u\|\geq 1$ and $r_{k_n+1}\leq n^{-1/2}\delta<n^{-1/2}$. Similar to the proof of \eqref{Ak0-Tukey} and \eqref{ED-Tukey}, we can show that
		\begin{align*}
			\text{Pr}\left(\sup_{\bm u \in B}\dfrac{D_{n,\tau}(\bm u)}{1+\sqrt{n}\|\bm u\|}>3\epsilon, E_n\right) \leq \text{Pr}\left(\max_{1 \leq i \leq K(\epsilon)}D_{n,\tau}^U(\bm u_i)>\epsilon, E_n\right) \leq \dfrac{K(\epsilon)\Upsilon_{\tau}(r_{k_n+1})}{\epsilon^2}.
		\end{align*}
		We can obtain the same bound for the lower tail. Therefore, we have
		\begin{align}\label{B-Tukey}
			\text{Pr}\left(\sup_{\bm u \in B}\dfrac{|D_{n,\tau}(\bm u)|}{1+\sqrt{n}\|\bm u\|}>3\epsilon\right)
			= &  \text{Pr}\left(\sup_{\bm u \in B}\dfrac{|D_{n,\tau}(\bm u)|}{1+\sqrt{n}\|\bm u\|}>3\epsilon, E_n\right)+ \text{Pr}(E_n^c) \nonumber \\
			\leq & \dfrac{2K(\epsilon)\Upsilon_{\tau}(r_{k_n+1})}{\epsilon^2} + \text{Pr}(E_n^c).
		\end{align}
		Note that $\Upsilon_{\tau}(r_{k_n+1})\to 0$ as $n\to \infty$. Moreover, by the ergodic theorem, $\text{Pr}(E_n)\rightarrow 1$ and thus $\text{Pr}(E_n^c)\rightarrow 0$ as $n\rightarrow \infty$. \eqref{B-Tukey} together with \eqref{Bc-Tukey} asserts \eqref{diffcond-Tukey}. The proof of this lemma is accomplished.	
	\end{proof}

	\begin{proof}[Proof of Lemma \ref{lem2-Tukey}]
		Denote $\bm u=\bm\varphi-\bm\varphi_0^*$, where $\bm\varphi=(\bm\phi^{\prime}, \lambda)^{\prime}=(a_0, a_1, b_1, \lambda)^{\prime}$ and $\bm\varphi_0^*=(\bm\phi_0^{\prime}, \lambda_0)^{\prime}=(a_{00}, a_{10}, b_{10}, \lambda_0)^{\prime}$. Recall that $L_n^*(\bm\varphi)=n^{-1}\sum_{k=1}^{K}\sum_{t=1}^{n}w_t\rho_{\tau_k}(y_t-q_{t,\tau_k}(\bm\varphi))$ and $e_{t,\tau}^*=y_t-q_{t,\tau}(\bm\varphi_0^*)$ with $q_{t,\tau}(\bm\varphi) = Q_{\tau}(\lambda) \left(\frac{a_0}{1-b_1}+a_1\sum_{j=1}^{\infty} b_1^{j-1}|y_{t-j}|\right):= Q_{\tau}(\lambda)h_t(\bm\phi)$. 
		Let $\xi_{t,\tau}(\bm u)=\int_{0}^{1}\left[I(e_{t,\tau}^*\leq \nu_{t,\tau}(\bm u)s)-I(e_{t,\tau}^*\leq 0)\right]ds$ with $\nu_{t,\tau}(\bm u)=q_{t,\tau}(\bm\varphi_0^*+\bm u)-q_{t,\tau}(\bm\varphi_0^*)$. 
		By the Knight identity \eqref{identity}, it holds that
		\begin{align}\label{Gnrep-Tukey}
			n[L_n^*(\bm\varphi_0^*+\bm u)-L_n^*(\bm\varphi_0^*)] =&\sum_{k=1}^{K}\sum_{t=1}^{n}w_t\left[\rho_{\tau_k}\left(e_{t,\tau_k}^*-\nu_{t,\tau_k}(\bm u)\right)-\rho_{\tau_k}\left(e_{t,\tau_k}^*\right)\right] \nonumber \\
			=& K_{1n}^*(\bm u)+K_{2n}^*(\bm u),
		\end{align}	
		where $\bm u\in\Lambda^*\equiv\{\bm u\in\mathbb{R}^4: \bm u+\bm\varphi_0^* \in \Phi\}$,
		\begin{align*}
			K_{1n}^*(\bm u)=-\sum_{k=1}^{K}\sum_{t=1}^{n}w_t\nu_{t,\tau_k}(\bm u)\psi_{\tau_k}(e_{t,\tau_k}^*) \hspace{2mm}\text{and}\hspace{2mm}  K_{2n}^*(\bm u)=\sum_{k=1}^{K}\sum_{t=1}^{n}w_t\nu_{t,\tau_k}(\bm u)\xi_{t,\tau_k}(\bm u).
		\end{align*}
		By Taylor expansion, we have $\nu_{t,\tau}(\bm u)=q_{1t,\tau}(\bm u)+q_{2t,\tau}(\bm u)$, where $q_{1t,\tau}(\bm u)=\bm u^{\prime}\dot{q}_{t,\tau}(\bm\varphi_0^*)$ and $q_{2t,\tau}(\bm u)=\bm u^{\prime}\ddot{q}_{t,\tau}(\bm\varphi^{\dagger})\bm u/2$ with $\bm\varphi^{\dagger}$ between $\bm\varphi_0^*+\bm u$ and $\bm\varphi_0^*$.
		Then it follows that
		\begin{align}\label{K1rep-Tukey}
			K_{1n}^*(\bm u)&=-\sum_{k=1}^{K}\sum_{t=1}^{n}w_tq_{1t,\tau_k}(\bm u)\psi_{\tau_k}(e_{t,\tau_k}^*)-\sum_{k=1}^{K}\sum_{t=1}^{n}w_tq_{2t,\tau_k}(\bm u)\psi_{\tau_k}(e_{t,\tau_k}^*) \nonumber \\
			&=-\sqrt{n}\bm u^{\prime}\bm T_n^*-\sqrt{n}\bm u^{\prime}R_{1n}^*(\bm\varphi^{\dagger})\sqrt{n}\bm u,
		\end{align}
		where
		\[\bm T_n^*=\dfrac{1}{\sqrt{n}}\sum_{k=1}^{K}\sum_{t=1}^{n}w_t\dot{q}_{t,\tau_k}(\bm\varphi_0^*)\psi_{\tau_k}(e_{t,\tau_k}^*) \hspace{2mm}\text{and}\hspace{2mm} R_{1n}^*(\bm\varphi^{\dagger})=\dfrac{1}{2n}\sum_{k=1}^{K}\sum_{t=1}^{n}w_t\ddot{q}_{t,\tau_k}(\bm\varphi^{\dagger})\psi_{\tau_k}(e_{t,\tau_k}^*).\]
		From $E(w_t)<\infty$ and $E(w_{t}|y_{t-j}|^3)<\infty$ for all $j\geq 1$ by Assumption \ref{assum-RandomWeight}, strictly stationarity and $\alpha$-mixing property of $\{y_t\}$ under Assumption \ref{assum-Process-Mixing} and $\max\{a_0^*, a_{1}^*\}\leq\overline{c}<\infty$ and $b_{1}^*\leq\rho<1$ by Assumption \ref{assum-SpaceTukey}, together with \eqref{2nd-derivative-qt-Tukey} and the fact that $|\psi_{\tau}(\cdot)|\leq 1$, we have
		\[E\left[\sup_{\bm\varphi^{\dagger}\in\Phi}\left\|w_t \ddot{q}_{t,\tau_k}(\bm\varphi^{\dagger})\psi_{\tau_k}(e_{t,\tau_k}^*)\right\|\right]\leq CE\left[\sup_{\bm\varphi^{\dagger}\in\Phi}\left\|w_t \ddot{q}_{t,\tau_k}(\bm\varphi^{\dagger})\right\|\right]<\infty.\]
		Moreover, since $\ddot{q}_{t,\tau}(\bm\varphi)$ is continuous with respect to $\bm\varphi\in\Phi$, then by ergodic theorem for strictly stationary and $\alpha$-mixing process under Assumption \ref{assum-Process-Mixing}, together with $\bm\varphi_n=\bm\varphi_0^*+\bm u_n=\bm\varphi_0^*+o_p(1)$ and $\bm\varphi^{\dagger}_n$ between $\bm\varphi_0^*+\bm u_n$ and $\bm\varphi_0^*$, we can show that
		\[R_{1n}^*(\bm\varphi^{\dagger}_n)=J_{1}^*+o_p(1),\]
		where $J_{1}^*=\sum_{k=1}^{K}E[w_t\ddot{q}_{t,\tau_k}(\bm\varphi_0^*)\psi_{\tau_k}(e_{t,\tau_k}^*)]/2$. 
		This together with \eqref{K1rep-Tukey}, implies that
		\begin{align}\label{Kn1-Tukey}
			K_{1n}^*(\bm u_n)=-\sqrt{n}\bm u_n^{\prime}\bm T_n^*-\sqrt{n}\bm u_n^{\prime}J_{1}^*\sqrt{n}\bm u_n+o_p(n\|\bm u_n\|^2).
		\end{align}
		Denote 
		$\xi_{t,\tau}(\bm u)=\xi_{1t,\tau}(\bm u)+\xi_{2t,\tau}(\bm u)$, where
		\begin{align*}
			\xi_{1t,\tau}(\bm u)&=\int_{0}^{1}\left[I(e_{t,\tau}^*\leq q_{1t,\tau}(\bm u)s)-I(e_{t,\tau}^*\leq 0)\right]ds \hspace{2mm}\text{and}\hspace{2mm} \\
			\xi_{2t,\tau}(\bm u)&=\int_{0}^{1}\left[I(e_{t,\tau}^*\leq \nu_{t,\tau}(\bm u)s)-I(e_{t,\tau}^*\leq q_{1t,\tau}(\bm u)s)\right]ds.
		\end{align*}
		For $K_{2n}^*(\bm u)$, by Taylor expansion, it holds that
		\begin{align}\label{K2rep-Tukey}
			K_{2n}^*(\bm u)=R_{2n}^*(\bm u)+R_{3n}^*(\bm u)+R_{4n}^*(\bm u)+R_{5n}^*(\bm u),
		\end{align}
		where
		\begin{align*}
			R_{2n}^*(\bm u)&=\bm u^{\prime}\sum_{k=1}^{K}\sum_{t=1}^{n}w_t\dot{q}_{t,\tau_k}(\bm\varphi_0^*)E[\xi_{1t,\tau_k}(\bm u)|\mathcal{F}_{t-1}], \\
			R_{3n}^*(\bm u)&=\bm u^{\prime}\sum_{k=1}^{K}\sum_{t=1}^{n}w_t\dot{q}_{t,\tau_k}(\bm\varphi_0^*)E[\xi_{2t,\tau_k}(\bm u)|\mathcal{F}_{t-1}], \\
			R_{4n}^*(\bm u)&=\bm u^{\prime}\sum_{k=1}^{K}\sum_{t=1}^{n}w_t\dot{q}_{t,\tau_k}(\bm\varphi_0^*)\{\xi_{t,\tau_k}(\bm u)-E[\xi_{t,\tau_k}(\bm u)|\mathcal{F}_{t-1}]\} \hspace{2mm}\text{and}\hspace{2mm} \\
			R_{5n}^*(\bm u)&=\dfrac{\bm u^{\prime}}{2}\sum_{k=1}^{K}\sum_{t=1}^{n}w_t\ddot{q}_{t,\tau_k}(\bm\varphi^{\dagger})\xi_{t,\tau_k}(\bm u)\bm u.
		\end{align*}
		Note that
		\begin{align}\label{xi1-Tukey}
			E[\xi_{1t,\tau}(\bm u)|\mathcal{F}_{t-1}]=\int_{0}^{1}[F_{t-1}(q_{t,\tau}(\bm\varphi_0^*)+q_{1t,\tau}(\bm u)s)-F_{t-1}(q_{t,\tau}(\bm\varphi_0^*))]ds.
		\end{align}
		Then by Taylor expansion, together with Assumption \ref{assum-ConditionalDensity}, it follows that
		\begin{align*}
			E[\xi_{1t,\tau}(\bm u)|\mathcal{F}_{t-1}]=&\dfrac{\bm u^{\prime}}{2}f_{t-1}(q_{t,\tau}(\bm\varphi_0^*))\dot{q}_{t,\tau}(\bm\varphi_0^*) \\
			&+q_{1t,\tau}(\bm u)\int_{0}^{1}[f_{t-1}(q_{t,\tau}(\bm\varphi_0^*)+q_{1t,\tau}(\bm u)s^*)-f_{t-1}(q_{t,\tau}(\bm\varphi_0^*))]sds,
		\end{align*}
		where $s^*$ is between 0 and $s$. Therefore, it holds that
		\begin{align}\label{R2nrep-Tukey}
			R_{2n}^*(\bm u)=\sqrt{n}\bm u^{\prime}J_{2n}^*\sqrt{n}\bm u+\sqrt{n}\bm u^{\prime}\Pi_{1n}^*(\bm u)\sqrt{n}\bm u,
		\end{align}
		where
		$J_{2n}^*=(2n)^{-1}\sum_{k=1}^{K}\sum_{t=1}^{n}w_tf_{t-1}(q_{t,\tau_k}(\bm\varphi_0^*))\dot{q}_{t,\tau_k}(\bm\varphi_0^*)\dot{q}_{t,\tau_k}^{\prime}(\bm\varphi_0^*)$ and
		\begin{equation*}
			\Pi_{1n}^*(\bm u)=\dfrac{1}{n}\sum_{k=1}^{K}\sum_{t=1}^{n}w_t\dot{q}_{t,\tau_k}(\bm\varphi_0^*)\dot{q}_{t,\tau_k}^{\prime}(\bm\varphi_0^*)\int_{0}^{1}[f_{t-1}(q_{t,\tau_k}(\bm\varphi_0^*)+q_{1t,\tau_k}(\bm u)s^*)-f_{t-1}(q_{t,\tau_k}(\bm\varphi_0^*))]sds.
		\end{equation*}
		By Taylor expansion, together with $\sup_{x}|\dot{f}_{t-1}(x)|<\infty$ by Assumption \ref{assum-ConditionalDensity} and $E(w_t)<\infty$ and $E(w_{t}|y_{t-j}|^3)<\infty$ for all $j\geq 1$ by Assumption \ref{assum-RandomWeight}, \eqref{1st-derivative-qt-Tukey}, strictly stationarity and $\alpha$-mixing property of $\{y_t\}$ under Assumption \ref{assum-Process-Mixing}, $\max\{a_{00},a_{10}\}\leq\overline{c}<\infty$ and $b_{10}\leq\rho<1$ by Assumption \ref{assum-SpaceTukey}, for any $\eta>0$, it holds that
		\begin{align*}
			E\left(\sup_{\|\bm u\|\leq \eta}\|\Pi_{1n}^*(\bm u)\|\right)&\leq \dfrac{1}{n}\sum_{k=1}^{K}\sum_{t=1}^{n}E\left[\sup_{\|\bm u\|\leq \eta} \|w_t\dot{q}_{t,\tau_k}(\bm\varphi_0^*)\dot{q}_{t,\tau_k}^{\prime}(\bm\varphi_0^*)\sup_{x}|\dot{f}_{t-1}(x)|\bm u^{\prime}\dot{q}_{t,\tau_k}(\bm\varphi_0^*)\|\right] \\
			&\leq C\eta\sup_{x}|\dot{f}_{t-1}(x)|\sum_{k=1}^{K} E[w_t\|\dot{q}_{t,\tau_k}(\bm\varphi_0^*)\|^3]
		\end{align*}
		tends to $0$ as $\eta \to 0$.
		Therefore, by Markov’s theorem, for any $\epsilon$, $\delta>0$, there exists $\eta_0=\eta_0(\epsilon)>0$ such that
		\begin{align}\label{epsdelta1-Tukey}
			\text{Pr}\left(\sup_{\|\bm u\|\leq \eta_0}\|\Pi_{1n}^*(\bm u)\|> \delta\right)<\dfrac{\epsilon}{2}
		\end{align}
		for all $n\geq 1$. Since $\bm u_n=o_p(1)$, it follows that
		\begin{align}\label{epsdelta2-Tukey}
			\text{Pr}\left(\|\bm u_n\|> \eta_0\right)<\dfrac{\epsilon}{2}
		\end{align}
		as $n$ is large enough. From \eqref{epsdelta1-Tukey} and \eqref{epsdelta2-Tukey}, we have
		\begin{align*}
			\text{Pr}\left(\|\Pi_{1n}^*(\bm u_n)\|> \delta\right)&\leq \text{Pr}\left(\|\Pi_{1n}^*(\bm u_n)\|> \delta, \|\bm u_n\|\leq \eta_0\right)+\text{Pr}\left(\|\bm u_n\|> \eta_0\right) \\
			&\leq \text{Pr}\left(\sup_{\|\bm u\|\leq \eta_0}\|\Pi_{1n}^*(\bm u)\|> \delta\right)+\dfrac{\epsilon}{2}<\epsilon
		\end{align*}
		as $n$ is large enough. Therefore, $\Pi_{1n}^*(\bm u_n)=o_p(1)$. This together with Assumption \ref{assum-Process-Mixing}, \eqref{R2nrep-Tukey} and $J_{2n}^*=J_2^*+o_p(1)$ by ergodic theorem for strictly stationary and $\alpha$-mixing process, implies that
		\begin{align}\label{R2n-Tukey}
			R_{2n}^*(\bm u_n)=\sqrt{n}\bm u_n^{\prime}J_2^*\sqrt{n}\bm u_n+o_p(n\|\bm u_n\|^2),
		\end{align}
		where $J_2^*=\sum_{k=1}^{K}E[w_tf_{t-1}(q_{t,\tau_k}(\bm\varphi_0^*))\dot{q}_{t,\tau_k}(\bm\varphi_0^*)\dot{q}_{t,\tau_k}^{\prime}(\bm\varphi_0^*)]/2$. 
		
		For $R_{3n}^*(\bm u)$, note that
		\begin{equation}\label{xi2-Tukey}
			E[\xi_{2t,\tau}(\bm u)|\mathcal{F}_{t-1}]=\int_{0}^{1}\left[F_{t-1}(q_{t,\tau}(\bm\varphi_0^*)+\nu_{t,\tau}(\bm u)s)-F_{t-1}(q_{t,\tau}(\bm\varphi_0^*)+q_{1t,\tau}(\bm u)s)\right]ds.
		\end{equation}
		Then by iterated-expectation, Taylor expansion and the Cauchy-Schwarz inequality, together with \eqref{1st-derivative-qt-Tukey}, \eqref{2nd-derivative-qt-Tukey}, $\sup_{x}f_{t-1}(x)<\infty$ by Assumption \ref{assum-ConditionalDensity}, $E(w_t)<\infty$ and $E(w_{t}|y_{t-j}|^3)<\infty$ for all $j\geq 1$ by Assumption \ref{assum-RandomWeight}, strictly stationarity and $\alpha$-mixing property of $\{y_t\}$ under Assumption \ref{assum-Process-Mixing}, $\max\{a_{00},a_{10},a_{0},a_{1}\}\leq\overline{c}<\infty$ and $b_{10},b_1\leq\rho<1$ by Assumption \ref{assum-SpaceTukey}, for any $\eta>0$, it holds that
		\begin{align*}
			E\left(\sup_{\|\bm u\|\leq \eta}\dfrac{|R_{3n}^*(\bm u)|}{n\|\bm u\|^2}\right)
			\leq & \dfrac{\eta}{n}\sum_{k=1}^{K}\sum_{t=1}^{n}E\left\{w_t\left\|\dot{q}_{t,\tau_k}(\bm\varphi_0^*)\right\|\dfrac{1}{2}\sup_{x}f_{t-1}(x)\sup_{\Phi}\left\|\ddot{q}_{t,\tau_k}(\bm\varphi^{\dagger})\right\| \right\} \\
			\leq &  C\eta\sum_{k=1}^{K} E\left\{\left\|\sqrt{w_t}\dot{q}_{t,\tau_k}(\bm\varphi_0^*)\right\|\sup_{\Phi}\left\|\sqrt{w_t}\ddot{q}_{t,\tau_k}(\bm\varphi^{\dagger})\right\| \right\} \\
			\leq & C\eta\sum_{k=1}^{K} \left[E\left(w_t\left\|\dot{q}_{t,\tau_k}(\bm\varphi_0^*)\right\|^2\right)\right]^{1/2}\left[E\left(w_t\sup_{\Phi}\left\|\ddot{q}_{t,\tau_k}(\bm\varphi^{\dagger})\right\|^2\right)\right]^{1/2}
		\end{align*}
		tends to $0$ as $\eta \to 0$. Similar to \eqref{epsdelta1-Tukey} and \eqref{epsdelta2-Tukey}, we can show that
		\begin{align}\label{R3n-Tukey}
			R_{3n}^*(\bm u_n)=o_p(n\|\bm u_n\|^2).
		\end{align}
		
		For $R_{4n}^*(\bm u)$, by Lemma \ref{lem1-Tukey}, it holds that
		\begin{align}\label{R4n-Tukey}
			R_{4n}^*(\bm u_n)=o_p(\sqrt{n}\|\bm u_n\|+n\|\bm u_n\|^2).
		\end{align}	
		
		Finally, we consider $R_{5n}^*(\bm u)$. Since $I(x\leq a)-I(x\leq b)=I(b\leq x \leq a)-I(b\geq x \geq a)$ and $\nu_{t,\tau}(\bm u)=\bm u^{\prime}\dot{q}_{t,\tau}(\bm\varphi^{\star})$ with $\bm\varphi^{\star}$ between $\bm\varphi_0^*$ and $\bm u+\bm\varphi_0^*$ by Taylor expansion, we have
		\begin{align*}
			\sup_{\|\bm u\|\leq \eta}|\xi_{t,\tau}(\bm u)| 
			\leq & \int_{0}^{1}\sup_{\|\bm u\|\leq \eta}\left|I(q_{t,\tau}(\bm\varphi_0^*)\leq y_{t}\leq q_{t,\tau}(\bm\varphi_0^*)+\nu_{t,\tau}(\bm u)s)\right|ds \\ 
			& + \int_{0}^{1}\sup_{\|\bm u\|\leq \eta}\left|I(q_{t,\tau}(\bm\varphi_0^*) \geq y_{t}\geq q_{t,\tau}(\bm\varphi_0^*)+\nu_{t,\tau}(\bm u)s)\right|ds \\
			\leq & I\left(q_{t,\tau}(\bm\varphi_0^*)\leq y_{t}\leq q_{t,\tau}(\bm\varphi_0^*)+ \eta\sup_{\bm\varphi^{\star}\in\Phi}\left\|\dot{q}_{t,\tau}(\bm\varphi^{\star})\right\|\right) \\ 
			& + I\left(q_{t,\tau}(\bm\varphi_0^*)\geq y_{t}\geq q_{t,\tau}(\bm\varphi_0^*)- \eta\sup_{\bm\varphi^{\star}\in\Phi}\left\|\dot{q}_{t,\tau}(\bm\varphi^{\star})\right\|\right).
		\end{align*}
		Then by iterated-expectation, the Cauchy-Schwarz inequality and the strict stationarity and ergodicity of $y_t$ under Assumption \ref{assum-Process-Mixing}, together with \eqref{1st-derivative-qt-Tukey}, \eqref{2nd-derivative-qt-Tukey}, $\max\{a_0^*, a_{1}^*,a_0^{\star}, a_{1}^{\star}\}\leq\overline{c}<\infty$ and $b_{1}^*,b_{1}^{\star}\leq\rho<1$ by Assumption \ref{assum-SpaceTukey}, $\sup_{x}f_{t-1}(x)<\infty$ by Assumption \ref{assum-ConditionalDensity} and $E(w_t)<\infty$ and $E(w_{t}|y_{t-j}|^3)<\infty$ for all $j\geq 1$ by Assumption \ref{assum-RandomWeight}, for any $\eta>0$, it follows that
		\begin{align*}
			E\left(\sup_{\|\bm u\|\leq \eta}\dfrac{|R_{5n}^*(\bm u)|}{n\|\bm u\|^2}\right)  
			\leq & \dfrac{1}{2n}\sum_{k=1}^{K}\sum_{t=1}^{n}E\left[w_t\sup_{\bm\varphi^{\dagger}\in\Phi}\left\|\ddot{q}_{t,\tau_k}(\bm\varphi^{\dagger})\right\|E\left(\sup_{\|\bm u\|\leq \eta}|\xi_{t,\tau_k}(\bm u)| | \mathcal{F}_{t-1}\right)\right] \\
			\leq & \eta\sup_{x}f_{t-1}(x) \sum_{k=1}^{K}E\left[w_t\sup_{\bm\varphi^{\dagger}\in\Phi}\left\|\ddot{q}_{t,\tau_k}(\bm\varphi^{\dagger})\right\|\sup_{\bm\varphi^{\star}\in\Phi}\left\|\dot{q}_{t,\tau_k}(\bm\varphi^{\star})\right\| \right] \\
			\leq & C\eta\sum_{k=1}^{K} \left[E\left(w_t\sup_{\bm\varphi^{\dagger}\in\Phi}\left\|\ddot{q}_{t,\tau_k}(\bm\varphi^{\dagger})\right\|^2\right)\right]^{1/2}\left[E\left(w_t\sup_{\bm\varphi^{\star}\in\Phi}\left\|\dot{q}_{t,\tau_k}(\bm\varphi^{\star})\right\|^2\right)\right]^{1/2} 
		\end{align*} 
		tends to $0$ as $\eta \to 0$. 
		Similar to \eqref{epsdelta1-Tukey} and \eqref{epsdelta2-Tukey}, we can show that 
		\begin{align}\label{R5n-Tukey}
			R_{5n}^*(\bm u_n)=o_p(n\|\bm u_n\|^2).
		\end{align}
		From \eqref{K2rep-Tukey}, \eqref{R2n-Tukey}, \eqref{R3n-Tukey}, \eqref{R4n-Tukey} and \eqref{R5n-Tukey}, we have
		\begin{align}\label{Kn2-Tukey}
			K_{2n}^*(\bm u_n)=\sqrt{n}\bm u_n^{\prime}J_2^*\sqrt{n}\bm u_n+o_p(\sqrt{n}\|\bm u_n\|+n\|\bm u_n\|^2).
		\end{align}
		In view of \eqref{Gnrep-Tukey}, \eqref{Kn1-Tukey} and \eqref{Kn2-Tukey}, we accomplish the proof of this lemma.
	\end{proof}

	\begin{proof}[Proof of Lemma \ref{lem3-Tukey}]
		Denote $\bm u=\bm\varphi-\bm\varphi_0^*$. Recall that $e_{t,\tau}^*=y_t-q_{t,\tau}(\bm\varphi_0^*)$, $\nu_{t,\tau}(\bm u)=q_{t,\tau}(\bm\varphi_0^*+\bm u)-q_{t,\tau}(\bm\varphi_0^*)$ and $\xi_{t,\tau}(\bm u)=\int_{0}^{1}[I(e_{t,\tau}^*\leq \nu_{t,\tau}(\bm u)s)-I(e_{t,\tau}^*\leq 0)]ds$. 
		Let $\widetilde{e}_{t,\tau}^*=y_t-\widetilde{q}_{t,\tau}(\bm\varphi_0^*)$, $\widetilde{\nu}_{t,\tau}(\bm u)=\widetilde{q}_{t,\tau}(\bm\varphi_0^*+\bm u)-\widetilde{q}_{t,\tau}(\bm\varphi_0^*)$
		and 
		$\widetilde{\xi}_{t,\tau}(\bm u)=\int_{0}^{1}[I(\widetilde{e}_{t,\tau}^*\leq \widetilde{\nu}_{t,\tau}(\bm u)s)-I(\widetilde{e}_{t,\tau}^*\leq 0)]ds$. 
		Similar to \eqref{Gnrep-Tukey}, by the Knight identity \eqref{identity} we can verify that
		\begin{align}\label{initialHnrep-Tukey}
			&n[\widetilde{L}_n^*(\bm\varphi_0^*+\bm u)-\widetilde{L}_n^*(\bm\varphi_0^*)]-n[L_n(\bm\varphi_0^*+\bm u)-L_n(\bm\varphi_0^*)] \nonumber\\
			=& \sum_{k=1}^{K}\sum_{t=1}^n w_t \{[-\widetilde{\nu}_{t,\tau_k}(\bm u)\psi_{\tau_k}(\widetilde{e}_{t,\tau_k}^*)+\widetilde{\nu}_{t,\tau_k}(\bm u)\widetilde{\xi}_{t,\tau_k}(\bm u)] \nonumber\\
			&\hspace{18mm} - [-\nu_{t,\tau_k}(\bm u)\psi_{\tau_k}(e_{t,\tau_k}^*)+\nu_{t,\tau_k}(\bm u)\xi_{t,\tau_k}(\bm u)]\} \nonumber\\
			=& \sum_{k=1}^{K}\left[\widetilde{A}_{1n,k}^*(\bm u)+\widetilde{A}_{2n,k}^*(\bm u)+\widetilde{A}_{3n,k}^*(\bm u)+\widetilde{A}_{4n,k}^*(\bm u)\right],
		\end{align}
		where $\bm u\in\Lambda^*\equiv\{\bm u\in\mathbb{R}^4: \bm u+\bm\varphi_0^* \in \Phi\}$,
		\begin{align*}
			\widetilde{A}_{1n,k}^*(\bm u)&= \sum_{t=1}^n w_t [\nu_{t,\tau_k}(\bm u)-\widetilde{\nu}_{t,\tau_k}(\bm u)]\psi_{\tau_k}(\widetilde{e}_{t,\tau_k}^*),\\ 
			\widetilde{A}_{2n,k}^*(\bm u)&=\sum_{t=1}^n w_t [\psi_{\tau_k}(e_{t,\tau_k}^*)-\psi_{\tau_k}(\widetilde{e}_{t,\tau_k}^*)]\nu_{t,\tau_k}(\bm u),\\
			\widetilde{A}_{3n,k}^*(\bm u)&= \sum_{t=1}^n w_t [\widetilde{\nu}_{t,\tau_k}(\bm u)-\nu_{t,\tau_k}(\bm u)]\widetilde{\xi}_{t,\tau_k}(\bm u) \;\;\text{and}\;\; \\ \widetilde{A}_{4n,k}^*(\bm u)&=\sum_{t=1}^n w_t [\widetilde{\xi}_{t,\tau_k}(\bm u)-\xi_{t,\tau_k}(\bm u)] \nu_{t,\tau_k}(\bm u).
		\end{align*}
		
		We first consider $\widetilde{A}_{1n,k}^*(\bm u)$. 
		Since $|\psi_{\tau}(\cdot)|\leq 1$, $\{y_t\}$ is strictly stationary and ergodic by Assumption \ref{assum-Process-Mixing} and $E(w_t)<\infty$ and $E(w_{t}|y_{t-j}|^3)<\infty$ for all $j\geq 1$ by Assumption \ref{assum-RandomWeight}, then by Taylor expansion and Lemma \ref{lem0-Tukey}(ii), we have
		\begin{align*}
			\sup_{\bm u\in\Lambda^*}\dfrac{|\widetilde{A}_{1n,k}^*(\bm u)|}{\sqrt{n}\|\bm u\|} 
			&\leq \dfrac{1}{\sqrt{n}}\sum_{t=1}^n w_t \sup_{\bm u\in\Lambda^*}\dfrac{|\nu_{t,\tau_k}(\bm u)-\widetilde{\nu}_{t,\tau_k}(\bm u)|}{\|\bm u\|}|\psi_{\tau_k}(\widetilde{e}_{t,\tau_k}^*)| \\
			&\leq \dfrac{1}{\sqrt{n}}\sum_{t=1}^n w_t \sup_{\Phi}\|\dot{q}_{t,\tau_k}(\bm\varphi^{\dagger})-\dot{\widetilde{q}}_{t,\tau_k}(\bm\varphi^{\dagger})\| \\
			&\leq \dfrac{C}{\sqrt{n}}\sum_{t=1}^n \rho^t w_t(\varsigma_{\rho}+\xi_{\rho})+ \dfrac{C}{\sqrt{n}}\sum_{t=1}^n t\rho^t w_t\varsigma_{\rho}=o_p(1),
		\end{align*}
		where $\bm\varphi^{\dagger}$ is between $\bm\varphi$ and $\bm\varphi_0^*$. Therefore, it holds that
		\begin{align}\label{Pi1tilde-Tukey}
			\widetilde{A}_{1n,k}^*(\bm u_n)=o_p(\sqrt{n}\|\bm u_n\|).
		\end{align}

		We next consider $\widetilde{A}_{2n,k}^*(\bm u)$. 
		Using $I(x < a)-I(x < b)=I(0 < x-b < a-b)-I(0> x-b > a-b)$ and $\psi_{\tau}(e_{t,\tau}^*)-\psi_{\tau}(\widetilde{e}_{t,\tau}^*)=I(y_t<\widetilde{q}_{t,\tau}(\bm\varphi_0^*))-I(y_t<q_{t,\tau}(\bm\varphi_0^*))$, we have 
		\begin{align*}
			E[|\psi_{\tau}(e_{t,\tau}^*)-\psi_{\tau}(\widetilde{e}_{t,\tau}^*)| |\mathcal{F}_{t-1}] 
			\leq & E\left[I(0< y_t-q_{t,\tau}(\bm\varphi_0^*)< |\widetilde{q}_{t,\tau}(\bm\varphi_0^*)-q_{t,\tau}(\bm\varphi_0^*)|) |\mathcal{F}_{t-1}\right] \\
			&+ E\left[I(0> y_t-q_{t,\tau}(\bm\varphi_0^*)> -|\widetilde{q}_{t,\tau}(\bm\varphi_0^*)-q_{t,\tau}(\bm\varphi_0^*)|)|\mathcal{F}_{t-1}\right] \\
			\leq & F_{t-1}\left(q_{t,\tau}(\bm\varphi_0^*)+|\widetilde{q}_{t,\tau}(\bm\varphi_0^*)-q_{t,\tau}(\bm\varphi_0^*)|\right) \\
			&-F_{t-1}\left(q_{t,\tau}(\bm\varphi_0^*)-|\widetilde{q}_{t,\tau}(\bm\varphi_0^*)-q_{t,\tau}(\bm\varphi_0^*)|\right).
		\end{align*}
		Then by iterative-expectation and Cauchy-Schwarz inequality, together with $\nu_{t,\tau}(\bm u)=\bm u^{\prime}\dot{q}_{t,\tau}(\bm\varphi^{\dagger})$ by Taylor expansion, Lemma \ref{lem0-Tukey}(i),  Assumption \ref{assum-Process-Mixing}, $\max\{a_{0}^{\dagger},a_{1}^{\dagger}\}\leq\overline{c}<\infty$ and $b_{1}^{\dagger}\leq\rho<1$ by Assumption \ref{assum-SpaceTukey}, $\sup_{x}f_{t-1}(x)<\infty$ by Assumption \ref{assum-ConditionalDensity} and $E(w_t)<\infty$ and $E(w_{t}|y_{t-j}|^3)<\infty$ for all $j\geq 1$ by Assumption \ref{assum-RandomWeight}, it holds that
		\begin{align*}
			E\sup_{\bm u\in\Lambda^*}\dfrac{|\widetilde{A}_{2n,k}^*(\bm u)|}{\sqrt{n}\|\bm u\|} 
			&\leq \dfrac{1}{\sqrt{n}}\sum_{t=1}^n E\left\{ w_t\sup_{\Phi}\|\dot{q}_{t,\tau_k}(\bm\varphi^{\dagger})\|\cdot E[|\psi_{\tau_k}(e_{t,\tau_k}^*)-\psi_{\tau_k}(\widetilde{e}_{t,\tau_k}^*)| |\mathcal{F}_{t-1}] \right\} \\
			&\leq 2C\sup_{x}f_{t-1}(x)\dfrac{1}{\sqrt{n}}\sum_{t=1}^n \rho^t E\left\{w_t\sup_{\Phi}\|\dot{q}_{t,\tau_k}(\bm\varphi^{\dagger})\|  \varsigma_{\rho} \right\} \\
			&\leq \dfrac{C}{\sqrt{n}}\sum_{t=1}^n \rho^t \cdot \left[E\left(w_t\sup_{\Phi}\|\dot{q}_{t,\tau_k}(\bm\varphi^{\dagger})\|^2\right)\right]^{1/2}\cdot \left[E(w_t\varsigma_{\rho}^2)\right]^{1/2}=o(1),
		\end{align*}
		where $\bm\varphi^{\dagger}$ is between $\bm\varphi_0^*+\bm u$ and $\bm\varphi_0^*$. 
		As a result, it follows that
		\begin{align}\label{Pi2tilde-Tukey}
			\widetilde{A}_{2n,k}^*(\bm u_n)=o_p(\sqrt{n}\|\bm u_n\|).
		\end{align}

		For $\widetilde{A}_{3n,k}^*(\bm u)$, since $|\widetilde{\xi}_{t,\tau}(\bm u)|<2$, similar to the proof of $\widetilde{A}_{1n,k}^*(\bm u)$, we can verify that
		\begin{align}\label{Pi3tilde-Tukey}
			\widetilde{A}_{3n,k}^*(\bm u_n)=o_p(\sqrt{n}\|\bm u_n\|).
		\end{align}

		Finally, we consider $\widetilde{A}_{4n,k}^*(\bm u)$. Denote $\widetilde{c}_{t,\tau}^*=I(y_t\leq\widetilde{q}_{t,\tau}(\bm\varphi_0^*))-I(y_t\leq q_{t,\tau}(\bm\varphi_0^*))$ and $\widetilde{d}_{t,\tau}^*= \int_{0}^1 \delta_{t,\tau}^*(s) ds$ with $\widetilde{\delta}_{t,\tau}^*(s)=I(y_t\leq\widetilde{q}_{t,\tau}(\bm\varphi_0^*)+\widetilde{\nu}_{t,\tau}(\bm u)s)-I(y_t\leq q_{t,\tau}(\bm\varphi_0^*)+\nu_{t,\tau}(\bm u)s)$. 
		Using $I(X\leq a)-I(X\leq b)=I(b\leq X \leq a)-I(b\geq X \geq a)$, it holds that
		\begin{align*}
			|\widetilde{c}_{t,\tau}^*| \leq & I\left(|y_t-q_{t,\tau}(\bm\varphi_0^*)| \leq |\widetilde{q}_{t,\tau}(\bm\varphi_0^*)-q_{t,\tau}(\bm\varphi_0^*)|\right) \quad \text{and} \\
			\sup_{\bm u\in\Lambda^*}|\widetilde{\delta}_{t,\tau}^*(s)| \leq & I\left(|y_t-q_{t,\tau}(\bm\varphi_0^*)-\nu_{t,\tau}(\bm u)s| \leq |\widetilde{q}_{t,\tau}(\bm\varphi_0^*)-q_{t,\tau}(\bm\varphi_0^*)|+\sup_{\Phi}|\widetilde{\nu}_{t,\tau}(\bm u)-\nu_{t,\tau}(\bm u)|s \right).
		\end{align*}
		Then by Taylor expansion, together with $\sup_{x}f_{t-1}(x)<\infty$ under Assumption \ref{assum-ConditionalDensity} and Lemma \ref{lem0-Tukey}, we have 
		\begin{align*}
			E\left(|\widetilde{c}_{t,\tau}^*| |\mathcal{F}_{t-1}\right) 
			\leq & 2\sup_{x}f_{t-1}(x) |\widetilde{q}_{t,\tau}(\bm\varphi_0^*)-q_{t,\tau}(\bm\varphi_0^*)| \leq C\rho^t \varsigma_{\rho} \quad \text{and} \\
			E\left(\sup_{\bm u\in\Lambda^*}|\widetilde{\delta}_{t,\tau}^*(s)| |\mathcal{F}_{t-1}\right) 
			\leq & 2\sup_{x}f_{t-1}(x) \left(|\widetilde{q}_{t,\tau}(\bm\varphi_0^*)-q_{t,\tau}(\bm\varphi_0^*)|+\sup_{\bm u\in\Lambda^*}|\widetilde{\nu}_{t,\tau}(\bm u)-\nu_{t,\tau}(\bm u)|\right) \\
			\leq & C\rho^t [\varsigma_{\rho} + \|\bm u\|(\varsigma_{\rho}+t\varsigma_{\rho}+\xi_{\rho})].
		\end{align*}
		These imply that
		\begin{align*}
			E\left(\sup_{\bm u\in\Lambda^*}|\widetilde{\xi}_{t,\tau}(\bm u)-\xi_{t,\tau}(\bm u)| |\mathcal{F}_{t-1}\right) \leq C\rho^t \varsigma_{\rho} + C\|\bm u\|\rho^t(\varsigma_{\rho}+t\varsigma_{\rho}+\xi_{\rho}).
		\end{align*}
		As a result, by iterative-expectation and Cauchy-Schwarz inequality, together with Assumption \ref{assum-Process-Mixing}, $\widetilde{\xi}_{t,\tau}(\bm u)-\xi_{t,\tau}(\bm u)=\widetilde{d}_{t,\tau}^*-\widetilde{c}_{t,\tau}^*$, $\nu_{t,\tau}(\bm u)=\bm u^{\prime}\dot{q}_{t,\tau}(\bm\varphi^{\dagger})$ by Taylor expansion, and $E(w_t)<\infty$ and $E(w_{t}|y_{t-j}|^3)<\infty$ for all $j\geq 1$ by Assumption \ref{assum-RandomWeight}, we have
		\begin{align*}
			E\sup_{\bm u\in\Lambda^*}\dfrac{|\widetilde{A}_{4n,k}^*(\bm u)|}{\sqrt{n}\|\bm u\|+n\|\bm u\|^2} 
			\leq & \sum_{t=1}^n E\left\{w_t E\left(\sup_{\bm u\in\Lambda^*}\dfrac{|\widetilde{\xi}_{t,\tau_k}(\bm u)-\xi_{t,\tau_k}(\bm u)|}{\sqrt{n}+n\|\bm u\|} |\mathcal{F}_{t-1}\right)\sup_{\bm u\in\Lambda^*}\dfrac{|\nu_{t,\tau_k}(\bm u)|}{\|\bm u\|} \right\} \\
			\leq & \dfrac{C}{\sqrt{n}}\sum_{t=1}^n \rho^t \cdot \left[E\left(w_t\sup_{\Phi}\|\dot{q}_{t,\tau_k}(\bm\varphi^{\dagger})\|^2\right)\right]^{1/2}\cdot \left[E(w_t\varsigma_{\rho}^2)\right]^{1/2} \\
			& + \dfrac{C}{n} \sum_{t=1}^n \rho^t \cdot \left[E\left(w_t\sup_{\Phi}\|\dot{q}_{t,\tau_k}(\bm\varphi^{\dagger})\|^2\right)\right]^{1/2}\cdot \left[E(w_t\varsigma_{\rho}^2)\right]^{1/2} \\
			& + \dfrac{C}{n} \sum_{t=1}^n t\rho^t \cdot \left[E\left(w_t\sup_{\Phi}\|\dot{q}_{t,\tau_k}(\bm\varphi^{\dagger})\|^2\right)\right]^{1/2}\cdot \left[E(w_t\varsigma_{\rho}^2)\right]^{1/2} \\
			& + \dfrac{C}{n} \sum_{t=1}^n \rho^t \cdot \left[E\left(w_t\sup_{\Phi}\|\dot{q}_{t,\tau_k}(\bm\varphi^{\dagger})\|^2\right)\right]^{1/2}\cdot \left[E(w_t\xi_{\rho}^2)\right]^{1/2} = o(1).
		\end{align*}
		Hence, it follows that
		\begin{align}\label{Pi4tilde-Tukey}
			\widetilde{A}_{4n,k}^*(\bm u_n)=o_p(\sqrt{n}\|\bm u_n\|+n\|\bm u_n\|^2).
		\end{align}
		Combining \eqref{initialHnrep-Tukey}--\eqref{Pi4tilde-Tukey}, we accomplish the proof of this lemma.
	\end{proof}

	\section{Additional Simulation Studies}
	
	\subsection{Unweighted QR estimator}
	
	In this experiment, we examine the robustness of the unweighted QR estimator when the data is heavy-tailed such that the condition $E|y_t|^3<\infty$ does not hold. Note that a simulation experiment is conducted in  Section 5.2 of the main paper to examine the performance of the self-weighted QR estimator $\widetilde{\bm\theta}_{wn}(\tau)$.
	For a direction comparison, we redo the experiment for the unweighted estimator $\widetilde{\bm\theta}_{n}(\tau)$ under the same settings.
	
	To determine whether the third-order moment of $y_t$ exists or not, we generate $\{y_t\}$ of length $10^5$ for Settings (5.2) and (5.3) with $F=F_N$ and $F=F_T$.
	By calculating the tail index of $y_t$ using Hill estimator \citep{Hill1975simple}, we conclude that possibly $E|y_t|^3=\infty$ if $F=F_T$ and $E|y_t|^3<\infty$ if $F=F_N$ for both settings; see Table \ref{tab:TailIndex} for the tail index of $\{y_t\}$.  To confirm this,  we further conduct tests for the null hypothesis that the $k$th moment of $y_t$ does not exist  \citep{Trapani2016testing} for $k=1,2$ and 3. From the $p$-values in  Table \ref{tab:TailIndex}, we confirm the above conclusion. 
	
	Tables \ref{tab.estimation.unweightedCQE.DGP1} and \ref{tab.estimation.unweightedCQE.DGP2} report the biases, empirical standard deviations (ESDs) and asymptotic standard deviations (ASDs) of the unweighted QR estimator $\widetilde{\bm\theta}_{n}(\tau)$ at quantile level $\tau=0.5\%,1\%$ or 5\%. From the results, we observe that the three main findings for the self-weighted estimator summarized in Section \ref{Sec-simulation-WCQE} of the main paper also hold true for the unweighted estimator. This indicates that the unweighted estimator is robust to heavy-tailed data without a finite third-order moment. Moreover, we can  compare Tables \ref{tab.estimation.unweightedCQE.DGP1} and \ref{tab.estimation.unweightedCQE.DGP2} to Tables \ref{tab.estimation.CQE.w2.DGP1} and \ref{tab.estimation.CQE.w2.DGP2} in the main paper, respectively. It can be observed that both estimators have similar performance for $F=F_N$, whereas the self-weighted estimator outperforms the unweighted estimator in terms of ESD and ASD for $F=F_T$. 
	Thus, when $E|y_t|^3=\infty$, although the unweighted estimator is still robust, it is less efficient than the self-weighted estimator.

	\subsection{Quantile rearrangement}\label{supp::Simulation rearrangement}
	
	To evaluate the effect of the quantile rearrangement method on prediction, we conduct a simulation experiment to compare the original quantile curve based on the pointwise quantile estimates $\{\widetilde{Q}_{\tau_k}(y_{t}|\mathcal{F}_{t-1})\}_{k=1}^K$ with the rearranged quantile curve based on the sorted quantile estimates denoted by $\{\widetilde{Q}_{\tau_k}^*(y_{t}|\mathcal{F}_{t-1})\}_{k=1}^K$.
	We consider the DGP in (5.1) with Settings (5.2) and (5.3) and $F$ being the standard normal distribution $F_N$ or Tukey-lambda distribution $F_T$, respectively. The sample size is set to $n=1000$ or 2000, and 1000 replications are generated for each sample size.	
	For evaluation, we use the $\ell_2$-loss to measure the prediction errors of $\{\widetilde{Q}_{\tau_k}(y_{t}|\mathcal{F}_{t-1})\}_{k=1}^K$ and $\{\widetilde{Q}_{\tau_k}^*(y_{t}|\mathcal{F}_{t-1})\}_{k=1}^K$. Specifically, the in-sample prediction error is defined as 
	\[ \left[\dfrac{1}{nMK}\sum_{m=1}^M\sum_{t=1}^n\sum_{k=1}^K \left|\widehat{Q}_{\tau_k}^{(m)}(y_{t}|\mathcal{F}_{t-1})-Q_{\tau_k}(y_{t}|\mathcal{F}_{t-1})\right|^2\right]^{1/2}, 
	\]
	and the out-of-sample prediction error is defined as
	\[ \left[\dfrac{1}{MK}\sum_{m=1}^M\sum_{k=1}^K \left|\widehat{Q}_{\tau_k}^{(m)}(y_{n+1}|\mathcal{F}_{t-1})-Q_{\tau_k}(y_{n+1}|\mathcal{F}_{t-1})\right|^2\right]^{1/2}, 
	\]
	where $\widehat{Q}_{\tau_k}^{(m)}(y_{t}|\mathcal{F}_{t-1})=\widetilde{Q}_{\tau_k}^{(m)}(y_{t}|\mathcal{F}_{t-1})$ or $\widetilde{Q}_{\tau_k}^{*(m)}(y_{t}|\mathcal{F}_{t-1})$ is the estimate in the $m$th replication, and $M=1000$ is the total number of replications.
	Table \ref{QtRearrangement} reports the prediction errors of estimated curves based on $\{\widetilde{Q}_{\tau_k}(y_{t}|\mathcal{F}_{t-1})\}_{k=1}^K$ and $\{\widetilde{Q}_{\tau_k}^*(y_{t}|\mathcal{F}_{t-1})\}_{k=1}^K$ for $\tau_k=0.7+0.005k$ with $k=1,\ldots,58$ (dense case) and $\tau_k=0.7+0.05k$ with $k=1,\ldots,5$ (sparse case).
	It can be observed that the rearranged quantile curve has no greater prediction errors than the original quantile curve in finite samples.

	\subsection{A Kolmogorov-Smirnov test and its finite-sample comparison with the CvM test}
	
	To test whether $\beta_1(\tau)$ is a constant or not, we can also construct the Kolmogorov-Smirnov (KS)-type test
	\[S_{n}^*=\sqrt{n}\sup_{\tau\in\mathcal{T}}|v_{n}(\tau)|,\] where $v_{n}(\tau)=R\widetilde{\bm\theta}_{wn}(\tau)-\widetilde{\beta}_1=R[\widetilde{\bm\theta}_{wn}(\tau)-\int_{\mathcal{T}}\widetilde{\bm\theta}_{wn}(\tau)d\tau]$. 
	Similar to Corollary \ref{thm-test1}, under the same regular conditions, we can show that $S_{n}^*\to_d S^*\equiv\sup_{\tau\in\mathcal{T}}|v_{0}(\tau)|$ as $n\to\infty$ under $H_0$, where $v_{0}(\tau)=R[\mathbb{G}(\tau)-\int_{\mathcal{T}}\mathbb{G}(\tau)d\tau]$ with $\mathbb{G}(\tau)$ defined in Theorem \ref{thm-WCQE-weak-convergence}. 
	Then the subsampling method in Section 3.2 of the main paper can be used to calculate the critical values of $S_{n}^*$ with $S_{k,b_n}$ replaced by $S_{k,b_n}^*=\sqrt{b_n}\sup_{\tau\in\mathcal{T}}|v_{k,b_n}(\tau)|$. 
	
	The same experiment is conducted using the same DGPs as in Section \ref{Sec-simulation-CvM}. To calculate $S_{n}$ in \eqref{CvM_constant} and $S_{n}^*$, we use a grid $\mathcal{T}_n$ with equal cell size $\delta_n=0.005$ in place of $\mathcal{T}$. 
	For the block size $b_n$ in subsampling, we consider $b_n=\lfloor cn^{1/2}\rfloor$ with $c=0.5, 1$ and 2; see also \cite{Shao2011bootstrap}.
	Tables \ref{tab_const_test_tau1} and \ref{tab_const_test_tau2} summarize the rejection rates of $S_n$ (the CvM test) and $S_{n}^*$ (the KS-type test) at $5\%$ significance level for $\mathcal{T}=[0.7,0.995]$ and $[0.8,0.995]$, respectively. It can be seen that the KS-type test has lower power than the CvM test for tail quantile intervals in finite samples. As a result, we recommend using the CvM test for testing $H_0:\ \forall\tau\in\mathcal{T}, \; R\bm\theta(\tau)=\beta_1$ in our model setting.
	
	\section{Additional results for the empirical analysis}\label{Sec-rearrangement}
	
	We have also re-calculated both the backtesting and empirical coverage results for the proposed self-weighted QR method after conducting the quantile rearrangement in \cite{Chernozhukov2010}. 
	Tables \ref{tab:CommonQuantiles}--\ref{tab:ExtremeQuantiles} report the results for the self-weighted QR method before and after quantile rearrangement.  
	We find that the $p$-values of the backtests, the empirical coverage rates and prediction errors do not change at all after the percentage points are rounded down to two decimal places for lower and upper $1\%, 2.5\%, 5\%$ conditional quantiles. 
	Moreover, the empirical coverage rates and prediction errors change very little for lower and upper $0.1\%, 0.25\%, 0.5\%$ conditional quantiles.
	As a result, almost the same results can be observed for Tables \ref{tabForecasting1}--\ref{tabForecasting2}. 
	
	Moreover, we also employ the monotone rearrangement method to ensure the monotonicity of estimated curves for $\omega(\cdot)$ and $\alpha_1(\cdot)$, respectively.
	Specifically, we sort the pointwise estimates $\{\widetilde{\omega}_{wn}(\tau_k)\}_{k=1}^K$ and $\{\widetilde{\alpha}_{1wn}(\tau_k)\}_{k=1}^K$ in Figure \ref{coef_comparison_ARCHinfty} respectively in increasing order to enforce the monotonicity. Moreover, the pointwise confidence intervals of $\omega(\cdot)$ and $\alpha_1(\cdot)$ can be rearranged accordingly by sorting the upper and lower endpoint functions. 
	Figure \ref{coef_comparison_ARCHinfty_rearrangement} illustrates the original curve estimates together with their 95\% confidence intervals, and the rearranged curve estimates together with rearranged confidence intervals for $\omega(\cdot)$ and $\alpha_1(\cdot)$. 
	It can be seen that the rearranged curves and confidence intervals for $\omega(\cdot)$ and $\alpha_1(\cdot)$ are monotonic and more smooth than the original estimated curves, and the rearranged confidence interval is shorter in length than the original interval.

	\begin{table}[htp]
		\centering
		\caption{\label{tab:TailIndex} Hill's estimator of the  tail index and $p$-values of \citeauthor{Trapani2016testing}'s tests with $k=1,2$ and 3 for $\{y_t\}$ generated from Settings (5.2) and (5.3) with $F=F_N$ and $F=F_T$.}
		\begin{tabular}{lcccccc}
			\hline
			&&\multicolumn{2}{c}{Setting (5.2)} && \multicolumn{2}{c}{Setting (5.3)}\\
			\cline{3-4}\cline{6-7}
			&& $F_N$ & $F_T$ && $F_N$ & $F_T$\\
			\hline
			Tail index && 8.449 & 1.971 && 5.078 & 2.016\\
			Trapani's test(3)&& $<0.01$ & $0.943$ && $<0.01$ & $0.992$\\
			Trapani's test(2)&& $<0.01$ & $0.082$ && $<0.01$ & $0.929$\\
			Trapani's test(1)&& $<0.01$ & $<0.01$ && $<0.01$ & $<0.01$\\
			\hline
		\end{tabular}
	\end{table}
	
	\begin{table}[htp]
		\caption{Biases, ESDs and ASDs of the unweighted QR estimator $\widetilde{\bm\theta}_{n}(\tau)$ at quantile level $\tau=0.5\%,1\%$ or $5\%$ for DGP (5.1) with Setting (5.2). ASD$_{1}$ and ASD$_{2}$ correspond to the bandwidths $\ell_{B}$ and $\ell_{HS}$, respectively. $F$ is the standard normal distribution $F_N$ or Tukey-lambda distribution $F_T$.}
		\label{tab.estimation.unweightedCQE.DGP1}
		\centering
		\begin{tabular}{lrrrrrrrrrrrrrr}
			\hline
			&&   && \multicolumn{5}{c}{$F=F_{N}$} && \multicolumn{5}{c}{$F=F_{T}$}\\
			\cline{5-9}\cline{11-15}
			&& $n$ &&True & Bias & ESD & ASD$_{\text{1}}$ & ASD$_{\text{2}}$&& True & Bias & ESD & ASD$_{\text{1}}$ & ASD$_{\text{2}}$\\ 
			\hline
			&&&&\multicolumn{11}{c}{$\tau=0.5\%$}\\
			$\omega$&& 1000&&-0.258 & -0.011 & 0.076 & 0.088 & 0.056 && -0.942 & -0.435 & 1.191 & 1.865 & 1.169 \\ 
			&& 2000&&-0.258 & -0.005 & 0.061 & 0.065 & 0.046 && -0.942 & -0.350 & 0.959 & 1.325 & 0.869 \\ 
			$\alpha_1$&& 1000&&-0.258 & -0.018 & 0.157 & 0.194 & 0.120 && -0.942 & -0.038 & 0.586 & 0.845 & 0.505 \\ 
			&& 2000&&-0.258 & -0.020 & 0.126 & 0.139 & 0.098 && -0.942 & -0.037 & 0.510 & 0.623 & 0.415 \\ 
			$\beta_1$&& 1000&&0.800 & -0.057 & 0.152 & 0.896 & 0.460 && 0.800 & -0.025 & 0.110 & 0.177 & 0.099 \\ 
			&& 2000&&0.800 & -0.045 & 0.132 & 0.197 & 0.164 && 0.800 & -0.019 & 0.096 & 0.128 & 0.079 \\ 
			\hline
			&&&&\multicolumn{11}{c}{$\tau=1\%$}\\
			$\omega$&& 1000&&-0.233 & -0.011 & 0.064 & 0.070 & 0.056 && -0.755 & -0.411 & 0.890 & 1.082 & 0.759 \\ 
			&& 2000&&-0.233 & -0.005 & 0.050 & 0.053 & 0.039 && -0.755 & -0.262 & 0.658 & 0.809 & 0.557 \\ 
			$\alpha_1$&& 1000&&-0.233 & -0.014 & 0.135 & 0.156 & 0.123 && -0.755 & -0.029 & 0.427 & 0.481 & 0.339 \\ 
			&& 2000&&-0.233 & -0.013 & 0.097 & 0.110 & 0.082 && -0.755 & -0.031 & 0.373 & 0.370 & 0.280 \\
			$\beta_1$&& 1000&&0.800 & -0.055 & 0.141 & 0.371 & 0.303 && 0.800 & -0.024 & 0.105 & 0.127 & 0.085 \\ 
			&& 2000&&0.800 & -0.040 & 0.127 & 0.201 & 0.135 && 0.800 & -0.016 & 0.085 & 0.092 & 0.067 \\ 
			\hline
			&&&&\multicolumn{11}{c}{$\tau=5\%$}\\
			$\omega$&& 1000&&-0.164 & -0.009 & 0.039 & 0.041 & 0.036 && -0.405 & -0.197 & 0.360 & 0.384 & 0.299 \\ 
			&& 2000&&-0.164 & -0.005 & 0.030 & 0.031 & 0.027 && -0.405 & -0.123 & 0.256 & 0.291 & 0.226 \\
			$\alpha_1$&& 1000&&-0.164 & -0.011 & 0.083 & 0.089 & 0.078 && -0.405 & -0.007 & 0.156 & 0.170 & 0.135 \\ 
			&& 2000&&-0.164 & -0.006 & 0.058 & 0.062 & 0.056 && -0.405 & -0.008 & 0.140 & 0.130 & 0.108 \\ 
			$\beta_1$&& 1000&&0.800 & -0.061 & 0.153 & 0.272 & 0.242 && 0.800 & -0.018 & 0.076 & 0.082 & 0.065 \\ 
			&& 2000&&0.800 & -0.035 & 0.113 & 0.107 & 0.096 && 0.800 & -0.011 & 0.060 & 0.060 & 0.049 \\ 
			\hline
		\end{tabular}
	\end{table}

	\begin{table}[htp]
		\caption{Biases, ESDs and ASDs of the unweighted QR estimator $\widetilde{\bm\theta}_{n}(\tau)$ at quantile level $\tau=0.5\%,1\%$ or $5\%$ for DGP (5.1) with Setting (5.3). ASD$_{1}$ and ASD$_{2}$ correspond to the bandwidths $\ell_{B}$ and $\ell_{HS}$, respectively. $F$ is the standard normal distribution $F_N$ or Tukey-lambda distribution $F_T$.}
		\label{tab.estimation.unweightedCQE.DGP2}
		\centering
		\begin{tabular}{lrrrrrrrrrrrrrr}
			\hline
			&&   && \multicolumn{5}{c}{$F=F_{N}$} && \multicolumn{5}{c}{$F=F_{T}$}\\
			\cline{5-9}\cline{11-15}
			&& $n$ &&True & Bias & ESD & ASD$_{\text{1}}$ & ASD$_{\text{2}}$&& True & Bias & ESD & ASD$_{\text{1}}$ & ASD$_{\text{2}}$\\ 
			\hline
			&&&&\multicolumn{11}{c}{$\tau=0.5\%$}\\
			$\omega$&& 1000&&-0.258 & -0.022 & 0.059 & 0.056 & 0.035 && -0.942 & -0.315 & 0.663 & 0.694 & 0.422 \\
			&& 2000&&-0.258 & -0.011 & 0.041 & 0.036 & 0.026 && -0.942 & -0.232 & 0.503 & 0.536 & 0.325 \\ 
			$\alpha_1$&& 1000&&-0.753 & 0.005 & 0.120 & 0.121 & 0.078 && -1.437 & 0.034 & 0.534 & 0.626 & 0.446 \\
			&& 2000&&-0.753 & 0.003 & 0.088 & 0.087 & 0.062 && -1.437 & 0.019 & 0.462 & 0.519 & 0.361 \\ 
			$\beta_1$&& 1000&&0.597 & -0.026 & 0.080 & 0.073 & 0.047 && 0.597 & -0.026 & 0.119 & 0.161 & 0.103 \\
			&& 2000&&0.597 & -0.013 & 0.055 & 0.049 & 0.037 && 0.597 & -0.016 & 0.097 & 0.123 & 0.082 \\ 
			\hline
			&&&&\multicolumn{11}{c}{$\tau=1\%$}\\
			$\omega$&& 1000&&-0.233 & -0.017 & 0.047 & 0.043 & 0.033 && -0.755 & -0.235 & 0.471 & 0.462 & 0.287 \\
			&& 2000&&-0.233 & -0.008 & 0.033 & 0.031 & 0.025 && -0.755 & -0.161 & 0.342 & 0.354 & 0.227 \\ 
			$\alpha_1$&& 1000&&-0.723 & 0.000 & 0.107 & 0.109 & 0.080 && -1.245 & 0.003 & 0.428 & 0.466 & 0.317 \\ 
			&& 2000&&-0.723 & 0.001 & 0.080 & 0.081 & 0.062 && -1.245 & 0.007 & 0.336 & 0.360 & 0.269 \\
			$\beta_1$&& 1000&&0.594 & -0.023 & 0.073 & 0.067 & 0.051 && 0.594 & -0.027 & 0.113 & 0.133 & 0.086 \\
			&& 2000&&0.594 & -0.010 & 0.050 & 0.048 & 0.037 && 0.594 & -0.014 & 0.088 & 0.097 & 0.070 \\
			\hline
			&&&&\multicolumn{11}{c}{$\tau=5\%$}\\
			$\omega$&& 1000&&-0.164 & -0.009 & 0.034 & 0.035 & 0.030 && -0.405 & -0.102 & 0.201 & 0.213 & 0.160 \\
			&& 2000&&-0.164 & -0.003 & 0.023 & 0.025 & 0.023 && -0.405 & -0.064 & 0.155 & 0.164 & 0.127 \\
			$\alpha_1$&& 1000&&-0.614 & 0.002 & 0.095 & 0.100 & 0.089 && -0.855 & 0.009 & 0.207 & 0.221 & 0.179 \\
			&& 2000&&-0.614 & -0.002 & 0.067 & 0.071 & 0.065 && -0.855 & 0.004 & 0.159 & 0.170 & 0.144 \\
			$\beta_1$&& 1000&&0.570 & -0.016 & 0.074 & 0.077 & 0.067 && 0.570 & -0.019 & 0.097 & 0.101 & 0.081 \\ 
			&& 2000&&0.570 & -0.007 & 0.052 & 0.053 & 0.049 && 0.570 & -0.012 & 0.072 & 0.076 & 0.063 \\ 
			\hline
		\end{tabular}
	\end{table}  
	
	\begin{table}
		\caption{\label{QtRearrangement} Prediction errors before and after rearrangement. `In' and `Out' represent in-sample and out-of-sample prediction errors, respectively. $F$ is the standard normal distribution $F_N$ or Tukey-lambda distribution $F_T$. }
		\centering
		\begin{tabular}{llcccccccccccc}
			\hline
			& && \multicolumn{5}{c}{Setting (5.2)} && \multicolumn{5}{c}{Setting (5.3)}\\
			\cline{4-8} \cline{10-14}
			& && \multicolumn{2}{c}{Before} && \multicolumn{2}{c}{After} &&  \multicolumn{2}{c}{Before} && \multicolumn{2}{c}{After} \\
			\cline{4-14}
			$n$& && $F_N$ & $F_T$ &&  $F_N$ & $F_T$  &&  $F_N$ & $F_T$ &&  $F_N$ & $F_T$\\
			\hline
			&&&\multicolumn{11}{c}{Dense case}\\
			1000 & In && 0.016 & 6.017 && 0.016 & 5.996 && 0.035 & 6.170 && 0.035 & 6.156 \\
			& Out && 0.016 & 1.699 && 0.015 & 1.618 && 0.039 & 1.810 && 0.039 & 1.801 \\ 
			2000 & In && 0.012 & 1.786 && 0.012 & 1.779 && 0.025 & 1.859 && 0.025 & 1.836 \\ 
			& Out && 0.012 & 1.619 && 0.012 & 1.614 && 0.025 & 0.596 && 0.025 & 0.594 \\ 
			&&&\multicolumn{11}{c}{Sparse case}\\
			1000 & In && 0.015 & 2.790 && 0.014 & 2.785 && 0.034 & 5.283 && 0.034 & 5.283 \\ 
			& Out && 0.014 & 1.046 && 0.014 & 1.046 && 0.038 & 1.415 && 0.038 & 1.415 \\ 
			2000 & In && 0.011 & 1.460 && 0.011 & 1.460 && 0.024 & 1.276 && 0.024 & 1.276 \\ 
			& Out && 0.011 & 1.280 && 0.011 & 1.280 && 0.024 & 0.516 && 0.024 & 0.516 \\ 
			\hline
		\end{tabular}				
	\end{table}
	
	\begin{table}
		\caption{\label{tab_const_test_tau1} Rejection rates of the CvM and KS-type tests at the 5\% significance level for $\mathcal{T}=[0.7,0.995]$, where $b_1$, $b_2$ and $b_3$ correspond to $\lfloor cn^{1/2}\rfloor$ with $c=0.5, 1$ and 2, respectively. $F$ is the standard normal distribution $F_N$ or Tukey-lambda distribution $F_T$.}
		\centering
		\begin{tabular}{ccrcccccccc}
			\hline
			&&&& \multicolumn{3}{c}{$F=F_N$} & & \multicolumn{3}{c}{$F=F_T$}\\
			\cline{5-7}\cline{9-11}
			&$N$&$d$&& $b_1$ & $b_2$ & $b_3$& & $b_1$ & $b_2$ & $b_3$ \\ 
			\hline
			&&   0 && 0.045 & 0.058 & 0.084 &  & 0.028 & 0.041 & 0.056 \\
			&1000&   1 && 0.101 & 0.116 & 0.140 &  & 0.098 & 0.122 & 0.167 \\ 
			CvM&&   1.6 && 0.236 & 0.278 & 0.332 &  & 0.259 & 0.325 & 0.403 \\
			&&   0 && 0.047 & 0.055 & 0.064 &  & 0.034 & 0.049 & 0.062 \\
			&2000&   1 && 0.190 & 0.214 & 0.236 &  & 0.240 & 0.274 & 0.334 \\
			&&   1.6 && 0.571 & 0.597 & 0.656 &  & 0.689 & 0.739 & 0.784 \\ 
			\hline
			&&   0 && 0.035 & 0.037 & 0.059 &  & 0.033 & 0.035 & 0.043 \\
			&1000&   1 && 0.052 & 0.075 & 0.100 &  & 0.058 & 0.057 & 0.088 \\
			KS&&   1.6 && 0.143 & 0.182 & 0.244 &  & 0.125 & 0.154 & 0.213 \\
			&&   0 && 0.034 & 0.038 & 0.052 &  & 0.034 & 0.046 & 0.061 \\
			&2000&   1 && 0.141 & 0.158 & 0.179 &  & 0.111 & 0.132 & 0.156 \\ 
			&&   1.6 && 0.465 & 0.503 & 0.558 &  & 0.444 & 0.474 & 0.530 \\ 
			\hline
		\end{tabular}
	\end{table}
	
	\begin{table}
		\caption{\label{tab_const_test_tau2} Rejection rates of the CvM and KS-type tests at the 5\% significance level for $\mathcal{T}=[0.8,0.995]$, where $b_1$, $b_2$ and $b_3$ correspond to $\lfloor cn^{1/2}\rfloor$ with $c=0.5, 1$ and 2, respectively. $F$ is the standard normal distribution $F_N$ or Tukey-lambda distribution $F_T$.}
		\centering
		\begin{tabular}{ccrcccccccc}
			\hline
			&&&& \multicolumn{3}{c}{$F=F_N$} & & \multicolumn{3}{c}{$F=F_T$}\\
			\cline{5-7}\cline{9-11}
			&$N$&$d$&& $b_1$ & $b_2$ & $b_3$& & $b_1$ & $b_2$ & $b_3$ \\ 
			\hline
			&&   0 &&  0.036 & 0.046 & 0.071 &  & 0.026 & 0.040 & 0.054 \\ 
			&1000&   1 && 0.071 & 0.093 & 0.124 &  & 0.061 & 0.073 & 0.121 \\ 
			CvM&&   1.6 && 0.169 & 0.223 & 0.284 &  & 0.147 & 0.191 & 0.274 \\
			&&   0 && 0.028 & 0.037 & 0.056 &  & 0.027 & 0.038 & 0.055 \\  
			&2000&   1 && 0.143 & 0.161 & 0.202 &  & 0.131 & 0.173 & 0.213 \\
			&&   1.6 && 0.481 & 0.558 & 0.601 &  & 0.473 & 0.530 & 0.610 \\
			\hline
			&&   0 && 0.030 & 0.036 & 0.061 &  & 0.035 & 0.037 & 0.045 \\
			&1000&   1 && 0.040 & 0.059 & 0.095 &  & 0.041 & 0.045 & 0.075 \\
			KS&&   1.6 && 0.103 & 0.137 & 0.196 &  & 0.082 & 0.083 & 0.125 \\
			&&   0 && 0.034 & 0.047 & 0.067 &  & 0.037 & 0.045 & 0.064 \\
			&2000&   1 && 0.097 & 0.119 & 0.150 &  & 0.074 & 0.093 & 0.110 \\
			&&   1.6 && 0.360 & 0.410 & 0.491 &  & 0.262 & 0.278 & 0.336 \\ 
			\hline
		\end{tabular}
	\end{table}
	
	\begin{table}
		\centering
		\caption{Empirical coverage rates (ECRs) in percentage, prediction errors (PEs) and $p$-values for correct conditional coverage (CC) and the dynamic quantile (DQ) tests for the proposed QR method at lower and upper $1\%, 2.5\%, 5\%$ conditional quantiles.}
		\label{tab:CommonQuantiles}
		\begin{tabular}{rrrrrrrrrrr}
			\hline
			&& \multicolumn{4}{c}{Before} && \multicolumn{4}{c}{After}\\
			\cline{3-6}\cline{8-11}
			$\tau$ && ECR & PE & CC & DQ &&  ECR & PE & CC & DQ \\
			\hline
			$1\%$ &&  1.26 & 0.65 & 0.74 & 0.96 && 1.26 & 0.65 & 0.74 & 0.96 \\
			$2.5\%$ &&  2.98 & 0.78 & 0.42 & 0.75 && 2.98 & 0.78 & 0.42 & 0.75 \\
			$5\%$ &&  6.12 & 1.30 & 0.42 & 0.01 &&  6.12 & 1.30 & 0.42 & 0.01 \\
			$95\%$ &&  94.51 & 0.57 & 0.11 & 0.62 && 94.51 & 0.57 & 0.11 & 0.62 \\
			$97.5\%$ &&  97.65 & 0.23 & 0.68 & 0.68 && 97.65 & 0.23 & 0.68 & 0.68 \\
			$99\%$ &&  99.06 & 0.15 & 0.93 & 1.00 && 99.06 & 0.15 & 0.93 & 1.00 \\
			\hline
		\end{tabular}
	\end{table}
	
	\begin{table}
		\centering
		\caption{Empirical coverage rates (ECRs) in percentage and prediction errors (PEs) for the proposed QR method at lower and upper $0.1\%, 0.25\%, 0.5\%$ conditional quantiles.}
		\label{tab:ExtremeQuantiles}
		\begin{tabular}{rcrrrrr}
			\hline
			&& \multicolumn{2}{c}{Before} && \multicolumn{2}{c}{After}\\
			\cline{3-4}\cline{6-7}
			$\tau$ && ECR & PE  &&  ECR & PE \\
			\hline
			$0.1\%$ &&  0.27 & 3.50 && 0.22 & 2.50 \\
			$0.25\%$ &&  0.55 & 3.80 && 0.52 & 3.48 \\
			$0.5\%$ &&  0.90 & 3.59 && 0.98 & 4.26 \\
			$99.5\%$ &&  99.42 & 0.67 && 99.42 & 0.67 \\
			$99.75\%$ &&  99.60 & 1.90 && 99.58 & 2.22 \\
			$99.9\%$ &&  99.78 & 2.50 && 99.80 & 2.00 \\
			\hline
		\end{tabular}
	\end{table}
	
	\begin{figure}[htp]
		\centering
		\includegraphics[width=6in]{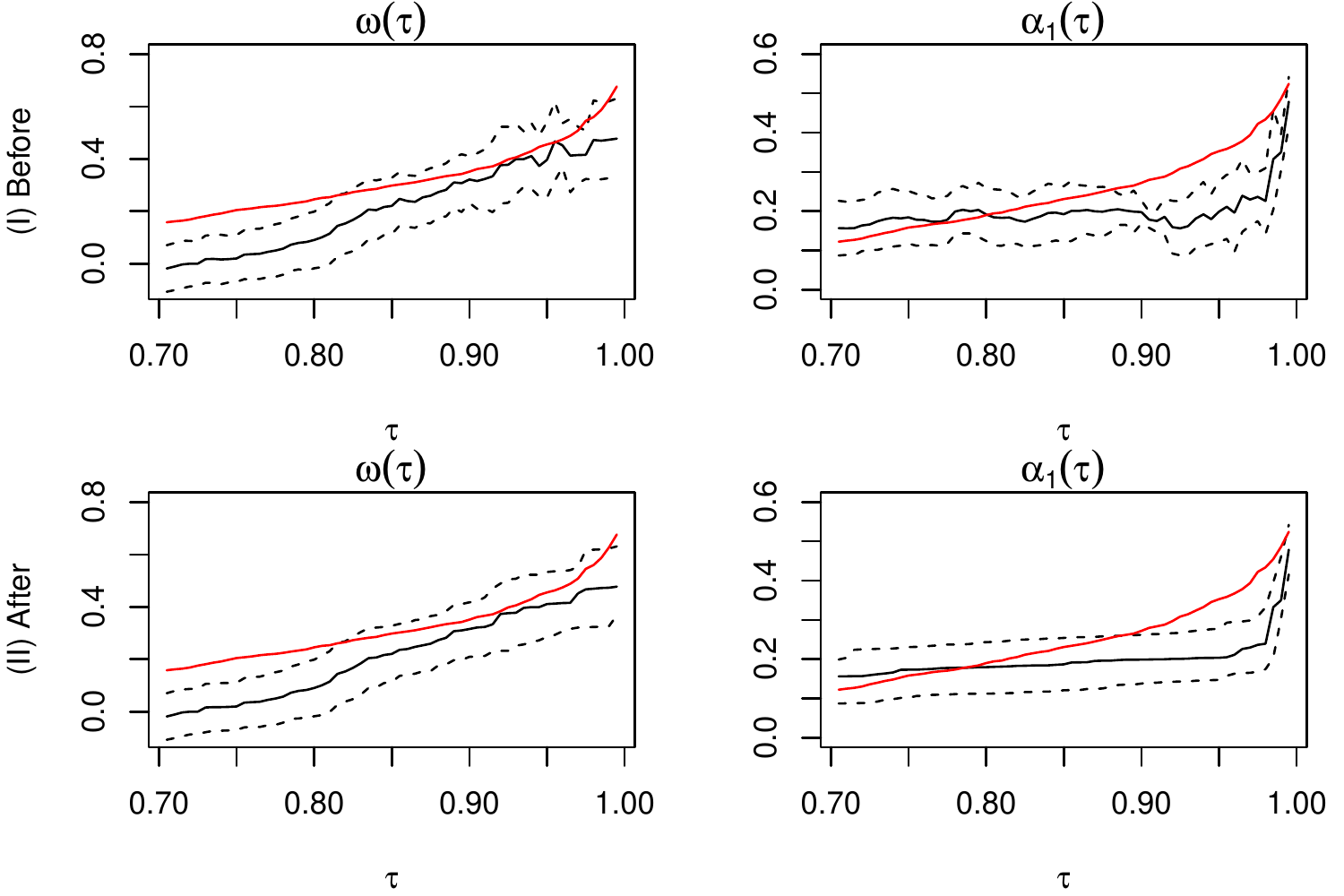}
		\caption{\label{coef_comparison_ARCHinfty_rearrangement} (I): Self-weighted QR estimates (black solid) of $\omega(\cdot)$ and $\alpha_1(\cdot)$ and 95\% confidence intervals (black dotted) for $\tau_k=k/200$ with $140\leq k \leq 199$. (II): rearranged curves (black solid) and rearranged confidence intervals (black dotted). The estimates of $\bm\theta_{\tau}=(a_0Q_{\tau}(\varepsilon_t)/(1-b_1),a_1Q_{\tau}(\varepsilon_t),b_1)$ (red solid) for the linear ARCH($\infty$) model in \eqref{lgarch11} using the FHS method are also provided for comparison.}
	\end{figure}
	
\end{document}